\documentclass[11pt,a4paper]{memoir}
\usepackage[utf8]{inputenc}
\usepackage[british]{babel}
\makeatletter\AtBeginDocument{\let\@elt\relax}\makeatother 
\usepackage[left=80pt, right=80pt, top=60pt]{geometry}
\usepackage[affil-it]{authblk} 
\usepackage{graphics}
\usepackage{color}
\usepackage[bookmarksopen,bookmarksdepth=2]{hyperref}
\usepackage{amsfonts,amsmath,amssymb}
\usepackage{enumitem} 
\usepackage{xfrac,slashed,bbm,mathtools,ulem,epigraph}	
\usepackage{savesym}			
\savesymbol{intertext}
\restoresymbol{TT}{intertext}	
\usepackage{upgreek}
\usepackage{threeparttable}
\usepackage{tikz}	
\usepackage{pict2e}	
\usepackage{eso-pic} 
\usepackage{caption}
\usepackage{subcaption}
\usepackage{tikz-feynman} 

\providecommand{\noopsort}[1]{} 

\newcommand{\N}{\mathbb N}
\newcommand{\Z}{\mathbb Z}
\newcommand{\R}{\mathbb R}
\newcommand{\C}{\mathbb C}
\newcommand{\T}{\mathbb T}
\renewcommand{\H}{\mathcal H}
\newcommand{\D}{\mathcal{D}}
\renewcommand{\S}{\mathcal{S}}

\newcommand{\A}{\mathcal{A}}
\renewcommand{\L}{\mathcal{L}}
\newcommand{\M}{\mathcal{M}}

\newcommand{\E}{\mathcal{E}}

\newcommand{\mR}{\mathcal{R}}
\newcommand{\Cr}{C_\mR}
\newcommand{\Sr}{\S_\mR}
\newcommand{\mB}{\mathcal{B}}
\newcommand{\B}{\mathcal{B}}
\newcommand{\W}{\mathcal{W}}
\newcommand{\Wr}{\mathcal{W}^0_\mR}

\newcommand{\mW}{\mathfrak{W}}
\newcommand{\mN}{\mathcal{N}}
\newcommand{\mK}{\mathcal{K}}
\newcommand{\Wsn}{\W_s^{n}}

\newcommand{\Esg}{\mathcal{E}_s^{\gamma}}

\renewcommand{\O}{\mathcal{O}}

\def\Re{{\mathrm{Re}}}
\def\Im{{\mathrm{Im}}}
\newcommand{\supnorm}[1]{\norm{ #1 }_\infty}
\newcommand{\absnorm}[1]{\norm{ #1 }_1}
\newcommand{\quadnorm}[1]{\norm{#1}_2}

\newcommand{\norm}[1]{\left\| #1 \right\|}
\newcommand{\nrm}[2]{\left\|#1\right\|_{#2}}
\newcommand{\tnrm}[2]{\|#1\|_{#2}}

\newcommand{\bra}[1]{\left\langle #1 \right|}
\newcommand{\ket}[1]{\left| #1 \right\rangle}
\newcommand{\braket}[2]{\left\langle #1 \middle| #2 \right\rangle}
\newcommand{\vect}[2]{\begin{pmatrix}#1\\#2\end{pmatrix}}
\newcommand{\p}[2]{\langle #1 , #2 \rangle}

\newcommand{\spn}{\textnormal{span}}
\DeclareMathOperator*{\sgn}{sgn}
\DeclareMathOperator*{\dom}{dom}

\newcommand{\Q}{Q_\hbar^\infty}
\newcommand{\Ql}{Q_\hbar^l}
\newcommand{\Qm}{Q_\hbar^m}
\newcommand{\Qn}{Q_\hbar^n}
\newcommand{\Qp}{Q_\hbar^p}

\renewcommand{\mid}{~\middle|~}

\newcommand{\1}{\mathbbm{1}}

\newcommand{\supp}{\operatorname{supp}}

\newcommand{\Tfn}{T_{f^{[n]}}}
\newcommand{\Tr}{\operatorname{Tr}}

\def\cs{\mathrm{cs}}
\newcommand{\sa}{\textnormal{sa}}
\newcommand{\br}[1]{\langle#1\rangle}
\newcommand{\bbr}[1]{\llangle#1\rrangle^{1L}_N}
\newcommand{\brr}[1]{\boldsymbol\prec\!#1\!\boldsymbol\succ}
\newcommand{\ev}{\textup{ev}}
\newcommand{\odd}{\textup{odd}}
\newcommand{\dt}{\Delta t}
\newcommand{\set}[2]{\{#1 \colon #2 \}}
\newcommand{\hatotimes}{\mathbin{\hat{\otimes}}}
\newcommand{\indicator}{1}
\newcommand{\unit}{\1}
\renewcommand{\th}{^{\text{th}}}
\newcommand{\image}[2]{\includegraphics[width=#1\textwidth, height=#1\textheight,keepaspectratio]{#2}}
\newcommand{\upcite}[1]{$^{\textnormal{\cite{#1}}}$}

\newcommand{\Weylalg}[1]{\mathcal{W}^0(#1)}
\newcommand{\Teunalg}[1]{\mathcal{W}^0_{\mathcal{R}}(#1)}

\usepackage{feynmp-auto}
\usepackage{simpler-wick}
\usepackage{soul} 

\makeatletter
\newsavebox{\@brx}
\newcommand{\llangle}[1][]{\savebox{\@brx}{\(\m@th{#1\langle}\)}%
  \mathopen{\copy\@brx\kern-0.5\wd\@brx\usebox{\@brx}}}
\newcommand{\rrangle}[1][]{\savebox{\@brx}{\(\m@th{#1\rangle}\)}%
  \mathclose{\copy\@brx\kern-0.5\wd\@brx\usebox{\@brx}}}
\makeatother

\tikzset{edge/.style={decorate, decoration=snake}}
\tikzset{puntjes/.style={dash, pattern=on 0pt off 2\pgflinewidth}}
\tikzset{streepjes/.style={dash pattern=on 8pt off 4pt}}
\newcommand{\ncvertex}[1]{\filldraw (#1) circle (12.5pt);
	{\color{white}\filldraw (#1) circle (10.2pt);}
	\node at (#1) {\huge $f$};}

\newcommand{\QW}{\mathcal{Q}_\hbar^\textnormal{W}}
\newcommand{\g}{\mathfrak{g}}
\newcommand{\add}{\textnormal{add}}
\newcommand{\sub}{\textnormal{sub}}
\newcommand{\conf}{\textnormal{conf}}
\newcommand{\mom}{\textnormal{mom}}
\newcommand{\vol}{\textnormal{vol}}
\newcommand{\id}{\textnormal{id}}
\newcommand{\uvp}{\underline{\varphi}}

\usepackage{amsthm}
\newtheorem{thm}{Theorem}[section]
\newtheorem{lem}[thm]{Lemma}
\newtheorem{prop}[thm]{Proposition}
\newtheorem{cor}[thm]{Corollary}
\newtheorem{defi}[thm]{Definition}

\newtheorem{exam}[thm]{Example}
\newtheorem{rema}[thm]{Remark}

\newcommand{\ncvertexcover}[1]{\filldraw (#1) circle (12.5pt);
	{\color{teal}\filldraw (#1) circle (10.2pt);}
	\node at (#1) {\scalebox{0.8}{\huge $f$}};}
	
\usepackage{array}
\newcolumntype{?}[1]{!{\vrule width #1}} 
\usepackage{color}
\usepackage{xcolor}
\definecolor{wrongultramarine}{rgb}{0.07, 0.04, 0.36}
\usepackage{background}
\usepackage{blindtext}
\usepackage{ marvosym } 

\begin{document}
%
\newgeometry{left=0cm,bottom=0cm,top=-10cm,right=0cm}
\begin{titlingpage}
\backgroundsetup{
scale=4,
angle=90,
opacity=1,
contents={\makebox[0pt]{\begin{tikzpicture}[remember picture,overlay]
\path [left color = black, right color = teal] (current page.south west)rectangle (current page.north east);
\end{tikzpicture}}}
}
\hspace{0pt}
\vspace{220pt}
\begin{center}
\makebox[26.6pt][l]{~}
\makebox[0pt]{
\scalebox{2.15}{
\begin{tikzpicture}[thick,baseline={(100,0)}]
\color{white}

	
	\draw[edge] (0,12.8) to (0.5,12);
	\draw[edge] (0.5,12) to (2,13);
	\draw[edge] (1,14) to (2,13);
	\draw[edge] (2,13) to (3.5,14);
	\draw[edge] (2,13) to (4,12);
	\draw[edge] (4,12) to (3.5,14);
	\draw[edge] (3.5,14) to (5.5,13.5);
	\draw[edge] (5.5,13.5) to (6.5,14.5);
	\draw[edge] (2,13) to (2,10.5);
	\draw[edge] (2,10.5) to (4,12);
	\draw[edge] (4,12) to (6,11);
	\draw[edge] (4,12) to (5.5,13.5);
	\draw[edge] (5.5,13.5) to (6,11);
	\draw[edge] (5.5,13.5) to (7.5,12.5);
	\draw[edge] (6,11) to (7.5,12.5);
	\draw[edge] (7.5,12.5) to (7.5,14);
	\draw[edge] (7.5,12.5) to (9.5,12);
	\draw[edge] (9.5,12) to (10,14);
	\draw[edge] (9.5,12) to (10.5,10);
	
	\ncvertexcover{0.5,12}
	\ncvertexcover{2,13}
	\ncvertexcover{4,12}
	\ncvertexcover{2,10.5}
	\ncvertexcover{6,11}
	\ncvertexcover{5.5,13.5}
	\ncvertexcover{7.5,12.5}
	\ncvertexcover{9.5,12}
	\ncvertexcover{3.5,14}
	
	\node at (5,8.7) {\huge C*-algebraic Results};
	\node at (5,7.7) {\large in the Search for};
	\node at (5,6.7) {\huge Quantum Gauge Fields};
	\node at (5,4.1) {\large Teun van Nuland};
	
	\draw (0,1) to (2,1);
	\draw (2,0) to (2,1);
	\draw (2,1) to (4,1);
	\draw (4,1) to (6,1);
	\draw (4,0) to (4,1);
	\draw (6,1) to (8,1);
	\draw (6,0) to (6,1);
	\draw (8,0) to (8,1);
	\draw (8,1) to (10,1);
	\filldraw (2,1) circle (3pt);
	\filldraw (4,1) circle (3pt);
	\filldraw (6,1) circle (3pt);
	\filldraw (8,1) circle (3pt);
	\filldraw (2,3) circle (3pt);
	\filldraw (4,3) circle (3pt);
	\filldraw (6,3) circle (3pt);
	\filldraw (8,3) circle (3pt);
	\filldraw (2,5) circle (3pt);
	\filldraw (8,5) circle (3pt);
	
	\draw (0,3) to (2,3);
	\draw (2,1) to (2,3);
	\draw (2,3) to (4,3);
	\draw (4,3) to (6,3);
	\draw (4,1) to (4,3);
	\draw (6,3) to (8,3);
	\draw (6,1) to (6,3);
	\draw (8,1) to (8,3);
	\draw (8,3) to (10,3);
	
	\draw (2,3) to (2,5);
	\draw (8,3) to (8,5);
	\draw (0,5) to (2,5);
\end{tikzpicture}
}
}
\end{center}
\newpage
\backgroundsetup{opacity=0.5, contents={}}

\end{titlingpage}
\restoregeometry
\newpage
\newpage
\color{black}
\backgroundsetup{opacity=0.5, contents={}}


\thispagestyle{empty}
\title{\huge C*-algebraic Results in the Search for Quantum Gauge Fields\\~\\~\\~\\~\\~\\~\\
 Proefschrift\\~\\~\\~\\~\\
\small ter verkrijging van de graad van doctor aan de\\
\vspace{2.5pt}
\large Radboud Universiteit Nijmegen\\
\vspace{4pt}
\small op gezag van de rector magnificus prof. dr. J.H.J.M. van Krieken,\\
volgens besluit van het college voor promoties\\
in het openbaar te verdedigen op\\
donderdag 7 april 2022\\
om 12:30 uur precies\\
~\\~\\~\\~\\~\\~door\\~\\~\\~\\
}
\author[van Nuland]{\Large Teun Dirk Henrik van Nuland\\
{\small geboren op 10 mei 1994\\
te Tiel}
}
\date{~}
\thispagestyle{empty}

\thispagestyle{empty}
\begin{center}
\thispagestyle{empty}
\resizebox{\textwidth}{!}{%
 \begin{tabular}{c}
 
 \small \textbf{C*-algebraic Results in the}\\
  \small \textbf{Search for Quantum Gauge Fields}\\
  \tiny \textbf{\textit{Online version}}
  ~\\
  ~\\
  ~\\
  ~\\~\\~\\
  \vspace{22pt}\\
  {
  \huge $Q_\hbar(\,\raisebox{7pt}{\uwave{\qquad}}\;)$}
 \end{tabular}%
 }
 \mbox{}
\thispagestyle{empty}
 \vfill
 {
 \huge \textbf{Teun D. H. van Nuland}}\\
 ~\\
 \small Questions and corrections very welcome at teunvn$\otimes$gmail$\cdot\hspace{0.5pt}$com
\thispagestyle{empty}
\end{center}
\thispagestyle{empty}
\newpage

\noindent \textbf{Onderzoek gesponsord door}\\
NWO Physics Projectruimte (680-91-101).

\vspace{563pt}

\noindent \phantom{\textbf{Colofon}\\
~\\
\textit{Design/Lay-out} \,\,Proeschriftenbalie.nl\\
~\\
\textit{Print} \,\,Ipskamp Printing\\
~\\
\textcopyright \,\,2022, Teun van Nuland}
\thispagestyle{empty}
\maketitle
\thispagestyle{empty}
\pagebreak
\thispagestyle{empty}
\noindent \textbf{Promotor:}\\
prof. dr. W.D. van Suijlekom\\
~\\
~\\

\noindent \textbf{Manuscriptcommissie:}\\
prof. dr. H.T. Koelink\\
dr. P. Hochs\\
prof. dr. J.M. Gracia-Bond\'ia (Universidad de Costa Rica, Costa Rica)\\
prof. dr. N.S. Larsen (Universitetet i Oslo, Noorwegen)\\
dr. B. Mesland (Universiteit Leiden)

\vspace{205pt}

\newpage
\begin{KeepFromToc}
\tableofcontents
\end{KeepFromToc}

\chapter*{Introduction}
\addcontentsline{toc}{chapter}{\protect\numberline{}Introduction}
\label{intro}
\markboth{INTRODUCTION}{}

%
This thesis consists of two parts, situated at neighboring branches of operator theory.
Although its content is mathematical, this thesis is inspired and motivated by physics at every step of the way. We will start this introduction with the mathematical and physical context that is important for both parts. 

\section*{General motivations}
\addcontentsline{toc}{section}{\protect\numberline{}General motivations}

Notions of space, becoming increasingly abstract over the years, have driven mathematics and physics in numerous ways.
%
Euclidean space is such a notion, which was in fact vital for mathematics to develop its axiomatic approach. It also provided the backdrop for most physical theories until the advent of general relativity. Then, flat Euclidean space was traded in for a curved notion, a manifold,\upcite{Lang} that had been introduced by Riemann some time before Einstein developed his groundbreaking theory.
Going further in abstraction, manifolds are now regarded as particular instances of topological spaces.\upcite{Bredon} To be more specific, they form a subclass of the locally compact Hausdorff spaces.\upcite{Bredon} The abstract notion of space provided by topology elegantly allows infinite dimensions and singularities. 
Besides its relevance in physics, like the well-known solid-state applications,\upcite{HK} the abstractness of topology has proven particularly useful in mathematics.
Indeed, once you show that a certain property of a topological space (like a separation axiom) implies another, your result can be applied in an enormous amount of instances.
Many useful properties of topological spaces are naturally formulated in terms of continuous functions on the space.
Gelfand duality\upcite{Murphy} explains the important role of continuous functions in topology, at the same time providing our final step in abstraction. In its simplest form, Gelfand duality states that the commutative unital C*-algebras${}^{\textnormal{\cite{Murphy}}}$ as a category are equivalent to the compact Hausdorff spaces, each C*-algebra being obtained as the collection of continuous functions on the respective compact Hausdorff space. In a stronger form it shows how \textit{all} commutative C*-algebras uniquely arise from \textit{locally} compact Hausdorff spaces. In short, C*-algebras provide an elegant generalized notion of space.

A first advantage is that results proven for all C*-algebras imply results for all locally compact Hausdorff spaces, implying results for all manifolds, implying results for all Euclidean spaces. 
One can even describe manifolds in a C*-algebraic way while retaining the geometric structure that is lost when passing to topology. This is done by noncommutative geometry${}^{\textnormal{\cite{C94}}}$, as we will discuss later. 
%
%

%



A second advantage of the C*-algebraic method is that it allows us to go beyond the classical notion of space, by using the fact that C*-algebras can be noncommutative.

\subsubsection*{Quantum mechanics and C*-algebras}
Where Gelfand duality classifies commutative C*-algebras, the Gelfand--Naimark theorem\upcite{Murphy} realizes all (possibly noncommutative) C*-algebras as spaces of bounded operators on a Hilbert space. 

This provides the perfect setup for quantum mechanics. The most basic objects of quantum mechanics are observables, which, mathematically speaking, are (unbounded) self-adjoint operators in a Hilbert space. A cornerstone of the quantum physical approach is that the measurable information of these observables -- their spectrum -- is determined uniquely by their commutation relations. Making such a statement mathematically rigorous is -- like with many other statements in quantum mechanics -- complicated by domain problems.\upcite{Hall} 
Luckily though, thanks to Stone's theorem,\upcite{Hall} we can equivalently work with bounded operators obtained from the unbounded ones by functional calculus (like the associated one-parameter unitary groups or the resolvents of the unbounded operators) and resolve all domain problems. For example, the commutation relations formulated in terms of such bounded operators give the easiest way to uniquely characterize a quantum system.\upcite{Hall} Because these bounded operators generate C*-algebras (like the Weyl C*-algebra\upcite{HR} or the resolvent algebra\upcite{BG}) we end up with an elegant C*-algebraic description of quantum mechanics. As a rule, quantum mechanical concepts (states, time evolutions, classical limits, etcetera) are most rigorously analyzed in such a C*-algebraic description.\upcite{landsman17} 

The rigorous description of quantum mechanics is just the start. In quantum field theory, and in quantum gauge theory in particular, even more compelling reasons to work with C*-algebras appear. 

\subsubsection*{Gauge theory and the structure of space}
Yang--Mills gauge theory\upcite{Atiyah} forms a completely geometrical basis of our understanding of forces in the Standard Model. 
In the classical form of the theory, the gauge fields that define forces are modeled as an internal structure of spacetime that influences the matter particles that move through it. Gauge fields can be represented by various mathematical objects (Lie-algebra-valued one-forms,\upcite{Ehresmann} covariant derivatives,\upcite{Lang} etcetera) describing this internal structure. That is to say, when a matter particle moves along a path through space, the values of the gauge field along that path determine the change of a symmetry parameter associated with the particle. This parameter lies in a compact Lie group called the gauge group, and is only observable relatively, as physics as a whole is invariant under (local) actions of this gauge group. The gauge field itself changes as well, and does so according to a set of differential equations -- the Euler--Lagrange equations\upcite{Taylor} of the Yang--Mills action\upcite{Atiyah} -- which allows waves in the gauge field to propagate through space. Electrodynamics is a simple but already very powerful example of a gauge theory, in which the gauge field comes from the Yang--Mills action for an abelian and one-dimensional gauge group. It describes the electromagnetic field as a gauge field (hence, a geometrical internal structure of space) which deflects any electron that moves through, effectively creating an acceleration of the particle by electric and magnetic forces. The waves in the gauge field quite accurately describe light of arbitrary polarization.

As we know, light is not only a wave, but sometimes behaves like a particle as well. This is one of the reasons why gauge theory should be quantized.

\subsubsection*{Quantum gauge theory}
Quantum field theory\upcite{PS} provides us with a well-motivated strategy to quantize gauge theories. Surprisingly, there exists no known mathematical model for quantum gauge theories, not even for quantum electrodynamics (QED) even though its gauge group is abelian. But although the current perturbative model of QED is ill-defined in a strict sense, it has been extremely successful. It has survived confrontation with experiments up to unprecedented precision. Only at very high energy scales this theory seems to predict an infinite coupling constant, causing it to be inconsistent.\upcite{Landau} The theory of quarks and the strong nuclear force, quantum chromodynamics, makes use of the Yang--Mills action of the gauge group $SU(3)$. Because this gauge group is nonabelian, one needs to include so-called Faddeev--Popov ghost fields\upcite{FP} in order to obtain experimentally testable results, but these ghost fields are manifestly nonphysical. This is one of the reasons that the mathematical construction of an $SU(3)$ quantum Yang--Mills theory is labeled a millennium problem by the Clay Mathematics Institute.\upcite{JW} 

Some would argue that the problems described above arise because the gauge groups underlying the Standard Model simply do not represent nature, and will disappear once we have a Grand Unified Theory.\upcite{FNTU} Others would say the problems are artifacts caused by looking at the Standard Model from a perturbative angle, and will disappear once we have a rigorous non-perturbative framework.\upcite{GHLRSS,JW}

This thesis would like to demonstrate that, whichever of these two hypotheses you adhere to, 
a C*-algebraic approach can help.

\section*{Operator trace functionals and noncommutative geometry (part I)}
\addcontentsline{toc}{section}{\protect\numberline{}Operator trace functionals and noncommutative geometry (part I)}

In quantum mechanics, all physical information of an observable is contained in its spectrum, regarded as a subset of the real line including multiplicity. For the Hamiltonian of 
the harmonic oscillator for instance, this spectrum (i.e., the set of eigenvalues) gives all possible outcomes of a measurement of the energy of the system, namely an infinite list of increasing energy levels. 
When a bounded potential term is added to this Hamiltonian (say a small perturbation of the system) solving the Schr\"odinger equation exactly might become very complicated.  However, one can perturbatively calculate how each of the energy levels in the list shifts due to the perturbation. To be concrete, we use the trace of a function $f$ of the perturbed observable $H+V$, namely
\begin{align}\label{eq:operator trace functional}
	\Tr(f(H+V)),
\end{align}
which we call an operator trace functional, and compare it to its unperturbed counterpart (where $V=0$). Because all physical information is spectral, these functionals can capture all physical information about $H+V$ by choosing the right $f$. (To see this, note that $f$ can select any region of the real line to see if $H+V$ has an eigenvalue there.) 
To mathematically capture the shift of the spectrum, one should subtract the perturbed operator trace function from its unperturbed counterpart, and separate the part that depends on the test function from the rest. The part independent from $f$ captures the spectral shift from $H$ to $H+V$. Remarkably, this can in many cases be described simply by a \textit{function}, called the spectral shift function.$^\textnormal{\cite{Krein53}}$ More generally, Koplienko$^\textnormal{\cite{Koplienko84}}$ introduced the higher-order spectral shift function, which captures not only the spectral information on the jump from $H$ to $H+V$, but also on the path that is taken (say, along the path $t\mapsto H+tV$) by considering higher-order Taylor remainders.

The spectral shift function (of order 1, 2, or any other $n\in\N$) has connections with many other fascinating mathematical and physical notions like spectral flow,\upcite{ACDS09} perturbation determinants,\upcite{Koplienko84} and scattering phases.\upcite{BP98}
To prove existence and properties of a (higher-order) spectral shift function under general conditions is a challenging analytical problem. This problem has sparked the creation and investigation of many fascinating mathematical objects and interrelations.\upcite{ACDS09,BP98,BS75,Neidhardt88,Peller,PSS13,S17,ST19,Yafaev92,Yafaev05,Yafaev10} One excellent object to deal with these kinds of problems is the multiple operator integral,\upcite{ST19} which is an $n$-multilinear map between operator spaces that extends the $n\th$ order derivative of an operator trace function.

Not only will multiple operator integrals be essential to the proof of existence of the spectral shift function in this thesis, they will also play a crucial role in our results on noncommutative geometry.

\subsubsection*{Noncommutative geometry}
Noncommutative geometry generalizes manifolds to spectral triples.
A spectral triple 
\begin{align*}
	(\A,\H,D)
\end{align*}
consists of an algebra $\A$ -- which should be thought of as a generalization of the underlying topology of the space -- together with a Hilbert space $\H$ and a self-adjoint operator $D$ in $\H$, satisfying certain properties.\upcite{GVF} The latter two objects enrich the generalized topological space with more structure, analogous to how a manifold is a topological space enriched with extra structure. In fact, Alain Connes\upcite{C08} showed that a natural class of commutative spectral triples (the word `commutative' referring to the algebra) is in bijection with a natural class of Riemannian manifolds called spin$^\text{c}$ manifolds.\upcite{Sui15}

In short, a possibly noncommutative spectral triple is a generalization of a manifold that, as opposed to a C*-algebra, retains the differentiable structure.

\subsubsection*{Spectral Action Principle}
In order to describe action principles\upcite{Taylor} in physics in a geometric way, Chamseddine and Connes\upcite{CC96,CC97} proposed an action principle for spectral triples. The action consists of a bosonic and fermionic part. The bosonic part describes the gauge fields and is called the spectral action. It is given by the operator trace functional (cf. \eqref{eq:operator trace functional})
\begin{align}\label{eq:full spectral action}
	\Tr\left(f\left(\frac{D+V}{\Lambda}\right)\right)
\end{align}
for a spectral triple $(\A,\H,D)$. Here, $\Lambda$ is a number indicating the cutoff scale, and $V$ takes a specific form derived from $\A$ and $D$, and should be thought of as a noncommutative gauge potential. Just like gauge fields on a manifold (Lie-algebra-valued differential one-forms) can be represented by a finite number of scalar functions,\upcite{Lang} the noncommutative gauge field $V$ can be represented by elements of $\A$ instead of scalar functions.\upcite{C85} 
By choosing the right spectral triple, and taking the limit $\Lambda\to\infty$, one can recover the classical action of the entire Standard Model of particle physics from the spectral action principle, including neutrino oscillations and a minimal coupling to gravity.\upcite{Sui15} The Higgs particle does not need to be added by hand, but arises naturally from the spectral action as a generalized gauge boson. 
As such, the spectral action gives an extraordinarily elegant geometrical 
explanation of the fundamental forces in nature.

Also, the spectral action gives a few natural suggestions for physics \textit{beyond} the standard model, some of them allowing for grand unification.\upcite{CvS}

Because we do not yet have a quantum analogue of the spectral action within the noncommutative framework, obtaining experimentally testable values from a given spectral triple proceeds according to the usual renormalization group techniques\upcite{PS} of quantum field theory.
%
%
%
%
Although we do not question the correctness of its outcomes, this derivation is still unsatisfactory from a theoretical perspective. 
If we accept the premise that the particle content of the Standard Model originates from a spectral triple, then there should be a spectral description of space and its internal structure at lower energies as well. Such a description has been wanted for several decades, and is the driving motivation for us to investigate the variation of the spectral action in a gauge field theoretical way.

But before thinking about a quantum version of the spectral action, there is still a lot to do on the classical side. Although the behavior of the spectral action \eqref{eq:full spectral action}
in the limit $\Lambda\to\infty$ is reasonably well understood by heat kernel methods,\upcite{EZ,Sui15} not much work has been done on the behavior of \eqref{eq:full spectral action} as $V$ is fluctuated (in the vicinity of $0$). Indeed, the latter behavior firstly poses a challenge on the analytical side. For instance, a Taylor series in $V$ can only be obtained when putting the right assumptions on $D$, $V$ and $f$ concerning e.g. their summability and differentiability.\upcite{ST19} A second challenge is on the algebraic side, as $D$ and $V$ do not commute, except in trivial cases. 
Still, a particular result of Chamseddine and Connes\upcite{CC06} made a great first step in describing the dependence of the spectral action on $V$, and thus inspired an important part of this thesis. They were able to express the so-called scale-invariant part\upcite{CC06} of the spectral action in terms of generalized versions of the Yang--Mills action and the Chern--Simons action.\upcite{Qui90,Witten} In these generalized action functionals, the integration over a part of noncommutative space is carried out by objects from noncommutative differential geometry\upcite{C85} called cyclic cocycles. In order to extend this intriguing result to the full spectral action \eqref{eq:full spectral action} we will make good use of the algebraic and analytical benefits of multiple operator integrals, just like we did for the spectral shift function.

%

\subsubsection*{Structure of Part I}
In Part \ref{part:I} of this thesis we will investigate the variation of the spectral action (or, more generally, the operator trace functional) and prove expressions for this variation. Multiple operator integrals will be the tools that we carry throughout, which we first sharpen in Chapter \ref{ch:MOI} in order to deal with the difficult problems facing us in Chapters \ref{ch:Spectral Shift for Relative Schatten Perturbations} and \ref{ch:Cyclic cocycles in the spectral action}.
\begin{itemize}
	\item In \textbf{Chapter \ref{ch:MOI}} we identify very basic summability properties that occur naturally in applications, and apply them to multiple operator integration in an efficient and powerful way. We obtain analytical results that are fundamentally stronger than those known before, something which is vital to the following chapters.
	\item \textbf{Chapter \ref{ch:Spectral Shift for Relative Schatten Perturbations}} will discuss the higher-order spectral shift function, which elegantly captures the essence of the higher-order variation of an operator trace functional in a real-valued locally integrable function, separating it from any dependence on the test function. As we will show, the spectral shift function exists for a much larger class of operators, and has a more tempered behavior, than was previously known. The proof will be a very technical tour in multiple operator integration.
\end{itemize}
Whereas Chapter \ref{ch:Spectral Shift for Relative Schatten Perturbations} extends a long known formula, Chapter \ref{ch:Cyclic cocycles in the spectral action} will prove a completely new one.
\begin{itemize}
	\item The formula proved in \textbf{Chapter \ref{ch:Cyclic cocycles in the spectral action}} expands the spectral action in the gauge fluctuation and expresses the result in terms of known noncommutative versions of higher-order Yang--Mills and Chern--Simons actions. The noncommutative analogues of integrals in these action functionals, cyclic cocycles, are shown to fit precisely in the framework of noncommutative differential geometry. Moreover, 
we construct an important quantum field theoretic application. Namely, we use the expansion as a starting point for a (one-loop) quantization of the spectral action.
\end{itemize}

\section*{Lattices and strict deformation quantization (part II)}
\addcontentsline{toc}{section}{\protect\numberline{}Lattices and strict deformation quantization (part II)}

The second part of this thesis, in contrast to the first, takes a completely non-perturbative approach to gauge theories, employing the framework of lattice gauge theory introduced by Wilson\upcite{Wilson} and adapted by Kogut and Susskind.\upcite{KS} It focuses on abelian Yang--Mills, and is driven in part by the ultimate goal of eventually constructing a rigorous model for QED.

C*-algebras are expected to provide the building blocks of a mathematical construction of 
gauge theories such as QED and quantum Yang--Mills.
Besides providing the most elegant description of (ordinary) quantum mechanics, C*-algebras feature in the Haag--Kastler axioms,\upcite{Haag} and could therefore be used to construct a local quantum field theory. Moreover, as C*-algebras can model both quantum and classical theories, a C*-algebraic model of a classical gauge theory might provide a good footing from which to take the leap towards a quantum gauge theory.
The direction of this leap, then, might be indicated by strict deformation quantization,\upcite{landsman98} for it gives a set of axioms that a quantization map between a classical C*-algebra and its quantum counterpart should satisfy. These axioms are stringent, and examples are mostly found in finite-dimensional configuration spaces,\upcite{bieliavsky15,landsman93,landsman98,LMvdV,Rieffel89,Rieffel} with a few exceptions that usually rely on finite-dimensional approximations.\upcite{BHR,vN19} To quantize a gauge theory, one is therefore advised to first quantize a finite-dimensional regularization.

Lattice gauge theory provides precisely such a finite-dimensional regularization. This theory, introduced by Wilson,\upcite{Wilson}  shows how to approximate gauge fields by their parallel transports on a lattice (where by `lattice' we mean a type of finite graph). Wilson's framework has been widely applied in theoretical and phenomenological physics. On the theoretical side, an important contribution was made by Kogut and Susskind,\upcite{KS} who took a Hamiltonian approach to Wilson's ideas, considering lattices in a time-slice -- typically $\R^3$ -- and showed that the parallel transports of a gauge field on the lattice can be interpreted as rigid rotors, and that Yang--Mills time evolution implies a certain coupled movement of these rotors. 
Important for us, the finite-dimensionality of this quantum Hamiltonian system makes it suitable for the C*-algebraic approach. C*-algebraic Hamiltonian quantum lattice gauge theory forms a lively research program.\upcite{ASS,BS19a,BS19b,GR,vNS20,Stienstra,ST} A central goal of this program is to describe the continuum limit (in which the lattices are replaced by the full Euclidean space or a subset thereof) by a C*-algebra invariant under a *-homomorphism coming from the Yang--Mills equations, and satisfying axioms such as Lorentz invariance, thereby potentially giving rise to a local quantum field theory.\upcite{Haag}

As others have done before,\upcite{ASS,Stienstra} we will use strict deformation quantization as a guiding principle towards obtaining such a continuum C*-algebra.

This approach can be divided into two steps. The first step is to construct good classical and quantum observable algebras at the finite level, i.e., associated to a (finite) lattice. By connecting each quantum algebra with a classical counterpart through a quantization map, the classical limit is kept close at hand. An important feature of these observable algebras is invariance under the restriction of (classical and quantum) Yang--Mills time evolution.\upcite{KS,Stienstra} If this feature is present, then the incorporation of full Yang--Mills dynamics is reduced to an approximation problem (although not necessarily an easy one).\upcite{KS}

The second step is to construct the continuum limit, by letting the lattices approach the Euclidean space in which they lie. This includes simultaneously the ultraviolet limit (in which the lattice spacing becomes arbitrarily small) and the infrared limit (in which the lattice covers an arbitrarily large volume). Step two can only be achieved if the classical and quantum systems associated to the lattices chosen in step one are sufficiently nice; in particular, they should admit the embedding maps one gets from replacing a coarse lattice by a finer one. It was noted by Stottmeister and Thiemann\upcite{ST} that the observable algebras at least need to admit the *-homomorphism given by tensoring with the identity, which already excludes the compact operators, for instance.

We will uncover a radical additional requirement. As we will see, in order to satisfy natural constraints, we are forced at the finite level to replace the algebras with operator systems. The operator systems have the same interpretation as the algebras, and have the same structure except that they are effectively `cut off' at a certain momentum scale, and are therefore no longer closed under multiplication. However, quite surprisingly, we will also argue that this replacement does not cause problems in the end, as the C*-algebraic structure can be fully retrieved in the continuum limit.

Part \ref{part:II} culminates in the construction of promising new field algebras for classical and quantum abelian lattice gauge theories in arbitrary dimension, and a strict deformation quantization between them.


\subsubsection*{Structure of Part II}
The first step, concerning the situation at the finite level, is taken in Chapter \ref{ch:cylinder}, whereas the second step, concerning the continuum limit, is taken in Chapter \ref{ch:SDQ}.

\begin{itemize}
\item In \textbf{Chapter \ref{ch:cylinder}} we will construct two new interesting C*-algebras and show that they satisfy very useful properties. 
A commutative one is intended to model classical abelian gauge theory, and a noncommutative one should model quantum abelian gauge theory, both considered on a finite lattice.
The noncommutative C*-algebra will be obtained via Weyl quantization from the commutative one. Both C*-algebras are shown to be invariant under their respective time evolutions. The Weyl quantization map (which first needs to be generalized to be defined for functions that do not vanish at infinity) turns out to satisfy many suitable properties. Almost enough properties are satisfied for it to be called a strict deformation quantization.
\item In \textbf{Chapter \ref{ch:SDQ}} we define a continuum limit that respects quantization. But, in order to do that without handing in on properties like gauge invariance, an unconventional approach is needed, namely to drop the multiplicative structure at the finite level. Although radical, this step allows us to indeed define a continuum limit, and, for the systems thus obtained, it turns out that the multiplicative structure is recovered. Better yet, the obstructions that barred Chapter \ref{ch:cylinder} from obtaining a strict deformation quantization, have melted away by passing to the limit.
\end{itemize}

%
%
%

\subsubsection*{Papers}
This thesis is based on five papers.\\

 \noindent Chapter \ref{ch:MOI} is based on the first Sections of \cite{vNvS21a} (joint with W. D. van Suijlekom and published in \textit{J. Noncommut. Geom.}) and \cite{vNS21} (joint with A. Skripka and accepted for publication in \textit{J. Spectr. Theor.}).\\
Chapter \ref{ch:Spectral Shift for Relative Schatten Perturbations} is based on \cite{vNS21}.\\
Chapter \ref{ch:Cyclic cocycles in the spectral action}, save for Section \ref{sct:One-Loop}, is based on \cite{vNvS21a}. 
Section \ref{sct:One-Loop} is based on \cite{vNvS21b} (joint with Van Suijlekom and published in \textit{J. High Energy Phys.}).\\
Chapter \ref{ch:cylinder} is based on \cite{vNS20} (joint with R. Stienstra and submitted for publication).\\
Chapter \ref{ch:SDQ} is based on \cite{vN21} (published in \textit{Lett. Math. Phys.}).

\chapter*{Notations}
\addcontentsline{toc}{chapter}{\protect\numberline{}Notations}
\markboth{NOTATIONS}{}

We write $\N=\{1,2,\ldots\}$ and $\N_0=\{0,1,2,\ldots\}$.

\paragraph*{Function classes.} 
All functions (on some space $X$ with suitable structure) are complex-valued unless stated otherwise. We denote by $C(X),C_0(X),C_\text{b}(X),C^n(X)$ respectively the continuous functions, the ones vanishing at infinity, the bounded ones, and the $n$ times continuously differentiable ones ($n\in\N_0$) on $X$. By $C^\infty(X),C^\infty_\text{c}(X),\S(X)$ we denote the smooth functions, the compactly supported ones, and the Schwartz functions on $X$. Let $L^p(X)$ denote the space of (measure-zero equivalence classes of) measurable functions $f$ on $X$ for which $|f|^p$ is Lebesgue integrable, equipped with the standard norm $\|f\|_p:=(\int_X|f(x)|^p\,dx)^{1/p}$ ($p\in[1,\infty)$) and let $L^\infty(X)$ denote the space of essentially bounded functions equipped with the essential supremum norm $\|\cdot\|_\infty$.
Let $L^1_{\textnormal{loc}}(X)$ denote the space of locally integrable functions. When $X$ equals $\R$, its dependency in the above function classes is suppressed. We write $C^n_\text{c}$ for the space of compactly supported functions in $C^n$, and write $\mathcal{D}:=C_\text{c}^\infty$.

\paragraph*{Operator theory.}
\label{General intro}
Throughout, we fix a separable Hilbert space $\H$. If we say $D$ is self-adjoint in $\H$, it is possibly unbounded and self-adjoint with a domain dense in $\H$. We denote by $E_D$ its spectral measure. We let $\mB(\H)$ be the C*-algebra of bounded operators on $\H$, $\mK(\H)$ the one of compact operators on $\H$, and $\norm{\cdot}$ the operator norm. When $N$ is a subset of one of the C*-algebras $\mB(\H)$ and $C_\textnormal{b}(X)$, $N_\sa$ denotes the self-adjoint elements in $N$. We denote the Schatten $p$-class (i.e., the Schatten-von Neumann ideal of $p$-summable operators on $\H$) by $\S^p$, and its norm by $\nrm{\cdot}{p}$ ($p\in[1,\infty)$). Basic properties of Schatten-von Neumann ideals can be found in, for instance, \cite{Simon05,ST19}. In some cases it will also be convenient to denote $\S^\infty:=\mB(\H)$ and $\nrm{\cdot}{\infty}=\norm{\cdot}$. 
We denote by $M_g\in\mB(L^2(X))$ the multiplication operator of the function $g\in L^\infty(X)$. For $\psi,\varphi\in\H$ we define $\ket{\psi}\bra{\varphi}\in\mB(\H)$ by $\ket{\psi}\bra{\varphi}(\chi):=\p{\varphi}{\chi}\psi$.

\paragraph*{Fourier transforms.}We define the Fourier transform of $f\in L^1(\R^n)$ by 
\begin{align}\label{eq:FT}
\hat{f}(x):=\int_{\R^n} \frac{dy}{(2\pi)^n}f(y)e^{-iyx}.
\end{align}
For a general (not necessarily tempered) distribution $f\in\mathcal{D}'$ we can still define the Fourier transform as a distribution $\hat{f}:\hat{\mathcal D}\to\C$ by $\langle \hat{f}|\varphi\rangle:=\braket{f}{\hat{\varphi}}$ for all Schwartz functions $\varphi$ with $\hat{\varphi}\in\mathcal D$. The restriction $~\hat{~}:\S'\to\S'$ is bijective with inverse denoted by $\check{~}:\S'\to\S'$. 

The Fourier transform on $\D'$ is less well-behaved, but will only be applied in the following way. For an arbitrary continuous function $f$ (which is in $\mathcal D'$ but not a priori in $\S'$) we will often assume that $\hat{f}\in L^1$. Because then $\hat{f}\in\S'$ and $\check{\hat{f}}\in C_0$ by the Riemann-Lebesgue lemma, we find that $\hat{\check{\hat{f}}}=\hat{f}$, which implies $\braket{f}{\varphi}=\langle\check{\hat{f}}|\varphi\rangle$ for all $\varphi\in\mathcal D$. Therefore $f=\check{\hat{f}}\in C_0$ and $f(x)=\int \hat{f}(y)e^{ixy}\,dy$.


\paragraph*{Specific functions.}
We define $$u(x):=x-i$$ for $x\in\R$ and write $u^{k}(x):=(u(x))^k$ for $k\in\Z$. 
The zeroth order divided difference $f^{[0]}$ of a function $f\in C$ is $f^{[0]}:=f$. Let $x_0,\dots,x_n$ be points in $\R$ and let $f\in C^n$. The divided difference $f^{[n]}$ of order $n$ is defined recursively by
\begin{align}\label{eq:divdiff}
f^{[n]}(x_0,\dots,x_n)
:=\lim\limits_{x\rightarrow x_n}\frac{f^{[n-1]}(x_0,\dots,x_{n-2},x)
-f^{[n-1]}(x_0,\dots,x_{n-2},x_{n-1})}{x-x_{n-1}}.
\end{align}
 
\paragraph*{Universal forms.}
When $\A\subseteq\mB(\H)$ is a *-algebra, we write $\Omega^\bullet(\A)=\oplus_{n\in\N_0}\Omega^n(\A)$ for the universal differential graded algebra over $\A=:\Omega^0(\A)$, endowed with grading $d$. When $D$ is self-adjoint in $\H$ with $[D,\A]\subseteq\mB(\H)$, we write $\Omega^1_D(\A):=\pi_D(\Omega^1(\A))$ where $\pi_D:\Omega^1(\A)\to\mB(\H)$ is the linear *-preserving map defined by $\pi_D(adb):=a[D,b]$. Whenever $A\in\Omega^1(\A)$, we write $F:=dA+A^2\in\Omega^2(\A)$ for the curvature of $A$.

\paragraph*{The torus.} The elements of the $n$-torus, $\T^n:=\R^n/\Z^n$, are usually denoted by $q$ or $[x]$, where $[x]:=x+\Z^n$ for $x\in\R^n$. We denote by $L_q$ the left-translation on $\T^n$, i.e., $L_{[x]}[y]=[x+y]$.  We denote by $e_a$ the functions $[x]\mapsto e^{ia\cdot x}$ for each $a\in\Z^n$, and by $\psi_a$ the equivalence class of $e_a$ in $L^2(\T^n)$. 

\part{C*-algebraic Perturbation Theory for Noncommutative Geometry}
\label{part:I}

%

\chapter{Multiple Operator Integration for Finitely Summable and Relative Schatten Operators}
\label{ch:MOI}
\chaptermark{M.O.I. for Finitely Summable \& Relative Schatten}

In this chapter, adapted from \cite[Sections 3 and 4]{vNS21} and \cite[Section 3]{vNvS21a}, we prove several bounds and continuity properties of the multiple operator integral in the case that the (unbounded) self-adjoint operator is either finitely summable, or together with its perturbation satisfies a constraint called the relative Schatten condition. These results will be crucial in Chapters \ref{ch:Spectral Shift for Relative Schatten Perturbations} and \ref{ch:Cyclic cocycles in the spectral action}. 

Results in Section \ref{sct:Relative Schatten} were obtained in collaboration with Anna Skripka.

\section{Introduction}
\label{sct:MOI Introduction}
Multiple operator integration is a powerful tool with numerous applications in the noncommutative realm, such as the theory of spectral shift, spectral flow, index theory, differentiation of operator functions, and expansions of trace functionals. Multiple operator integrals emerged in the 1950s from the groundbreaking work of \cite{DK,Krein53,Lifshits}. Gradually, for example in \cite{ACDS09,BS75,DK,Koplienko84,Peller,PSS13}, their properties and applicability became better understood. A full treatment can be found in \cite{ST19}. For us, multiple operator integrals will be of interest mainly because of the operator trace functional
\begin{align}\label{eq:Operator Trace Functional}
	V\mapsto \Tr(f(H+V)).
\end{align}
Here $H$ is a (possibly unbounded) self-adjoint operator in a Hilbert space $\H$, $V\in\mB(\H)_\sa$ is called a perturbation, and the test function $f:\R\to\C$ acts by functional calculus on $H+V$. The functional \eqref{eq:Operator Trace Functional} arises as the spectral action in noncommutative geometry, but also occurs in quantum mechanics, most notably when $H$ is some non-interactive Hamiltonian, and $V$ is an interaction or potential term. A challenge in analyzing \eqref{eq:Operator Trace Functional} lies in the noncommutativity of $H$ and $V$ which occurs for instance when $H$ is a differential operator and $V$ is a multiplication operator (like a bounded potential term being added to a non-interactive Hamiltonian) or when $H$ and $V$ lie in some more abstract noncommutative space (for instance, when they are matrices). In these examples, one often wishes to analyze \eqref{eq:Operator Trace Functional} in the vicinity of $V=0$, which can be done with the Taylor expansion
\begin{align}\label{eq:Taylor 1}
	\Tr(f(H+V))\sim\sum_{n=0}^\infty\Tr\left(\frac{1}{n!}\frac{d^n}{dt^n}f(H+tV)\big|_{t=0}\right).
\end{align}
Already here one runs into a number of analytical problems. For a start, it is not a priori clear whether the higher-order derivatives of $t\mapsto f(H+tV)$ exist. Moreover, if they do, it might still not hold that
	$$\frac{d^n}{dt^n}f(H+tV)\big|_{t=0}\in\S^1,$$
and hence the trace on the right-hand side of \eqref{eq:Taylor 1} might not be defined. Lastly, for a series like the one in \eqref{eq:Taylor 1} to converge, one would hope for a trace-norm bound on $\frac{d^n}{dt^n}f(H+tV)\big|_{t=0}$ that decays rapidly as $n$ goes to infinity. In order to solve these problems, as well as obtain more insight into expansions like \eqref{eq:Taylor 1}, we will use multiple operator integrals.

	

The multiple operator integral $T^H_{f^{[n]}}$ is an $n$-multilinear operator
\begin{align*}
	T^H_{f^{[n]}}:\mB(\H)\times\cdots\times\mB(\H)\to\mB(\H)
\end{align*}
(occasionally defined only on subsets of $\mB(\H)$) depending on the function $f$ by means of its $n\th$ divided difference (defined by \eqref{eq:divdiff}). It satisfies in particular 
\begin{align*}
	T_{f^{[n]}}^H(V,\ldots,V)=\frac{1}{n!}\frac{d^n}{dt^n}f(H+tV)\big|_{t=0}.
\end{align*}
A key analytical asset of the multiple operator integral is that it satisfies a H\"older-type bound, namely,
%
\begin{align}\label{eq:Holder type bound}
	\big\|T_{f^{[n]}}^{H}(V_1,\ldots,V_n)\big\|_{\alpha}
\leq \frac{1}{n!}\big\|\widehat{f^{(n)}}\big\|_1\nrm{V_1}{\alpha_1}\cdots\nrm{V_n}{\alpha_n},
\end{align}
for $\alpha,\alpha_1,\ldots,\alpha_n$ satisfying $\frac{1}{\alpha}=\frac{1}{\alpha_1}+\ldots+\frac{1}{\alpha_n}$ and $V_j\in\S^{\alpha_j}$.
In particular, the multilinear operator $T^H_{f^{[n]}}:\S^{\alpha_1}\times\cdots\times\S^{\alpha_n}\to\S^\alpha$ is continuous, which turns out to provide a powerful interpretation of the derivatives of the functional \eqref{eq:Operator Trace Functional}. Moreover, by \eqref{eq:Holder type bound}, we have that
\begin{align*}
	\frac{d^n}{dt^n}f(H+tV)\big|_{t=0}=n!T^H_{f^{[n]}}(V,\ldots,V)\in\S^1,
\end{align*}
whenever $V\in\S^n$ and $f\in C^n$ such that $\widehat{f^{(n)}}\in L^1$. 
The approach described above has proven incredibly useful, as shown by e.g. \cite{ACDS09,LS,PSS13}.
However, for many examples, it is too restrictive to ask that the perturbation is compact -- let alone Schatten class. When $V$ is for instance a multiplication operator on a manifold, compactness implies that either $V=0$ or the manifold has dimension 0.

This chapter will show how the above results can be extended to noncompact perturbations with some simple but powerful techniques. Instead of assuming $V\in\S^n$ we will explore the separate implications of two assumptions that are more often satisfied by examples encountered in the wild.

The first assumption we use is
\begin{align}\label{eq:finitely summable}
	(H-i)^{-1}\in\S^s,
\end{align}
for some number $s$, called the summability. This condition is called \textit{finite summability} or \textit{$s$-summability}, and is satisfied for instance by differential operators on a compact manifold. It is also often assumed for spectral triples (which are generalizations of compact manifolds), in which case $H=D$ is the generalized Dirac operator. For simplicity, we assume $s\in\N$, which for first-order differential operators means that $s$ should be at least the dimension of the (noncommutative) manifold plus one, for reasons having to do with the fact that $\sum_{n=1}^\infty 1/n=\infty$ but $\sum_{n=1}^\infty 1/n^{1+\epsilon}<\infty$ for every $\epsilon>0$. Adding one (or $\epsilon$) to the dimension can be avoided considering a weak noncommutative Lebesgue space instead of $\S^s$, but this is too technical to consider here. When considering Taylor approximations of high enough order, the distinction does not matter, as we will see. 
In Section \ref{sct:Finitely summable}, we will show that, under condition \eqref{eq:finitely summable}, we have
\begin{align*}
	\nrm{T^H_{f^{[n]}}(V_1,\ldots,V_n)}{1}\leq c_{s,n}(f)\norm{V_1}\cdots\norm{V_n}\big\|(H-i)^{-1}\big\|_{s}^s,
\end{align*}
for an explicit seminorm $c_{s,n}$ and all $f\in C^n$ satisfying $\widehat{(fu^k)^{(m)}}\in L^1$ for all $k\leq n$ and $m\leq s$, where $u(x)=x-i$. The result is based on a formula that expresses the multiple operator integral into a finite sum of terms that are clearly trace class. Simultaneously, this expression implies strong continuity properties of the multiple operator integral, and an explicit expression for the terms of the (noncommutative) Taylor expansion \eqref{eq:Taylor 1}. Our results generalize (parts of) \cite{S14} and \cite{Sui11}. The results in Section \ref{sct:Finitely summable} will lay the analytical foundation of Chapter \ref{ch:Cyclic cocycles in the spectral action}, in which multiple operator integrals play a key role, because of their analytical, but also their algebraic properties.

The second assumption we use is strictly more general than the first, \eqref{eq:finitely summable}, namely
\begin{align}\label{eq:relative Schatten}
	V(H-i)^{-1}\in \S^n.
\end{align}
This condition is called \textit{relative Schatten} or \textit{relative Schatten class}. It is applicable also for differential operators on \textit{locally} compact spaces and generalized Dirac operators $H=D$ of \textit{nonunital} spectral triples. Precise accounts of applications can be found in Section \ref{sec4}. 
Under the assumption \eqref{eq:relative Schatten}, we find in Section \ref{sct:Relative Schatten} that
\begin{align*}
	\nrm{T^H_{f^{[n]}}(V,\ldots,V)}{1}\leq c_{n}(f)\big\|V(H-i)^{-1}\big\|_{n}^n,
\end{align*}
for an explicit seminorm $c_n$ and all $f\in C^n$ satisfying $\widehat{f^{(n)}},\widehat{(fu^p)^{(p)}}\in L^1$, $p\leq n$.
The proof will again rely on an explicit expression in terms of trace-class summands. This expression will be used extensively in Chapter \ref{ch:Spectral Shift for Relative Schatten Perturbations}. It also generalizes a few partial results in \cite{CS18}. 



First, in Section \ref{sct:MOI}, we will give a brief introduction to multiple operator integrals, establishing conventions and results in a form that will be useful for us.

\section{Multiple operator integration: preliminaries}
\label{sct:MOI}
We will give an introduction to the theory of multiple operator integration. A more thorough discussion is found in \cite{ST19}. 
Other good references are  \cite{ACDS09,PSS13}. 
%
%

The following very general definition of a multilinear operator integral was introduced in \cite{PSS13} (see also \cite[Definition 4.3.3]{ST19}). Recall that $E_H$ denotes the spectral measure of $H$ and that $\S^\alpha$ denotes the Schatten $\alpha$-class, with the convention that $\S^\infty=\mB(\H)$.

\begin{defi}
\label{def:nosep}
For $n\in\N$, let $\phi:\R^{n+1}\to\C$ be a bounded Borel function and fix $\alpha,\alpha_1,\ldots,\alpha_n\in [1,\infty]$ such that $\tfrac{1}{\alpha}=\tfrac{1}{\alpha_1}+\ldots+\tfrac{1}{\alpha_n}$.
Let $H_0,\dots,H_n$ be self-adjoint operators in $\H$.
Denote $E^j_{l,m}:=E_{H_j}\big(\big[\frac{l}{m},\frac{l+1}{m}\big)\big)$.
If for all $V_j\in\S^{\alpha_j}$, $j=1,\ldots,n$,
the double limit
$$T^{H_0,\ldots,H_n}_{\phi}(V_1,\ldots,V_n)
:=\lim_{m\to\infty}\lim_{N\to\infty}\sum_{|l_0|,\ldots,|l_n|<N}
\phi\left(\frac{l_0}{m},\ldots,\frac{l_n}{m}\right)E^0_{l_0,m}V_1E^1_{l_1,m}\cdots V_nE^n_{l_n,m}\,$$
exists in $\S^\alpha$, then the linear map $T^{H_0,\ldots,H_n}_{\phi}:\S^{\alpha_1}\times\cdots\times\S^{\alpha_n}\to\S^\alpha$, which is bounded by the Banach--Steinhaus theorem, is called a \textbf{multilinear operator integral}, and we write $T^{H_0,\ldots,H_n}_\phi\in\mB_{\alpha_1,\ldots,\alpha_n}^\alpha$. In the case that $H_j=H$ for all $j$, we also write $T_\phi^H:=T_\phi^{H,\ldots,H}$.
\end{defi}

An important class of examples is given by the following result, which also explains the appearance of the word \textit{integral} in \textit{multiple operator integral}. It concerns functions $\phi$ admitting a certain separation of variables, and is proven in \cite[Lemma 3.5]{PSS13}.

\begin{thm}\label{thm:coincidence}
Let $H_0,\ldots,H_n$ be self-adjoint operators in $\H$. Let $\phi:\R^{n+1}\to\C$ be a function admitting the representation
\begin{align}\label{phi representation}
	\phi(x_0,\ldots,x_n)=\int_\Omega a_0(x_0,s)\cdots a_n(x_n,s)\,d\nu(s),
\end{align}
where $(\Omega,\nu)$ is a finite measure space,
$a_j(\cdot,s):\R\to\C$ is a continuous function for every $s\in\Omega,$ and there is a sequence $\{\Omega_k\}_{k=1}^\infty$ of growing measurable subsets
of $\Omega$ such that $\Omega=\cup_{k=1}^\infty\Omega_k$ and the families
$$\{a_j(\cdot,s)\}_{s\in\Omega_k},\quad j=0,\dots,n$$
are uniformly bounded and uniformly equicontinuous. Then, for all $\alpha,\alpha_1,\ldots,\alpha_n\in[1,\infty]$ such that $\tfrac{1}{\alpha}=\tfrac{1}{\alpha_1}+\ldots+\tfrac{1}{\alpha_n}$, we have
$T^{H_0,\ldots,H_n}_\phi\in\mB_{\alpha_1,\ldots,\alpha_n}^\alpha$ and
\begin{align*}
T^{H_0,\ldots,H_n}_\phi(V_1,\ldots,V_n)\psi=\int_\Omega a_0(H_0,s)V_1 a_1(H_1,s)\cdots V_n a_n(H_n,s)\psi\,d\nu(s),\quad \psi\in\H,
\end{align*}
as well as
\begin{align*}
\big\|T_{\phi}^{H_0,\dots,H_n}(V_1,\ldots,V_n)\big\|_{\alpha}
\le \inf\bigg\{\int_\Omega\prod_{j=0}^n\|a_j(\cdot,s)\|_\infty\,d|\nu|(s)\bigg\}\nrm{V_1}{\alpha_1}\cdots\nrm{V_n}{\alpha_n},
\end{align*}
where the infimum is taken over all possible representations \eqref{phi representation}.
\end{thm}

The above theorem in particular shows that $T^H_{f}()=f(H)$. 
More generally, we will use the above theorem, although not exclusively, in the case that $\phi$ equals a divided difference $f^{[n]}$, as defined by \eqref{eq:divdiff}. To explain how, let $\sigma$ denote the standard measure on the $n$-simplex, $$\Delta_n:=\bigg\{(s_0,\ldots,s_n)\in\R^{n+1}_{\geq0}:~ \sum_{j=0}^n s_j=1\bigg\}.$$  In order to apply Theorem \ref{thm:coincidence} (and obtain Corollary \ref{cor:nice def of MOI}) we need the following lemma.
\begin{lem}\label{lem:divided diff}
Whenever $f\in C^n$ is such that $\widehat{f^{(n)}}\in L^1$  (cf. the discussion following \eqref{eq:FT}) we can write $f^{[n]}$ as
	$$f^{[n]}(x_0,\ldots,x_n)=\int_{\Delta_n}\int_\R e^{its_0x_0}\cdots e^{its_nx_n}\widehat{f^{(n)}}(t)\,dt\,d\sigma(s_0,\ldots,s_n).$$
	As such, $\phi=f^{[n]}$ satisfies the assumptions of Theorem \ref{thm:coincidence}.
\end{lem}
\begin{proof}
One simply combines the proofs of \cite[Lemma 5.1]{PSS13} and \cite[Lemma 5.2]{PSS13}.
It is easily seen that $(\Sigma,\sigma_f):=(\Delta_n\times\R,\sigma\times \widehat{f^{(n)}})$ is a finite measure space with total variation equal to $\tfrac{1}{n!}\|\widehat{f^{(n)}}\|_1$.
\end{proof}

\begin{cor}\label{cor:nice def of MOI}
	Let $H_0\ldots,H_n$ be self-adjoint in $\H$ and let  $f \in C^{n}$ such that $\widehat{f^{(n)}} \in L^1$. For all $\alpha,\alpha_1,\ldots,\alpha_n\in[1,\infty]$ such that $\tfrac{1}{\alpha}=\tfrac{1}{\alpha_1}+\ldots+\tfrac{1}{\alpha_n}$ we have $\Tfn^{H_0,\ldots,H_n}\in\mB_{\alpha_1,\ldots,\alpha_n}^\alpha$ and for all $V_1,\ldots,V_n\in\mB(\H)$ and $\psi\in\H$ we have
\begin{align}\label{def:Tfn}
	\Tfn^{H_0,\ldots,H_n}(V_1,\ldots,V_n)\psi=\int_{\Delta_n}\int_\R   e^{its_0H_0}V_1e^{its_1H_1}\cdots V_n e^{its_nH_n}\psi\:\widehat{f^{(n)}}(t)\,dt\,d\sigma(s_0,\ldots,s_n).
\end{align}
Moreover, denoting $\norm{\cdot}_\infty=\norm{\cdot}$, we have
\begin{align}
\label{fourierbound}
\Big\|T_{f^{[n]}}^{H_0,\dots,H_n}(V_1,\ldots,V_n)\Big\|_{\alpha}
\leq \frac{1}{n!}\Big\|\widehat{f^{(n)}}\Big\|_1\nrm{V_1}{\alpha_1}\cdots\nrm{V_n}{\alpha_n}.
\end{align}
\end{cor}
%

The next theorem, discovered in \cite{PSS13}, shows that, in some cases, the norm of the multiple operator integral $T_{f^{[n]}}$ can be bounded not just by $\|\widehat{f^{(n)}}\|_1$, but even by $\|f^{(n)}\|_\infty$.


\begin{thm}\label{thm:PSS Higher Order}
Let $\alpha,\alpha_1,\ldots,\alpha_n\in (1,\infty)$ such that $\tfrac{1}{\alpha}=\tfrac{1}{\alpha_1}+\ldots+\tfrac{1}{\alpha_n}$. If $f\in C^n$ is such that $f^{(n)}\in C_\text{b}$ then $T^{H_0,\ldots,H_n}_{f^{[n]}}\in\mB_{\alpha_1,\ldots,\alpha_n}^\alpha$ and 
\begin{align}
\nrm{T^{H_0,\ldots,H_n}_{f^{[n]}}(V_1,\ldots,V_n)}{\alpha}\leq c_{\alpha_1,\ldots,\alpha_n}\supnorm{f^{(n)}}\nrm{V_1}{\alpha_1}\cdots\nrm{V_n}{\alpha_n}.
\end{align}	

\end{thm}
\begin{proof}
The result for $H_0=\ldots=H_n$ is proved in \cite[Theorem 5.6]{PSS13}. Its extension to the case of distinct $H_0,\ldots,H_n$ is explained in the proof of \cite[Theorem 4.3.10]{ST19}.
\end{proof}

\subsection{Continuity}
Equation \eqref{fourierbound} shows in particular that $T^{H_0,\ldots,H_n}_{f^{[n]}}:\B(\H)^{\times n}\to \B(\H)$ is $\norm{\cdot}$-continuous. This is known to still hold true when we replace $(\B(\H),\norm{\cdot})$ by $(\B(\H)_1,\text{s.o.t.})$ (see \cite[Proposition 4.9]{ACDS09}). Here $\mB(\H)_1$ denotes the closed unit ball in $\mB(\H)$, and $\text{s.o.t.}$ the strong operator topology. These results can be unified and generalized by writing $\L^\alpha:=(\S^\alpha,\nrm{\cdot}{\alpha})$ for $\alpha\in[1,\infty)$ and $\L^\infty:=(\B(\H)_1,\text{s.o.t.})$. We can then make use of the following straightforward result in operator theory.

\begin{lem}\label{lem:Schatten product is continuous}
Let $\alpha,\alpha_j\in[1,\infty]$ satisfy $\tfrac{1}{\alpha}=\tfrac{1}{\alpha_1}+\ldots+\tfrac{1}{\alpha_n}$. If either $\alpha_n<\infty$ or $\alpha_1=\ldots=\alpha_n=\infty$, then
the function $$(A_1,\ldots,A_n)\mapsto A_1\cdots A_n$$
is a continuous map from $\L^{\alpha_1}\times\cdots\times\L^{\alpha_n}$ to $\L^\alpha$.
\end{lem}

This implies the following slight strengthening of \cite[Proposition 4.9]{ACDS09}.

\begin{lem}\label{lem:continuity with alpha relation}
	Let $f\in C^n$ with $\widehat{f^{(n)}}\in L^1$ and let $\alpha,\alpha_j\in[1,\infty]$ with $\tfrac{1}{\alpha}=\tfrac{1}{\alpha_1}+\ldots+\tfrac{1}{\alpha_n}$. If either $\alpha_n<\infty$ or $\alpha_1=\ldots=\alpha_n=\infty$, then
		$$T^{H_0,\ldots,H_n}_{f^{[n]}}:\L^{\alpha_1}\times\cdots\times\L^{\alpha_n}\to\L^\alpha$$
	is continuous.
\end{lem}
\begin{proof}
	If $\alpha=\infty$, then all $\alpha_j=\infty$, and the result is proven in \cite[Proposition 4.9]{ACDS09}. If $\alpha<\infty$, we define
		$$A_{s,t}:=e^{its_0H_0}V_1e^{its_1H_1}\cdots V_ne^{its_nH_n},$$
	for all $(s,t)\in\Delta_n\times\R=: \Sigma$.	
	We find
\begin{align}\label{eq:Tr(TfnC)=int Tr(AC)}
	\Tr(\Tfn^{H_0,\ldots,H_n}(V_1,\ldots,V_n)B)=\int_{\Delta_n}\int_\R \Tr(A_{s,t}B)\widehat{f^{(n)}}(t)\,dt\,d\sigma(s_0,\ldots,s_n),
\end{align}
	for every $B\in\S^{\alpha'}$, where $\alpha'=(1-1/\alpha)^{-1}$. By Lemma \ref{lem:Schatten product is continuous}, for every $(s,t)$, the operator $A_{s,t}B\in\L^1$ depends continuously on $(V_1,\ldots,V_n)\in\L^{\alpha_1}\times\cdots\times\L^{\alpha_n}$. Since any convergent sequence in $\L^p$ is bounded with respect to $\nrm{\cdot}{p}$ (where $\nrm{\cdot}{\infty}=\norm{\cdot}$) an application of the dominated convergence theorem shows that \eqref{eq:Tr(TfnC)=int Tr(AC)} depends sequentially continuously on $(V_1,\ldots,V_n)$. By our specific choice of $\L^\infty$, every $\L^p$ is a metric space, hence sequential continuity implies continuity.
\end{proof}

The relation $\tfrac{1}{\alpha}=\tfrac{1}{\alpha_1}+\ldots+\tfrac{1}{\alpha_n}$ is central to the above Lemma. When the resolvent of $H$ is $s$-Schatten, however, for an explicit class $\Wsn$ of functions $f$, defined in \eqref{eq:Wsn}, we can remove this restrictive relation. We will prove this in \textsection\ref{sct:MOI continuity}.

\subsection{Taylor remainder via operator integrals}

The following result shows that the multiple operator integral of order $n$ is a multilinear extension of the $n\th$ derivative of an operator function. We refer the interested reader to \cite{ST19} for additional details.

\begin{thm}
\label{dm}
Let $n\in\N$ and let $f\in C^n(\R)$ be such that $\widehat{f^{(k)}}\in L^1(\R)$, $k=1,\dots,n$. Let $H$ be a self-adjoint operator in $\H$, let $V\in\mB(\H)_{\sa}$.
Then, the Fr\'{e}chet derivative $\frac{1}{n!}\frac{d^n}{dt^n}f(H+tV)|_{t=0}$ exists in the operator norm and admits the multiple operator integral representation
\begin{align}
\label{dermoi}
\frac{1}{n!}\frac{d^n}{ds^n}f(H+sV)\big|_{s=t}=T_{f^{[n]}}^{H+tV,\dots,H+tV}(V,\dots,V).
\end{align}
The map $t\mapsto\frac{d^n}{ds^n}f(H+sV)|_{s=t}$
is continuous in the strong operator topology and, when $V\in\S^n$, in the $\S^1$-norm.
\end{thm}
\begin{proof}
The first assertion is given in \cite[Theorem 5.3.5]{ST19} and, in fact, holds for a larger set of functions. 
The second assertion follows from \cite[Proposition 4.3.15]{ST19}. The proof relies on Theorems \ref{thm:coincidence} and Lemma \ref{lem:divided diff}.
\end{proof}

Thanks to this theorem, we can formally express the perturbation of an operator function in terms of operator integrals, via the Taylor series
\begin{align*}
	f(H+V)\sim\sum_{k=0}^\infty\frac{1}{k!}\frac{d^k}{dt^k}f(H+tV)\big|_{t=0}=\sum_{k=0}^\infty T^H_{f^{[k]}}(V,\ldots,V),
\end{align*}
thereby giving access to a wide variety of algebraic and analytic results on multiple operator integration. The formal expansion above can be made even more powerful by using Taylor remainders.
Given a function $f\in C^n(\R)$ satisfying $\widehat{f^{(k)}}\in L^1(\R)$, $k=1,\dots,n$, a self-adjoint operator $H$ in $\H$, and $V\in\mB(\H)_{\sa}$, we denote the $n^{\text{th}}$ Taylor remainder by
\begin{align}\label{eq:def of the remainder}	R_{n,H,f}(V):=f(H+V)-\sum_{k=0}^{n-1}\frac{1}{k!}\frac{d^k}{dt^k}f(H+tV)\big|_{t=0}.
\end{align}

Like the individual terms of the Taylor series, the remainder can also be expressed in terms of a multiple operator integral.
\begin{thm}
\label{rm}
Let $n\in\N$ and let $f\in C^n(\R)$ be such that $\widehat{f^{(k)}}\in L^1(\R)$, $k=1,\dots,n$.
Let $H$ be a self-adjoint operator in $\H$, let $V\in\mB(\H)_{\sa}$. We then have
\begin{align}
\label{remmoi}
R_{n,H,f}(V)&=T^{H+V,H,\ldots,H}_{f^{[n]}}(V,\ldots,V)\nonumber\\
&=T^{H,H+V,H,\ldots,H}_{f^{[n]}}(V,\ldots,V).
\end{align}
\end{thm}
\begin{proof}
By \cite[Theorem 3.3.8]{ST19} for $k=0$ and \cite[Theorem 4.3.14]{ST19} for
$k\ge 1$,
\begin{align}
\label{pf}
T^{H_0+V_0,H_1,\ldots,H_k}_{f^{[k]}}(V_1,\ldots,V_k)-T^{H_0,\ldots,H_k}_{f^{[k]}}(V_1,\ldots,V_k)
=T^{H_0+V_0,H_0,\ldots,H_k}_{f^{[k+1]}}(V_0,\ldots,V_k),
\end{align}
where $H_0,\ldots,H_k$ are self-adjoint operators in $\H$ and $V_0,\ldots,V_k\in\mB(\H)_{\sa}$. In particular,
\begin{align}\label{pf2}
	T^{H+V,H,\ldots,H}_{f^{[k]}}(V,\ldots,V)-T^{H,\ldots,H}_{f^{[k]}}(V,\ldots,V)
=T^{H+V,H,\ldots,H}_{f^{[k+1]}}(V,\ldots,V).
\end{align}
Combining \eqref{pf2} with \eqref{dermoi} and proceeding by induction on $k$ yields \eqref{remmoi}. The second equality follows similarly.
\end{proof}

\section{Finitely summable}
\label{sct:Finitely summable}

We specialize the class of functions $f$ that appear in Theorem \ref{rm} and consider for $s,n\in\N_0$:
\begin{align}\label{eq:Wsn}
	\Wsn:=\{f\in C^{n}:~ \widehat{(fu^m)^{(k)}}\in L^1\text{ for all $m=0,\ldots, s$ and $k=0,\ldots,n$}\},
\end{align}
where $u(x):=x-i$. Examples of functions in $\Wsn$ are $n+1$-differentiable functions such that $(fu^s)^{(k)}\in L^2$ for all $k\leq n+1$, such as Schwartz functions, or functions in $C_\text{c}^{n+1}$. 
\subsection{Bound on the multiple operator integral}\label{sct:core results}

In this section we will use the $s$-summability of $H$, as well as the above function class, to obtain a trace-class estimate on the multiple operator integral $\Tfn^H$, i.e., Theorem \ref{thm:Schatten estimate}. For summability $s=2$, a similar estimate was found by Anna Skripka in \cite[Lemma 3.6]{S14}. 


The core idea used in our proof is inspired by the proof of \cite[Lemma 3.6]{S14}, namely to expand $T_{f^{[n]}}(V_1,\ldots,V_n)$ as a sum of operator integrals, of which the trace norm can be bounded using (a noncommutative) H\"older's inequality. However, for general values of $s$, the expansion process needs to be repeated, and the increasingly complicated summands need to be controlled. As an intermediate step, we prove the following lemma.

\begin{lem}\label{lem:added one weight}
When $f\in C^n$ and $\widehat{(fu)^{(n)}},\widehat{f^{(n-1)}}\in L^1$, we have
\begin{align*}
	T^{H_0,\ldots,H_n}_{f^{[n]}}(V_1,\ldots,V_n)=&T^{H_0,\ldots,H_n}_{(fu)^{[n]}}(V_1,\ldots,V_n)(H-i)^{-1}\\
	&-T^{H_0,\ldots,H_{j-1},H_{j+1},\ldots,H_{n}}_{f^{[n-1]}}(V_1,\ldots,V_{n-1})V_n(H-i)^{-1}.
\end{align*}
\end{lem}
\begin{proof}
Since $u^{[1]}=\indicator_{\R^2}$ and $u^{[p]}=0$ for all $p\geq2$, the Leibniz rule for divided differences gives
\begin{align*}
(fu)^{[n]}(x_0,\ldots,x_n)=f^{[n]}(x_0,\ldots,x_n)u(x_n)+f^{[n-1]}(x_0,\ldots,x_{n-1}),
\end{align*}
hence,
\begin{align}\label{eq:divdiff Leibniz}
	f^{[n]}(x_0,\ldots,x_n)=(fu)^{[n]}(x_0,\ldots,x_n)u^{-1}(x_n)-f^{[n-1]}(x_0,\ldots,x_{n-1})u^{-1}(x_n).
\end{align}
By Lemma \ref{lem:divided diff}, the functions $(fu)^{[n]}$ and $f^{[n-1]}$ admit the representation \eqref{phi representation}. Hence, the function on the right-hand side of \eqref{eq:divdiff Leibniz} also admits the representation \eqref{phi representation}. Therefore, by Theorem \ref{thm:coincidence} applied
to $\phi=f^{[n]}$, $\phi=(fu)^{[n]}$, and $\phi=f^{[n-1]}$,
we obtain the lemma.
\end{proof}

For brevity, we write $V_k^{\{j\}}:=V_k(H_k-i)^{-j}$, and, similarly, $T^{H_0,\ldots,H_k}_\phi(V_1,\ldots,V_k)^{\{j\}}:=T^{H_0,\ldots,H_k}_\phi(V_1,\ldots,V_k)(H_k-i)^{-j}$.

\begin{prop}\label{prop:added weights}
	For $s,n\in\N_0$, $f\in\Wsn$, $V_1,\ldots,V_n\in\mB(\H)$, and $H_0,\ldots,H_n$ self-adjoint in $\H$, we have
	\begin{align*}
		 T^{H_0,\ldots,H_n}_{f^{[n]}}(V_1,\ldots,V_n)=\!
		\sum_{k=0}^{\min(s,n)}\!\!(-1)^k\!\!\!\!\!\!\sum_{\substack{ j_0\geq0,\, j_1,\ldots,j_k\geq 1,\\ j_0+\ldots+j_k=s}}\!\!\!\!\!
		 T^{H_0,\ldots,H_{n-k}}_{(fu^{s-k})^{[n-k]}} (V_1,\ldots,V_{n-k})^{\{j_0\}}V_{n-k+1}^{\{j_1\}}\cdots V_n^{\{j_k\}}.
	\end{align*}
\end{prop}
\begin{proof}
	Let $n\in\N_0$ be fixed. We prove the proposition by induction on $s$. If $s=0$, the statement follows directly. Now suppose the claim of the proposition holds for a certain $s\in\N_0$, i.e., we have the displayed formula above. To each of its terms, we can apply Lemma \ref{lem:added one weight}, namely
		$$T^{H_0,\ldots,H_n}_{f^{[n]}}(V_1,\ldots,V_n)^{\{j_0\}}=T^{H_0,\ldots,H_n}_{(fu)^{[n]}}(V_1,\ldots,V_n)^{\{j_0+1\}}-T^{H_0,\ldots,H_{n-1}}_{f^{[n-1]}}(V_1,\ldots,V_{n-1})V_n^{\{j_0+1\}},$$
 for all $f\in\W^n_{s+1}$. We obtain
	\begin{align*}
		&T^{H_0,\ldots,H_n}_{f^{[n]}}(V_1,\ldots,V_n)\\
		&\quad=\sum_{k=0}^{\min(s,n)}(-1)^k\!\sum_{\substack{j_0\geq0,\, j_1,\ldots,j_k\geq 1\\ j_0+\ldots+j_k=s}}  T^{H_0,\ldots,H_{n-k}}_{(fu^{s-k+1})^{[n-k]}}(V_1,\ldots,V_{n-k})^{\{j_0+1\}}V_{n-k+1}^{\{j_1\}}\cdots V_n^{\{j_k\}}\\
		&\qquad+\sum_{k=0}^{\min(s,n-1)}(-1)^{k+1}\!\!\sum_{\substack{j_0\geq0,\, j_1,\ldots,j_k\geq 1\\ j_0+\ldots+j_k=s}}T^{H_0,\ldots,H_{n-k-1}}_{(fu^{s-k})^{[n-k-1]}}(V_1,\ldots,V_{n-k-1})V_{n-k}^{\{j_0+1\}}V_{n-k+1}^{\{j_1\}}\cdots V_n^{\{j_k\}}\\
		&\quad=\sum_{k=0}^{\min(s,n)}(-1)^k\!\sum_{\substack{j_0\geq1,\,j_1,\ldots,j_k\geq1\\j_0+\ldots+j_k=s+1}} T^{H_0,\ldots,H_{n-k}}_{(fu^{s+1-k})^{[n-k]}}(V_1,\ldots,V_{n-k})^{\{j_0\}}V_{n-k+1}^{\{j_1\}}\cdots V_n^{\{j_k\}}\\
		&\qquad+\sum_{k=1}^{\min(s+1,n)}(-1)^k\!\sum_{\substack{j_0=0,\,j_1,\ldots,j_k\geq1\\ j_0+\ldots+j_k=s+1}}T^{H_0,\ldots,H_{n-k}}_{(fu^{s+1-k})^{[n-k]}}(V_1,\ldots,V_{n-k})^{\{j_0\}}V_{n-k+1}^{\{j_1\}}\cdots V_n^{\{j_k\}}.
	\end{align*}
	In the first term, instead of letting $k$ run from $0$ to $\min(s,n)$, we can freely let $k$ run from $0$ to $\min(s+1,n)$, because the adjacent sum over $j_0,\ldots,j_k$ is trivial for $k=s+1$. Similarly, in the second term, we can freely let $k$ run from 0 to $\min(s+1,n)$. Combining the two terms gives the claim of the lemma for $s+1$, which completes the induction step.
\end{proof}

\begin{rema}
	As is done in \cite{S14} to handle the case $s=2$, one could use the real weight $\tilde u(x):=\sqrt{x^2+1}$ instead of the complex weight $u(x)=x-i$ to obtain a version of Proposition \ref{prop:added weights}. However, because $\tilde{u}^{[2]}\neq 0$, the obtained summands will become horribly convoluted, and the results do not seem to be as strong as when using $u$.
\end{rema}

Thanks to Proposition \ref{prop:added weights} we can now prove the main result of this section, which is vital to Chapter \ref{ch:Cyclic cocycles in the spectral action}. 
\begin{thm}\label{thm:Schatten estimate}
	Let $H$ be self-adjoint in $\H$ such that $(H-i)^{-1}\in\S^s$ for $s\in\N$. For every $n\in\N_0$, every $f\in\Wsn$ and every $V_1,\ldots,V_n\in \B(\H)$, the multiple operator integral $T^H_{f^{[n]}}(V_1,\ldots,V_n)$ is trace-class and satisfies the bound
		$$\nrm{T^H_{f^{[n]}}(V_1,\ldots,V_n)}{1}\leq c_{s,n}(f)\norm{V_1}\cdots\norm{V_n}\big\|(H-i)^{-1}\big\|_{s}^s,$$
	where
		$$c_{s,n}(f):=\sum_{k=0}^{\min(s,n)}\vect{s}{k}\frac{\absnorm{\widehat{(fu^{s-k})^{(n-k)}}}}{(n-k)!}.$$
	More generally, when $V\in\mB(\H)_\sa$,
		$$\nrm{T^{H+V,H,\ldots,H}_{f^{[n]}}(V_1,\ldots,V_n)}{1}\leq c_{s,n}(f)\norm{V_1}\cdots\norm{V_n}(1+\norm{V})^{s}\big\|(H-i)^{-1}\big\|_{s}^s.$$
\end{thm}
\begin{proof}
	We apply Proposition \ref{prop:added weights}, and find
	\begin{align*}
		&\nrm{T^H_{f^{[n]}}(V_1,\ldots,V_n)}{1}\\
		&\quad\leq\sum_{k=0}^{\min(s,n)}\sum_{\substack{ j_0\geq0,\, j_1,\ldots,j_k\geq 1,\\ j_0+\ldots+j_k=s}}
		\nrm{ T^H_{(fu^{s-k})^{[n-k]}} (V_1,\ldots,V_{n-k})^{\{j_0\}}}{\frac{s}{j_0}} \nrm{V_{n-k+1}^{\{j_1\}}}{\frac{s}{j_1}}\cdots \nrm{V_n^{\{j_k\}}}{\frac{s}{j_k}}\!\!.
	\end{align*}
	Apply \eqref{fourierbound}, to find
	\begin{align*}
		\nrm{T^H_{f^{[n]}}(V_1,\ldots,V_n)}{1}\leq\sum_{k=0}^{\min(s,n)}\sum_{\substack{ j_0\geq0,\, j_1,\ldots,j_k\geq 1,\\ j_0+\ldots+j_k=s}}\frac{\nrm{\widehat{(fu^{s-k})^{(n-k)}}}{1}}{(n-k)!}\norm{V_1}\cdots\norm{V_n}\|(H-i)^{-1}\|_{s}^s.
	\end{align*}
	A bit of combinatorics shows that the sum over $j_0,\ldots,j_k$ adds a factor $\vect{s}{k}$, which implies the first statement of the theorem. The second statement follows similarly, with the added remark that
		$$\|(H+V-i)^{-1}\|_{s}^s\leq(1+\norm{V})^{s}\|(H-i)^{-1}\|_{s}^s.$$
	This inequality follows from the second resolvent identity. For more specific bounds see \cite[Appendix B, Lemma 6]{CP}.
\end{proof}

\subsection{Continuity of the multiple operator integral}
\label{sct:MOI continuity}
A second application of Proposition \ref{prop:added weights} is the following strong continuity property.

\begin{thm}\label{thm:continuity for L's}
Let $s\in\N$, $H$ self-adjoint in $\H$ with $(H-i)^{-1}\in\S^s$, $n\in\N_0$, and $f\in\Wsn$. The map $$T^H_{f^{[n]}}:\L^\infty\times\cdots\times\L^\infty\to\L^1$$ is continuous. Recall here that $\L^1=\S^1$, and that $\L^\infty=\B(\H)_1$, endowed with the strong operator topology.
\end{thm}
\begin{proof}
	Suppose that 
	$V^m_1\to V_1,\ldots,V^m_n\to V_n$ in $\L^\infty$.
	By Lemma \ref{lem:Schatten product is continuous}, we obtain that 
		$$(V^m_{n-k+l})^{\{j_l\}}\to V_{n-k+l}^{\{j_l\}}\quad\text{in $\L^{s/j_l}$.}$$
	We invoke Lemma \ref{lem:continuity with alpha relation} to find that
		$$T^H_{(fu^{s-k})^{[n-k]}}(V^m_1,\ldots,V^m_{n-k})^{\{j_0\}}\to T^H_{(fu^{s-k})^{[n-k]}}(V_1,\ldots,V_{n-k})^{\{j_0\}}\quad\text{in $\L^{s/j_0}$.}$$
	By Proposition \ref{prop:added weights} and Lemma \ref{lem:Schatten product is continuous}, we find that
		$$T^H_{f^{[n]}}(V^m_1,\ldots,V^m_n)\to T^H_{f^{[n]}}(V_1,\ldots,V_n)\quad\text{in $\L^{1}$},$$
	so we are done.
\end{proof}

To emphasize the strength of this result, we compare it to Lemma \ref{lem:continuity with alpha relation} which was already known (at least in the cases $\alpha_1,\ldots,\alpha_n<\infty$ and $\alpha=\infty$). By applying the continuity of the inclusion $\L^\alpha\hookrightarrow\L^\beta$ ($\alpha<\beta$ in $[1,\infty]$) to Theorem \ref{thm:continuity for L's} we obtain the following clear improvement of Lemma \ref{lem:continuity with alpha relation}.
\begin{cor}
	Let $s\in\N$, $H$ self-adjoint in $\H$ with $(H-i)^{-1}\in\S^s$, $n\in\N_0$, and $f\in\Wsn$. For any $\alpha\in[1,\infty]$ and any $\alpha_1,\ldots,\alpha_n\in[1,\infty]$ (no relation between $\alpha$ and the $\alpha_j$'s assumed) the map
		$$\Tfn^H:\L^{\alpha_1}\times\cdots\times\L^{\alpha_n}\to\L^\alpha$$
	is continuous.
\end{cor}
\subsection{Taylor series in terms of divided differences}
\label{sct:Taylor divdiff}

For a self-adjoint operator $H$ in $\H$ with compact resolvent, we let $\varphi_1,\varphi_2,\ldots$ be an orthonormal basis of eigenvectors of $H$, with corresponding eigenvalues $\lambda_1,\lambda_2,\ldots$.
When $n\in\N$ and
\begin{align*}
	V_1,\ldots,V_n\in\spn\{\ket{\varphi_i}\bra{\varphi_j}:~i,j\in\N\},
\end{align*}
we have, by Definition \ref{def:nosep}, the finite sum
\begin{align}
	T_{f^{[n]}}^H(V_1,\ldots,V_n)=\sum_{i_0,\ldots,i_n\in\N}f^{[n]}(\lambda_{i_0},\ldots,\lambda_{i_n})(V_1)_{i_0i_1}\cdots(V_n)_{i_{n-1}i_n}\ket{\varphi_{i_0}}\bra{\varphi_{i_n}},\label{eq:Trace function divdiff}
\end{align}
where $W_{kl}:=\p{\varphi_k}{W\varphi_l}$ denote the matrix elements of $W$. In particular, assuming that $V_1=V_2=\ldots=V_n\equiv V$, and taking the trace, a standard computation (cf. \cite{Sui11,Sui15}) gives
\begin{align}
	\frac{1}{n!}\frac{d^n}{dt^n}\Tr(f(H+tV))\big|_{t=0}&=\Tr(T_{f^{[n]}}^H(V,\ldots,V))\nonumber\\
	&=\sum_{i_1,\ldots,i_n\in\N}f^{[n]}(\lambda_{i_1},\ldots,\lambda_{i_n},\lambda_{i_1})V_{i_1i_2}\cdots V_{i_{n-1}i_n} V_{i_{n}i_1}\nonumber\\
	&=\frac{1}{n}\sum_{i_1,\ldots,i_n\in\N}(f')^{[n-1]}(\lambda_{i_1},\ldots,\lambda_{i_n})V_{i_1i_2}\cdots V_{i_{n-1}i_n}V_{i_{n}i_1}.\label{eq:SA divdiff}
\end{align}
This formula appears in \cite[Corollary 3.6]{Hansen} and, in higher generality, in \cite[Theorem 18]{Sui11}.
The formula \eqref{eq:SA divdiff} gives a very concrete way to calculate derivatives of the spectral action, as well as calculate the Taylor series of a perturbation of the spectral action. 
One needs to be careful, however, when applying this formula in a general setting. When the perturbation $V$ is not of finite rank, writing \eqref{eq:SA divdiff} as a sum over $i_1,\ldots,i_n$ is misleading, as the series is often not absolutely convergent and there is no reason for a Fubini theorem to hold. The best way to generalize \eqref{eq:SA divdiff} is arguably by using the machinery we developed in Section \ref{sct:Finitely summable}.
\begin{thm}\label{thm:Taylor expansion divdiffs}
	For $n,s\in\N$, $H$ self-adjoint in $\H$ with $(H-i)^{-1}\in\S^s$, $V\in\mB(\H)$, $f\in\Wsn$, and $\{\varphi_i\}_{i\in\N}$ an orthonormal basis of eigenvectors of $H$ with corresponding eigenvalues $\{\lambda_i\}_{i\in\N}$, we have
	\begin{align*}
		\frac{1}{(n-1)!}\frac{d^n}{dt^n}\Tr(f(H+tV))\big|_{t=0}=\lim_{N\to\infty}\sum_{i_1,\ldots, i_n<N}(f')^{[n-1]}(\lambda_{i_1},\ldots,\lambda_{i_n})V_{i_1i_2}\cdots V_{i_ni_1}.
	\end{align*}
\end{thm}
\begin{proof}
	Write $E^N:=\sum_{i<N}\ket{\varphi_i}\bra{\varphi_i}$, and notice that $E^N\to 1$ strongly. Defining
	\begin{align*}
		V^N:=E^NVE^N,
	\end{align*}
	we obtain $V^N\to V$ strongly by, for instance, Lemma \ref{lem:Schatten product is continuous}. By Theorem \ref{thm:continuity for L's} and \eqref{eq:SA divdiff}, we find
	\begin{align*}
		\Tr(T_{f^{[n]}}^H(V,\ldots,V))&=\lim_{N\to\infty}\Tr(T_{f^{[n]}}^H(V^N,\ldots,V^N))\\
		&=\lim_{N\to\infty}\frac{1}{n}\sum_{i_1,\ldots,i_n\in\N}(f')^{[n-1]}(\lambda_{i_1},\ldots,\lambda_{i_n})V^N_{i_1i_2}\cdots V^N_{i_{n}i_1}\\
		&=\lim_{N\to\infty}\frac{1}{n}\sum_{i_1,\ldots,i_n<N}(f')^{[n-1]}(\lambda_{i_1},\ldots,\lambda_{i_n})V_{i_1i_2}\cdots V_{i_{n}i_1},
	\end{align*}
	which implies the theorem.
\end{proof}
A similar argument can be used to generalize \eqref{eq:Trace function divdiff} to arbitrary perturbations. 

\section{Relative Schatten}
\label{sct:Relative Schatten}
The relative Schatten case, being a generalization of the finitely summable case, is slightly more subtle. However, we can still obtain an analogue of Proposition \ref{prop:added weights}, namely Theorem \ref{thm:adding resolvents}.
%
This theorem will be used throughout Chapter \ref{ch:Spectral Shift for Relative Schatten Perturbations}, in particular to apply the bound from Theorem \ref{thm:PSS Higher Order} to the relative Schatten case, in which the perturbation $V$ is generally noncompact.

\begin{thm}\label{thm:adding resolvents}
Let $n\in\N$, let $H_0,\ldots,H_n$ be self-adjoint in $\H$, and let $V_1,\ldots,V_n\in\mB(\H)$.
\begin{enumerate}[label=\textnormal{(\roman*)}]	
	\item\label{weight}
For each $j\in\{0,\ldots,n\}$ and every $f\in C^n$ satisfying $\widehat{(fu)^{(n)}},\widehat{f^{(n-1)}}\in L^1$ we have
\begin{align*}
T^{H_0,\ldots,H_n}_{f^{[n]}}(V_1,\ldots,V_n)
=&T^{H_0,\ldots,H_n}_{(fu)^{[n]}}(V_1,\ldots,V_j(H_j-i)^{-1},V_{j+1},\ldots,V_n)
\\&-T^{H_0,\ldots,H_{j-1},H_{j+1},\ldots,H_n}_{f^{[n-1]}}(V_1,\ldots,V_j(H_j-i)^{-1}V_{j+1},\ldots,V_n).
\end{align*}

\item\label{weights} Denoting $\tilde{V}_{j,l}:=V_{j+1}(H_{j+1}-i)^{-1}\cdots V_l(H_l-i)^{-1}$, we have$$T^{H_0,\ldots,H_n}_{f^{[n]}}(V_1,\ldots,V_n)=\sum_{p=0}^n(-1)^{n-p}\!\!\!\sum_{0<j_1<\cdots<j_{p}\leq n}\! T^{H_0,H_{j_1},\ldots,H_{j_p}}_{(fu^p)^{[p]}}
(\tilde{V}_{0,j_1},\ldots,\tilde{V}_{j_{p-1},j_p})\,
\tilde{V}_{j_{p},n}\,,$$
for every $f\in C^n$ satisfying $\widehat{(fu^p)^{(p)}}\in L^1$ for every $p=0,\ldots,n$.

\item\label{trace class} If $V_k(H_k-i)^{-1}\in\S^n$ for every $k=1,\dots,n$, then
$$T^{H_0,\ldots,H_n}_{f^{[n]}}(V_1,\ldots,V_n)\in\S^1,$$
for every $f\in C^n$ satisfying $\widehat{(fu^p)^{(p)}}\in L^1$ for every $p=0,\ldots,n$.
\end{enumerate}
\end{thm}
\begin{proof}
	If, in \eqref{eq:divdiff Leibniz}, we swap $x_n$ and $x_j$ (for any $j\in\{0,\ldots,n\}$), we obtain, by symmetry of the divided difference,
\begin{align}\label{eq:adding one weight}
f^{[n]}(x_0,\ldots,x_n)
=&(fu)^{[n]}(x_0,\ldots,x_n)u^{-1}(x_j)\\
\nonumber
&-f^{[n-1]}(x_0,\ldots,x_{j-1},x_{j+1},\ldots,x_n)u^{-1}(x_j).
\end{align}
Applying \eqref{eq:adding one weight} repeatedly, similar to the proof of Proposition \ref{prop:added weights}, we obtain
\begin{align*}
	f^{[n]}(x_0,\ldots,x_n)=\sum_{p=0}^s (-1)^{s-p}\!\!\!\!\sum_{n-s<j_1<\cdots<j_p\leq n}&(fu^p)^{[n-s+p]}(x_0,\ldots,x_{n-s},x_{j_1},\ldots,x_{j_p})\\
	&\cdot u^{-1}(x_{n-s+1})\cdots u^{-1}(x_n)
\end{align*}
by induction to $s\in\{0,\ldots,n\}$. Taking $s=n$, we find
\begin{align}	
\label{eq:ddweights}	
f^{[n]}(x_0,\ldots,x_n)
=\sum_{p=0}^n(-1)^{n-p}\!\!\!\!\sum_{0<j_1<\cdots<j_{p}\leq n}
(fu^p)^{[p]}(x_0,x_{j_1},\ldots,x_{j_p})
u^{-1}(x_1)\cdots u^{-1}(x_n).
\end{align}
To prove \ref{weight}, we use Lemma \ref{lem:divided diff} to see that the functions $f^{[n-1]}$ and $(fu)^{[n]}$ admit the representation \eqref{phi representation}. Hence, the function on the right-hand side of \eqref{eq:adding one weight} also admits the representation \eqref{phi representation}. Therefore, by Theorem \ref{thm:coincidence} applied
to $\phi=f^{[n]}$, $\phi=f^{[n-1]}$, and $\phi=(fu)^{[n]}$,
we obtain \ref{weight}.
Similarly, applying Theorem \ref{thm:coincidence} and Lemma \ref{lem:divided diff} to \eqref{eq:ddweights} gives \ref{weights}.
	Corollary \ref{cor:nice def of MOI} shows that the right-hand side of \ref{weights} is trace-class, which gives \ref{trace class}.
\end{proof}

\begin{rema}
Although the condition $V(H-i)^{-1}\in\S^n$ is equivalent to $V(H^2+\1)^{-1/2}\in\S^n$, we made use of the complex weight $u(x)=x-i$ rather than the real weight $\tilde u(x)=\sqrt{x^2+1}$ because there is no suitable analog of Theorem \ref{thm:adding resolvents} for the latter. For instance, an analog of \eqref{eq:adding one weight} for $\tilde{u}$ with $n=4$ and $j=1$ contains terms like
\begin{align}
f^{[2]}(x_0,x_2,x_4)\,\tilde{u}^{[2]}(x_1,x_2,x_3)\,\tilde{u}^{-1}(x_1).
\end{align}
The latter term does not allow a separation of variables like in \eqref{eq:ddweights} which enables us to write the result in terms of multiple operator integrals.
\end{rema}

An immediate corollary of Theorem \ref{thm:adding resolvents} is given by applying Corollary \ref{cor:nice def of MOI} (but one could also apply Theorem \ref{thm:Schatten estimate} here).
\begin{cor}
	We have
	\begin{align*}
		\absnorm{T_{f^{[n]}}^{H_0,\ldots,H_n}(V_1,\ldots,V_n)}\leq \frac{1}{n!} \sum_{p=0}^n \vect{n}{p}\absnorm{\widehat{(fu^p)^{(p)}}}\nrm{V_1(H_1-i)^{-1}}{n}\cdots\nrm{V_n(H_n-i)^{-1}}{n}.
	\end{align*}
\end{cor}
In practice, one might want to estimate $T_{f^{[n]}}^H(V,\ldots,V)$ when $V(H-i)^{-1}\in\S^s$ for $s\geq n$. This can easily be done, since then $V(H-i)^{-1}\in\S^n$ and we can use H\"older's inequality to find
\begin{align*}
	\absnorm{T^H_{f^{[n]}}(V,\ldots,V)}\leq \frac{1}{n!} \sum_{p=0}^n \vect{n}{p}\absnorm{\widehat{(fu^p)^{(p)}}}\norm{V}^{n-s}\nrm{V(H-i)^{-1}}{s}^s.
\end{align*}

It may be clear that for similar assumptions (resolvent comparability, local compactness,  possibly using regularity, weak Schatten classes, von Neumann algebras, etcetera) similar change of variables formulas like Proposition \ref{prop:added weights} and Theorem \ref{thm:adding resolvents} can be derived. It may also be clear that Theorem \ref{thm:adding resolvents} has many more applications than the ones given here, like continuity properties similar to the ones of \textsection\ref{sct:MOI continuity}. This chapter will not pursue this further, trusting that by now, the reader has already absorbed the necessary techniques to obtain the best results available in their specific context.
By the same philosophy, the results and techniques of this chapter will now be used, in Chapter \ref{ch:Spectral Shift for Relative Schatten Perturbations}, in the context of the spectral shift function.

\chapter{Spectral Shift Function for Relative Schatten Perturbations}
\label{ch:Spectral Shift for Relative Schatten Perturbations}
\chaptermark{Spectral Shift Function for Relative Schatten}

This chapter, adapted from \cite{vNS21}, affirmatively settles the question on existence of a real-valued higher-order spectral shift function for relative Schatten class perturbations. Besides showing that the spectral shift function satisfies the same trace formula as in the known case of $V\in\S^n$, we show that it is unique up to a polynomial summand of order $n-1$. Our results significantly advance earlier partial results where counterparts of the spectral shift function for noncompact perturbations lacked real-valuedness and aforementioned uniqueness as well as appeared in more complicated trace formulas for much more restrictive sets of functions.
Our result applies to models arising in noncommutative geometry and mathematical physics.

Results in this chapter were obtained in collaboration with Anna Skripka.

\section{Introduction}
\label{sec1}

The spectral shift function originates from the foundational work \cite{Krein53} of M.G.~Krein which followed I.M.~Lifshits's physics research summarized in \cite{Lifshits}.
It is a central object in perturbation theory that allows to approximate a perturbed operator function by the unperturbed one, while controlling noncommutativity in the remainder.
In \cite{Koplienko84}, Koplienko suggested an interesting and useful generalization by considering higher-order Taylor remainders and conjecturing existence of higher-order spectral shift functions. Many partial results were obtained in that direction, but they were confined to either lower order approximations, weakened trace functionals and representations, or compact perturbations.
This chapter closes a gap between theory and applications, where perturbations are often noncompact, by proving existence of a higher-order spectral shift function under the relative Schatten class condition and obtaining bounds and properties stricter than previously known.

Our prime result is that, given a self-adjoint operator $H$ densely defined in a separable Hilbert space $\H$ and a bounded self-adjoint operator $V$ on $\H$ satisfying
\begin{align}
\label{vresinsn}
V(H-i)^{-1}\in\S^n,
\end{align}
there exists a real-valued spectral shift function $\eta_n=\eta_{n,H,V}$ of order $n$. Namely, the trace formula
\begin{align}
\label{trvinsn}
\Tr\left(f(H+V)-\sum_{k=0}^{n-1}\frac{1}{k!}\frac{d^k}{dt^k} f(H+tV)\big|_{t=0}\right)
=\int_\R f^{(n)}(x)\,\eta_n(x)\,dx
\end{align}
holds for a wide class of functions $f$ and the function $\eta_n$ satisfies suitable uniqueness and summability properties and bounds, as detailed below.
The relative Schatten class condition \eqref{vresinsn} applies, in particular, to
\begin{enumerate}[label=\textnormal{(\Roman*)}]
	\item\label{vinsn} $V\in\S^n;$
	\item\label{rinsn} $(H-i)^{-1}\in\S^n$;
	\item\label{inner perturbations} inner fluctuations of $H=D$ in a regular locally compact spectral triple $(\mathcal{A},\H,D)$ (see Section \ref{sec4a});
	\item\label{diff operators} differential operators on manifolds perturbed by multiplication operators (see Section \ref{sec4b}).
\end{enumerate}


\noindent Under the assumption \ref{vinsn}, the problem on existence of higher-order spectral shift functions has been resolved in \cite{PSS13}.
Namely, \eqref{trvinsn} was established in \cite{Krein53}, \cite{Koplienko84}, \cite{PSS13} for $n=1$, $n=2$, $n\ge 3$, respectively, for different classes of test functions $f$ (see, e.g., \cite[Section 5.5]{ST19} for details), where the function $\eta_n=\eta_{n,H,V}$ is unique, real-valued, and satisfies the bound
\begin{align*}
\|\eta_n\|_1\le c_n\|V\|_n^n.
\end{align*}

Taylor approximations and respective trace formulas were also derived in the study of the spectral action functional $\Tr(f(H))$ occurring in noncommutative geometry \cite{CC97} for operators $H$ with compact resolvent $(H-i)^{-1}$. The case of \ref{rinsn} and functions $f$ in the form $f(x)=g(x^2)$, where $g$ is the Laplace transform of a regular Borel measure, was investigated in \cite{Sui11}. The case of compact $(H-i)^{-1}$ and
$f\in C_\text{c}^{n+1}(\R)$ was handled in \cite{S14,S18}. In particular, the existence of a locally integrable spectral shift function was established in \cite{S18}.

In this chapter we generalize the results in the cases \ref{vinsn} and \ref{rinsn} while including the interesting cases \ref{inner perturbations} and \ref{diff operators} without significant compromises.

\paragraph*{Precise formulation of the result.}
In our main result, Theorem \ref{rsmain}, given $n\in\N$ and $H,V$ satisfying \eqref{vresinsn}, we establish the existence of a real-valued function $\eta_n=\eta_{n,H,V}$ such that $\eta_n\in L^1\big(\R,\tfrac{dx}{(1+|x|)^{n+\epsilon}}\big)$ for every $\epsilon>0$ and such that \eqref{trvinsn} holds for every $f\in\mW_n$, where the class $\mW_n$ is given by Definition \ref{fracw}. In particular, $\mW_n$ includes all $(n+1)$-times continuously differentiable functions whose derivatives decay at infinity at the rate $f^{(k)}(x)=\O\left(|x|^{-k-\alpha}\right)$, $k=0,\ldots,n+1$, for some $\alpha>\frac12$ (see Proposition \ref{inclusions}\ref{inclusionsi}). The weighted $L^1$-norm of the spectral shift function $\eta_n$ admits the bound
\begin{align*}
\int_\R|\eta_n(x)|\,\frac{dx}{(1+|x|)^{n+\epsilon}}
\le c_n(1+\epsilon^{-1})(1+\norm{V})\|V(H-i)^{-1}\|_n^n
\end{align*}
for every $\epsilon>0$. Moreover, the locally integrable spectral shift function $\eta_n$ is unique up to a polynomial summand of degree at most $n-1$.

Below we briefly summarize advantages of our main result in comparison to most relevant prior results.
Other results on approximation of operator functions and omitted details can be found in \cite[Chapter 5]{ST19} and references cited therein.

\paragraph*{Prior results.}
The existence of a real-valued function $\eta_1\in L^1\big(\R,\tfrac{dx}{1+x^2}\big)$ satisfying the trace formula \eqref{trvinsn} with $n=1$ for bounded rational functions was established in \cite[Theorem~3]{Krein62} (see also \cite[p.~48, Corollary 0.9.5]{Yafaev10}). The formula \eqref{trvinsn} was extended to twice-differentiable $f$ with bounded $f',f''$ such that
\begin{align}
\label{yafaevcond}
\frac{d^k}{dx^k}(f(x)-c_f x^{-1})
=\O(|x|^{-k-1-\epsilon})\quad\text{ as }\,\;|x|\rightarrow\infty,\quad k=0,1,2,\quad \epsilon>0,
\end{align}
where $c_f$ is a constant, in \cite[p.~47, Theorem 0.9.4]{Yafaev10}.
The respective function $\eta_1$ was determined by \eqref{trvinsn} uniquely up to a constant summand.
We prove that \eqref{trvinsn} with $n=1$ holds for all $\mathfrak{W}_1$, which contains all functions satisfying \eqref{yafaevcond} (see Proposition \ref{inclusions}\ref{inclusionsi}) as well as functions not included in \eqref{yafaevcond} (see, e.g., Remark \ref{mncontains}).
Moreover, we prove that $\eta_1$ is integrable with a smaller weight, namely $\eta_1\in L^1\big(\R,\frac{dx}{(1+|x|)^{1+\epsilon}}\big)$ for $\epsilon>0$. Thus, the results of \cite[Theorem~3]{Krein62} and \cite[p.~47, Theorem 0.9.4]{Yafaev10} are strengthened by our Theorem \ref{rsmain} in the case \eqref{vresinsn}.

In \cite[Corollary 3.7]{Neidhardt88}, the trace formula \eqref{trvinsn} with $n=2$ and a real-valued $\eta_2\in L^1\big(\R,\tfrac{dx}{(1+x^2)^2}\big)$ was proved for a set of functions including Schwartz functions
along with $\spn\{(z-\cdot)^{-k}:\, \Im(z)\neq 0,\, k\in\N,\, k\ge 2\}$.
The respective $\eta_2\in L^1\big(\R,\tfrac{dx}{(1+x^2)^2}\big)$ was determined by \eqref{trvinsn} uniquely up to a linear summand.
We prove that \eqref{trvinsn} with $n=2$ holds for all $f\in\mW_2$, which contains the functions $(z-\cdot)^{-1}$, $\Im(z)\neq 0$ not included in \cite[Corollary 3.7]{Neidhardt88} and the Schwartz functions included in \cite[Corollary 3.7]{Neidhardt88}, and that
$\eta_2$ is integrable with a significantly smaller weight, namely, $\eta_2\in L^1\big(\R,\frac{dx}{(1+|x|)^{2+\epsilon}}\big)$ for $\epsilon>0$.

Let $n\ge 2$. The existence of a complex-valued $\tilde\eta_n\in L^1\big(\R,\tfrac{dx}{(1+x^2)^{n/2}}\big)$ satisfying the trace formula
\begin{align}
\label{trcs18}
\Tr\Big(f(H+V)-\sum_{k=0}^{n-1}\frac{1}{k!}\frac{d^k}{dt^k}f(H+tV)|_{t=0}\Big)
=\int_\R\frac{d^{n-1}}{dx^{n-1}}\big((x-i)^{2n}f'(x)\big)\tilde\eta_n(x)\,dx
\end{align}
for a set of functions $f$ including $\spn\{(z-\,\cdot)^{-k},\; \Im(z)>0,\; k\in\N,\; k\ge 2n\}$ was established in \cite[Theorem 4.6]{CS18} (see also \cite[Remark 4.8(ii)]{CS18}). The weighted $L^1$-norm of $\tilde\eta_n$ satisfies the bound
\begin{align*}
\int_\R|\tilde\eta_n(x)|\,\frac{dx}{(1+x^2)^{\frac{n}{2}}}
\le c_n(1+\|V\|)^{n-1}\|V(H-i)^{-1}\|_n^n.
\end{align*}

As distinct from the aforementioned result of \cite{CS18} for $n\ge 2$, the function $\eta_n$ in our main result is real-valued and satisfies the  simpler trace formula \eqref{trvinsn} for the larger class $\mW_n$ of functions $f$ described in terms of familiar function classes. Moreover, the set of functions $\mW_n$ is large enough to ensure the uniqueness of $\eta_n$ up to a polynomial term of degree at most $n-1$.
\medskip

Other assumptions on $H$ and $V$, all having their own merits and limitations, were also considered in the literature. For instance, the existence of a nonnegative function $\eta_2=\eta_{2,H,V}\in L^1\big(\R,\tfrac{dx}{(1+x^2)^\gamma}\big)$, $\gamma>1/2$, satisfying the trace formula \eqref{trvinsn} with $n=2$ for bounded rational functions $f$ was established in \cite[Theorem 2]{Koplienko84} under the assumption $V|H-i|^{-\frac12}\in\S^2$. A more relaxed condition $(H+V-i)^{-1}-(H-i)^{-1}\in\S^n$ was traded off for a more restrictive set of functions $f$ and, when $n\ge 2$, for more complicated trace formulas where both the left and right-hand sides of \eqref{trvinsn} are modified. The respective results for $n=1$ can be found in \cite[Theorem 3]{Krein62} and \cite[Theorem 2.2]{Yafaev05}; for $n=2$ in \cite[Theorem 3.5, Corollary 3.6]{Neidhardt88}; for $n\ge 2$ in \cite[Theorem 3.5]{PSS15} and \cite{S17}.


\paragraph*{Methods.}
The technical scheme leading to the representation \eqref{trvinsn} under the assumption \eqref{vresinsn} is more subtle
than the one under the assumption \ref{vinsn}. The derivatives and Taylor approximations of operator functions are known to be expressible in terms of multiple operator integrals (see Theorems \ref{dm} and \ref{rm}). The prime technique to handle these multiple operator integrals (see Theorem \ref{thm:PSS Higher Order}) only applies to compact perturbations satisfying \ref{vinsn}. To bridge the gap between existing results for \ref{vinsn} and our setting \eqref{vresinsn} we impose suitable weights on the perturbations and involve multi-stage approximation arguments for functions and perturbations.

In Theorem \ref{thm:adding resolvents 2} we create Schatten class perturbations out of relative Schatten class perturbations \eqref{vresinsn} inside a multiple operator integral whose integrand is the $n$th order divided difference $f^{[n]}$ of a function $f\in C^n(\R)$ satisfying the properties $f^{(k)}(x)=o(|x|^{-k})$ as $|x|\rightarrow\infty$, $k=0,\ldots,n$, and $\widehat{f^{(n)}}\in L^1(\R)$.
Our Theorem \ref{thm:adding resolvents 2} significantly generalizes and extends earlier attempts in that direction made in \cite[Lemma 3.6]{S14}, \cite[Proposition 2.7]{S18}, \cite[Lemma 4.1]{CS18}.
The proof of Theorem \ref{thm:adding resolvents 2} involves the introduction of novel function classes (see Definition \ref{fracw}, \eqref{mBn}, and \eqref{mbn}), approximation arguments (see Lemma \ref{lem:density}), and analysis of multilinear operator integrals.

Based on the aforementioned results and analysis of distributions,
in Proposition \ref{prop:SSM with growth} we establish the trace formula
\begin{align}
\label{trexist}
\Tr\left(f(H+V)-\sum_{k=0}^{n-1}\frac{1}{k!}\frac{d^k}{dt^k} f(H+tV)\big|_{t=0}\right)
=\int_\R f^{(n)}(x)\,d\mu_n(x)
\end{align}
for every $f\in\mW_n$, where $\mu_n$ is a Borel measure determined uniquely up to an absolutely continuous term whose density is a polynomial of degree at most $n-1$ and such that for every $\epsilon>0$ the measure $(x-i)^{-n-\epsilon}\,d\mu_n(x)$ is finite and satisfies
\begin{align}
\|(\cdot-i)^{-n-\epsilon}\,d\mu_n\|\le c_n\,(1+\epsilon^{-1})(1+\norm{V})\nrm{V(H-i)^{-1}}{n}^n.
\end{align}

In order to obtain absolute continuity of $\mu_n$ (and hence obtain a spectral shift \textit{function}) we apply the change of variables provided by Theorem~\ref{thm:adding resolvents}, in this case to multiple operator integrals of order $n-1$. This entails new terms for which the trace is defined only when perturbations satisfy additional summability requirements.
We establish an auxiliary result for finite rank perturbations in Proposition \ref{prop36} and then extend it
to relative Schatten class perturbations appearing in our main result with help of two new approximation results, one for operators obtained in Lemma \ref{lem:approximating V by V_k} and the other for Taylor remainders obtained in Lemma \ref{prop:|Tr(R)|}. In order to apply those approximation results, in Lemma \ref{lem:pth part of remainder} we derive a new representation for the remainder of the Taylor approximation of $f(H+V)$ in terms of handy components that are continuous in $V$ in a very strong sense.

In order to strengthen \eqref{trexist}, in Proposition \ref{prop:SSF locally} we establish another weaker version of \eqref{trvinsn} for $f\in C_\text{c}^{n+1}(\R)$, where on the left-hand side we have a certain component of the Taylor remainder and on the right-hand side in place of $f$ we have its product with some complex weight. By combining advantages of the results of Propositions \ref{prop:SSM with growth} and \ref{prop:SSF locally} we derive the trace formula \eqref{trvinsn}.
\medskip

\paragraph*{Examples.}

The relative Schatten class condition \eqref{vresinsn} arises in noncommutative geometry; see, for instance, \cite{Sui11,SZ18}.
In that setting, $H$ is a generalized Dirac operator occurring in a (possibly nonunital) spectral triple and $V$ a generalized vector potential \cite[Section IV.1]{C94}, which is also known as an inner fluctuation or Connes' differential one-form \cite{CC97,Sui11}. For unital spectral triples, the condition \ref{rinsn}, which is known as finite summability, is often assumed. For nonunital spectral triples, conditions similar to \ref{inner perturbations} are discussed in Section \ref{sec4a}. Both in the unital and nonunital case, it is important to relax assumptions on the function $f$ appearing in the spectral action \cite{CC97} since that function might be prescribed by the model \cite{CCS}. Sometimes it is impossible or at least inconvenient to assume that $f$ is  given by a Laplace transform, as it was done in \cite{Sui11}, and an explicit class of functions like we consider in this chapter is more beneficial.

The condition \eqref{vresinsn} is also satisfied by many Dirac as well as random and deterministic Schr\"{o}dinger operators $H$ with $L^p$-potentials $V$. Appearance of such operators in problems of mathematical physics is discussed in, for instance, \cite{S21,Yafaev10} and references cited therein.
Sufficient conditions for \eqref{vresinsn} are discussed in Section \ref{sec4}.

%
%

\paragraph*{Notations.}
Given a self-adjoint operator $H$ in $\H$ and $V\in\mB(\H)$, we denote
$$\tilde{V}:=V(H-i)^{-1}.
$$
If $H_0,\ldots,H_m$ are self-adjoint operators in $\H$, and $V_1,\ldots,V_m$ are bounded operators, we denote
$$\tilde{V}_j:=V_j(H_j-i)^{-1}.$$
We denote positive constants by letters $c,C$ with superscripts indicating dependence on their parameters. For instance, the symbol $c_\alpha$ denotes a constant depending only on the parameter $\alpha$. We write $f(x)=\O(g(x))$ if there exists $M>0$ such that $|f(x)|\leq Mg(x)$ for all $x$ outside a compact set. We write $f(x)=o(g(x))$ if for all $\epsilon>0$, we have $|f(x)|\leq\epsilon g(x)$ for all $x$ outside a compact set.

\section{Auxiliary technical results}

In this section we set a technical foundation for the proof of our main result.

\subsection{New function classes}
\label{sec2}

In this subsection we introduce a new class of functions $\mW_n$, for which our main result holds, along with auxiliary classes $\mathfrak{B}_n$ and $\mathfrak{b}_n$ and derive their properties.

\begin{defi}
\label{fracw}
Let $\mathfrak{W}_n$ denote the set of functions $f\in C^n(\R)$ such that
\begin{enumerate}[label=\textnormal{(\roman*)}]	
\item $\widehat{f^{(k)}u^k}\in L^1(\R),~k=0,\ldots,n$,
\item\label{fracwii} $f^{(k)}\in L^1\big(\R,(1+|x|)^{k-1}\,dx\big)$, $k=1,\ldots,n$.
\end{enumerate}
\end{defi}

The following sufficient condition for integrability of the Fourier transform of a function is a standard exercise and, thus, its proof is omitted.

\begin{lem}\label{lem:FT in L1}
If $f\in L^2(\R)\cap C^1(\R)$ and $f'\in L^2(\R)$, then $\hat{f}\in L^1(\R)$.
\end{lem}

\begin{prop}
\label{inclusions}
Let $n\in\N$. Then, the following assertions hold.
\begin{enumerate}[label=\textnormal{(\roman*)}]	
\item\label{inclusionsi}
For every $\alpha>\frac12$,
\begin{align*}
\mW_n\supseteq\left\{f\in C^{n+1}:~ f^{(k)}(x)=\O\left(|x|^{-k-\alpha}\right)
\text{ as } |x|\to\infty,\;k=0,\ldots,n+1\right\}.
\end{align*}

\item\label{inclusionsii}
Furthermore,
\begin{align*}
\mW_n\subseteq\left\{f\in C^n:\;f^{(k)},\widehat{f^{(k)}}\in L^1(\R),\; k=1,\ldots,n\right\}.
\end{align*}
\end{enumerate}
\end{prop}

\begin{proof}
The inclusion in (i) is straightforward, as it follows from Lemma \ref{lem:FT in L1}.

(ii) The properties $f^{(k)}\in L^1(\R),$ $k=1,\ldots,n$ follow immediately from the definition of $\mW_n$. To prove $\widehat{f^{(k)}}\in L^1(\R)$, $k=1,\dots,n-1$, firstly we note that
\begin{align}
\label{8}
\widehat{u^{-k}}\in L^1,\quad k=1,\ldots,n,
\end{align}
by Lemma \ref{lem:FT in L1}. By the convolution theorem we find
\begin{align*}
\widehat{f^{(k)}}=\widehat{f^{(k)}u^k}*\widehat{u^{-k}},\quad k=1,\ldots,n,
\end{align*}
which is in $L^1$ because $L^1$ is closed under the convolution product.
Therefore, the proof of (ii) is complete.
\end{proof}

\begin{rema}
\label{mncontains}
It follows from Proposition \ref{inclusions}\ref{inclusionsii} that
$\mW_n$ contains all bounded rational functions except for linear combinations with constant functions, which are trivial in the context of our main result. In particular, $\mW_n$ contains the space $\spn\{(z-\,\cdot)^{-k},\; \Im(z)>0,\; k\in\N,\; k\ge 2n\}$ considered in \cite{CS18}. In addition, $\mW_n$ contains all Schwartz functions and every $f\in C^{n+1}$ such that $f(x)=|x|^{-\alpha}$ outside a bounded neighborhood of zero for some $\alpha>\frac12$.
\end{rema}

We will need the auxiliary function classes
\begin{align}
\label{mBn}
\mathfrak{B}_n:=&\left\{f\in C^n:~f^{(k)}u^k\in C_0(\R),~k=0,\ldots,n,~\widehat{f^{(n)}}\in L^1(\R)\right\}
\end{align}
and
\begin{align}
\label{mbn}
\mathfrak{b}_n:=&\left\{f\in \mathfrak{B}_n:~\widehat{f^{(p)}u^p}\in L^1(\R),~p=0,\ldots,n\right\}.
\end{align}
It follows from Definition \ref{fracw} and Proposition \ref{inclusions}\ref{inclusionsii} that
\begin{align*}
\mW_n\subseteq\mathfrak{b}_n\subseteq\mathfrak{B}_n.
\end{align*}
We also have the following result relating $\mathfrak{b}_n$ and $\mathfrak{B}_n$.

\begin{lem}\label{lem:density}
	The space $\mathfrak{b}_n$ is dense in $\mathfrak{B}_n$ with respect to the norm
		$$\norm{f}_{\mathfrak{B}_n}:=\sum_{p=0}^n\supnorm{f^{(p)}u^p}+\absnorm{\widehat{f^{(n)}}}.$$
\end{lem}
\begin{proof}
	Let $f\in\mathfrak{B}_n$. Fix a Schwartz function $\phi$ such that $\hat{\phi}\in C^\infty_\text{c}(\R)$ and $\phi(0)=1$. For every $k\in\N$, define
\begin{align*}
\phi_k(x):=\phi(x/k),\quad x\in\R.
\end{align*}
We note that $\big\{\,\widehat{\phi_k}\,\big\}_{k=1}^\infty$ is an approximate identity.
In particular, it satisfies the property
\begin{align}
\label{ai1}
\|\widehat{\phi_k}*g-g\|_1\rightarrow 0\quad\text{as }\;k\rightarrow\infty
\end{align}
for every $g\in L^1$.
Define
$$f_k:=\phi_kf.$$
Because every $\phi_k^{(m)}$ is of rapid decrease, it is obvious that $f_k^{(p)}u^p=\sum_{m=0}^p \vect{p}{m} \phi_k^{(m)}f^{(p-m)}u^p$ is integrable for every $p\in\{0,\ldots,n\}$.
	By Lemma \ref{lem:FT in L1} and the rapid decrease of every $\phi_k^{(m)}$, we obtain that $\widehat{f_k^{(p)}u^p}\in L^1$ for every $p\in\{0,\ldots,n-1\}$.
	In the same way, we obtain that $(f^{(p)}\phi_k^{(n-p)}u^n)\hat{~}\in L^1$ for every $p\in\{0,\ldots,n-1\}$.
	Moreover, we have $(f^{(n)}\phi_ku^n)\hat{~}=\widehat{f^{(n)}}*\widehat{\phi_ku^n}\in L^1$. Hence,
		$$\widehat{f_k^{(n)}u^n}=\sum_{p=0}^n\vect{n}{p}(f^{(p)}\phi_k^{(n-p)}u^n)\hat{~}\in L^1.$$
	We conclude that $f_k\in\mathfrak{b}_n$.

	In order to prove that $\|f^{(p)}u^p-f_k^{(p)}u^p\|_\infty\to0$ as $k\to\infty$, we write
	\begin{align}\label{eq:supnorm bounds phi_k}
		\supnorm{f^{(p)}u^p-f_k^{(p)}u^p}\leq \supnorm{(1-\phi_k)f^{(p)}u^p}+\sum_{m=1}^p\vect{p}{m}\supnorm{\phi_k^{(m)}u^mf^{(p-m)}u^{p-m}}.
	\end{align}
Since $f^{(p)}u^p\in C_0(\R)$, we obtain
\begin{align}
\label{term1}
\supnorm{(1-\phi_k)f^{(p)}u^p}\rightarrow 0\quad\text{as }\, k\rightarrow\infty.
\end{align}
By using $\phi^{(m)}_k(x)=\phi^{(m)}(x/k)/k^m$, we obtain
\begin{align}
\label{phik}
|\phi_k^{(m)}(x)u^m(x)|\leq\sqrt{2}^{\,m}\supnorm{\phi^{(m)}}k^{-m/2}\quad\text{for }\, x\in[-\sqrt{k},\sqrt{k}]
\end{align}
and
\begin{align}
\label{phiglobal}
\supnorm{\phi^{(m)}_ku^m}\leq \sqrt{2}^{\,m}\supnorm{\phi^{(m)}u^m}.
\end{align}
We now analyze the terms on the right-hand side of \eqref{eq:supnorm bounds phi_k} as $k\rightarrow\infty$.
By \eqref{phik}, \eqref{phiglobal}, and the assumption $f^{(p-m)}u^{p-m}\in C_0$, we obtain $\|\phi_k^{(m)}u^mf^{(p-m)}u^{p-m}\|_\infty\to0$ as $k\rightarrow\infty$.
Combining the latter with \eqref{eq:supnorm bounds phi_k} and \eqref{term1} implies
$$\supnorm{f^{(p)}u^p-f_k^{(p)}u^p}\to0\quad\text{as }\, k\rightarrow\infty,\quad p=0,\dots,n.$$

We are left to prove that $\|\widehat{f^{(n)}}-\widehat{f_k^{(n)}}\|_1\to0$.
Applying $f_k^{(n)}=\sum_{m=0}^n \begin{psmallmatrix}n\\ m\end{psmallmatrix}\phi_k^{(m)}f^{(n-m)}$ along with standard properties of the Fourier transform and convolution yields
\begin{align}
\label{above}
\absnorm{\widehat{f^{(n)}}-\widehat{f_k^{(n)}}}\leq\absnorm{\widehat{f^{(n)}}-\widehat{\phi_k}*\widehat{f^{(n)}}}+\sum_{m=1}^n\vect{n}{m}\frac{\absnorm{\widehat{\phi^{(m)}}}}{k^m}\absnorm{\widehat{f^{(n-m)}}}.
\end{align}
The first term on the right-hand side of \eqref{above} converges to $0$ as $k\rightarrow\infty$ by \eqref{ai1} applied to $g=\widehat{f^{(n)}}$. The other terms on the right-hand side of \eqref{above} converge to $0$ as $k\rightarrow\infty$ because $1/k^m\to 0$.
\end{proof}

\subsection{Change of variable formula}
By the above results we can now generalize Theorem \ref{thm:adding resolvents} to a larger class of functions.

\begin{thm}\label{thm:adding resolvents 2} 
Let $H_0,\ldots,H_n$ be self-adjoint in $\H$ and let $V_1,\ldots,V_n$ be such that $\tilde V_k=V_k(H_k-i)^{-1}\in\S^n$ for all $k=1,\ldots,n$. Then $T^{H_0,\ldots,H_n}_{f^{[n]}}(V_1,\ldots,V_n)$ is trace-class for all $f\in\mathfrak{b}_n$, and we have, denoting $\tilde{V}_{j,l}:=\tilde V_{j+1}\cdots \tilde V_l$,
$$T^{H_0,\ldots,H_n}_{f^{[n]}}(V_1,\ldots,V_n)=\sum_{p=0}^n\sum_{0<j_1<\cdots<j_{p}\leq n}\!(-1)^{n-p}\, T^{H_0,H_{j_1},\ldots,H_{j_p}}_{(fu^p)^{[p]}}
(\tilde{V}_{0,j_1},\ldots,\tilde{V}_{j_{p-1},j_p})\,
\tilde{V}_{j_{p},n}\,$$
for every $f\in \mathfrak{B}_n$, and hence, for every $f\in\mW_n$.
\end{thm}
\begin{proof}
Theorem \ref{thm:adding resolvents}\ref{trace class} (which is applicable to $f\in\mathfrak{b}_n$ by the generalized Leibniz rule) gives the desired statement that $T^{H_0,\ldots,H_n}_{f^{[n]}}(V_1,\ldots,V_n)\in\S^1$ for $f\in\mathfrak{b}_n$. 

Let $f\in\mathfrak{B}_n$. By Lemma \ref{lem:density} we can choose $f_k\in\mathfrak{b}_n$ for all $k\in\N$ such that
\begin{align}\label{eq:density}
	\|\widehat{f_k^{(n)}}-\widehat{f^{(n)}}\|_1\to0\quad\text{and}\quad\|(f_ku^p)^{(p)}-(fu^p)^{(p)}\|_\infty\to0.
\end{align}
Theorem \ref{thm:adding resolvents}\ref{weights} in particular gives
\begin{align}\label{weights with hats}
	T^{H_0,\ldots,H_n}_{f_k^{[n]}}(V_1,\ldots,V_n)
=&\sum_{p=0}^n\sum_{0<j_1<\cdots<j_{p}\leq n}\!(-1)^{n-p}\, T^{H_0,H_{j_1},\ldots,H_{j_p}}_{(f_ku^p)^{[p]}}
(\tilde{V}_{0,j_1},\ldots,\tilde{V}_{j_{p-1},j_p})\,
\tilde{V}_{j_{p},n}.
\end{align}	
	The $L^1$-convergence in \eqref{eq:density} implies that the left-hand side converges (in operator norm) to $T^{H_0,\ldots,H_n}_{f^{[n]}}(V_1,\ldots,V_n)$ by \eqref{fourierbound}. Moreover, we find that $\tilde V_{j_{m-1},j_m}\in \S^{\alpha_m}$, where $\alpha_m:=n/(j_m-j_{m-1})\in(1,\infty)$ for $m=2,\ldots,p$, and $\tilde V_{0,j_1}\in\S^{\alpha_1}$, $\tilde V_{j_p,n}\in\S^{\alpha_{p+1}}$, where $\alpha_1:=n/j_1\in[1,\infty)$, $\alpha_{p+1}=n/(n-j_p)\in(1,\infty]$. On the strength of Theorem \ref{thm:PSS Higher Order} applied to $\S^{2\alpha_m}$, the supnorm-convergence in \eqref{eq:density} implies that the right-hand side of \eqref{weights with hats} converges to the right-hand side of \ref{weights} in the operator norm (since convergence in Schatten norms implies uniform convergence). By uniqueness of limits in $\mB(\H)$, we conclude \ref{weights}.
\end{proof}


\section{Existence of the spectral shift function}
\label{sec3}

In this section we establish the main result of Chapter \ref{ch:Spectral Shift for Relative Schatten Perturbations}. 

\begin{thm}
\label{rsmain}
Let $n\in\N$, let $H$ be a self-adjoint operator in $\H$, and let $V\in\mB(\H)_{\sa}$ be such that $V(H-i)^{-1}\in\S^n$. Then, there exists $c_n>0$ and a real-valued function $\eta_n$
such that
\begin{align}\label{eta estimate}
\int_\R |\eta_n(x)|\,\frac{dx}{(1+|x|)^{n+\epsilon}}\leq c_n\,(1+\epsilon^{-1})(1+\norm{V})\nrm{V(H-i)^{-1}}{n}^n\quad\text{for all } \epsilon >0
\end{align} and
\begin{align}
\label{tff}
\Tr(R_{n,H,f}(V))=\int_\R f^{(n)}(x)\eta_n(x)\,dx\,
\end{align}
for every $f\in\mW_n$. The locally integrable function $\eta_n$ is determined by \eqref{tff} uniquely up to a polynomial summand of degree at most $n-1$.
\end{thm}


We start by outlining major steps and ideas of the proof of Theorem \ref{rsmain}.

In Proposition \ref{prop:SSM with growth} we establish a weaker version of \eqref{tff} with  measure $d\mu_n$ on the right-hand side of \eqref{tff} in place of the desired absolutely continuous measure $\eta_n(x)\,dx$. The measure $\mu_n$, which we call the spectral shift measure, satisfies the bound \eqref{eta estimate}. In Proposition \ref{prop:SSF locally} we establish another weaker version of \eqref{tff} for compactly supported $f$, where on the left-hand side we have a certain component of the remainder and on the right-hand side instead of $f$ we have its product with some complex weight. By combining advantages of the results of Propositions \ref{prop:SSM with growth} and \ref{prop:SSF locally} we derive the trace formula \eqref{tff}.


One of our main tools is multiple operator integration theory developed for Schatten class perturbations. This theory is not \textit{directly} applicable in our setting, however, because our perturbations are not compact.
To bridge the gap between the existing theory and our setting we combine the powerful results of Chapter \ref{ch:MOI} with multistage approximation arguments. In particular, the proof of
Proposition \ref{prop:SSF locally} requires two novel techniques. The first one is a new expression for the remainder $R_{n,H,f}(V)$ in terms of handy components that are continuous in $V$ in a very strong sense. The second one is an approximation argument that allows replacing relative Schatten $V$ by finite rank $V_k$ and strengthens convergence arguments present in the literature.

\subsection{Existence of the spectral shift measure}
\label{sct:Existence of the spectral shift measure}

The following result is our first major step in the proof of the representation \eqref{tff}.

\begin{prop}\label{prop:SSM with growth}
Let $n\in\N$, let $H$ be a self-adjoint operator in $\H$, and let $V\in\mB(\H)_{\sa}$ be such that $V(H-i)^{-1}\in\S^n$.
Then, there exists a Borel measure $\mu_n$ such that
\begin{align}\label{mu tilde}
	\Tr(R_{n,H,f}(V))=\int_\R f^{(n)}\,d\mu_n\,
\end{align}
for every $f\in\mW_n$ and
\begin{align}
\label{munfla}
d\mu_n(x)=u^n(x)\,d\nu_n(x)+\xi_n(x)\,dx,
\end{align}
where $\nu_n$ is a finite measure satisfying
\begin{align}\label{eq:nu bound}
\|\nu_n\|\leq c_n(1+\norm{V})\nrm{V(H-i)^{-1}}{n}^n,
\end{align}
and $\xi_n$ is a continuous function
satisfying
\begin{align}\label{eq:xi bound}
	|\xi_n(x)|\leq c_n(1+\norm{V})\nrm{V(H-i)^{-1}}{n}^n(1+|x|)^{n-1},\quad x\in\R,
\end{align}
for some constant $c_n>0$ independent from $V$, $H$, and $x$.
If $\tilde\mu_n$ is another locally finite Borel measure such that \eqref{mu tilde} holds for all $f\in C^{n+1}_\text{c}$, then $d\tilde{\mu}_n(x)=d\mu_n(x)+p_{n-1}(x)\,dx$, where $p_{n-1}$ is a polynomial of degree at most $n-1$.
\end{prop}

To prove Proposition \ref{prop:SSM with growth} we need the estimate stated below.

\begin{lem}\label{lem:Hahn Banach Borel measures}
Let $k\in\N$, let $H_0,\ldots,H_k$ be self-adjoint operators in $\H$, let  $\alpha_1\ldots,\alpha_k\in(1,\infty)$ be such that $1=\tfrac{1}{\alpha_1}+\ldots+\tfrac{1}{\alpha_k}$.
Then, there exists $c_\alpha:=c_{\alpha_1,\dots,\alpha_k}>0$ such that for all $B_j\in\S^{\alpha_j}$, $j=1,\dots,k$, we have
\begin{align*}
|\Tr(T^{H_k,H_1,\ldots,H_k}_{f^{[k]}}(B_1,\ldots,B_k))|&\leq c_\alpha\supnorm{f^{(k)}}\nrm{B_1}{\alpha_1}\cdots\nrm{B_k}{\alpha_k}\qquad ( f\in C^k, \widehat{f^{(k)}}\in L^1)
\end{align*}
and
\begin{align*}
|\Tr(B_1T^{H_1,\ldots,H_{k}}_{f^{[k-1]}}(B_2,\ldots,B_k))|&\leq c_\alpha\supnorm{f^{(k-1)}}\nrm{B_1}{\alpha_1}\cdots\nrm{B_k}{\alpha_k}\quad ( f\in C^{k-1},f^{(k-1)}\in C_\textnormal{b}).
\end{align*}
Consequently, there exist unique (complex) Borel measures $\mu_1,\mu_2$ with total variation bounded by $c_\alpha\nrm{B_1}{\alpha_1}\cdots\nrm{B_k}{\alpha_k}$ such that
\begin{align*}
\Tr(T^{H_k,H_1,\ldots,H_k}_{f^{[k]}}(B_1,\ldots,B_k))&=\int_\R f^{(k)} \,d\mu_1\qquad (f\in C^k,\widehat{f^{(k)}}\in L^1)
\end{align*}
and
\begin{align*}
\Tr(B_1 T^{H_1,\ldots,H_{k}}_{f^{[k-1]}}(B_2,\ldots,B_k))&=\int_\R f^{(k-1)}\,d\mu_2\,\qquad (f\in C^{k-1},f^{(k-1)}\in C_0).
\end{align*}
\end{lem}

\begin{proof}
The first assertion of the lemma, at least when $H_0=\ldots=H_k$, follows from \cite[Theorem 5.3 and Remark 5.4]{PSS13}. The extension to distinct $H_0,\ldots,H_k$ follows from \cite[Theorem 4.3.10]{ST19} and H\"older's inequality. The second assertion of the lemma is subsequently obtained by the Riesz--Markov representation theorem for a bounded linear functional on the space $C_0(\R)$.
\end{proof}

\begin{proof}[Proof of Proposition \ref{prop:SSM with growth}]
Let $n\geq2$.
Using \eqref{remmoi} and Theorem \ref{thm:adding resolvents 2}, we obtain
\begin{align}
\label{eq:new expansion of remainder}
\nonumber
R_{n,H,f}(V)=&T^{H,H+V,H,\ldots,H}_{f^{[n]}}(V,\ldots,V)\\
	=&\sum_{p=0}^n\sum_{\substack{j_1,\ldots,j_{p}\geq1,j_{p+1}\geq0\\j_1+\ldots+j_{p+1}=n}}(-1)^{n-p}\,
T^{H,H_{j_1},H,\ldots,H}_{(fu^p)^{[p]}}(\tilde{V}^{j_1},\ldots,\tilde{V}^{j_{p}})\tilde{V}^{j_{p+1}},
\end{align}
where $H_{1}=H+V$ and $H_{j_1}=H$ for $j_1\neq1$, and in which the first factor of $\tilde V$ in the first input of the multilinear
operator integral should be interpreted as $V(H+V-i)^{-1}$.
By the second resolvent identity,
\begin{align*}
\|V(H+V-i)^{-1}\|_n\leq (1+\norm{V})\|V(H-i)^{-1}\|_n.
\end{align*}
By the definition of $\mW_n$ (see Definition \ref{fracw}), 
we obtain $\widehat{(fu^p)^{(p)}}\in L^1(\R)$
for every $f\in\mW_n$, $p=0,\dots,n$. Hence, by Lemma \ref{lem:Hahn Banach Borel measures}
applied to each term of \eqref{eq:new expansion of remainder},
there exist unique Borel measures $\breve\mu_0,\ldots,\breve\mu_n$ such that
\begin{align}
\label{mupbound}
\|\breve\mu_p\|\le C_n\,(1+\norm{V})\,\|V(H-i)^{-1}\|_n^n
\end{align}
and
\begin{align}\label{eq:Tr(R) in measures 1}
\Tr( R_{n,H,f}(V))=&\sum_{p=0}^n\int (fu^p)^{(p)}\,d\breve\mu_p
\end{align}
for every $f\in\mW_n$, $n\ge 2$.

Let $n=1$. Denote $H_t=H+tV$.
By Theorem \ref{dm}, continuity of the transformation $t\mapsto T^{H_t,H_t}_{f^{[1]}}(V)$
(see \cite[Proposition 3.3.9]{ST19}), and the fundamental theorem of calculus,
\begin{align*}
R_{1,H,f}(V)
=f(H+V)-f(H)=\int_0^1 T^{H_t,H_t}_{f^{[1]}}(V)\,dt
\end{align*}
for $f\in\mW_1$.
By \eqref{weight1} of Theorem \ref{thm:adding resolvents}\ref{weight} applied to $T^{H_t,H_t}_{f^{[1]}}(V)$  we obtain
\begin{align}\label{eq:no trace yet}
R_{1,H,f}(V)
=\int_0^1 (T^{H_t,H_t}_{(fu)^{[1]}}(V(H_t-i)^{-1})-f(H_t)V(H_t-i)^{-1})\,dt.
\end{align}
Noticing that
\begin{align*}
\sup_{t\in [0,1]}\|V(H_t-i)^{-1}\|_1\le(1+\|V\|)\|V(H-i)^{-1}\|_1,
\end{align*}
using the property of the double operator integral $\Tr(T^{H,H}_{g^{[1]}}(V))=\Tr(g'(H)V)$,
and applying H\"{o}lder's inequality and the Riesz--Markov representation theorem
completes the proof of \eqref{eq:Tr(R) in measures 1} for $n=1$.

Let $n\in\N$. Applying a higher-order differentiation product rule on the right-hand side of \eqref{eq:Tr(R) in measures 1} gives
\begin{align}
\Tr( R_{n,H,f}(V))&=\sum_{p=0}^n\sum_{k=0}^p\vect{p}{k}\frac{p!}{k!}\int f^{(k)}u^k\,d\breve\mu_p\nonumber\\
	&=\sum_{k=0}^{n-1}\int f^{(k)}u^k\,d\grave{\mu}_k+\int f^{(n)}u^n\,d\nu_n,\label{Tr(R) in integrals}
\end{align}
for some Borel measures $\grave{\mu}_0,\ldots\grave\mu_{n-1},\nu_n$ satisfying
\begin{align}
\label{normnun} 
\|\grave\mu_0\|,\dots,\|\grave\mu_{n-1}\|,\|\nu_n\|\le \tilde C_n\,(1+\norm{V})\,\|V(H-i)^{-1}\|_n^n.
\end{align}
Integrating by parts in \eqref{Tr(R) in integrals} and applying
\begin{align}
\label{partbounds}
\lim\limits_{x\rightarrow\pm\infty} f^{(k)}(x)u^k(x)=0,\quad k=0,\dots,n-1,
\end{align}
yields
\begin{align*}
\Tr(R_{n,H,f}(V))=-\sum_{k=0}^{n-1}\int_{-\infty}^\infty (f^{(k+1)}u^k+kf^{(k)}u^{k-1})(x)\,\grave{\mu}_k((-\infty,x))\,dx+\int f^{(n)}u^n\,d\nu_n.
\end{align*}
Since $$f^{(k)}u^{k-1}\in L^1(\R),\quad k=1,\dots,n,$$ we rearrange the terms above to obtain
\begin{align}
\label{bp1}
\Tr(R_{n,H,f}(V))=\sum_{k=1}^{n}\int f^{(k)}(x)u^{k-1}(x)\,\tilde\xi_k(x)\,dx+\int f^{(n)}u^n\,d\nu_n,
\end{align}
where $\tilde\xi_k$ are continuous functions defined by
\begin{align*}
&\tilde{\xi}_k(x)=-\grave\mu_{k-1}((-\infty,x))-k\,\grave\mu_k((-\infty,x)),\quad k=1,\dots,n-1,\\
&\tilde\xi_n(x)=-\grave\mu_{n-1}((-\infty,x)),
\end{align*}
so that
\begin{align}\label{eq:xi_k bound}
	\big\|\tilde\xi_k\big\|_\infty\leq c_{n,k}(1+\norm{V})\nrm{V(H-i)^{-1}}{n}^n,\quad k=1,\dots,n.
\end{align}
By a repeated partial integration in \eqref{bp1} and application of \eqref{partbounds}, we obtain
\begin{align*}
\Tr(R_{n,H,f}(V))=\int_\R f^{(n)}\,d\mu_n\quad (f\in\mW_n)
\end{align*}
with
\begin{align}
\label{mugravexi}
d\mu_n(x)=u^n(x)\,d\nu_n(x)+\xi_n(x)\,dx,
\end{align}
where
\begin{align}\label{eq:def xi}	
\xi_n(s_0):=\sum_{k=1}^n(-1)^{n-k}\int_0^{s_0}ds_1\cdots\int_0^{s_{n-k-1}}\,
u^{k-1}(s_{n-k})\,\tilde\xi_k(s_{n-k})\,ds_{n-k}.
\end{align}
The function $\xi_n$ given by \eqref{eq:def xi} is continuous. To confirm \eqref{eq:xi bound} we note that,
for all $m\in\N$,
\begin{align}
\label{add u^-1}
\sup_{x\in\R}\left|u^{-m}(x)\int_0^xg(t)\,dt\right|\leq \sup_{x\in\R}\left(\left|\frac{x}{u(x)}\,\right|
\sup_{|t|\leq|x|}|u^{1-m}(x)g(t)|\right)\le\|u^{1-m}g\|_\infty.
\end{align}
By applying \eqref{add u^-1} $(n-k)$-times in \eqref{eq:def xi} and using the bound \eqref{eq:xi_k bound}, we obtain
\begin{align}\label{xiprelbound}
|\xi_n(x)|\leq c_n(1+\norm{V})\nrm{V(H-i)^{-1}}{n}^n(1+|x|)^{n-1},\quad x\in\R.
\end{align}
We have thereby proven the first part of the proposition.

To prove the second part of the proposition, we let $\tilde{\mu}_n$ be a locally finite measure such that \eqref{mu tilde} holds for all $f\in C_\textnormal{c}^{n+1}$
and denote
$$\rho_n:=\mu_n-\tilde\mu_n.$$
Then,
\begin{align}\label{eq:measure nth derivative}
	\int f^{(n)}\,d\rho_n=0\qquad(f\in C_\textnormal{c}^{n+1}).
\end{align}
We are left to confirm that
\begin{align}
\label{goal}
d\rho_n(x)=p_{n-1}(x)\,dx,
\end{align}
where $p_{n-1}$ is a polynomial of degree at most $n-1$. Consider the distribution $T$ defined by 	
$$T(g):=\int g\, d\rho_n$$ for all $g\in C^\infty_\textnormal{c}(\R)$.
By \eqref{eq:measure nth derivative} and the definition of the derivative of a distribution, $T^{(n)}=0$.
Since the primitive of a distribution is unique up to an additive constant (see, e.g., \cite[Theorem 3.10]{gwaiz}),
by an inductive argument (see, e.g., \cite[Example 2.21]{gwaiz}) we obtain \eqref{goal}.
\end{proof}

\subsection{Alternative trace formula}
\label{sct:Alternative trace formula}

The following result is our second major step in the proof of the representation \eqref{tff}.
It provides an alternative to \eqref{tff} with weighted $f$ on the right-hand side.
It also provides an alternative to \eqref{mu tilde} with weighted $f$ on the right-hand side, thereby effectively replacing the measure $\mu_n$ with functions $\breve\eta_0,\ldots,\breve\eta_{n-1}\in L^1_{\textnormal{loc}}$.
\begin{prop}\label{prop:SSF locally}
Let $n\in\N$, $n\ge 3$, let $H$ be a self-adjoint operator in $\H$, and let $V\in\mB(\H)_{\sa}$ satisfy $V(H-i)^{-1}\in\S^n$. Then, for every $p=0,\dots,n-1$,
there exists $\breve\eta_p\in L^1_{\textnormal{loc}}$ such that
\begin{align}
\label{trfp}
\Tr(R_{n,H,f}(V))=\sum_{p=0}^{n-1}(-1)^{n-1-p}\int_\R (fu^p)^{(p+1)}(x)\breve\eta_p(x)\,dx
\end{align}
for all $f\in C_\textnormal{c}^{n+1}$.
\end{prop}

In order to prove \eqref{trfp} firstly we decompose $R_{n,H,f}(V)$ into more convenient components for which we can derive trace formulas by utilizing the method of the previous subsection, partial integration, and approximation arguments.


\begin{lem}\label{lem:pth part of remainder}
Let $H$ be a self-adjoint operator in $\H$, let $V\in\mB(\H)_{\text{sa}}$, let $n\in\N$, and let $f\in C_\textnormal{c}^{n+1}$. Then, $$R_{n,H,f}(V)=\sum_{p=0}^{n-1} (-1)^{n-1-p}\tilde R^p_{n,H,f}(V),$$
where
\begin{align}
\label{R0}
\nonumber
&\tilde R^0_{1,H,f}(V):=f(H+V)-f(H),\\
&\tilde R^0_{n,H,f}(V):=f(H)V((H+V-i)^{-1}-(H-i)^{-1})\tilde V^{n-2}
\end{align}
for $n\geq 2$ and
\begin{align}
\label{rpdef}
\tilde R^p_{n,H,f}(V):=\sum_{\substack{j_1,\ldots,j_{p}\geq1,j_{p+1}\geq0\\j_1+\ldots+j_{p+1}=n-1}}\!\, \Big(& T^{H,H_{j_1},H,\ldots,H}_{(fu^p)^{[p]}}(V(H+V-i)^{-1}\tilde{V}^{j_1-1}, \ldots,\tilde{V}^{j_p})\,\tilde{V}^{j_{p+1}}\nonumber\\
&-T^{H,\ldots,H}_{(fu^p)^{[p]}}(\tilde{V}^{j_1},\ldots,\tilde{V}^{j_p})\tilde{V}^{j_{p+1}}\Big)
\end{align}
for $ p=1,\dots, n-1$, with $H_1=H+V$ and $H_{j_1}=H$ for $j_1\neq1$.
\end{lem}

\begin{proof}\label{prop36}
Using \eqref{dermoi} and \eqref{remmoi}, we get
\begin{align}
R_{n,H,f}(V)=&R_{n-1,H,f}(V)-\frac{1}{(n-1)!}\frac{d^{n-1}}{dt^{n-1}}f(H+tV)|_{t=0}\nonumber\\
=&T^{H,H+V,H,\ldots,H}_{f^{[n-1]}}(V,\ldots,V)
-T^{H,\ldots,H}_{f^{[n-1]}}(V,\ldots,V)\label{birem}.
\end{align}
An application of Theorem \ref{thm:adding resolvents}\ref{weights} to each of the terms in \eqref{birem} completes the proof.
\end{proof}

Firstly we show that \eqref{trfp} holds when $V$ is a finite-rank operator. This is done by establishing an analog of \eqref{trfp} for $\tilde R^p_{n,H,f}(V)$ and then extending \eqref{trfp} to $R_{n,H,f}(V)$ with help of Lemma \ref{lem:pth part of remainder}.

\begin{prop}\label{prop:SSF locally V_k}
Let $n\in\N$, $n\ge 3$, let $H$ be a self-adjoint operator in $\H$, and let $V\in\mB(\H)_{\text{sa}}$
be of finite rank. Then, for $p=0,\dots, n-1$, there exists $\breve\eta_{p}\in L^1_{\textnormal{loc}}$ such that
\begin{align*}
\Tr(\tilde R^p_{n,H,f}(V))=\int_\R(fu^p)^{(p+1)}(x)\breve\eta_{p}(x)\,dx\,
\end{align*}
for all $f\in C^{n+1}_\textnormal{c}$,  where $\tilde R^p_{n,H,f}$ is given by \eqref{rpdef}.
\end{prop}
\begin{proof}

By the definition of $\tilde R^p_{n,H,f}(V)$ in Lemma~\ref{lem:pth part of remainder},
\begin{align}
\label{sumabstrt}
&|\Tr(\tilde R^p_{n,H,f}(V))|\\ \nonumber
&\leq\sum_{\substack{j_1,\ldots,j_{p}\geq1,j_{p+1}\geq0\\j_1+\ldots+j_{p+1}=n-1}}
\Big(\big|\Tr\big(T^{H,H_{j_1},H,\ldots,H}_{(fu^p)^{[p]}}(V(H+V-i)^{-1}\tilde{V}^{j_1-1},\ldots,
\tilde{V}^{j_p})\tilde{V}^{j_{p+1}}\big)\big|\\ \nonumber
&\quad+\big|\Tr\big(T^{H,\ldots,H}_{(fu^p)^{[p]}}
(\tilde{V}^{j_1},\ldots,\tilde{V}^{j_p})\tilde{V}^{j_{p+1}}\big)\big|\Big).
\end{align}
By Lemma \ref{lem:Hahn Banach Borel measures} applied to each summand on the right-hand side of \eqref{sumabstrt},
\begin{align}	
|\Tr(\tilde R^p_{n,H,f}(V))|
\leq&\sum_{\substack{j_1,\ldots,j_{p}\geq1,j_{p+1}\geq0\\
j_1+\ldots+j_{p+1}=n-1}}2c_{n,j}\supnorm{(fu^p)^{(p)}}
(1+\|V\|)\nrm{V(H-i)^{-1}}{n-1}^{n-1}\nonumber\\
=:&\,c_n\supnorm{(fu^p)^{(p)}}(1+\|V\|)\nrm{V(H-i)^{-1}}{n-1}^{n-1}.
\label{eq:bound R^p}
\end{align}
Hence, by the Riesz--Markov representation theorem, there exist unique Borel measures $\breve\mu_{p}$ such that
$$\|\breve\mu_{p}\|\le c_n(1+\|V\|)\nrm{V(H-i)^{-1}}{n-1}^{n-1}$$
and $$\Tr(\tilde R^p_{n,H,f}(V))=\int(fu^p)^{(p)}\,d\breve\mu_{p}$$
for all $f\in C^{n+1}_\textnormal{c}\subseteq \mW_n$. Hence, $\eta_{p}(x):=-\breve\mu_{p}((-\infty,x))$ is a bounded function in $L^1_{\textnormal{loc}}(\R)$ and the proposition follows by the partial integration formula for distribution functions.
\end{proof}

Proposition \ref{prop:SSF locally V_k} will be extended from finite rank to relative Schatten class perturbations by an approximation argument. To carry out the latter we build some technical machinery below.

%

The following approximation of weighted perturbations is an important step in the approximation of the trace formula given by Proposition \ref{prop:SSF locally V_k}.

\begin{lem}\label{lem:approximating V by V_k}
	Let $\H$ be a Hilbert space, $H$ a self-adjoint operator in $\H$, and let $V\in \mB(\H)_{\text{sa}}$ be such that $V(H-i)^{-1}\in\S^n$. Then, there exists a sequence $(V_k)_k\subset\mB(\H)_{\text{sa}}$ of finite-rank operators such that $(V_k)_k$ converges strongly to $V$, such that
\begin{align}\label{eq:V_k Schatten conv}
\nrm{V_k(H-i)^{-1}-V(H-i)^{-1}}{n}\to0\, \text{ as } k\rightarrow\infty,	
\end{align}
and such that
\begin{align}
\label{c=1}
\norm{V_k}\leq\norm{V}\quad\text{and}
\quad\nrm{V_k(H-i)^{-1}}{n}\leq\nrm{V(H-i)^{-1}}{n}.
\end{align}
\end{lem}
\begin{proof}
We start with a sequence of spectral projections, denoted $$P_k:=E_H((-k,k)),$$ which by the functional calculus converges strongly to $\1$.
Applying subsequently the property of orthogonal projections and standard functional calculus we obtain
\begin{align}
\label{pk}
(P_kVP_k)((H-i)^{-1}P_k+(\1-P_k))
\nonumber
&=(P_kVP_k)((H-i)^{-1}P_k)\\
&=P_kV(H-i)^{-1}P_k\in\S^n
\end{align}
for each $k\in\N$.
By the functional calculus, $(H-i)^{-1}P_k+(\1-P_k)$ is invertible.
Therefore, from \eqref{pk} we derive $$P_kVP_k=P_kV(H-i)^{-1}P_k\left((H-i)^{-1}P_k+(\1-P_k)\right)^{-1}\in\S^n.$$
For a fixed $k$, by the spectral theorem of compact self-adjoint operators, there exists a sequence $(E_l)_{l=1}^\infty$ of finite-rank projections, each $E_l$ commuting with $P_kVP_k$, such that $E_lP_kVP_k$ converges to $P_kVP_k$ in $\S^n$ as $l\to\infty$. For all $k\in\N$, there exists $l_k\in\N$ such that $$\nrm{E_{l_k}P_kVP_k-P_kVP_k}{n}<1/k.$$
Define $$V_k:=E_{l_k}P_kVP_k.$$ Then $\|V_k\|\leq\|V\|$ holds, $V_k$ is self-adjoint, $V_k\to V$ strongly, and
\begin{align*}
\nrm{V_k(H-i)^{-1}-V(H-i)^{-1}}{n}\leq& \nrm{E_{l_k}P_kVP_k-P_kVP_k}{n}\norm{(H-i)^{-1}}\\
&+\nrm{P_kV(H-i)^{-1}P_k-V(H-i)^{-1}}{n}.
\end{align*}
By Lemma \ref{lem:Schatten product is continuous}, the latter expression converges to $0$ as $k\rightarrow\infty$.
The estimate
$$\nrm{E_{l_k}P_kVP_k(H-i)^{-1}}{n}
\leq\norm{E_{l_k}}\norm{P_k}\nrm{V(H-i)^{-1}}{n}\norm{P_k}.$$
implies the second inequality in \eqref{c=1}.
\end{proof}


Our approximation on the left-hand side of the trace formula in Proposition \ref{prop:SSF locally V_k} is based on the next estimate.
	
\begin{lem}\label{prop:|Tr(R)|}
Let $H$ be a self-adjoint operator in $\H$, let $n\in\N$, $n\neq2$, and let $V\in \mB(\H)_{\sa}$ be such that $V(H-i)^{-1}\in\S^n$.
Let $(V_k)_k\subset\mB(\H)_{\sa}$ be a sequence satisfying the assertions of Lemma \ref{lem:approximating V by V_k}.
Let $W\in \{V,V_m\}$, where $m\in\N$.
Then, given $a>0$, there exists $c_{n,H,V,a}>0$ such that
$$|\Tr (\tilde R^p_{n,H,f}(V_k)-\tilde R^p_{n,H,f}(W))|\leq c_{n,H,V,a}\supnorm{(fu^p)^{(p+1)}}\|\tilde{V}_k-\tilde{W}\|_n$$	
for all $p=0,\dots,n-1$, $k\in\N$, and $f\in C^{n+1}$ with $\supp(f)\subseteq[-a,a]$, where $\tilde R^p_{n,H,f}$ is given by \eqref{rpdef}.
In addition,
\begin{align*}
\Tr(R_{2,H,f}(V_k)-R_{2,H,f}(W))=\sum_{p=0}^2\int_0^1\Tr(R_{t,t,H,W,V_k,f}^p+R_{t,H,W,V_k,f}^p)\,dt
\end{align*}
for some operators $R_{t,t,H,W,V_k,f}^p$ and $R_{t,H,W,V_k,f}^p$ satisfying
$$|\Tr(R_{t,t,H,W,V_k,f}^p+R_{t,H,W,V_k,f}^p)|\leq c_{H,V}\|(fu^p)^{(p)}\|_\infty\|\tilde{V}_k-\tilde{W}\|_2$$
for all $f\in C^3_\textnormal{c}[-a,a]$.
\end{lem}

\begin{proof}
Let $n\geq3$. By \eqref{rpdef} in Lemma \ref{lem:pth part of remainder},
\begin{align}
\label{rvw0}
&\tilde R^p_{n,H,f}(V_k)-\tilde R^p_{n,H,f}(W)\\
\nonumber
&=\sum_{\substack{j_1,\ldots,j_{p}\geq1,j_{p+1}\geq0\\j_1+\ldots+j_{p+1}=n-1}}\!\,
\Big(T^{H,H+V_{k,j_1},H,\ldots,H}_{(fu^p)^{[p]}}(V_{k}(H+V_{k}-i)^{-1}\tilde{V}_k^{j_1-1}, \ldots,\tilde{V}_k^{j_p})\,\tilde{V}_k^{j_{p+1}}\\ \nonumber
&\quad-T^{H,H+W_{j_1},H,\ldots,H}_{(fu^p)^{[p]}}(W(H+W-i)^{-1}\tilde{W}^{j_1-1}, \ldots,\tilde{W}^{j_p})\,\tilde{W}^{j_{p+1}}\\
\nonumber
&\quad-T^{H,\ldots,H}_{(fu^p)^{[p]}}(\tilde{V}_k^{j_1},\ldots,\tilde{V}_k^{j_p})\tilde{V}_k^{j_{p+1}}
+T^{H,\ldots,H}_{(fu^p)^{[p]}}(\tilde{W}^{j_1}, \ldots,\tilde{W}^{j_p})\,\tilde{W}^{j_{p+1}}\Big),
\end{align}
where $V_{k,1}=V_k$, $W_1=W$ and $V_{k,j}=W_j=0$ for $j\neq1$.
Below we also use the notations $\breve{V}_k^{j}=V_k(H+V_k-i)^{-1}\tilde{V}_k^{j-1}$ and $\breve{W}^{j}=W(H+W-i)^{-1}\tilde{W}^{j-1}$.

Firstly we handle the summands in \eqref{rvw0} with $j_1=1$.
By \eqref{pf},
\begin{align}
\label{rvw1}
&T^{H,H+V_k,H,\ldots,H}_{(fu^p)^{[p]}}(\breve{V}_k,\tilde{V}_k^{j_2}, \ldots,\tilde{V}_k^{j_p})\,\tilde{V}_k^{j_{p+1}}
-T^{H,H+W,H,\ldots,H}_{(fu^p)^{[p]}}(\breve{V}_k,\tilde{V}_k^{j_2},\ldots,\tilde{V}_k^{j_p})\tilde{V}_k^{j_{p+1}}\\
\nonumber
&=T^{H,H+V_k,H+W,H,\ldots,H}_{(fu^p)^{[p+1]}}
(\breve{V}_k,V_k-W,\tilde{V}_k^{j_2},\ldots,\tilde{V}_k^{j_p})\tilde{V}_k^{j_{p+1}}.
\end{align}
By telescoping we obtain
\begin{align}
\label{rvw2}
&T^{H,H+W,H,\ldots,H}_{(fu^p)^{[p]}}(\breve{V}_k,\tilde{V}_k^{j_2},\ldots,\tilde{V}_k^{j_p})\tilde{V}_k^{j_{p+1}}
-T^{H,H+W,H,\ldots,H}_{(fu^p)^{[p]}}(\breve{W},\tilde{W}^{j_2},\ldots,\tilde{W}^{j_p})\tilde{W}^{j_{p+1}}\\  \nonumber
&\quad-T^{H,\ldots,H}_{(fu^p)^{[p]}}(\tilde{V}_k,\tilde{V}_k^{j_2},\ldots,\tilde{V}_k^{j_p})\tilde{V}_k^{j_{p+1}}
+T^{H,\ldots,H}_{(fu^p)^{[p]}}(\tilde{W},\tilde{W}^{j_2},\ldots,\tilde{W}^{j_p})\tilde{W}^{j_{p+1}}\\ \nonumber
=&\sum_{l=1}^{p+1}T^{H,H+W,H,\ldots,H}_{(fu^p)^{[p]}}(\breve{V}_k,\tilde{V}^{j_2}_k,\ldots,\tilde{V}_k^{j_{l-1}},\tilde{V}_k^{j_l}-\tilde{W}^{j_l},\tilde{W}^{j_{l+1}},\ldots,\tilde{W}^{j_p})\tilde{W}^{j_{p+1}}\\
\nonumber
&-\sum_{l=1}^{p+1}T^{H,\ldots,H}_{(fu^p)^{[p]}}( \tilde{V}_k,\tilde{V}_k^{j_2},\ldots,\tilde{V}_k^{j_{l-1}},\tilde{V}_k^{j_l}-\tilde{W}^{j_l},\tilde{W}^{j_{l+1}},\ldots,\tilde{W}^{j_p})\tilde{W}^{j_{p+1}}\\
\nonumber
=&\sum_{l=1}^{p+1}T^{H,H+W,H,\ldots,H}_{(fu^p)^{[p]}}(\breve{V}_k-\tilde{V}_k,\tilde{V}^{j_2}_k,\ldots,\tilde{V}_k^{j_{l-1}},\tilde{V}_k^{j_l}-\tilde{W}^{j_l},\tilde{W}^{j_{l+1}},\ldots,\tilde{W}^{j_p})\tilde{W}^{j_{p+1}}\\
\nonumber
&+\sum_{l=1}^{p+1}(T^{H,H+W,H,\ldots,H}_{(fu^p)^{[p]}}-T^{H,\ldots,H}_{(fu^p)^{[p]}})( \tilde{V}_k,\tilde{V}_k^{j_2},\ldots,\tilde{V}_k^{j_{l-1}},\tilde{V}_k^{j_l}-\tilde{W}^{j_l},\tilde{W}^{j_{l+1}},\ldots,\tilde{W}^{j_p})\tilde{W}^{j_{p+1}}.
\end{align}
Noticing that $\breve{V}^j-\tilde{V}^j=-\breve{V}^{j+1}$ and applying \eqref{pf} in the last sum in \eqref{rvw2} yields
\begin{align}
\label{rvw2b}
&T^{ H,H+W,H,\ldots,H}_{(fu^p)^{[p]}}(\breve{V}_k,\tilde{V}_k^{j_2},\ldots,\tilde{V}_k^{j_p})\tilde{V}_k^{j_{p+1}}
-T^{ H,H+W,H,\ldots,H}_{(fu^p)^{[p]}}(\breve{W},\tilde{W}^{j_2},\ldots,\tilde{W}^{j_p})\tilde{W}^{j_{p+1}}\\  \nonumber
&\quad-T^{H,\ldots,H}_{(fu^p)^{[p]}}(\tilde{V}_k,\tilde{V}_k^{j_2},\ldots,\tilde{V}_k^{j_p})\tilde{V}_k^{j_{p+1}}
+T^{H,\ldots,H}_{(fu^p)^{[p]}}(\tilde{W},\tilde{W}^{j_2},\ldots,\tilde{W}^{j_p})\tilde{W}^{j_{p+1}}\\
\nonumber
=&\sum_{l=1}^{p+1}\Big(-T^{H,H+W,H,\ldots,H}_{(fu^p)^{[p]}}(\breve{V}_k^2,\tilde{V}^{j_2}_k,\ldots,\tilde{V}_k^{j_{l-1}},\tilde{V}_k^{j_l}-\tilde{W}^{j_l},\tilde{W}^{j_{l+1}},\ldots,\tilde{W}^{j_p})\tilde{W}^{j_{p+1}}\\
\nonumber
&+T^{H,H+W,H,\ldots,H}_{(fu^p)^{[p+1]}}( \tilde{V}_k,W,\tilde{V}_k^{j_2},\ldots,\tilde{V}_k^{j_{l-1}},\tilde{V}_k^{j_l}-\tilde{W}^{j_l},\tilde{W}^{j_{l+1}},\ldots,\tilde{W}^{j_p})\tilde{W}^{j_{p+1}}\Big).
\end{align}
Secondly we handle the summands in \eqref{rvw0} with $j_1\neq1$. By telescoping we obtain
\begin{align}
\label{rvw2c}
&T^{H,\ldots,H}_{(fu^p)^{[p]}}(\breve{V}_k^{j_1},\tilde{V}_k^{j_2},\ldots,\tilde{V}_k^{j_p})\,\tilde{V}_k^{j_{p+1}} -T^{H,\ldots,H}_{(fu^p)^{[p]}}(\breve{W}^{j_1},\tilde{W}^{j_2},\ldots,\tilde{W}^{j_p})\,
\tilde{W}^{j_{p+1}}\\
\nonumber
&-T^{H,\ldots,H}_{(fu^p)^{[p]}}(\tilde{V}_k^{j_1},\tilde{V}_k^{j_2},\ldots,\tilde{V}_k^{j_p})\,\tilde{V}_k^{j_{p+1}} +T^{H,\ldots,H}_{(fu^p)^{[p]}}(\tilde{W}^{j_1},\tilde{W}^{j_2},\ldots,\tilde{W}^{j_p})\,
\tilde{ W}^{j_{p+1}}\\
\nonumber
&=\sum_{l=1}^{p+1}\Big(T^{H,\ldots,H}_{(fu^p)^{[p]}}(\breve{V}_k^{j_1},\tilde{V}^{j_2}_k,\ldots,\tilde{V}_k^{j_{l-1}},\tilde{V}_k^{j_l}-\tilde{W}^{j_l},\tilde{W}^{j_{l+1}},\ldots,\tilde{W}^{j_p})\tilde{W}^{j_{p+1}}\\
\nonumber
&\quad-T^{H,\ldots,H}_{(fu^p)^{[p]}}( \tilde{V}_k^{j_1},\tilde{V}_k^{j_2},\ldots,\tilde{V}_k^{j_{l-1}},\tilde{V}_k^{j_l}-\tilde{W}^{j_l},\tilde{W}^{j_{l+1}},\ldots,\tilde{W}^{j_p})\tilde{W}^{j_{p+1}}\Big)\\
\nonumber
&=-\sum_{l=1}^{p+1}T^{H,\ldots,H}_{(fu^p)^{[p]}}(\breve{V}_k^{j_1+1},\tilde{V}^{j_2}_k,\ldots,\tilde{V}_k^{j_{l-1}},\tilde{V}_k^{j_l}-\tilde{W}^{j_l},\tilde{W}^{j_{l+1}},\ldots,\tilde{W}^{j_p})\tilde{W}^{j_{p+1}}.
\end{align}

Combining \eqref{rvw0}, \eqref{rvw1}, \eqref{rvw2b}, and \eqref{rvw2c} yields
\begin{align}
\label{rvw3a}
&\tilde R^p_{n,H,f}(V_k)-\tilde R^p_{n,H,f}(W)\\
\nonumber
&=\sum_{\substack{j_2,\ldots,j_{p}\geq1\\j_{p+1}\geq0\\j_2+\ldots+j_{p+1}=n-2}}
\Big(T^{H,H+V_k,H+W,H,\ldots,H}_{(fu^p)^{[p+1]}}
(\breve{V}_k,V_k-W,\tilde{V}_k^{j_2},\ldots,\tilde{V}_k^{j_p})\tilde{V}_k^{j_{p+1}}\\
\nonumber
&\quad+\sum_{l=1}^{p+1}\big(T^{H,H+W,H,\ldots,H}_{(fu^p)^{[p+1]}}( \tilde{V}_k,W,\tilde{V}_k^{j_2},\ldots,\tilde{V}_k^{j_{l-1}},\tilde{V}_k^{j_l}-\tilde{W}^{j_l},\tilde{W}^{j_{l+1}},\ldots,\tilde{W}^{j_p})\tilde{W}^{j_{p+1}}\\
\nonumber
&\quad-T^{H,H+W,H,\ldots,H}_{(fu^p)^{[p]}}(\breve{V}_k^2,\tilde{V}^{j_2}_k,\ldots,\tilde{V}_k^{j_{l-1}},\tilde{V}_k^{j_l}-\tilde{W}^{j_l},\tilde{W}^{j_{l+1}},\ldots,\tilde{W}^{j_p})\tilde{W}^{j_{p+1}}\big)\Big)\\
\nonumber
&\quad-\sum_{\substack{j_2,\ldots,j_{p}\geq1\\j_1\ge 2,j_{p+1}\geq0\\j_1+\ldots+j_{p+1}=n-1}}
\sum_{l=1}^{p+1}T^{H,\ldots,H}_{(fu^p)^{[p]}}(\breve{V}_k^{j_1+1},\tilde{V}^{j_2}_k,\ldots,\tilde{V}_k^{j_{l-1}},\tilde{V}_k^{j_l}-\tilde{W}^{j_l},\tilde{W}^{j_{l+1}},\ldots,\tilde{W}^{j_p})\tilde{W}^{j_{p+1}}.
\end{align}
By \eqref{weight1} of Theorem \ref{thm:adding resolvents}\ref{weight}, for $p\geq1$ we have
\begin{align}
\label{rvw3b}
&T^{H,H+V_k,H+W,H,\ldots,H}_{(fu^p)^{[p+1]}}
(\breve{V}_k,V_k-W,\tilde{V}_k^{j_2},\ldots,\tilde{V}_k^{j_p})\tilde{V}_k^{j_{p+1}}\\
\nonumber
&=T^{H,H+V_k,H+W,H,\ldots,H}_{(fu^{p+1})^{[p+1]}}
(\breve{V}_k,(V_k-W)(H+W-i)^{-1},\tilde{V}_k^{j_2},\ldots,\tilde{V}_k^{j_p})\tilde{V}_k^{j_{p+1}}\\
\nonumber
&\quad-T^{H,H+V_k,H,\ldots,H}_{(fu^p)^{[p]}}
(\breve{V}_k,(V_k-W)(H+W-i)^{-1}\tilde{V}_k^{j_2},\tilde{V}_k^{j_3},\ldots,\tilde{V}_k^{j_p})\tilde{V}_k^{j_{p+1}}
\end{align}
and
\begin{align}
\label{rvw3c}
&T^{H,H+W,H,\ldots,H}_{(fu^p)^{[p+1]}}( \tilde{V}_k,W,\tilde{V}_k^{j_2},\ldots,\tilde{V}_k^{j_{l-1}},\tilde{V}_k^{j_l}-\tilde{W}^{j_l},\tilde{W}^{j_{l+1}},\ldots,\tilde{W}^{j_p})\tilde{W}^{j_{p+1}}\\
&\nonumber
=T^{H,H+W,H,\ldots,H}_{(fu^{p+1})^{[p+1]}}( \tilde{V}_k,\tilde{W},\tilde{V}_k^{j_2},\ldots,\tilde{V}_k^{j_{l-1}},\tilde{V}_k^{j_l}-\tilde{W}^{j_l},\tilde{W}^{j_{l+1}},\ldots,\tilde{W}^{j_p})\tilde{W}^{j_{p+1}}\\
\nonumber
&\quad-T^{H,H+W,H,\ldots,H}_{(fu^p)^{[p]}}( \tilde{V}_k,\tilde{W}\tilde{V}_k^{j_2},\tilde{V}_k^{j_3},\ldots,\tilde{V}_k^{j_{l-1}},\tilde{V}_k^{j_l}-\tilde{W}^{j_l},\tilde{W}^{j_{l+1}},\ldots,\tilde{W}^{j_p})\tilde{W}^{j_{p+1}}.
\end{align}
Combining \eqref{rvw3a}--\eqref{rvw3c} yields
\begin{align}
\label{rvw3}
&\tilde R^p_{n,H,f}(V_k)-\tilde R^p_{n,H,f}(W)\\
\nonumber
&=\sum_{\substack{j_1,\ldots,j_{p}\geq1,j_{p+1}\geq0\\
j_1+\ldots+j_{p+1}=n-1}}\bigg(T^{ H,H+V_{k,j_1},H+W_{j_1},H,\ldots,H}_{(fu^{p+1})^{[p+1]}}
(\breve{V}^{j_1}_k,(V_{k,j_1}-W_{j_1})(H+W-i)^{-1},\tilde{V}^{j_2}_k,\ldots,\tilde{V}^{j_p}_k)\tilde{V}^{j_{p+1}}_k\\
\nonumber
&\quad-T^{H,H+V_{k,j_1},H,\ldots,H}_{(fu^p)^{[p]}}(\breve{V}_k^{j_1},(V_{k,j_1}-W_{j_1})(H+W-i)^{-1}\tilde{V}_k^{j_2},\ldots,\tilde{V}_k^{j_p})\tilde{V}_k^{j_{p+1}}\\
\nonumber
&\quad+\sum_{l=1}^{p+1}\Big(T^{ H,H+W_{j_1},H,\ldots,H}_{(fu^{p+1})^{[p+1]}}
(\tilde{V}^{j_1}_k,\tilde{W}_{j_1},\tilde{V}_k^{j_2},\ldots,\tilde{V}_k^{j_{l-1}},\tilde{V}_k^{j_l}-\tilde{W}^{j_l},
\tilde{W}^{j_{l+1}},\ldots,\tilde{W}^{j_p})\tilde{W}^{j_{p+1}}\\
\nonumber
&\quad-T^{H,H+W_{j_1},H,\ldots,H}_{(fu^{p})^{[p]}}
(\tilde{V}^{j_1}_k,\tilde W_{j_1} \tilde{V}_k^{j_2},\ldots,\tilde{V}_k^{j_{l-1}},\tilde{V}_k^{j_l}-\tilde{W}^{j_l},
\tilde{W}^{j_{l+1}},\ldots,\tilde{W}^{j_p})\tilde{W}^{j_{p+1}}\\ \nonumber
&\quad-T^{H,H+W_{j_1},H,\ldots,H}_{(fu^{p})^{[p]}}
(\breve{V}^{j_1+1}_k,\tilde{V}_k^{j_2},\ldots,\tilde{V}_k^{j_{l-1}},\tilde{V}_k^{j_l}-\tilde{W}^{j_l},
\tilde{W}^{j_{l+1}},\ldots,\tilde{W}^{j_p})\tilde{W}^{j_{p+1}}\Big)\bigg).
\end{align}

A straightforward application of the second resolvent identity implies
\begin{align*}
(V_k-W)(H+W-i)^{-1}=(V_k-W)(H-i)^{-1}(\1-W(H+W-i)^{-1}).
\end{align*}
For each $W\in\{V, V_m\}$, by the estimates \eqref{c=1} of
Lemma \ref{lem:approximating V by V_k}, we obtain
\begin{align}
\label{vfromw}
&\tnrm{\tilde{W}}{n}\leq\tnrm{\tilde{V}}{n}.
\end{align}
and
\begin{align*}
\norm{\1-W(H+W-i)^{-1}}\leq 1+\norm{V}.
\end{align*}
By the latter estimate,
\begin{align*}
\nrm{(V_k-W)(H+W-i)^{-1}}{n}\leq(1+\norm{V})\|\tilde V_k-\tilde W\|_n.
\end{align*}
It follows from \eqref{vfromw} and the telescoping identity
$\tilde V_k^j-\tilde W^j=\sum_{i=0}^{j-1}\tilde V_k^i(\tilde V_k-\tilde W)\tilde W^{j-1-i}$ that
$$\tnrm{\tilde{V}_k^j-\tilde{W}^j}{n/j}\leq j\tnrm{\tilde{V}}{n}^{j-1}\tnrm{\tilde{V}_k-\tilde{W}}{n}.$$
Applying the latter bound
and Lemma \ref{lem:Hahn Banach Borel measures} in \eqref{rvw3} implies
\begin{align}\label{eq:diff pth remainders}
&|\Tr(\tilde R^p_{n,H,f}(V_k)-\tilde R^p_{n,H,f}(W))|\nonumber\\
&\quad\leq\sum_{\substack{j_1,\ldots,j_{p}\geq1,j_{p+1}\geq0\\j_1+\ldots+j_{p+1}=n-1}}\Big(c^1_{n,j}
\supnorm{(fu^{p+1})^{(p+1)}}+c^2_{n,j}\supnorm{(fu^p)^{(p)}}\Big)C_{n,V,H}
\tnrm{\tilde{V}_k-\tilde{W}}{n},
\end{align}
for some constants $c^1_{n,j}$ and $c^2_{n,j}$ depending only on $n$ and $j_1,\ldots,j_{p+1}$, and the constant $$C_{n,V,H}:=(1+\norm{V})^{2}\,\|\tilde{V}\|_{n}^{n-1}.$$
If $\supp f\subseteq [-a,a]$, then the fundamental theorem of calculus gives
\begin{align*}
\supnorm{(fu^{p})^{(p)}}\leq2a\supnorm{(fu^p)^{(p+1)}}.
\end{align*}
Since $(fu^{p+1})^{(p+1)}=(fu^p)^{(p+1)}u+(p+1)(fu^p)^{(p)}$, we obtain
	$$\supnorm{(fu^{p+1})^{(p+1)}}\leq(|u(a)|+2a(p+1))\supnorm{(fu^p)^{(p+1)}}.$$
Along with \eqref{eq:diff pth remainders}, the latter two inequalities yield the result for $n\geq 3$.

If $n=1$, then $p=0$ and \eqref{R0} gives $\tilde R^0_{1,H,f}(V_k)-\tilde R^0_{1,H,f}(W)=f(H+V_k)-f(H+W)$.
Hence, by Theorem \ref{dm} and the fundamental theorem of calculus,
\begin{align*}
\tilde R^0_{1,H,f}(V_k)-\tilde R^0_{1,H,f}(W)
&=\int_0^1 T^{H_t,H_t}_{f^{[1]}}(V_k-W)\,dt,
\end{align*}
where $H_t=H+W+t(V_k-W)$. By \eqref{weight1} of Theorem \ref{thm:adding resolvents}\ref{weight} for $j=1$ applied to $T^{H_t,H_t}_{f^{[1]}}(V_k-W)$ and by continuity of the trace, we obtain
\begin{align*}
\Tr(\tilde R^0_{1,H,f}(V_k)-\tilde R^0_{1,H,f}(W))=\int_0^1 \big(&\Tr(T^{H_t,H_t}_{(fu)^{[1]}}((V_k-W)(H_t-i)^{-1}))\\
&-\Tr(f(H_t)(V_k-W)(H_t-i)^{-1})\big)\,dt.
\end{align*}
Noticing that
\begin{align}
\label{tvianot}
\sup_{t\in [0,1]}\|(V_k-W)(H_t-i)^{-1}\|_1&\le(1+\|V_k-W\|)\|\tilde{V}_k-\tilde{W}\|_1\\
\nonumber
&\le(1+2\|V\|)\|\tilde{V}_k-\tilde{W}\|_1
\end{align}
and applying H\"{o}lder's inequality and
the Riesz--Markov representation theorem completes the proof of the result for $n=1$.

If $n=2$, then by Theorem \ref{dm} and the fundamental theorem of calculus,
\begin{align*}
&R_{2,H,f}(V_k)-R_{2,H,f}(W)\\
&=f(H+V_k)-f(H)-T^{H,H}_{f^{[1]}}(V_k)-(f(H+W)-f(H)-T^{H,H}_{f^{[1]}}(W))\\
&=f(H+V_k)-f(H+W)-T^{H,H}_{f^{[1]}}(V_k-W)\\
&=\int_0^1T^{H_t,H_t}_{f^{[1]}}(V_k-W)dt-\int_0^1 T^{H,H}_{f^{[1]}}(V_k-W)dt.
\end{align*}
By \eqref{pf},
\begin{align*}
&R_{2,H,f}(V_k)-R_{2,H,f}(W)\\
&=\int_0^1(T^{H_t,H_t}_{f^{[1]}}(V_k-W)-T^{H,H_t}_{f^{[1]}}(V_k-W)+T^{H,H_t}_{f^{[1]}}(V_k-W)-T^{H,H}_{f^{[1]}}(V_k-W))dt\\
&=\int_0^1(T^{H_t,H,H_t}_{f^{[2]}}(W+t(V_k-W),V_k-W)+T^{H,H_t,H}_{f^{[2]}}(V_k-W,W+t(V_k-W)))dt.
\end{align*}
By Theorem \ref{thm:adding resolvents}\ref{weights},
\begin{align*}
&T^{H_t,H,H_t}_{f^{[2]}}(W+t(V_k-W),V_k-W)\\
&=f(H_t)(W+t(V_k-W))(H-i)^{-1}(V_k-W)(H_t-i)^{-1}\\
&\quad-T_{(fu)^{[1]}}^{H_t,H}\big((W+t(V_k-W))(H-i)^{-1}\big)(V_k-W)(H_t-i)^{-1}\\
&\quad-T_{(fu)^{[1]}}^{H_t,H_t}\big((W+t(V_k-W))(H-i)^{-1}(V_k-W)(H_t-i)^{-1}\big)\\
&\quad+T_{(fu^2)^{[2]}}^{H_t,H,H_t}\big((W+t(V_k-W))(H-i)^{-1},(V_k-W)(H_t-i)^{-1}\big)
\end{align*}
and
\begin{align*}
&T^{H,H_t,H}_{f^{[2]}}(V_k-W,W+t(V_k-W))\\
&=f(H)(V_k-W)(H_t-i)^{-1}(W+t(V_k-W))(H-i)^{-1}\\
&\quad-T_{(fu)^{[1]}}^{H,H_t}\big((V_k-W)(H_t-i)^{-1}\big)(W+t(V_k-W))(H-i)^{-1}\\
&\quad-T_{(fu)^{[1]}}^{H,H}\big((V_k-W)(H_t-i)^{-1}(W+t(V_k-W))(H-i)^{-1}\big)\\
&\quad+T_{(fu^2)^{[2]}}^{H,H_t,H}\big((V_k-W)(H_t-i)^{-1},(W+t(V_k-W))(H-i)^{-1}\big).
\end{align*}
Denote
\begin{align*}
R_{t,t,H,W,V_k,f}^0=&f(H_t)(W+t(V_k-W))(H-i)^{-1}(V_k-W)(H_t-i)^{-1},\\
R_{t,t,H,W,V_k,f}^1=&-T_{(fu)^{[1]}}^{H_t,H}\big((W+t(V_k-W))(H-i)^{-1}\big)(V_k-W)(H_t-i)^{-1}\\
&\quad-T_{(fu)^{[1]}}^{H_t,H_t}\big((W+t(V_k-W))(H-i)^{-1}(V_k-W)(H_t-i)^{-1}\big),\\
R_{t,t,H,W,V_k,f}^2=&T_{(fu^2)^{[2]}}^{H_t,H,H_t}\big((W+t(V_k-W))(H-i)^{-1},(V_k-W)(H_t-i)^{-1}\big),\\
R_{t,H,W,V_k,f}^0=&f(H)(V_k-W)(H_t-i)^{-1}(W+t(V_k-W))(H-i)^{-1},\\
R_{t,H,W,V_k,f}^1=&-T_{(fu)^{[1]}}^{H,H_t}\big((V_k-W)(H_t-i)^{-1}\big)(W+t(V_k-W))(H-i)^{-1}\\
&\quad-T_{(fu)^{[1]}}^{H,H}\big((V_k-W)(H_t-i)^{-1}(W+t(V_k-W))(H-i)^{-1}\big),\\
R_{t,H,W,V_k,f}^2=&T_{(fu^2)^{[2]}}^{H,H_t,H}\big((V_k-W)(H_t-i)^{-1},(W+t(V_k-W))(H-i)^{-1}\big).
\end{align*}
Applying continuity of $t\mapsto\Tr(R_{t,t,H,W,V_k,f}^p)$ and $t\mapsto\Tr(R_{t,H,W,V_k,f}^p)$
(see \cite[Proposition 4.3.15]{ST19}) yields
\begin{align*}
\Tr(R_{2,H,f}(V_k)-R_{2,H,f}(W))=\sum_{p=0}^2\int_0^1\Tr(R_{t,t,H,W,V_k,f}^p+R_{t,H,W,V_k,f}^p)\,dt.
\end{align*}
By Lemma \ref{lem:Hahn Banach Borel measures}, \eqref{c=1}, and an analog of \eqref{tvianot} for the Hilbert-Schmidt norm, we obtain
\begin{align*}
|\Tr(R_{t,t,H,W,V_k,f}^p+R_{t,H,W,V_k,f}^p)|
\leq c_{H,V}\|(fu^p)^{(p)}\|_\infty\|\tilde{V}_k-\tilde{W}\|_2
\end{align*}
for every $t\in[0,1]$, completing the proof of the lemma.
\end{proof}


Below we extend the result of Proposition \ref{prop:SSF locally V_k} to relative Schatten class perturbations.

\begin{proof}[Proof of Proposition \ref{prop:SSF locally}]
Let $(V_k)_k$ be a sequence provided by Lemma \ref{lem:approximating V by V_k}.
For every $p\in\{0,\dots,n-1\}$ and $k\in\N$, let $\breve\eta_{p,k}$ be a function satisfying
$$\Tr(\tilde R^p_{n,H,f}(V_k))=\int(fu^p)^{(p+1)}(x)\breve\eta_{p,k}(x)\,dx,$$
which exists by Proposition \ref{prop:SSF locally V_k}. By Lemma \ref{prop:|Tr(R)|} applied to $W=V_m$, we have
\begin{align*}
\norm{\breve\eta_{p,k}-\breve\eta_{p,m}}_{L^1((-a,a))}
&=\sup_{\substack{ f\in C^{n+1},\\ \supp(f)\subseteq[-a,a],\\
\supnorm{(fu^p)^{(p+1)}}\leq1}}|\Tr(\tilde R^p_{n,H,f}(V_k)-\tilde R^p_{n,H,f}(V_m))|\\
&\leq c_{n,H,V,a}\|\tilde{V}_k-\tilde{V}_m\|_n.
\end{align*}
By Lemma \ref{lem:approximating V by V_k}, the latter expression approaches $0$ as $k\geq m\to\infty$. Thus, $(\breve\eta_{p,k})_k$ is Cauchy with respect to seminorms in which $L^1_\text{loc}$ is complete. Let $\breve\eta_p$ be its $L^1_{\textnormal{loc}}$-limit.

Assume that $f\in C^{n+1}_\textnormal{c}$. We obtain
\begin{align*}
\int_\R(fu^p)^{(p+1)}(x)\,\breve\eta_p(x)\,dx
=&\int_{\supp f}(fu^p)^{(p+1)}(x)\,\breve\eta_p(x)\,dx\\
=&\lim_{k\rightarrow\infty}\int_{\supp f}(fu^p)^{(p+1)}(x)\,\breve\eta_{p,k}(x)\,dx\\
=&\lim_{k\rightarrow\infty}\Tr(\tilde R^p_{n,H,f}(V_k)).
\end{align*}
By Lemma \ref{prop:|Tr(R)|} applied to $W=V$,
$$|\Tr (\tilde R^p_{n,H,f}(V_k)-\tilde R^p_{n,H,f}(V))|\leq c_{n,H,V,a}\supnorm{(fu^p)^{(p+1)}}\|\tilde{V}_k-\tilde{V}\|_n\,$$	
for every $k\in\N$.
Hence, by Lemma \ref{lem:approximating V by V_k},
\begin{align*}
\Tr(\tilde R^p_{n,H,f}(V))=\lim_{k\rightarrow\infty}\Tr(\tilde R^p_{n,H,f}(V_k))
=\int_\R(fu^p)^{(p+1)}(x)\,\breve\eta_p(x)\,dx,
\end{align*}
completing the proof of the result.
\end{proof}

\subsection{Absolute continuity of the spectral shift measure}

In this subsection we prove our main result; existence and properties of a spectral shift function for relative Schatten class perturbations. We will combine the results of \textsection\ref{sct:Existence of the spectral shift measure} and \textsection\ref{sct:Alternative trace formula}, namely, Propositions \ref{prop:SSM with growth} and \ref{prop:SSF locally}.

\begin{proof}[Proof of Theorem \ref{rsmain}]
Let $f\in C_\textnormal{c}^{n+1}$. We provide a proof in the case $n\ge 3$; the cases $n=1$ and $n=2$ can be proved completely analogously.

Applying the general Leibniz differentiation rule on the right-hand side of \eqref{trfp} (see Proposition \ref{prop:SSF locally}) gives
\begin{align*}
\Tr(R_{n,H,f}(V))
&=\sum_{p=0}^{n-1}(-1)^{n-1-p}\int_\R(fu^p)^{(p+1)}(x)\,\breve\eta_p(x)\,dx.\\
&=\sum_{p=0}^{n-1}(-1)^{n-1-p}\sum_{k=0}^{p+1}\int_\R\vect{p+1}{k} f^{(k)}(x)(u^p)^{(p+1-k)}(x)\breve\eta_p(x)\,dx\,\\
&=\sum_{p=0}^{n-1}(-1)^{n-1-p}\sum_{k=1}^{p+1}\int_\R f^{(k)}(x) \vect{p+1}{k}\frac{p!}{(k-1)!}u^{k-1}(x)\breve\eta_p(x)\,dx.
\end{align*}
Integration by parts gives
\begin{align}
\nonumber
\Tr(R_{n,H,f}(V))=&\sum_{p=0}^{n-1}\int_\R f^{(p+1)}(x)\tilde{\eta}_p(x)\,dx,
\end{align}
where
\begin{align*}
\tilde{\eta}_p(t)
=\sum_{k=1}^{p+1}\frac{(-1)^{n-k}\,(p+1)!\,p!}{(p+1-k)!\,k!\,(k-1)!}\int_0^t ds_1\int_0^{s_1}ds_2\cdots\int_0^{s_{p-k}}u^{k-1}(x)\breve\eta_p(x)\,dx.
\end{align*}
Subsequent integration by parts gives
\begin{align}
\label{eta}
\nonumber&\Tr(R_{n,H,f}(V))\\
\nonumber
&=\int_\R f^{(n)}(x)\Bigg(\sum_{p=0}^{n-1}(-1)^{n-1-p}\int_0^xds_1\int_0^{s_1}ds_2\cdots
\int_0^{s_{n-p-2}}\tilde{\eta}_p(t)\,dt\Bigg)\,dx\quad\\
 &=:\int_\R f^{(n)}(x)\grave\eta_n(x)\,dx
\end{align}
for every $f\in C_\textnormal{c}^{n+1}$.
Since $\breve{\eta}_p\in L^1_{\textnormal{loc}}$ (see Proposition \ref{prop:SSF locally}),
we have that $\tilde{\eta}_p\in L^1_{\textnormal{loc}}$ and, hence, $\grave\eta_n\in L^1_{\textnormal{loc}}$.

By Proposition \ref{prop:SSM with growth}, there exists a locally finite Borel measure $\mu_n$ satisfying \eqref{mu tilde} and determined by \eqref{mu tilde} for every $f\in C_\textnormal{c}^{n+1}$ uniquely up to an absolutely continuous measure whose density is a polynomial of degree at most $n-1$.
Combining the latter with \eqref{eta} implies
\begin{align}
\label{311}
d\mu_n(x)=\grave\eta_n(x)dx+p_{n-1}(x)dx =:\acute\eta_n(x)dx,
\end{align}
where $p_{n-1}$ is a polynomial of degree at most $n-1$. By Proposition \ref{prop:SSM with growth}, the function $\acute\eta_n:=\grave\eta_n+p_{n-1}$ satisfies \eqref{tff} for every $f\in\mW_n$. The fact that $u^{-n-\epsilon}d\mu_n$ is a finite measure translates to $\acute\eta_n\in L^1(\R,u^{-n-\epsilon}(x)dx)$.

It follows from \eqref{munfla} that
\begin{align*}
\|u^{-n-\epsilon}\,d\mu_n\|\le\|u^{-\epsilon}\|_\infty\|\nu_n\|
+\|u^{-n-\epsilon}\,\xi_n\|_1.
\end{align*}
Along with \eqref{eq:nu bound} and \eqref{eq:xi bound}, the latter implies
\begin{align*}
\|u^{-n-\epsilon}\,d\mu_n\|\le c_n(1+\|u^{-1-\epsilon}\|_1)(1+\norm{V})\nrm{V(H-i)^{-1}}{n}^n.
\end{align*}
Since
\begin{align}
\label{312}
\int_0^1(1+x^2)^{(-1-\epsilon)/2}\,dx\le 1\quad\text{and}\quad
\int_1^\infty(1+x^2)^{(-1-\epsilon)/2}\,dx\le\int_1^\infty x^{-1-\epsilon}\,dx=\epsilon^{-1},
\end{align}
we obtain the bound
\begin{align}
\label{313}
\|u^{-n-\epsilon}\,d\mu_n\|\le c_n\,(1+\epsilon^{-1})(1+\norm{V})\nrm{V(H-i)^{-1}}{n}^n,
\end{align}
which translates to
	$$\int_\R |\acute\eta_n(x)|\,\frac{dx}{(1+|x|)^{n+\epsilon}}\leq c_n(1+\epsilon^{-1})(1+\norm{V})\|V(H-i)^{-1}\|_n^n.$$
We define $$\eta_n:=\Re(\acute\eta_n),$$ and obtain \eqref{eta estimate} by using $|\eta_n|\leq|\acute\eta_n|$.
As we have seen, $\acute\eta_n$ satisfies \eqref{tff} for all $f\in\mW_n$. Therefore,
\begin{align}
\label{eq:tr formula real and im part}
	\Tr(R_{n,H,f}(V))=\int_\R f^{(n)}(x)\eta_n(x)\,dx+i\int_\R f^{(n)}(x)\Im(\acute\eta_n(x))\,dx.
\end{align}
When $f\in \mW_n$ is real-valued, the left-hand side of \eqref{eq:tr formula real and im part} is real, and consequently the second term on the right-hand side of \eqref{eq:tr formula real and im part} vanishes. The latter implies \eqref{tff} for real-valued $f\in\mW_n$. By applying \eqref{tff} to the real-valued functions $\Re(f)$ and $\Im(f)$, we extend \eqref{tff} to all $f\in\mW_n$.

The uniqueness of $\eta_n$ satisfying \eqref{tff} up to a polynomial summand of order at most $n-1$ can be established completely analogously to the uniqueness of the measure $\mu_n$ established in Proposition \ref{prop:SSM with growth}.
\end{proof}

\section{Examples}
\label{sec4}

In this section we discuss two classes of examples -- arising in mathematical physics and noncommutative geometry, respectively -- that satisfy the condition \eqref{vresinsn}.

\subsection{Differential operators}
\label{sec4b}
In this section we consider conditions sufficient for perturbations of Dirac and Schr\"{o}dinger operators to satisfy \eqref{vresinsn}.

We will consider self-adjoint perturbations $V=M_v$ given by multiplication by a real-valued function $v\in L^\infty(\R^d)$.
%
Let $$\Delta=\sum_{k=1}^d\frac{\partial^2}{\partial x_k^2}$$ denote the Laplacian operator densely defined in the Hilbert space $ L^2(\mathbb{R}^d)$.

For $m\geq0$, let $D_m$ denote the free massive Dirac operator defined as follows. For $d\in\mathbb{N},$ let $N(d):=2^{\lfloor (d+1)/2\rfloor}$. Let $e_k\in M_{N(d)}(\mathbb{C}),$ $0\leq k\leq d,$ be the Clifford generators, that is, self-adjoint matrices satisfying $e_k^2=\1$ for $0\leq k\leq d$ and
$e_{k_1}e_{k_2}=-e_{k_2}e_{k_1}$ for $0\leq k_1,k_2\leq d,$ such that
$k_1\neq k_2.$ Let $\frac{\partial}{i\partial x_k}:=\frac{1}{i}\frac{\partial}{\partial x_k}$. Then, the operator
$$D_m:=e_0\otimes m\1+\sum_{k=1}^de_k\otimes \frac{\partial}{i\partial x_k}$$
is densely defined in the Hilbert space $\mathbb{C}^{N(d)}\otimes L^2(\mathbb{R}^d).$
We note that $D_0$ is unitarily equivalent to $\1\otimes D$, where $\1\in M_{N(d)/N(d-1)}(\C)$ and $D$ is the usual massless Dirac operator. We also note that, in the case when $d=1$, the Dirac operator $D_0=\1\otimes\frac{\partial}{i\partial x}$ can be identified with the differential operator $\frac{\partial}{i\partial x}$ in the Hilbert space $L^2(\R)$.

The Schatten class membership of the weighted resolvents below was derived in \cite[Theorem 3.3 and Remark 3.6]{S21}.
To estimate the respective Schatten norms one just needs to carefully follow the proof of the latter result. The respective result for $p\in [1,2)$ is found in \cite{vNS20}, see also \cite{S21}.

\begin{thm}
\label{src}
Let $d\in\N$, $2\leq p<\infty$. Let
$v\in L^p(\R^d)\cap L^\infty(\R^d)$
be real-valued.
\begin{enumerate}[label=\textnormal{(\roman*)}]
\item \label{srci}
If $p>d$ and $m\ge 0$, then
$(\1\otimes M_v)(D_m-i)^{-1}\in\S^p$
and
\begin{align}
\label{wrd}
\|(\1\otimes M_v)(D_m-i)^{-1}\|_p\le c_{d,p}\|v\|_{p}.
\end{align}

\item \label{srcii}
If $p>\frac{d}{2}$, then
$M_v(-\Delta-i)^{-1}\in\S^p$
and
\begin{align}
\label{wrs}
\|M_v(-\Delta-i)^{-1}\|_p\le c_{d,p}\|v\|_{p}.
\end{align}
\end{enumerate}
\end{thm}

\begin{rema}
\label{perturbedr}
The bounds analogous to \eqref{wrd} and \eqref{wrs} can also be established for perturbed Dirac $D_m+W$ and perturbed Schr\"{o}dinger $-\Delta+W$ operators, respectively. The respective results follow from Theorem \ref{src} and
Proposition \ref{perturbedp} below. In particular, we have the following bound for a massive Dirac operator with
electromagnetic potential in the case $p>d$:
\begin{align*}
&\Big\|(\1\otimes M_v)\Big(D_m+\sum_{k=1}^{d}e_k\otimes M_{w_k}+\1\otimes M_{w_{d+1}}-i\Big)^{-1}\Big\|_p
\le c_{d,p}\big(1+\max_{1\le k\le d+1}\|w_k\|_{\infty}\big)\|v\|_{p},
\end{align*}
for all real-valued functions $w_1,\ldots,w_{d+1}\in C_\textnormal{b}(\R)$.
The same reasoning applies to generalized Dirac operators $\1\otimes D+W$, where $\1\in M_k(\C)$ for $k\in\N$ and $W\in\mB(\C^k\otimes\H)_{\sa}$, that are associated with almost-commutative spectral triples (see \cite[Chapter 8]{Sui15}).
\end{rema}

\begin{prop}
\label{perturbedp}
Let $H,V$ be self-adjoint operators in $\H$ and $W\in\mB(\H)_{\sa}$. Let $1\le p<\infty$ and assume that
$\|V(H-i)^{-1}\|_p<\infty$. Then,
\begin{align*}
\|V(H+W-i)^{-1}\|_p\le\|V(H-i)^{-1}\|_p(1+\|W\|).
\end{align*}
\end{prop}

\begin{proof}
The result follows from the second resolvent identity
\begin{align*}
(H+W-i)^{-1}=(H-i)^{-1}-(H-i)^{-1}W(H+W-i)^{-1}
\end{align*}
upon multiplying it by $V$ and applying H\"{o}lder's inequality for Schatten norms.
\end{proof}

\subsection{Noncommutative geometry}
\label{sec4a}
In this subsection we show that the relative Schatten class condition occurs naturally in noncommutative geometry, namely, in inner perturbations of regular locally compact spectral triples, according to Definition \ref{lcst} below. A locally compact spectral triple is a generalization of a finitely summable spectral triple (Definition \ref{defi:st}) to the case where the algebra is possibly nonunital. Variations on this definition occur in, e.g., \cite[Definitions 2.4 and 2.5]{CGPRS}, \cite[Definitions 2.1 and 2.15 and Proposition 2.14]{CGRS}, \cite[Definition 7.7]{CPR}, \cite[Hypothesis 1.2.1 and \textsection 2.2.3]{SZ18}. In any case, many examples (including noncommutative field theory, \cite{GGISV}) satisfy the definition that is given below.

Let $\operatorname{dom}(D)$ denote the domain of any operator $D$
and let \begin{align*}
\delta_D(T):=[|D|,T]
\end{align*}
be defined on those $T\in\mB(\H)$ for which $\delta_D(T)$ extends to a bounded operator.

\begin{defi}
\label{lcst}
A \textbf{locally compact spectral triple} $(\mathcal{A},\H,D)$ consists of a separable Hilbert space $\H$, a self-adjoint operator $D$ in $\H$ and a *-algebra $\mathcal{A}\subseteq \B(\H)$ such that $a(\operatorname{dom}(D))\subseteq\operatorname{dom}(D)$, $[D,a]$ extends to a bounded operator, $a(D-i)^{-1}$ is compact, and $a(D-i)^{-s}\in \S^1$ for all $a\in\mathcal{A}$ and some $s\in\N$, called the summability of $(\mathcal{A},\H,D)$.
	A (locally compact) spectral triple $(\mathcal{A},\H,D)$ is called \textbf{regular} if for all $a\in\mathcal{A}$, we have $a,[D,a]\in\bigcap_{k=1}^\infty \dom(\delta_D^k)$.
\end{defi}

The following result might be known, but it seems like it was never explicitly proven, although a similar statement is made in \cite{SZ18}, and  \cite{CGRS} proves very related results.
Let $\Omega^1_D(\mathcal{A}):=\{\sum_{j=1}^n a_j[D,b_j]: a_j,b_j\in\mathcal{A}, n\in\N\}$ denote the set of inner fluctuations \cite{CC97} or \textit{Connes' differential one-forms}.

\begin{thm}
	A regular locally compact spectral triple $(\mathcal{A},\H, D)$ of summability $s$ satisfies $V(D-i)^{-1}\in\S^s$ for all $V\in\Omega^1_D(\mathcal{A})$.
\end{thm}

\begin{proof}
Let $V=\sum_{j=1}^n a_j[D,b_j]\in\Omega^1_D(\mathcal{A})$ be arbitrary and let $\delta:=\delta_D$.
For all $X\in\bigcap_{k=1}^\infty \dom(\delta^k)$ we have
\begin{align*}
X(|D|-i)^{-1}=(|D|-i)^{-1}X+(|D|-i)^{-1}\delta(X)(|D|-i)^{-1}.
\end{align*}
By induction, for all $X\in\bigcap_{k=1}^\infty\dom(\delta^k)$ there exists some $Y\in\bigcap_{k=1}^\infty\dom(\delta^k)$ such that
\begin{align}\label{eq:Commute with powers of resolvents of |D|}
X(|D|-i)^{-s}=(|D|-i)^{-s}Y.
\end{align}

Since $[D,b_j]\in \bigcap_{k=1}^\infty \dom(\delta^k)$ for all $j$ and since $g:\R\to\C,t\mapsto (|t|-i)/(t-i)$ is continuous and bounded, we have $g(D)\in\mB(\H)$ and there exist some $Y_j\in\mB(\H)$ such that
\begin{align*}
	V(D-i)^{-s}&=\sum_j a_j[D,b_j](|D|-i)^{-s}g(D)^s\\
	&=\sum_j a_j(|D|-i)^{-s}Y_jg(D)^s=\sum_j a_j(D-i)^{-s}g(D)^{-s}Y_jg(D)^s\in\S^1.
\end{align*}

More generally, let $X_1,\ldots,X_m\in\bigcap_{k=1}^\infty \dom(\delta^k)$, let $k_1,\ldots,k_m\in\N$ and set $k=\sum_{j=1}^m k_j$. By induction, noting that $\bigcap_{k=1}^\infty\dom(\delta^k)$ is an algebra, and applying \eqref{eq:Commute with powers of resolvents of |D|} to $s=k_j$, we obtain
	\begin{align*}
		\prod_{j=1}^m X_j(D-i)^{-k_j}=(D-i)^{-k}Y,
	\end{align*}
	for some $Y\in\bigcap_{k=1}^\infty \dom(\delta^k)$. If $s$ is even, we obtain
	\begin{align*}
		|(D+i)^{-1}V^*|^{s}&=V(D^2+\1)^{-1}V^*\cdots V(D^2+\1)^{-1}V^*\\
		&=V(D-i)^{-s}Y\in\S^1,
	\end{align*}
	for some $Y\in\bigcap_{k=1}^\infty\dom(\delta^k)$. Therefore, $V(D-i)^{-1}=((D+i)^{-1}V^*)^*\in\S^{s}$.

If $s$ is odd, we use polar decomposition to obtain $U\in\mB(\H)$ such that $|V(D-i)^{-1}|=UV(D-i)^{-1}$. Hence,
\begin{align*}
		|V(D-i)^{-1}|^{s}&=UV(D-i)^{-1}|V(D-i)^{-1}|^{s-1}\\
		&=UV(D^2+\1)^{-1}V^*\cdots V(D^2+\1)^{-1}V^*V(D-i)^{-1}\\
		&=UV(D-i)^{-s}Y'\in\S^1
	\end{align*}
	for some $Y'\in\bigcap_{k=1}^\infty\dom(\delta^k)$. Therefore, $V(D-i)^{-1}\in\S^s$.
\end{proof}

\chapter{Cyclic Cocycles in the Spectral Action}
\label{ch:Cyclic cocycles in the spectral action}
In this chapter, adapted from \cite{vNvS21a} and \cite{vNvS21b}, we show that the spectral action, when perturbed by an inner fluctuation, can be written as a series of Chern--Simons actions and Yang--Mills actions of all orders. In the odd orders, generalized Chern--Simons forms are integrated against an odd $(b,B)$-cocycle, whereas, in the even orders, powers of the curvature are integrated against $(b,B)$-cocycles that are Hochschild cocycles as well. In both cases, the Hochschild cochains are derived from the Taylor series expansion of the spectral action $\Tr(f(D+V))$ in powers of $V=\pi_D(A)$, but unlike the Taylor expansion we expand in increasing order of the forms in $A$. This extends \cite{CC06}, which computes only the scale-invariant part of the spectral action, works in dimension at most 4, and assumes the vanishing tadpole hypothesis. In our situation, we obtain a truly infinite odd $(b,B)$-cocycle. The analysis involved draws from results in multiple operator integration obtained in Chapter \ref{ch:MOI}, which also allows us to give conditions under which this cocycle is entire, and under which our expansion is absolutely convergent. As a first application of our expansion and of the gauge invariance of the spectral action, we show that the odd $(b,B)$-cocycle pairs trivially with $K_1$. As a second application, we show that a natural proposal for a quantum effective spectral action at one loop satisfies a similar expansion formula, and is hence an indication of a renormalization flow in the space of cyclic cocycles.

Results in this chapter were obtained in collaboration with Walter van Suijlekom.

\section{Introduction}
\label{sct:Cyclic Introduction}
The spectral action \cite{CC96,CC97} is one of the key instruments in the applications of noncommutative geometry to particle physics. With inner fluctuations \cite{C96} of a noncommutative manifold playing the role of gauge potentials, the spectral action principle yields the corresponding Lagrangians. Indeed, the asymptotic behavior of the spectral action for small momenta leads to experimentally testable field theories, by interpreting the spectral action as a classical action and applying the usual renormalization group techniques. In particular, this provides the simplest way known to geometrically explain the dynamics and interactions of the gauge bosons and the Higgs boson in the Standard Model Lagrangian as an effective field theory \cite{CCM07} (see also the textbooks \cite{CM07,Sui15}). More general noncommutative manifolds (spectral triples) can also be captured by the spectral action principle, leading to models beyond the standard model as well. As shown in \cite{CC06}, if one restricts to the scale-invariant part, one may naturally identify a Yang--Mills term and a Chern--Simons term to elegantly appear in the spectral action. From the perspective of quantum field theory, the appearance of these field-theoretic action functionals sparks hope that we might find a way to go beyond the classical framework provided by the spectral action principle. It is thus a natural question whether we can also field-theoretically describe the full spectral action, without resorting to the scale-invariant part. 


Motivated by this, we study the spectral action when it is expanded in terms of inner fluctuations associated to an arbitrary noncommutative manifold, without resorting to heat-kernel techniques. Indeed, the latter are not always available and an understanding of the full spectral action could provide deeper insight into how gauge theories originate from noncommutative geometry. Let us now give a more precise description of our setup.

We let $(\A,\H,D)$ be an $s$-summable spectral triple (\textit{cf.} Definition \ref{defi:st} below). If $f : \R \to \C$ is a suitably nice function we may define the spectral action \cite{CC97}:
$$
\Tr (f(D)).
$$
An inner fluctuation, as explained in \cite{C96}, is given by a Hermitian universal one-form
\begin{align}\label{eq:A uitgeschreven}
	A=\sum_{j=1}^n a_jdb_j\in\Omega^1(\A),
\end{align}
for elements $a_j,b_j\in\A$. The terminology `fluctuation' comes from representing $A$ on $\H$ as
\begin{align}\label{eq:V uitgeschreven}
	V:=\pi_D(A)=\sum_{j=1}^n a_j[D,b_j]\in\mB(\H)_\sa,
\end{align}
and fluctuating $D$ to $D+V$ in the spectral action.
 %
The variation of the spectral action under the inner fluctuation is then given by
\begin{align}\label{variation of SA}
	\Tr(f(D+V))-\Tr(f(D)).
\end{align}
As spectral triples can be understood as noncommutative spin$^\text{c}$ manifolds (see \cite{C08}) encoding the gauge fields as an inner structure, one could hope that perturbations of the spectral action could be understood in terms of noncommutative versions of geometrical, gauge theoretical concepts. Hence we would like to express \eqref{variation of SA} in terms of universal forms constructed from $A$. To express an action functional in terms of universal forms, one is naturally led to cyclic cohomology. As it turns out, hidden inside the spectral action we will identify an odd $(b,B)$-cocycle $(\tilde\psi_1,\tilde\psi_3,\ldots)$ and an even $(b,B)$-cocycle $(\phi_2,\phi_4,\ldots)$ for which $b\phi_{2k}=B\phi_{2k}=0$, i.e., each Hochschild cochain $\phi_{2k}$ forms its own $(b,B)$-cocycle $(0,\ldots,0,\phi_{2k},0,\ldots)$. On the other hand, the odd $(b,B)$-cocycle $(\tilde\psi_{2k+1})$ is truly infinite (in the sense of \cite{C94}). 

 The main result of this chapter is that for suitable $f:\R\to\C$ we may expand
\begin{align}\label{eq:expansion intro}
	\Tr(f(D+V)-f(D))=\sum_{k=1}^\infty\left(\int_{\psi_{2k-1}}\cs_{2k-1}(A)+\frac{1}{2k}\int_{\phi_{2k}}F^{k}\right),
\end{align}
	in which the series converges absolutely. Here $\psi_{2k-1}$ is a scalar multiple of $\tilde\psi_{2k-1}$, $F_t=tdA+t^2A^2$, so that $F=F_1$ is the curvature of $A$, and $\cs_{2k-1}(A)=\int_0^1 AF_t^{k-1}dt$ is a generalized noncommutative Chern--Simons form. We also give a bound on the remainder of this expansion.

As already mentioned, a similar result was shown earlier to hold for the scale-invariant part $\zeta_D(0)$ of the spectral action. Indeed,  Connes and Chamseddine \cite{CC06} expressed the variation of the scale-invariant part in dimension $\leq 4 $ as
\begin{equation*}
\zeta_{D+V}(0) - \zeta_D(0) = - \frac 1 4\int_{\tau_0} (dA+A^2) + \frac 12 \int_\psi \left(A d A + \frac 2 3 A^3\right),
\end{equation*}
for a certain Hochschild 4-cocycle $\tau_0$ and cyclic 3-cocycle $\psi$.

Interestingly, a key role in our extension of this result to the full spectral action will be played by multiple operator integrals. It is the natural replacement of residues in this context, and also allows to go beyond dimension $4$.
For our analysis of the cocycle structure that appears in the full spectral action we take the Taylor series expansion as a starting point. 
This explains the appearance of multiple operator integrals, as traces thereof are multilinear extensions of the derivatives of the spectral action. This viewpoint is also studied in \cite{S14,Sui11}, where multiple operator integrals are used to investigate the Taylor expansion of the spectral action. As we will show, multiple operator integrals can also be used to define cyclic cocycles, because of some known properties of the multiple operator integral that have been proved in increasing generality in the last decades (e.g., in \cite{ACDS09,CS18,vNS21,S14,Sui11}). In Section \ref{sct:Finitely summable} we have pushed these results even further, by proving estimates and continuity properties for the multiple operator integral when the self-adjoint operator has an $s$-summable resolvent, thereby supplying the discussion here with a strong functional analytic foundation. Applying the results of Section \ref{sct:Relative Schatten} and \textsection\ref{sec4a} in order to obtain \eqref{eq:expansion intro} for locally compact spectral triples is left open for future research.


We work out two interesting possibilities for application of our main result and the techniques used to obtain it. The first application is to index theory. The analytically powerful multiple operator integration techniques used for the absolute convergence of our expansion also allow us to show that the found $(b,B)$-cocycles are \textit{entire} in the sense of \cite{C88a}. This makes it meaningful to analyze their pairing with K-theory, which we find to be trivial in Section \ref{sct:vanishing pairing}. 

The second application is to quantization. In Section \ref{sct:One-Loop}, though evading analytical difficulties, we will take a first step towards the quantization of the spectral action within the framework of spectral triples. Using the asymptotic expansion proved in Theorem \ref{thm:asymptotic expansion}, and some basic quantum field theoretic techniques, we will propose a one-loop quantum effective spectral action and show that it satisfies a similar expansion formula, featuring in particular a new pair of cyclic cocycles.

\section{Multiple operator integrals and a new function class}

As our goal is to understand the structure of the gauge fluctuations in the spectral action, a good starting point is the (noncommutative) Taylor series expansion of $\Tr(f(D+V))$, expanded in $V$. In the sense of Chapter \ref{ch:MOI}, we may replace the $n\th$ order derivatives occurring in the Taylor series of the spectral action by a multiple operator integral, and obtain
\begin{align}\label{eq:Taylor}
	\Tr(f(D+V))\sim\sum_{n=0}^\infty \Tr(\Tfn^D(V,\ldots,V)),
\end{align}
which allows us to apply the powerful toolkit of multiple operator integration.

Indeed, we recall that the main result of Section \ref{sct:Finitely summable} was the following bound on the multiple operator integral:
\begin{align}\label{eq:s-summable bound nog een keer}
	\nrm{T^{D+V,D,\ldots,D}_{f^{[n]}}(V_1,\ldots,V_n)}{1}\leq c_{s,n}(f)\norm{V_1}\cdots\norm{V_n}(1+\norm{V})^{s}\big\|(D-i)^{-1}\big\|_{s}^s.
\end{align}
It was proven for functions $f\in\Wsn$, i.e., all $f\in C^n$ satisfying $((fu^m)^{(k)})\hat{~}\in L^1$ for all $m\leq s$ and $k\leq n$, where $u(x)=x-i$.

The analytical result \eqref{eq:s-summable bound nog een keer} allows us to freely work with the traces of multiple operator integrals up to order $n$. This is actually the sole analytical ingredient for a truncated version (Theorem \ref{thm:asymptotic expansion}) of our main result (Theorem \ref{thm:main thm}). However, if we want the expansion \eqref{eq:expansion intro} to converge, we will need to impose infinite differentiability of $f$, as well as a growth condition on $c_{s,n}(f)$ as $n$ goes to infinity.
We therefore introduce the space
\begin{align*}
	\Esg:=\left\{f\in C^\infty \mid
	\begin{aligned}
	&\text{there exists $C_f\geq 1$ s.t. } \|\widehat{(fu^m)^{(n)}}\|_1\leq (C_f)^{n+1}n!^\gamma \\
	&\text{for all $m=0,\ldots,s$ and $n\in\mathbb N_0$}
	\end{aligned}
	\right\},
\end{align*}
for $\gamma\in (0,1]$. Our main result is that the expansion \eqref{eq:expansion intro} holds for all functions $f\in\Esg$, and certain perturbations $A$. If $\gamma=1$, the expansion converges absolutely whenever the perturbation $A$ is sufficiently small. If $\gamma<1$ the expansion converges absolutely for all perturbations. The following Lemma underlies both results.
\begin{lem}\label{lem:Cs and Es bounds}
	Let $s\in\N$, $D$ self-adjoint in $\H$ with $(D-i)^{-1}\in\S^s$, and $\gamma\in(0,1]$. For any $f\in\Esg$ there exists a $C\geq1$ such that for all $n\in\N_0$, $V_1,\ldots,V_n\in\mB(\H)$, and $V\in\mB(\H)_\sa$, we have
		$$\absnorm{\Tfn^{D+V,D,\ldots,D}(V_1,\ldots,V_n)}\leq \Big(C^{n+1}n!^{\gamma-1}\Big)\norm{V_1}\cdots\norm{V_n}(1+\norm{V})^{s}\big\|(D-i)^{-1}\big\|_s^s.$$
\end{lem}
\begin{proof}
Apply the definition of $\Esg$ to Theorem \ref{thm:Schatten estimate}, and absorb $2^s$ into the constant $C$.
\end{proof}

This lemma will be used in \textsection\ref{sct:convergence} and Section \ref{sct:vanishing pairing}.

Examples of functions in $\E_s^{1}$ are Schwartz functions with compactly supported Fourier transform. The following proposition gives more examples.

\begin{prop}
	Let $f\in C^\infty$ and $s,t\in\N_0$.
	\begin{enumerate}[label=\textnormal{(\roman*)}]
		\item\label{mult E} If $f\in \E_s^1$ and $g\in \E_t^1$, then $fg\in \E_{s+t}^1$.
		\item\label{exp1 E} If $\hat{f}\in L^1$ with $|\hat{f}(x)|\leq e^{-c|x|}$ a.e. for some $c>0$, then $f\in\E_0^{1}$.
		\item\label{exp2 E} If $\widehat{fu^s}\in L^1$ with $|\widehat{fu^s}(x)|\leq e^{-c|x|}$ a.e. for some $c>0$, then $f\in\E_s^{1}$. 
		\item\label{rati E} Rational functions in $\O(|x|^{-s-1})$ are in $\E_s^1$.
		\item\label{Gauss E} The function $x\mapsto e^{-x^2}$ is in $\E_s^{1/2}$ for any $s\in\N_0$.
	\end{enumerate}
\end{prop}
\begin{proof}
	\begin{enumerate}[label=\textnormal{(\roman*)}]
	\item For $m\leq s$ and $p\leq t$, Young's inequality gives
	\begin{align*}
		\|\widehat{(fgu^{m+p})^{(n)}}\|_1&\leq\sum_{k=0}^n\vect{n}{k}\|\widehat{(fu^m)^{(k)}}\|_{1}\|\widehat{(gu^p)^{(n-k)}}\|_{1}\\
		&\leq (n+1)n!(C_fC_g)^{n+1}.
	\end{align*}
	Any polynomial in $n$ is $\mathcal{O}(C^n)$ for some $C\geq1$.
	\item As
		$$\frac{1}{n!}(\tfrac{1}{2}c)^n|x|^n\leq\sum_{m=0}^\infty \frac{1}{m!}(\tfrac12c|x|)^m=e^{\tfrac12c|x|},$$
	we find $\|\widehat{f^{(n)}}\|_1=\||x|^n\hat{f}\|_1\leq \|e^{\tfrac12c|x|}\hat{f}\|_1(\tfrac12c)^{-n}n!$, thereby obtaining $f\in \E_0^1$.
	\item 
	 Item \ref{exp1 E} gives that $fu^s\in\E_0^{1}$. It is easy to see that $u^{-s}\in\E_{s-1}^1$. Therefore \ref{mult E} gives $f\in\E_{s-1}^1$, i.e., $\|\widehat{(fu^m)^{(n)}}\|_1\leq (C_f)^{n+1}n!$ for $m\leq s-1$. Similar to \ref{exp1 E} we get $\|\widehat{(fu^s)^{(n)}}\|_1\leq C^{n+1}n!$ for some $C\geq 1$.
	\item Follows from \ref{exp2 E}.
	
	\item Let $f(x)=e^{-x^2}$ and $m\in\N_0$. The Fourier transform of $fu^m$ is a polynomial times a Gaussian, say $(fu^m)\hat{~}(x)=p(x)e^{-x^2/c^2}$. Therefore,
\begin{align*}
	\absnorm{\widehat{(fu^m)^{(n)}}}=\absnorm{|x|^n p(x)e^{-x^2/c^2}}.
\end{align*}
Furthermore, $|x|^n=c^n\sqrt{(x^2/c^2)^{n}}\leq c^n\sqrt{n!e^{x^2/c^2}}=\sqrt{n!}c^ne^{\frac{x^2}{2c^2}}$, so
\begin{align*}
	\absnorm{\widehat{(fu^m)^{(n)}}}&\leq \sqrt{n!}c^n\absnorm{p(x)e^{-\frac{x^2}{2c^2}}}.
\end{align*}
Therefore, $f\in\E_s^{1/2}$ for any $s\in\N_0$.\qedhere
	\end{enumerate}
\end{proof}

\let\mysectionmark\sectionmark
\renewcommand\sectionmark[1]{}
\section{Cyclic cocycles and universal forms underlying the spectral action}
\let\sectionmark\mysectionmark
\sectionmark{Cyclic cocycles and universal forms underlying the S.A.}
\label{sct:CC and UF underlying the SA}
Mainly to fix our conventions, we start with the definition of a finitely summable spectral triple, which is the situation in which our main result is stated.
\begin{defi}
  \label{defi:st}
	Let $s\in\N$. An \textbf{$s$-summable spectral triple} $(\A,\H,D)$ consists of a separable Hilbert space $\H$, a self-adjoint operator $D$ in $\H$ and a unital *-algebra $\A\subseteq \B(\H)$, such that, for all $a\in\A$, $a\operatorname{dom} D\subseteq\operatorname{dom} D$ and $[D,a]$ extends to a bounded operator, and $(D-i)^{-1}\in\mathcal{S}^s$.
\end{defi}


Throughout the rest of this chapter, we let $(\A,\H,D)$ be an $s$-summable spectral triple for $s\in\N$, and we let $f\in\W_s^{n}$ for $n\in\N_0$, unless stated otherwise.

\begin{defi}\label{def:bracket}
	Define a multilinear function $\br{\cdot}:\B(\H)^{\times n}\to\C$ by
	\begin{align}\label{eq:br cycl}
	\br{V_1,\ldots,V_n}:=\sum_{j=1}^n\Tr( T^D_{f^{[n]}}(V_j,\ldots,V_n,V_1,\ldots,V_{j-1})).
	\end{align}
\end{defi}

For our algebraic results (which make up most of Section \ref{sct:CC and UF underlying the SA} and \textsection\ref{sct:truncated expansion}) we only need two simple properties of the bracket $\br{\cdot}$, stated in the following lemma. After proving this lemma, all analytical subtleties (related to the unboundedness of $D$) are taken care of, and we can focus on the algebra that ensues from these simple rules.
\begin{lem}\label{cycl bracket}
For $V_1,\ldots,V_n\in\mB(\H)$ and $a\in\A$ we have
\begin{enumerate}[label=\textnormal{(\Roman*)}]
	\item $\br{V_1,\ldots,V_n}=\br{V_n,V_1,\ldots,V_{n-1}},$\label{cyclicity}
	\item $\br{V_1,\ldots,aV_j,\ldots,V_n}-\br{V_1,\ldots,V_{j-1}a,\ldots,V_n}=\br{V_1,\ldots,V_{j-1},[D,a],V_j,\ldots,V_n}$,\label{commutation}
\end{enumerate}
          where it is understood that for the edge case $j=1$ we need to substitute $n$ for $j-1$ on the left-hand side, and $f\in\W_s^{n+1}$ is assumed to define the right-hand side. 
\end{lem}
\begin{proof}
	Property \ref{cyclicity} follows immediately from Definition \ref{def:bracket}. By using \eqref{eq:Trace function divdiff} for finite-rank operators $V_1,\ldots,V_n$, we have,
	\begin{align}
		&T^D_{f^{[n]}}(V_1,\ldots,V_j,aV_{j+1},\ldots,V_n)-T^D_{f^{[n]}}(V_1,\ldots,V_ja,V_{j+1},\ldots,V_n)\nonumber\\
		&\quad=T^D_{f^{[n+1]}}(V_1,\ldots,V_j,[D,a],V_{j+1},\ldots,V_n),\label{com1}
	\end{align}
	and the two edge cases,
	\begin{align}\label{com2}
		T^D_{f^{[n]}}(aV_1,\ldots,V_n)-aT^D_{f^{[n]}}(V_1,\ldots,V_n)&=T^D_{f^{[n+1]}}([D,a],V_1,\ldots,V_n),\\
		T^D_{f^{[n]}}(V_1,\ldots,V_n)a-T^D_{f^{[n]}}(V_1,\ldots,V_na)&=T^D_{f^{[n+1]}}(V_1,\ldots,V_n,[D,a]).\label{com3}
	\end{align}
	By Theorem \ref{thm:continuity for L's}, and the fact that the finite-rank operators lie strongly dense in $\B(\H)$, we find that formulas \eqref{com1}, \eqref{com2}, and \eqref{com3}
        hold for all $V_1,\ldots,V_n\in \B(\H)$. Hence,
	\begin{align*}
		&\br{aV_1,V_2,\ldots,V_n}-\br{V_1,V_2,\ldots,V_na}\\
		&\quad = \sum_{j=2}^{n}\Tr (T^D_{f^{[n]}}(V_j,\ldots,V_n,[D,a],V_1,\ldots,V_{j-1}))\\
		&\qquad +\Tr(T^D_{f^{[n]}}(aV_1,\ldots,V_n))-\Tr(T^D_{f^{[n]}}(V_1,\ldots,V_na))\\
		&\quad=	\sum_{j=2}^{n}\Tr (T^D_{f^{[n+1]}}(V_{j},\ldots,V_n,[D,a],V_1,\ldots,V_{j-1}))\\
		&\qquad +\Tr(T^D_{f^{[n+1]}}([D,a],V_1,\ldots,V_n))+\Tr(aT^D_{f^{[n]}}(V_1,\ldots,V_n))-\Tr(T^D_{f^{[n]}}(V_1,\ldots,V_na))\\
		&\quad=	\sum_{j=2}^{n}\Tr (T^D_{f^{[n+1]}}(V_{j},\ldots,V_n,[D,a],V_1,\ldots,V_{j-1}))\\
		&\qquad +\Tr(T^D_{f^{[n+1]}}([D,a],V_1,\ldots,V_n))+\Tr(T^D_{f^{[n+1]}}(V_1,\ldots,V_n,[D,a]))\\
		&\quad= \br{[D,a],V_1,\ldots,V_n},
	\end{align*}
	and therefore \ref{commutation} follows by applying \ref{cyclicity}.
\end{proof}

\begin{rema}\label{rem:contour integral}
Under additional assumptions -- for instance when $V_1,\ldots, V_n\in\S^1$ and $f\in\W_s^{n}$ is such that $f'$ is compactly supported and analytic in a region of $\C$ containing a rectifiable curve $\gamma$ which surrounds the support of $f$ in $\R$ -- we have
	\begin{align*}
		\br{V_1,\ldots,V_n}= \frac{1}{2\pi i}\Tr(\oint_\gamma f'(z)\prod_{j=1}^n V_j(z-D)^{-1}).
	\end{align*}
This occurs in \cite[Corl. 20]{Sui11} in the case where $V_1=V_2=\cdots=V_n$. It would be interesting to compare these resolvent formulas with the ones appearing in \cite{Pay07}.
\end{rema}

\subsection{Hochschild and cyclic cocycles}
When the above brackets $\br{\cdot}$ are evaluated at one-forms $a[D,b]$ associated to a spectral triple, the relations found in Lemma \ref{cycl bracket} can be translated nicely in terms of the coboundary operators appearing in cyclic cohomology. This is very similar to the structure appearing in the context of index theory, see for instance \cite{GS89,Hig06}.

Let us start by recalling the definition of the boundary operators $b$ and $B$ from \cite{C85}.

\begin{defi}
If $\A$ is an algebra, and $n\in\N_0$, we define the space of \textit{Hochschild $n$-cochains}, denoted by $\mathcal{C}^n(\A)$, as the space of $(n+1)$-linear
functionals $\phi$ on $\A$ with the property that if $a_j =1$ for some $j \geq 1$, then $\phi(a_0,\ldots,a_n) = 0$. Define operators $b : \mathcal{C}^{n}(\A) \to \mathcal{C}^{n+1}(\A)$ and $B: \mathcal{C}^{n+1}(\A) \to \mathcal{C}^{n}(\A)$ by
\begin{align*}
b\phi(a_0, a_1,\dots, a_{n+1})
:=& \sum_{j=0}^n (-1)^j \phi(a_0,\dots, a_j a_{j+1},\dots, a_{n+1})\\
& + (-1)^{n+1} \phi(a_{n+1} a_0, a_1,\dots, a_n) ,\\
B \phi(a_0 ,a_1, \ldots, a_n) :=& 
\sum_{j=0}^n (-1)^{nj}\phi(1,a_j,a_{j+1},\ldots, a_{j-1}).
\end{align*}
\end{defi}
Note that $B = A B_0$ in terms of the operator $A$ of cyclic anti-symmetrization and the operator defined by $B_0 \phi (a_0, a_1, \ldots, a_n) = \phi(1,a_0, a_1,\ldots, a_n)$.

One may check that the pair $(b,B)$ defines a double complex, \textit{i.e.} $b^2 = 0,~ B^2=0,$ and $bB +Bb =0$. Hochschild cohomology then arises as the cohomology of the complex  $(\mathcal{C}^n(\A),b)$, while the for us relevant \textit{periodic cyclic cohomology} is defined as the cohomology of the totalization of the $(b,B)$-complex. That is to say, 
\begin{align*}
\mathcal{C}^\ev(\A) = \bigoplus_k \mathcal{C}^{2k} (\A) ; \qquad \mathcal{C}^{\odd}(\A) = \bigoplus_k \mathcal{C}^{2k+1} (\A),
\end{align*}
form a complex with differential $b+B$ and the cohomology of this complex is called periodic cyclic cohomology. We will also refer to a periodic cyclic cocycle as a $(b,B)$-cocycle. Explicitly, an odd $(b,B)$-cocycle is thus given by a sequence
$$
(\phi_1, \phi_3, \phi_5, \ldots),
$$
where $\phi_{2k+1} \in \mathcal{C}^{2k+1}(\A)$ and 
$$
b \phi_{2k+1} + B \phi_{2k+3} = 0 ,
$$
for all $k \geq 0$, and also $B \phi_1 = 0$. An analogous statement holds for even $(b,B)$-cocycles.

\subsection{Cyclic cocycles associated to multiple operator integrals}
\label{sct:Cyclic cocycles associated to multiple operator integrals}

We define the following Hochschild $n$-cochain:
\begin{align}\label{eq:def phi_n}
	\phi_n(a_0,\ldots,a_n):=\br{a_0[D,a_1],[D,a_2],\ldots,[D,a_{n}]} \qquad  (a_0, \ldots, a_n \in \A).
\end{align}
	We easily see that $B_0\phi_n$ is invariant under cyclic permutations, so that $B\phi_n=nB_0\phi_n$ for odd $n$ and $B\phi_n=0$ for even $n$. Also, $\phi_n(a_0,\ldots,a_n)=0$ when $a_j=1$ for some $j\geq1$. We put $\phi_0:=0$.
\begin{lem}\label{lem:b}
	We have $b\phi_n=\phi_{n+1}$ for odd $n$ and we have $b\phi_n=0$ for even $n$.
\end{lem}
\begin{proof}
	As $b\phi_0=0$ by definition, and $b^2=0$, we need only check the case in which $n$ is odd.
	
	We find, by splitting up the sum, and shifting the second appearing sum by one, that
	\begin{align*}
		&b\phi_n(a_0,\ldots,a_{n+1})\\
		&\quad=\br{a_0a_1[D,a_1],[D,a_2],\ldots,[D,a_{n+1}]}
		-\br{a_0a_1[D,a_1],[D,a_2],\ldots,[D,a_{n+1}]}\\
		&\qquad+\sum_{j=2}^n(-1)^j\br{a_0[D,a_1],[D,a_2],\ldots,a_j[D,a_{j+1}],\ldots,[D,a_{n+1}]}\\
		&\qquad-\sum_{j=2}^{n+1}(-1)^j\br{a_0[D,a_1],[D,a_2],\ldots,[D,a_{j-1}]a_j,\ldots,[D,a_{n+1}]}\\
		&\qquad+\br{a_{n+1}a_0[D,a_1],[D,a_2],\ldots,[D,a_n]}\\
		&\quad=\sum_{j=2}^n(-1)^j\br{a_0[D,a_1],[D,a_2],\ldots,[D,a_{n+1}]}
		-\br{a_0[D,a_1],[D,a_2],\ldots,[D,a_n]a_{n+1}}\\
		&\qquad+\br{a_{n+1}a_0[D,a_1],\ldots,[D,a_n]}\\
		&\quad=\br{[D,a_{n+1}],a_0[D,a_1],[D,a_2],\ldots,[D,a_{n}]}\\
		&\quad=\phi_{n+1}(a_0,\ldots,a_{n+1}),
	\end{align*}
	by \ref{cyclicity} and \ref{commutation} of Lemma \ref{cycl bracket}.
\end{proof}
\begin{lem}\label{lem:c}
Let $n$ be even. We have $bB_0\phi_n=2\phi_n-B_0\phi_{n+1}$.
\end{lem}
\begin{proof}
	Splitting the sum in two, and shifting the index of the second sum, we find
	\begin{align*}
		&bB_0\phi_n(a_0,\ldots,a_n)\\
		&\quad=\sum_{j=0}^{n-1}(-1)^j\br{[D,a_0],\ldots,a_j[D,a_{j+1}],\ldots,[D,a_n]}\\
		&\qquad-\sum_{j=1}^n(-1)^j\br{[D,a_0],\ldots,[D,a_{j-1}]a_j,\ldots,[D,a_n]}+\br{[D,a_na_0],\ldots,[D,a_{n-1}]}\\
		&\quad=\br{a_0[D,a_1],[D,a_2],\ldots,[D,a_n]}+\sum_{j=1}^{n-1}(-1)^j\br{[D,a_0],\ldots,[D,a_n]}\\
		&\qquad-\br{[D,a_0],\ldots,[D,a_{n-2}],[D,a_{n-1}]a_n}+\br{[D,a_na_0],\ldots,[D,a_{n-1}]}\\
		&\quad=\phi_n(a_0,\ldots,a_n)-\br{[D,a_0],\ldots,[D,a_n]}+\br{[D,a_n],[D,a_0],\ldots,[D,a_{n-1}]}\\
		&\qquad+\br{[D,a_n]a_0,[D,a_1],\ldots,[D,a_{n-1}]}\\
		&\quad=2\phi_n(a_0,\ldots,a_n)-B_0\phi_{n+1}(a_0,\ldots,a_n),
	\end{align*}
	by using both properties of the bracket $\br{\cdot}$ in the last step.
\end{proof}

Motivated by this we define 
	$$\psi_{2k-1}:=\phi_{2k-1}-\tfrac{1}{2}B_0\phi_{2k},$$
so that
$$B\psi_{2k+1}=2(2k+1)b\psi_{2k-1}.$$
We can rephrase this property in terms of the $(b,B)$-complex as follows. 
\begin{prop}
  \label{prop:bB}
  Let $\phi_n$ and $\psi_{2k-1}$ be as defined above and set
  $$\tilde{\psi}_{2k-1}:=(-1)^{k-1}\frac{(k-1)!}{(2k-1)!}\psi_{2k-1}\,.$$
  \begin{enumerate}[label=\textnormal{(\roman*)}]
  \item The sequence $(\phi_{2k})$ is a $(b,B)$-cocycle and each $\phi_{2k}$ defines an even Hochschild cocycle: $b \phi_{2k} = 0$. 
    \item The sequence $(\tilde \psi_{2k-1})$ is an odd $(b,B)$-cocycle. 
    \end{enumerate}
  \end{prop}
We use an integral notation that is defined by linear extension of
	$$\int_\phi a_0da_1\cdots da_n :=\int_{\phi_n}a_0da_1\cdots da_n:=  \phi(a_0,a_1,\ldots,a_n),$$
and similarly for $\psi$.

\subsection{Derivatives of the spectral action in terms of universal forms}\label{sct:from V to A}
In this section we will express the derivatives of the fluctuated spectral action (occurring in the Taylor series) in terms of universal forms that are integrated along $\phi$. We thus make the jump from an expression in terms of $V=\pi_D(A)\in\Omega^1_D(\A)_\sa$ to an expression in terms of $A\in\Omega^1(\A)$. By \eqref{dermoi} and the definition of $\br{V,\ldots,V}$, we have, for $n\in\N$,
\begin{align}
	\frac{1}{n!}\frac{d^n}{dt^n}\Tr(f(D+tV))\big|_{t=0}&=\Tr(\Tfn^D(V,\ldots,V))\nonumber\\
	&=\frac{1}{n}\br{V,\ldots,V}.\label{eq:<V,...,V>/n}
\end{align}
As $V$ decomposes as a finite sum $V=\sum a_j[D,b_j]$, our task is to express $$\br{a_{j_1}[D,b_{j_1}],\ldots,a_{j_n}[D,b_{j_n}]}$$ in terms of universal forms $a_0da_1\cdots da_n$ integrated along $\phi$. This is possible by just using \ref{commutation} and $[D,a_1a_2]=a_1[D,a_2]+[D,a_1]a_2$. Written out explicitly for increasing values of $n$, the resulting expressions quickly become horribly convoluted. Thankfully, though, by lifting them to the algebra $M_2(\Omega^\bullet(\A))=M_2(\C)\otimes\Omega^\bullet(\A)$, these expressions take a tractable form. 

\begin{prop}
	Let $n\in\N$. For $a_1,\ldots,a_n,b_1,\ldots,b_n\in\A$, denoting $A_j:=a_jdb_j$, we have
		$$\br{a_1[D,b_1],\ldots,a_n[D,b_n]}=\int_\phi \begin{pmatrix}
A_1&0
\end{pmatrix}\prod_{j=2}^n\begin{pmatrix}
A_j+dA_j&-A_j\\dA_j&-A_j
\end{pmatrix}\begin{pmatrix}
1\\0
\end{pmatrix}.$$
\end{prop}
\begin{proof}
 If we combine, for every $n\in\N_0$, the $n$-multilinear function $\br{\cdot}$ from \eqref{eq:br cycl}, we obtain a linear function
	$$\br{\cdot}:T\mB(\H)\to\C$$
	on the tensor algebra $T\mB(\H)$ of $\mB(\H)$. 
	For any $\omega,\nu\in T\mB(\H)$, a straightforward calculation using the commutation rule \ref{commutation} from Lemma \ref{cycl bracket} shows that
\begin{align}
	&\br{\omega\otimes a_{j-1}[D,b_{j-1}]\otimes\begin{pmatrix}
	a_j&a_jb_j
	\end{pmatrix}\nu}
	=\br{\omega\otimes\begin{pmatrix}
	a_{j-1}&a_{j-1}b_{j-1}
	\end{pmatrix}
	M_j\otimes\nu},\label{eq:look at the stars}
\end{align}
where $M_j\in M_2(T\mB(\H))$ is defined by
\begin{align}\label{eq:M_j}
	M_j:=\begin{pmatrix}
	[D,b_{j-1}a_j]+[D,b_{j-1}]\otimes[D,a_j]&[D,b_{j-1}a_jb_j]+[D,b_{j-1}]\otimes[D,a_jb_j]\\-[D,a_j]&-[D,a_jb_j]
	\end{pmatrix}.
\end{align}
Repeating \eqref{eq:look at the stars}, and subsequently using \eqref{eq:def phi_n}, it follows that
\begin{align*}
	\br{a_1[D,b_1],\ldots, a_n[D,b_n]}&=\br{a_1[D,b_1]\otimes\ldots\otimes a_{n-1}[D,b_{n-1}]\otimes\begin{pmatrix}a_n& a_nb_n\end{pmatrix}\begin{pmatrix}
	[D,b_n]\\0
	\end{pmatrix}}\\
	&=\br{\begin{pmatrix}
	a_1&a_1b_1
	\end{pmatrix}\bigg(\prod_{j=2}^nM_j\bigg)\begin{pmatrix}
	[D,b_n]\\0
	\end{pmatrix}}\\
	&=\int_\phi \begin{pmatrix}
	a_1&a_1b_1
	\end{pmatrix}\bigg(\prod_{j=2}^nN_j\bigg)\begin{pmatrix}
	db_n\\0
	\end{pmatrix},
\end{align*}
where from \eqref{eq:M_j} we obtain
\begin{align*}	
	N_j&=\begin{pmatrix}
	d(b_{j-1}a_j)+db_{j-1} da_j&d(b_{j-1}a_jb_j)+db_{j-1} d(a_jb_j)\\-da_j&-d(a_jb_j)
	\end{pmatrix}\\
	&=\begin{pmatrix}
	db_{j-1}&b_{j-1}\\
	0&-1
	\end{pmatrix}
	\begin{pmatrix}
	a_j+da_j&a_jb_j+da_jb_j+a_jdb_j\\
	da_j&da_jb_j+a_jdb_j
	\end{pmatrix}.
\end{align*}
By also writing $\begin{pmatrix}
	db_n\\0
	\end{pmatrix}=\begin{pmatrix}
	db_n&b_n\\0&-1
	\end{pmatrix}\begin{pmatrix}
	1\\0
	\end{pmatrix}$, we find that
\begin{align*}
	&\br{a_1[D,b_1],\ldots a_n[D,b_n]}\\
	&\quad=\int_\phi \begin{pmatrix}
	a_1&a_1b_1
	\end{pmatrix}\begin{pmatrix}
	db_{1}&b_{1}\\
	0&-1
	\end{pmatrix}\left(\prod_{j=2}^n\begin{pmatrix}
	a_j+da_j&a_jb_j+da_jb_j+a_jdb_j\\
	da_j&da_jb_j+a_jdb_j
	\end{pmatrix}\begin{pmatrix}
	db_{j}&b_{j}\\
	0&-1
	\end{pmatrix}\right)\begin{pmatrix}
	1\\0
	\end{pmatrix}\\
	&\quad=\int_\phi \begin{pmatrix}
	A_1&0
	\end{pmatrix}\left(\prod_{j=2}^n\begin{pmatrix}
	A_j+dA_j&-A_j\\
	dA_j&-A_j
	\end{pmatrix}\right)\begin{pmatrix}
	1\\0
	\end{pmatrix},
\end{align*}
	which concludes the proof.
\end{proof}

\begin{cor}\label{cor:2x2 matrix}
	If $n\in\N$, $A\in\Omega^1(\A)$ and $V:=\pi_D(A)\in\Omega^1_D(\A)$, then
\begin{align}\label{eq:2x2 matrix}
	\br{V,\ldots,V}=\int_\phi \begin{pmatrix}
	A & 0
	\end{pmatrix}
	\begin{pmatrix}
	A+dA & -A\\
	dA   & -A
	\end{pmatrix}^{n-1}\begin{pmatrix}
	1\\
	0
	\end{pmatrix}.
\end{align}
\end{cor}
	
\begin{exam}\label{ex:first terms}
	Using \eqref{eq:2x2 matrix}, we obtain in particular
\begin{align*}
	\br{V}&=\int_{\phi_1}A,\\
	\br{V,V}&=\int_{\phi_2}A^2+\int_{\phi_3}AdA,\\
	\br{V,V,V}&=\int_{\phi_3}A^3+\int_{\phi_4}AdAA+\int_{\phi_5}AdAdA,\\
	\br{V,V,V,V}&=\int_{\phi_4}A^4+\int_{\phi_5}(A^3dA+AdAA^2)+\int_{\phi_6}AdAdAA+\int_{\phi_7}AdAdAdA.
\end{align*}
	With \eqref{eq:<V,...,V>/n}, and in the sense of \eqref{eq:Taylor}, this implies that
	\begin{align*}
		\Tr(f(D+V)-f(D))=&\int_{\phi_1} A+\frac{1}{2}\int_{\phi_2}A^2+\int_{\phi_3}\Big(\frac{1}{2}AdA+\frac{1}{3}A^3\Big)+\int_{\phi_4}\Big(\frac{1}{3}AdAA+\frac{1}{4}A^4\Big)\\
		&+\ldots,
	\end{align*}
	where the dots indicate terms of degree 5 and higher. Using $\phi_{2k-1}=\psi_{2k-1}+\frac{1}{2}B_0\phi_{2k}$, this becomes
	\begin{align*}
		\Tr(f(D+V)-f(D))=&\int_{\psi_1} A+\frac{1}{2}\int_{\phi_2}(A^2+dA)+\int_{\psi_3}\Big(\frac{1}{2}AdA+\frac{1}{3}A^3\Big)\\
		&+\frac{1}{4}\int_{\phi_4}\Big(dAdA+\frac{2}{3}(dAA^2+AdAA+A^2dA)+A^4\Big)+\ldots.
	\end{align*}
	Notice that, if $\phi_4$ would be tracial, we would be able to identify the terms $dAA^2$, $AdAA$ and $A^2dA$, and thus obtain the Yang--Mills form $F^2=(dA+A^2)^2$, under the fourth integral. In the general case, however, cyclic permutations under $\int_\phi$ produce correction terms, of which we will need to keep track.
\end{exam}

\subsection{Near-tracial behavior of $\int_\phi$}\label{sct:Cyclicity}

In \textsection\ref{sct:from V to A}, we have not yet used the cyclicity property \ref{cyclicity} from Lemma \ref{cycl bracket}. Now applying that property yields the following proposition, which shows how $\int_\phi$ differs from being tracial. This proposition and its corollary are crucial for Section \ref{sct:main thm}.

For a universal $n$-form $X\in\Omega^n(\A)$, define $\odd(X):=1$ if $n$ is odd, and $\odd(X):=0$ if $n$ is even.
\begin{prop}\label{prop:cycl}
	Let $X$ and $Y$ be universal forms. Then
		$$\int_\phi XY-\int_\phi YX=\odd(Y)\int_\phi Y dX-\odd(X)\int_\phi X dY.$$
\end{prop}
\begin{proof}
Without loss of generality, assume that $X=x_0dx_1\dots dx_n$ and $Y=y_0dy_1\dots dy_k$ for some $x_0,\ldots,x_n,y_0,\ldots,y_k\in\A$.
By using $da b=d(ab)-adb$ repeatedly, we get
\begin{align*}
	\int_\phi XY =& \int_\phi x_0 dx_1\cdots dx_{n-1}(d(x_ny_0)-x_ndy_0)dy_1\cdots dy_k\\
	=&\int_\phi x_0\big( dx_1\cdots dx_{n-1}d(x_ny_0)-dx_1\cdots dx_{n-2}d(x_{n-1}x_n)dy_0+\dots\\
	&\ldots+(-1)^{n-1}d(x_1x_2)dx_3\cdots dx_ndy_0+(-1)^nx_1dx_2\cdots dx_ndy_0\big)dy_1\cdots dy_k\\
	=&\big\langle x_0[D,x_1],\ldots,[D,x_{n-1}],[D,x_ny_0],[D,y_1],\ldots,[D,y_k]\big\rangle\\
	&-\big\langle x_0[D,x_1],\ldots,[D,x_{n-2}],[D,x_{n-1}x_n],[D,y_0],[D,y_1],\ldots,[D,y_k]\big\rangle+\ldots\\
	&\ldots+(-1)^{n-1}\big\langle x_0[D,x_1x_2],[D,x_3],\ldots,[D,x_n],[D,y_0],[D,y_1],\ldots,[D,y_k]\big\rangle\\
	&+(-1)^n\big\langle x_0x_1[D,x_2],[D,x_3],\ldots,[D,x_n],[D,y_0],[D,y_1],\ldots,[D,y_k]\big\rangle\\
	=&\big\langle x_0[D,x_1],\ldots,[D,x_{n}]y_0,[D,y_1],\ldots,[D,y_k]\big\rangle\\
	&+\sum_{j=0}^{n-2}(-1)^j\big\langle x_0[D,x_1],\ldots,[D,x_n],[D,y_0],\ldots,[D,y_k]\big\rangle\\
	=&\big\langle x_0[D,x_1],\ldots,[D,x_{n}],y_0[D,y_1],\ldots,[D,y_k]\big\rangle-\odd(X)\int_\phi X dY.
\end{align*}
	Doing the same for $\int_\phi YX$ and using cyclicity (Lemma \ref{cycl bracket}\ref{cyclicity}) yields the proposition.
\end{proof}

A quick check shows that the above proposition implies the following handy rules.

\begin{cor}\label{cor:cycl}Let $X,Y\in\Omega^\bullet(\A)$, and $A\in\Omega^1(\A)$.
\begin{enumerate}[label=\textnormal{(\roman*)}]
	\item\label{cycl2} If $X$ and $Y$ are both of even degree, then
		$$\int_\phi XY=\int_\phi YX.$$
	\item\label{cycl3} If $X$ has odd degree, then
		$$\int_\phi (AX-XA)=\int_\phi d(AX).$$
	\item\label{cycl4} If $X$ has even degree, then
		$$\int_\phi(XA-AX)=\int_\phi dXA+\int_\phi dA\,dX.$$
\end{enumerate}	

\end{cor}


\subsection{Higher-order generalized Chern--Simons forms}
As a final preparation for the formulation of our main result, we briefly recall from \cite{Qui90} the definition of Chern--Simons forms. 

         \begin{defi}
           \label{defi:cs}
           Let $(\Omega^\bullet, d)$ be a differential graded algebra. The \textbf{Chern--Simons form} of degree $2k-1$ is given for $A \in \Omega^1$ by 
           \begin{equation}\label{eq:cs}
\cs_{2k-1}(A) := \int_0^1 A (F_t)^{k-1} \,dt,
           \end{equation}
           where $F_t = t dA + t^ 2 A^2$ is the curvature two-form of the (connection) one-form $A_t = t A$.
           \end{defi}
         We will only work with the universal differential graded algebra $\Omega^\bullet=\Omega^\bullet(\A)$ for the algebra $\A$.

         \begin{exam}
           For the first three Chern--Simons forms one easily derives the following explicit expressions:
           \begin{gather*}
             \cs_1(A) = A; \qquad  \cs_3(A) = \frac 12 \left( A dA + \frac 2 3 A^3 \right);\\
             \cs_5(A) = \frac 13 \left( A (dA)^2 + \frac 3 4 A dA A^ 2 + \frac 3 4 A^3 dA + \frac 3 5 A^5 \right).
             \end{gather*}
These are well-known expressions from the physics literature (\textit{cf.} \cite[Section 11.5.2]{Nak90}).
           \end{exam}

\section{Expansion of the spectral action in terms of $(b,B)$-cocycles}\label{sct:main thm}
\sectionmark{Expansion of the spectral action in terms of cocycles}
In this section we prove our main theorem, which is stated as follows.
\begin{thm}\label{thm:main thm}
  Let $(\A,\H,D)$ be an $s$-summable spectral triple, and let $f\in\E_s^{\gamma}$ for $\gamma\in(0,1)$. The spectral action fluctuated by $V=\pi_D(A)\in\Omega_D^1(\A)_\sa$ 
  can be written as
 \begin{align*}
  \Tr(f(D+V)-f(D)) = \sum_{k=1}^\infty \left( \int_{\psi_{2k-1}}  \cs_{2k-1} (A) +\frac 1 {2k} \int_{\phi_{2k}}  F^{k} \right),
  \end{align*}
  where the series converges absolutely.
\end{thm}
We prove this theorem in two steps. Firstly, we deal with the algebraic part of this statement, in \textsection\ref{sct:truncated expansion}. Here we only need to assume $f\in\Wsn$ for a finite $n\in\N$. Secondly, in \textsection\ref{sct:convergence}, we tackle the analytical part. We there obtain a strong estimate on the remainder of the above expansion in Theorem \ref{thm:main bounds} for a function $f\in\E_s^{\gamma}$ for general $\gamma\in(0,1]$. This estimate will imply that the conclusion of Theorem \ref{thm:main thm} still holds in the case of $\gamma=1$, when the perturbation $V$ is sufficiently small. When $f\in\E_s^{\gamma}$ for $\gamma\in(0,1)$, the expansion follows for all perturbations, and thus we prove Theorem \ref{thm:main thm}.

\subsection{Asymptotic expansion}\label{sct:truncated expansion}
Let $K\in\N$, $f\in\W_{s}^{2K}$, and $V=\pi_D(A)\in\Omega^1_D(\A)_\sa$.
We prove an asymptotic (or, one might say, truncated) version of Theorem \ref{thm:main thm}, showing that the fluctuation of the spectral action can be expressed in terms of Chern--Simons and Yang--Mills forms, up to a remainder which involves forms of degree higher than $K$.
To enumerate the remainder forms we use the index set
\begin{align}\label{eq:T_K}
  T_K:=\left\{(v,w,p)\in\coprod_{m\in\N_0} (\N_0\times\N^{m-1})\times\N^m\times\N_0\mid\begin{aligned}|v|+|w|+\left\lfloor\frac{p}{2}\right\rfloor&< K,\\ 2|v|+|w|+p&\geq K\end{aligned}\right\}.
\end{align}
A good first step is made by the following proposition.
\begin{prop}\label{prop:k}We have the asymptotic expansion
\begin{align*}
	\Tr(f(D+V)-f(D))\sim&\sum_{k=1}^\infty\int_\phi\bigg(\cs_{2k-1}(A)+\int_0^1AF_t^{k-1}tA\,dt\bigg),
\end{align*}
by which we mean that we can write the $K^\text{th}$ remainder of this expansion as
\begin{align*}
	&\Tr(f(D+V)-f(D))-\sum_{k=1}^K\int_\phi\bigg(\cs_{2k-1}(A)+\int_0^1AF_t^{k-1}tA\,dt\bigg)\\
	&\quad= \Tr\left(T_{f^{[K+1]}}^{D+V,D,\ldots,D}(V,\ldots,V)\right)\\&\qquad-\sum_{(v,w,p)\in T_K}\frac{1}{2|v|+|w|+p+1}\int_\phi AA^{2v_1}(dA)^{w_1}\cdots A^{2v_m}(dA)^{w_m}A^{p},
\end{align*}
where $T_K$, defined by \eqref{eq:T_K}, satisfies $|T_K|\leq 2^{K+1}$, and where $f\in\W_s^{2K}$.
\end{prop}
\begin{proof}
We start with the 2x2 matrix equation from Corollary \ref{cor:2x2 matrix} and separate the 1-forms $A$ from the two-forms $dA$. The $n$-th term in the Taylor expansion of $\Tr (f(D+V))$ is given (by use of \eqref{eq:<V,...,V>/n} and Corollary \ref{cor:2x2 matrix}) by
\begin{align}
	\frac{1}{n!}\frac{d^{n}}{dt^{n}}\Tr(f(D+tV))\Big|_{t=0}
	=&
	\frac{1}{n}\int_\phi\begin{pmatrix}
	A & 0
	\end{pmatrix}
	\left(\begin{pmatrix}
	A & -A\\
	0   & -A
	\end{pmatrix}+
	\begin{pmatrix}
	dA & 0\\
	dA   & 0
	\end{pmatrix}
	\right)^{n-1}\begin{pmatrix}
	1\\
	0
	\end{pmatrix}\nonumber\\
	\equiv&
	\frac{1}{n}\int_\phi\begin{pmatrix}
	A & 0
	\end{pmatrix}
	\left(\alpha A+
	\beta dA\right)^{n-1}\begin{pmatrix}
	1\\
	0
	\end{pmatrix}\nonumber\\
	=& \frac{1}{n}\int_\phi A\,e_1^t (\alpha A+\beta dA)^{n-1}e_1.
	\label{2x2 formula Tr}
\end{align}
	for some scalar-valued 2x2 matrices $\alpha$ and $\beta$, and $e_1=\vect{1}{0}$. The $\alpha$'s and $\beta$'s have very nice algebraic properties, which can be used to regroup the terms in the expansion in $n$.
When summing \eqref{2x2 formula Tr} from $n=1$ to infinity, and grouping the universal forms by their degree as in Example \ref{ex:first terms}, we need some machinery to keep track of the coefficient $1/n$. We will work in the space of (finite) polynomials $M_2(\Omega^\bullet(\A))[t]$, and define an integration with respect to $t$ as the linear map $\int_0^1dt\,:M_2(\Omega^\bullet(\A))[t]\to M_2(\Omega^\bullet(\A))$ given by integration of polynomials. We thus obtain
\begin{align}
	\frac{1}{n!}\frac{d^{n}}{dt^{n}}\Tr(f(D+tV))\Big|_{t=0}=& \int_0^1dt\,t^{n-1}\int_\phi A \,e_1^t(\alpha A+\beta dA)^{n-1}e_1\nonumber\\
	=&\int_0^1dt\int_\phi A\,e_1^t(\alpha tA+\beta t dA)^{n-1}e_1.\label{eq:M_2[t]}
\end{align}

We now expand the $(n-1)$-th power, which is complicated because $\alpha$ and $\beta$ do not commute. To avoid notational clutter, let us denote $X:=tA$ and $Y:=tdA$. We find
\begin{align}
	e_1^t (\alpha X+\beta Y)^{n-1}e_1=\sum_{k=0}^{\lceil\frac{n-1}{2}\rceil}\sum_{\substack{v_1\geq0,~v_2,\ldots,v_k\geq1\\ w_1,\ldots,w_k\geq 1,~p\geq0\\ |v|+|w|+p=n-1}}e_1^t (\alpha^{v_1}\beta^{w_1}\cdots\alpha^{v_k}\beta^{w_k}\alpha^{p})e_1\, X^{v_1}Y^{w_1}\cdots X^{v_k}Y^{w_k}X^{p}.\label{ncbinom}
\end{align}

We can summarize the identities involving $\alpha$ and $\beta$ that we will use as
\begin{align*}
	\alpha^2&=1;	&
	\beta^2&=\beta;	&
	\beta\alpha\beta&=0;& e_1^t(\alpha) e_1&=1;\\
	e_1^t (\alpha\beta\alpha) e_1&=0;	&	
	e_1^t (\alpha\beta) e_1 &=0;	&
	e_1^t (\beta\alpha) e_1 &= 1;	&
	e_1^t(\beta) e_1&=1.
\end{align*}
From these identities follow the following two remarks:
\begin{itemize}
\item If $k\geq2$ and $v_i$ is odd for a certain $i\in\{2,\ldots,k\}$, then somewhere in the string $\alpha^{v_1}\beta^{w_1}\cdots\alpha^{v_k}\beta^{w_k}\alpha^{p}$ a factor $\beta\alpha\beta=0$ occurs, so in particular $$e_1^t( \alpha^{v_1}\beta^{w_1}\cdots\alpha^{v_k}\beta^{w_k}\alpha^{p})e_1=0.$$
\item If $v_1$ is odd and $v_2,\ldots,v_k$ are all even, then
	$$e_1^t(\alpha^{v_1}\beta^{w_1}\cdots\alpha^{v_k}\beta^{w_k}\alpha^{p})e_1=e_1^t( \alpha\beta\alpha^{p})e_1=0.$$
\end{itemize}
Therefore, for all $k\geq 0$, we conclude that in \eqref{ncbinom} only terms remain in which $v_1,\ldots,v_k$ are even. In fact, we find
\begin{align*}
	e_1^t (\alpha X+\beta Y)^{n-1}e_1&=\sum_{k=0}^{\lceil\frac{n-1}{2}\rceil}\sum_{\substack{v_1\in2\N_0,\\v_2,\ldots,v_k\in2\N 
	\\w_1,\ldots,w_k\geq 1,~p\geq0,\\|v|+|w|+p=n-1}}e_1^t( \alpha^{v_1}\beta^{w_1}\cdots\alpha^{v_k}\beta^{w_k}\alpha^{p})e_1\, X^{v_1}Y^{w_1}\cdots X^{v_k}Y^{w_k}X^{p}\\
	&=\sum_{k=0}^{\lceil{\frac{n-1}{2}}\rceil}\sum_{\substack{v_1\in2\N_0,~v_2,\ldots,v_k\in2\N\\w_1,\ldots,w_k\geq 1,~p\geq0,\\|v|+|w|+p=n-1}}e_1^t( \beta\alpha^{p})e_1\, X^{v_1}Y^{w_1}\cdots X^{v_k}Y^{w_k}X^{p}\\
	&=\sum_{k=0}^{\lceil\frac{n-1}{2}\rceil}\sum_{\substack{v_1\geq0,~v_2\ldots,v_k\geq1,\\w_1,\ldots,w_k\geq 1,~ p\geq0,\\2|v|+|w|+p=n-1}} (X^2)^{v_1}Y^{w_1}\cdots (X^2)^{v_k}Y^{w_k}X^{p}.
\end{align*}
Summing this from $n=1$ to $K$, we can write
\begin{align}\label{eq:one expansion}
	\sum_{n=1}^K e_1^t (\alpha X+\beta Y)^{n-1}e_1 = \sum_{(v,w,p)\in P_K}(X^2)^{v_1}Y^{w_1}\cdots (X^2)^{v_m}Y^{w_m}X^{p}
	,
\end{align}
where $P_K$ is the set of $(v,w,p)\in\coprod_{m}(\N_0\times\N^{m-1})\times\N^m\times\N_0$ such that $2|v|+w+p< K$.
%
In this last expression we can almost recognize an expansion of $(X^2+Y)^{k-1}=F_t^{k-1}$. Indeed, we have
\begin{align}\label{eq:other expansion}
	\sum_{k=1}^K(X^2+Y)^{k-1}(1+X)=\sum_{(v,w,p)\in S_K} (X^2)^{v_1}Y^{w_1}\cdots(X^2)^{v_m}Y^{w_m}X^{p},
\end{align}
where $S_K$ is the set of $(v,w,p)\in\coprod_{m}(\N_0\times\N^{m-1})\times\N^m\times\N_0$ such that $|v|+|w|+\lfloor \tfrac{p}{2}\rfloor< K.$ By \eqref{eq:other expansion} we have $|S_K|\leq 2^{K+1}$.
By using $T_K=S_K\setminus P_K$, we can combine \eqref{eq:M_2[t]}, \eqref{eq:one expansion} and \eqref{eq:other expansion}, and obtain
\begin{align}
	&\sum_{n=1}^{K}\frac{1}{n!}\frac{d^n}{dt^n}\Tr\big(f(D+tV)\big)\Big|_{t=0}-\sum_{k=1}^K\int_\phi\int_0^1AF_t^{k-1}(1+tA)\,dt\nonumber\\
	&\quad= -\sum_{(v,w,p)\in T_K}\frac{1}{2|v|+|w|+p+1}\int_\phi AA^{2v_1}(dA)^{w_1}\cdots A^{2v_m}(dA)^{w_m}A^{p}.\label{eq:difference of expansions}
\end{align}
Together with \eqref{remmoi}, and the definition \eqref{eq:cs} of $\cs_{2k-1}(A)$, \eqref{eq:difference of expansions} implies the proposition.
\end{proof}

We will now state the asymptotic version of our main result, and spend the rest of \textsection\ref{sct:truncated expansion} to prove this version.
\begin{thm}\label{thm:asymptotic expansion}
For every $k\in\N$ we have
	$$\int_\phi\bigg(\cs_{2k-1}(A)+\int_0^1AF_t^{k-1}tA\,dt\bigg)=\int_{\psi_{2k-1}}\cs_{2k-1}(A)+\frac{1}{2k}\int_{\phi_{2k}} F^{k}.$$
	Therefore, with the same remainder term as in Proposition \ref{prop:k}, we have
	$$\Tr(f(D+V)-f(D))\sim\sum_{k=1}^\infty\bigg(\int_{\psi_{2k-1}}\cs_{2k-1}(A)+\frac{1}{2k}\int_{\phi_{2k}} F^{k}\bigg).$$
\end{thm}

The Chern--Simons term in Proposition \ref{prop:k}, integrated along $\phi$, yields the correct Chern--Simons term integrated along $\psi$, plus an additional term. Indeed, recall that $\phi_{2k-1} = \psi_{2k-1} + \frac 12 B_0 \phi_{2k}$ so that we find
\begin{align*}
  &\int_{\phi_{2k-1}}  A F_t^{k-1} +  \int_{\phi_{2k-1}} t A F_t^{k-1} A \\
  &\quad=\int_{\psi_{2k-1}}  A F_t^{k-1} +  \int_{\phi_{2k}} \Big( \frac 12  d(A F_t^{k-1} ) + t A F_t^{k-1} A\Big)\\
  &\quad=  \int_{\psi_{2k-1}}  A F_t^{k-1} + \frac 12  \int_{\phi_{2k}} ( d A F_t^{k-1} + t A^2 F_t^{k-1}+ t A F_t^{k-1} A),
\end{align*}
where we used the repeated Bianchi identity $d(F_t^{k-1}) = - [A_t,F_t^{k-1}]$ in going to the last line.

We arrive at the following formula:
\begin{align*}
&\Tr(f(D+V)-f(D)) \\
&\quad\sim \sum_{k=1}^\infty \left( \int_{\psi_{2k-1}}  \cs_{2k-1} (A)  + \frac 12 \int_0^1 dt \int_{\phi_{2k}} ( d A F_t^{k-1} + t A^2 F_t^{k-1}+ t A F_t^{k-1} A)\right).
\end{align*}
We are now to show that the second term, namely
\begin{align*}
	YM_k:=&\frac{1}{2}\int_0^1 dt\int_{\phi_{2k}}(dA F_t^{k-1}+tA^2F_t^{k-1}+tAF_t^{k-1}A)\\
	=&\int_0^1 dt\,\frac{1}{2t}\int_{\phi_{2k}}(dA_t F_t^{k-1}+A_t^2F_t^{k-1}+A_tF_t^{k-1}A_t),
\end{align*}
equals $\frac{1}{2k}\int_{\phi_{2k}} F^{k}$.
After some rearrangement we can use Corollary \ref{cor:cycl}\ref{cycl3}, to find
\begin{align}
	YM_k&=\int_0^1 dt\,\frac{1}{2t}\int_{\phi_{2k}}(dA_t+2A_t^2) F_t^{k-1}+\int_0^1 dt\,\frac{1}{2t}\int_{\phi_{2k}}\left(A_tF_t^{k-1}A_t-A_t^2F_t^{k-1}\right)\nonumber\\
	&=\int_0^1 dt\,\frac{1}{2t}\int_{\phi_{2k}}(dA_t+2A_t^2) F_t^{k-1} - \int_0^1 dt\,\frac{1}{2t}\int_{\phi_{2k+1}} d(A_t^2 F_t^{k-1}).\label{eq:M}
\end{align}
We will first show that the second term of \eqref{eq:M} vanishes. We use the following rule, which allows us to replace the integrand by a form which is two degrees lower.
\begin{lem}\label{lem:een}
	For every $m\geq0$, we have
		$$\int_{\phi_{2m+3}} d(A_t^2F_t^m)=-\int_{\phi_{2m+1}} \left(d(A_t^2F_t^{m-1})+dA_td(F_t^{m-1})\right).$$
\end{lem}
\begin{proof}
	We use the definition of $F_t$, the repeated Bianchi identity $d(F_t^m)=[F_t^m,A_t]$, and subsequently Proposition \ref{prop:cycl}, to obtain
	\begin{align*}
		\int_{\phi_{2m+1}} d(A_t^2F_t^{m-1})=&\int_{\phi_{2m+1}} \left(d(F_tF_t^{m-1})-d(dA_tF_t^{m-1})\right)\\
		=&\int_{\phi_{2m+1}} \left(d(F_t^m)-dA_td(F_t^{m-1})\right)\\
		=&\int_{\phi_{2m+1}}\left(F_t^mA_t-A_tF_t^m\right)-\int_{\phi_{2m+1}} dA_td(F_t^{m-1})\\
		=&\int_{\phi_{2m+2}} A_td(F_t^m)-\int_{\phi_{2m+1}} dA_t d(F_t^{m-1})\\
		=&\int_{\phi_{2m+2}} \left(A_tF_t^mA_t-A_t^2F_t^m\right)-\int_{\phi_{2m+1}} dA_td(F_t^{m-1}).
	\end{align*}
	Applying Corollary \ref{cor:cycl}\ref{cycl3} to the first term gives
	\begin{align*}
		\int_{\phi_{2m+1}} d(A_t^2F_t^{m-1})=& -\int_{\phi_{2m+3}} d(A_t^2F_t^m)-\int_{\phi_{2m+1}} dA_td(F_t^{m-1}),
	\end{align*}
	which implies the lemma.
\end{proof}
We obtain the following result.
\begin{lem}\label{lem:een b}
	For every $m\geq0$, we have $$\int_{\phi_{2m+3}} d(A_t^2F_t^m)=0.$$
\end{lem}
\begin{proof}
	This is easily checked when $m=0$, and when $m=1$, it follows from Lemma \ref{lem:een}. If $m\geq 2$, we apply Lemma \ref{lem:een} twice, and find
	\begin{align*}
		\int_{\phi_{2m+3}} d(A_t^2F_t^m)=&-\int_{\phi_{2m+1}} d(A_t^2F_t^{m-1})-\int_{\phi_{2m+1}} dA_td(F_t^{m-1})\\
		=&\int_{\phi_{2m-1}} d(A_t^2F_t^{m-2})+\int_{\phi_{2m-1}} dA_td(F_t^{m-2})-\int_{\phi_{2m+1}} dA_td(F_t^{m-1}).
	\end{align*}
	We recognize $F_t=A_t^2+dA_t$ in the first two terms above, so by the Bianchi identity we obtain
	\begin{align*}
		\int_{\phi_{2m+3}} d(A_t^2F_t^m)=&\int_{\phi_{2m-1}} d(F_t^{m-1})-\int_{\phi_{2m+1}} dA_td(F_t^{m-1})\\
		=&\int_{\phi_{2m-1}} (F_t^{m-1}A_t-A_tF_t^{m-1})-\int_{\phi_{2m+1}} dA_td(F_t^{m-1}).
	\end{align*}
	By Corollary \ref{cor:cycl}\ref{cycl4}, the Bianchi identity, and Corollary \ref{cor:cycl}\ref{cycl2}, this gives
	\begin{align*}
		\int_{\phi_{2m+3}} d(A_t^2F_t^m)=&\int_{\phi_{2m}} d(F_t^{m-1})A_t\\
		=&\int_{\phi_{2m}} (F_t^{m-1}A_t^2-A_tF_t^{m-1}A_t)\\
		=&\int_{\phi_{2m}} (A_t^2F_t^{m-1}-A_tF_t^{m-1}A_t).
	\end{align*}
	We apply Corollary \ref{cor:cycl}\ref{cycl3}, to find
	\begin{align*}
		\int_{\phi_{2m+3}} d(A_t^2F_t^m)=&\int_{\phi_{2m+1}} d(A_t^2F_t^{m-1}).
	\end{align*}
	By induction, it follows that $\int_\phi d(A_t^2F_t^m)=0$ for all $m$.
\end{proof}
By the above Lemma, only the first term of \eqref{eq:M} remains, namely,
\begin{align}\label{eq:last term of M}
	YM_k=&\int_0^1 dt\,\frac{1}{2t}\int_{\phi_{2k}}(dA_t+2A_t^2) F_t^{k-1}	=\int_0^1 dt\int_{\phi_{2k}}(\tfrac{1}{2}dA+tA^2) F_t^{k-1}.
\end{align}
To express $YM_k$ in an even simpler form, we now remove the integral over $t$, which is possible by the following lemma.
\begin{lem}\label{lem:twee}
	We have
		$$\int_0^1 dt\int_{\phi_{2k}} (\tfrac{1}{2}dA+tA^2)F_t^{k-1}=\frac{1}{2k}\int_{\phi_{2k}}F^{k}.$$
\end{lem}
\begin{proof}
	Recall that $\Omega^\bullet(\A)[t]$ is the space of polynomials with coefficients in the algebra $\Omega^\bullet(\A)$. The linear map $\frac{d}{dt}:\Omega^\bullet(\A)[t]\rightarrow\Omega^\bullet(\A)[t]$ is defined by $\frac{d}{dt}(t^nB):=nt^{n-1}B$ for $B\in\Omega^\bullet(\A)$, and satisfies the Leibniz rule. Therefore,
	$$\frac{d}{dt}(F_t^{k})=\frac{d}{dt}(F_t)F^{k-1}_t+F_t\frac{d}{dt}(F_t)F_t^{k-2}+\ldots+F_t^{k-1}\frac{d}{dt}(F_t).$$
	Both $F_t$ and $\frac{d}{dt}(F_t)$ are $2$-forms, so, after a few applications of Corollary \ref{cor:cycl}\ref{cycl2}, we arrive at
	\begin{align*}
		\int_{\phi_{2k}}\frac{d}{dt}(F_t^{k})&=k\int_{\phi_{2k}}\frac{d}{dt}(F_t)F_t^{k-1}=k\int_{\phi_{2k}}(dA+2tA^2)F_t^{k-1}.
	\end{align*}
	The fundamental theorem of calculus (for polynomials) gives
	\begin{align*}
		\int_0^1 dt\int_{\phi_{2k}}\,(dA+2tA^2)F_t^{k-1}=&\frac{1}{k}\int_{\phi_{2k}}\int_0^1dt \,\frac{d}{dt}(F_t^{k})	=\frac{1}{k}\int_{\phi_{2k}}(F_1^{k}-F_0^{k})=\frac{1}{k}\int_{\phi_{2k}}F^{k},
	\end{align*}
	from which the lemma follows.
\end{proof}

\begin{proof}[Proof of Theorem \ref{thm:asymptotic expansion}]
Applying Lemma \ref{lem:twee} to our earlier expression for $YM_k$ (equation \eqref{eq:last term of M}), we find that
\begin{align*}
	YM_k=\frac{1}{2k}\int_{\phi_{2k}}F^{k}.
\end{align*}
We therefore obtain the desired asymptotic expansion.
\end{proof}

\subsection{Convergence}\label{sct:convergence}
We prove a strong bound on the asymptotic expansion given by Theorem \ref{thm:asymptotic expansion}, in particular giving sufficient conditions for the series to converge, effectively replacing $\sim$ by $=$. A crucial ingredient is Lemma \ref{lem:Cs and Es bounds}.

\begin{thm}\label{thm:main bounds}
	Let $(\A,\H,D)$ be an $s$-summable spectral triple, let $n\in\N$, and fix $f\in \Esg$ for $\gamma\in(0,1]$. Then there exists $C_{f,s,n,\gamma}$ such that, for $A=\sum_{j=1}^na_jdb_j$ and $V=\sum_{j=1}^n a_j[D,b_j]$ self-adjoint with $\|a_j\|,\|b_j\|,\|[D,a_j]\|,\|[D,b_j]\|\leq R$, we have
	\begin{align}
		&\left|\Tr(f(D+V)-f(D))-\sum_{k=1}^K\left(\int_{\psi_{2k-1}}\cs_{2k-1}(A)+\frac{1}{2k}\int_{\phi_{2k}}F^{k}\right)\right|\nonumber\\
		&\quad\leq \frac{(C_{f,s,n,\gamma})^{K+1}}{K!^{1-\gamma}}\max(R^{2K+2},R^{4K+2+2s})\Tr |(D-i)^{-s}|,\label{eq:convergence}
	\end{align}
	for all $K\in\N_0$. Moreover, we have
		$$\left|\int_{\psi_{2k-1}}\cs_{2k-1}(A)\right|+\left|\int_{\phi_{2k}}F^{k}\right|\leq \frac{(C_{f,s,n,\gamma})^{k}}{k!^{1-\gamma}}\max(R^{2k},R^{4k})\Tr|(D-i)^{-s}|.$$
\end{thm}
\begin{proof}
Theorem \ref{thm:asymptotic expansion} gives
\begin{align}
	&\left|\Tr(f(D+V)-f(D))-\sum_{k=1}^K \left(\int_{\psi_{2k-1}}\cs_{2k-1}(A)+\frac{1}{2k}\int_{\phi_{2k}}F^{k}\right)\right|\nonumber\\
	&\quad\leq \left\|T_{f^{[K+1]}}^{D+V,D,\ldots,D}(V,\ldots,V)\right\|_1+\sum_{(v,w,p)\in T_K}\left|\int_\phi AA^{2v_1}(dA)^{w_1}\cdots A^{2v_m}(dA)^{w_m}A^{p}\right|.\label{eq:propk}
\end{align}
We first focus on the first term. Lemma \ref{lem:Cs and Es bounds} gives a $C\geq1$ such that
\begin{align*}
	&\left\|T_{f^{[K+1]}}^{D+V,D,\ldots,D}(V,\ldots,V)\right\|_1\\
	&\quad\leq \frac{C^{K+1}}{K!^{1-\gamma}}
	\sum_{j_1,\ldots,j_{K+1}\in\{1,\ldots,n\}}\prod_{m=1}^{K+1}\norm{a_{j_m}}\norm{[D,b_{j_m}]}(1+\norm{V})^{s}\Tr|(D-i)^{-s}|\\
	&\quad\leq \frac{C^{K+1}}{K!^{1-\gamma}}n^{K+1}R^{2K+2}(1+\norm{V})^{s}\Tr|(D-i)^{-s}|,
\end{align*}
for all $K\in\N_0$. We conclude that there exists $\tilde C_{f,s,n,\gamma}$ such that
\begin{align*}
	\left\|T_{f^{[K+2]}}^{D+V,D,\ldots,D}(V,\ldots,V)\right\|_1\leq \frac{\tilde C_{f,s,n,\gamma}^{K+1}}{K!^{1-\gamma}}\max(R^{2K+2},R^{2K+2+2s})\Tr|(D-i)^{-s}|.
\end{align*}

We now move on to the second term (the finite sum) on the right-hand side of \eqref{eq:propk}. It contains terms of the form
\begin{align*}
	\left|\int_\phi B_1\cdots B_M\right|,
\end{align*}
for $B_1,\ldots,B_M\in \{a_jdb_j,da_jdb_j:j\in\{1,\ldots,n\}\}$. Let $l$ be the degree of $B_1\cdots B_M$. By the definition of $T_K$ (equation \eqref{eq:T_K}) $K+1\leq l\leq 2K+1$ and $K+1\leq M\leq 2K+1$. By using the Leibniz rule repeatedly, we can write
\begin{align*}
	\int_\phi B_1\cdots B_M=\sum_{j\in J}\int_\phi e_{0,j}de_{1,j}\cdots de_{l,j},
\end{align*}
for a set $J$ with $|J|\leq 3^M\leq 3^{2K+1}$, and $e_{i,j}\in\A$ such that $e_{0,j}\cdots e_{l,j}=\prod_{m=1}^M a_{j_m}b_{j_m}$.
We get
\begin{align}
	\left|\int_\phi B_1\cdots B_M\right|&\leq\sum_{j\in J}|\int_\phi e_{0,j}de_{1,j}\cdots de_{l,j}|\nonumber\\
	&=\sum_{j\in J}|\phi_l(e_{0,j},\ldots,e_{l,j})|\nonumber\\
	&\leq \sum_{j\in J}\sum_{i=1}^l|\Tr(T^D_{f^{[l]}}([D,e_i],\ldots,[D,e_l],e_0[D,e_1],[D,e_2],\ldots,[D,e_{i-1}])|,\label{eq:phi to MOI}
\end{align}
where we suppressed the index $j$ for readability.
We now apply Lemma \ref{lem:Cs and Es bounds} with $V=0$ to \eqref{eq:phi to MOI} and obtain
\begin{align*}
	\left|\int_\phi B_1\cdots B_M\right|
	&\leq  lC^{l+1}l!^{\gamma-1}\sum_{j\in J}\norm{e_0}\Big(\prod_{i=1}^l\norm{[D,e_i]}\Big)\Tr|(D-i)^{-s}|,
\end{align*}
for a constant $C\geq 1$. Because we have $\norm{a_j},\norm{b_j}$, $\norm{[D,a_j]},\norm{[D,b_j]}\leq R$ by assumption, and $e_{0}\cdots e_{l}=\prod_{m=1}^M a_{j_m}b_{j_m}$, with $K+1\leq M\leq 2K+1$, we find
\begin{align*}
	\left|\int_\phi B_1\cdots B_M\right|
	&\leq \tilde C^{l+1}l!^{\gamma-1}\sum_{j\in J} R^{2M}\Tr|(D-i)^{-s}|\\
	&\leq \hat C^{K+1}|J|K!^{\gamma-1}\max(R^{2K+2},R^{4K+2})\Tr|(D-i)^{-s}|\\
	&\leq \check C^{K+1}K!^{\gamma-1}\max(R^{2K+2},R^{4K+2})\Tr|(D-i)^{-s}|.
\end{align*}
We can now bound the second term on the right-hand side of \eqref{eq:propk}. We use that $|T_K|\leq 2^{K+1}$, and that $n^M\leq(n^2)^{K+1}$,
to find
\begin{align*}
	&\sum_{(v,w,p)\in T_K}\left|\int_\phi AA^{2v_1}(dA)^{w_1}\cdots A^{2v_m}(dA)^{w_m}A^{p}\right|\\
	&\quad\leq 2^{K+1}(n^2)^{K+1}\check C^{K+1}K!^{\gamma-1}\max(R^{2K+2},R^{4K+2})\Tr|(D-i)^{-s}|\\
	&\quad\leq \breve C_{f,s,n,\gamma}^{K+1}K!^{\gamma-1}\max(R^{2K+2},R^{4K+2})\Tr|(D-i)^{-s}|.
\end{align*}
Combining the first and second term of \eqref{eq:propk}, we obtain a number $C_{f,s,n,\gamma}$ such that \eqref{eq:convergence} holds. 

Moving on to the last claim of the theorem, we notice that, because $\psi_{2k-1}=\phi_{2k-1}-\frac{1}{2}B_0\phi_{2k}$,
	$$\left|\int_{\psi_{2k-1}}\cs_{2k-1}(A)\right|\leq\sum_{j\in J}\left|\int_\phi e_{0,j}de_{1,j}\cdots de_{l_j,j}\right|,$$
	where the sum is over certain $e_{i,j}\in\mathcal{A}$ (because $\mathcal{A}$ is unital) with $e_{0,j}\cdots e_{l_j,j}=\prod_{m=1}^{M}a_{j_m}b_{j_m}$ for some $M$ with $k\leq M\leq 2k-1$. The number of elements in $J$ is exponential in $k$. We obtain
	$$\left|\int_{\psi_{2k-1}}\cs_{2k-1}(A)\right|\leq \check C_{f,s,n,\gamma}^{k}k!^{\gamma-1}\max(R^{2k},R^{4k-2})\Tr|(D-i)^{-s}|,$$
	for some number $\check C_{f,s,n,\gamma}\geq1$. Similarly, we obtain a number $\hat C_{f,s,n,\gamma}\geq1$ such that
	$$\left|\int_{\phi_{2k}}F^{k}\right|\leq \hat C_{f,s,n,\gamma}^{k}k!^{\gamma-1}\max(R^{2k},R^{4k})\Tr|(D-i)^{-s}|,$$
	thereby proving the theorem.
\end{proof}
This theorem has two important consequences; for $f\in\E_s^1$ (hence, for all $f\in\E_s^\gamma$) we obtain the following corollary, and for $f\in\E_s^\gamma$, $\gamma<1$ we obtain our main theorem.

\begin{cor}
	Let $(\A,\H,D)$ be an $s$-summable spectral triple, let $f\in\E_s^1$ and $V=\pi_D(A)\in\Omega^1_D(\A)_\sa$. Then there exists a $\delta>0$ such that for all $t\in\R$ with $|t|<\delta$, we have
	$$\Tr(f(D+tV)-f(D))=\sum_{k=1}^\infty\left(\int_{\psi_{2k-1}}\cs_{2k-1}(tA)+\frac{1}{2k}\int_{\phi_{2k}}F_t^{k}\right),$$
	and the series converges absolutely.
\end{cor}
\begin{proof}
	Write $V=\sum_{j=1}^na_j[D,b_j]$. First take $C_{f,s,n,1}\geq 1$ from Theorem \ref{thm:main bounds}, define $R:=1/(C_{f,s,n,1}+1)$ such that $C_{f,s,n,1}R<1$, and define $\delta:= \left(\frac{R}{\max_j\{\|a_j\|,\|b_j\|,\|[D,a_j]\|,\|[D,b_j]\|\}}\right)^2$. Writing
	$$tV=\sum_{j=1}^n\sqrt{|t|}a_j[D,\text{sign}(t)\sqrt{|t|}b_j],$$
	the corollary follows by applying (the first and second part of) Theorem \ref{thm:main bounds} to $tV$ instead of $V$.
\end{proof}
\begin{proof}[Proof of Theorem \ref{thm:main thm}.]
	This follows from Theorem \ref{thm:main bounds} by taking $\gamma<1$.
\end{proof}

\section{Gauge invariance and the pairing with K-theory}\label{sct:vanishing pairing}

 Since the spectral action is a spectral invariant, it is in particular invariant under conjugation of $D$ by a unitary $U\in \A$. More generally, in the presence of an inner fluctuation we find that the spectral action is invariant under the transformation
$$
D+V \mapsto U (D+V) U^* = D + V^U; \qquad V^U = U[D,U^*] + U V U^*.
$$
This transformation also holds at the level of the universal forms, with a gauge transformation of the form $A \mapsto A^U = U d U^* + U A U ^*$. Let us analyze the behavior of the Chern--Simons and Yang--Mills terms appearing in Theorem \ref{thm:main thm} under this gauge transformation, and derive an interesting consequence for the pairing between the odd $(b,B)$-cocycle $\tilde \psi$ with the odd K-theory group of $\A$.

\begin{lem}
  The Yang--Mills terms $\int_{\phi_{2k}} F^k$ with $F = dA +A^2$ are invariant under the gauge transformation $A \mapsto A^U$ for every $k \geq 1$. 
  \end{lem}
\proof
Since the curvature of $A^U$ is simply given by $U F U^*$, the claim follows 
from Corollary \ref{cor:cycl}\ref{cycl2}.
\endproof


We are thus led to the conclusion that the Chern--Simons forms are gauge invariant as well. Indeed, arguing as in \cite{CC06}, since both $\Tr (f(D+V))$ and the Yang-Mills terms are invariant under $V \mapsto V^U$, we find that, under the assumptions stated in Theorem \ref{thm:main thm}:
$$
\sum_{k=0}^\infty \int_{\psi_{2k+1}} \cs_{2k+1}(A^U ) =  
\sum_{k=0}^\infty \int_{\psi_{2k+1}} \cs_{2k+1} (A ).
$$
Each individual Chern--Simons form behaves non-trivially under a gauge transformation. Nevertheless, it turns out that we can conclude, just as in \cite{CC06}, that the pairing of the whole $(b,B)$-cocycle with K-theory is trivial. 
Since the $(b,B)$-cocycle $\tilde \psi$ is given as an infinite sequence, we should first carefully study the analytical behavior of $\tilde \psi$. In fact, we should show that it is an \textit{entire cyclic cocycle} in the sense of \cite{C88a} (see also \cite[Section IV.7.$\alpha$]{C94}). For this purpose, we can without loss of generality assume that $\A$ is complete in the Banach algebra norm defined by $\|a\|_1:=\|a\|+\|[D,a]\|$, because $(\overline{\A}^{\|\cdot\|_1},\H,D)$ is also a spectral triple, and the resulting $\tilde\psi_{2k+1}\in \mathcal{C}^{2k+1}(\overline{\A}^{\|\cdot\|_1})$ is an extension of the one in $\mathcal{C}^{2k+1}(\A)$. Recall that for Banach algebras $\A$ an odd cochain such as $\tilde \psi$ is called entire if the power series $\sum_k  \frac{(2k+1)!}{k!}\| \tilde \psi_{2k+1} \| z^k$ converges everywhere in $\C$. This is equivalent \cite[Remark IV.7.7a,c]{C94} to the condition that for any bounded subset $\Sigma \subset \A$ there exists a constant $C_\Sigma$ such that
$$
\left|\tilde \psi_{2k+1} (a_0, \ldots, a_{2k+1}) \right|  \leq \frac{C_\Sigma}{ k!} \qquad (\forall a_j \in \Sigma).
$$
In our case it turns out that Lemma \ref{lem:Cs and Es bounds} implies the following growth condition, guaranteeing that indeed $\tilde\psi$ is entire.

\begin{lem}
	Fix $f\in\Esg$ for $\gamma<1$ and equip $\A$ with the norm $\| a \|_1 = \| a \| + \| [D,a]\|$. Then, for any bounded subset $\Sigma\subset\mathcal{A}$ there exists $C_\Sigma$ such that
		$$\left |\tilde{\psi}_{2k+1}(a_0,\ldots, a_{2k+1}) \right|\leq \frac{C_\Sigma}{k!},$$
	for all $a_j\in\Sigma$.
\end{lem}

\begin{proof}
	Assume that $\| a_j \|_1 \leq R$ for all $a_j \in \Sigma$ so that both $\|a_j\|,\|[D,a_j]\|\leq R$. By definition of $\phi$, the expression $\psi_{2k+1}(a_0,\ldots,a_{2k+1})$ is given by a linear combination of multiple operator integrals with arguments in $\{V\in\mB(\H):\|V\|\leq R\}$ except for $a_0[D,a_1]$, which is bounded by $R^2$. By applying Lemma \ref{lem:Cs and Es bounds}, we obtain the estimate
	\begin{align}\label{eq:bound psi}
	\left	|\psi_{2k+1}(a_0,\ldots,a_{2k+1}) \right|\leq\bigg((2k+1)\frac{C^{2k+2}}{(2k+1)!^{1-\gamma}}+(k+1)\frac{C^{2k+3}}{(2k+2)!^{1-\gamma}}\bigg)R^{2k+2}\|(D-i)^{-1}\|_{s}^s.
	\end{align}
	We recall from Proposition \ref{prop:bB} that
	\begin{align}\label{tilde psi reprise}
		\tilde{\psi}_{2k+1}=(-1)^{k}\frac{k!}{(2k+1)!}\psi_{2k+1},
	\end{align}
	so that \eqref{eq:bound psi} in particular implies the lemma by use of, for instance, Stirling's approximation.
\end{proof}
For $U\in M_q(\A)$, define a pairing
\begin{align}\label{eq:pairing}
	\langle U,\tilde\psi\rangle:=(2\pi i)^{-1/2}\sum_{k=0}^\infty (-1)^{k}k!\tilde{\psi}^q_{2k+1}(U^*,U,\ldots,U^*,U),
\end{align}
where $$\tilde{\psi}^q_{2k+1}:=\Tr \# \tilde\psi_{2k+1}:(\mu_0\otimes a_0,\ldots,\mu_{2k+1}\otimes a_{2k+1})\mapsto\Tr(\mu_0\cdots\mu_{2k+1})\tilde\psi_{2k+1}(a_0,\ldots,a_{2k+1})$$ for $\mu_0,\ldots,\mu_{2k+1}\in M_q(\C)$ and $a_0,\ldots,a_{2k+1}\in\A$. Since $\tilde \psi$ is a $(b,B)$-cocycle, it follows from \cite[Corollary IV.7.27]{C94} (see also \cite[Sections III.3 and IV.7]{C94}) that this pairing only depends on the class of $U$ in $K_1(\A)$. 
We now prove an interesting consequence of our main theorem.

\begin{thm}
  Let $f\in\Esg$ for $\gamma<1$. Then the pairing of the odd entire cyclic cocycle $\tilde \psi$ with $K_1(\A)$ is trivial, \textit{i.e.}
  $$\langle U,\tilde\psi\rangle=0
  $$
  for all unitary $U\in M_q(\A)$.
\end{thm}
\begin{proof}
	Apply Theorem \ref{thm:main thm} to a bigger spectral triple, namely $(\A^q,\H^{q},D^q):=(M_q(\C)\otimes \A,\C^q\otimes\H,I_q\otimes D)$. Take $A=U^*dU$ for $U$ unitary in $M_q(\A)=M_q(\C)\otimes\A$.
	Clearly, then $V=U^*[D^q,U]$, and because the multiple operator integral behaves naturally with respect to tensor products, we obtain
		$$\Tr(f(D^q+U^*[D^q,U])-f(D^q))=\sum_{k=0}^\infty\left(\int_{\psi^q_{2k+1}}\cs_{2k+1}(U^*dU)+\frac{1}{2k+2}\int_{\phi^q_{2k+2}}F^{k+1}\right),$$
	where $F=d(U^*dU)+(U^*dU)^2=0$. The left-hand side equals $\Tr(f(U^*D^qU)-f(D^q))=0.$
	Therefore,
	\begin{align}\label{eq:cs is 0}
		\sum_{k=0}^\infty\int_{\psi^q_{2k+1}}\cs_{2k+1}(U^*dU)=0.
	\end{align}
        From the definition of the Chern--Simons form (Definition \ref{defi:cs}) and the fact that $F_t=tdA+t^2A^2=(t-t^2)dA+t^2F=(t-t^2)dU^*dU$ we find that
	\begin{align*}
		\cs_{2k+1}(U^*dU)&=\int_0^1dt\, (t-t^2)^{k}U^*dUdU^*dU\cdots dU^*dU,
	\end{align*}
	so that by a straightforward integration we may conclude that
	\begin{align*}
		\int_{\psi^q_{2k+1}}\cs_{2k+1}(U^*dU)=\frac{k!^2}{(2k+1)!}\psi^q_{2k+1}(U^*,U,\ldots,U^*,U).
	\end{align*}
Combining this with \eqref{tilde psi reprise}, \eqref{eq:pairing} and \eqref{eq:cs is 0}, the theorem follows.
\end{proof}


\section{One-loop corrections to the spectral action}
\label{sct:One-Loop}

As a last application of the expansion obtained in Section \ref{sct:main thm}, in this section we will show how the asymptotic expansion allows us to formulate a quantum version of the spectral action. To do this, we must first interpret the spectral action, expanded in terms of generalized Chern--Simons and Yang--Mills actions by Theorem \ref{thm:main thm}, as a classical action, which leads us naturally to a noncommutative geometric notion of a vertex. Enhanced with a spectral gauge propagator derived from the formalism of random matrices (and in particular, random finite noncommutative geometries) this gives us a concept of one-loop counterterms and a proposal for a one-loop \textit{quantum effective spectral action}, without leaving the spectral framework. We will show here that, at least in a finite-dimensional setting, these counterterms can again be written as Chern--Simons and Yang--Mills forms integrated over (quantum corrected) cyclic cocycles. We therefore discern a renormalization flow in the space of cyclic cocycles. 


\begin{fmffile}{graphs}

  \fmfset{wiggly_len}{2mm}

  \fmfset{dot_len}{1.5mm}

  \subsection{Conventions}



%

We let $\varphi_1,\varphi_2,\ldots$ be an orthonormal basis of eigenvectors of $D$, with corresponding eigenvalues $\lambda_1,\lambda_2,\ldots$. For any $N\in\N$, we define
\begin{align*}
	H_N:=(M_N)_\sa,\quad M_N:=\spn\left\{\ket{\varphi_i}\bra{\varphi_j}:~i,j\in\{1,\ldots,N\}\right\},
\end{align*}
and endow $H_N$ with the Lebesgue measure on the coordinates $Q\mapsto \Re( Q_{ij})$ ($i\leq j$) and $Q\mapsto\Im(Q_{ij})$ ($i<j$). Here and in the following, $Q_{ij}:=\p{\varphi_i}{Q\varphi_j}$ are the matrix elements of $Q$.
For simplicity, we will assume that the perturbations $V_1,\ldots,V_n$ are in $\cup_K H_K$. Of course, as explained in Section \ref{sct:MOI Introduction}, we would like to eventually consider noncompact perturbations as well. This would be a challenging analytic endeavor and will not be pursued here, but we note that the techniques developed in Chapter \ref{ch:MOI} 
might provide essential help.

For us, a \textit{Feynman diagram} is a finite multigraph with a number of marked vertices of degree 1 called external vertices, all other vertices being called internal vertices or, by abuse of terminology, vertices. An edge, sometimes called a propagator, is called external if it connects to an external vertex, and internal otherwise. The external vertices are simply places for the external edges to attach to, and are often left out of the discussion. An $n$-point diagram is a Feynman diagram with $n$ external edges. A Feynman diagram is called one-particle-irreducible if any multigraph obtained by removing one of the internal edges is connected.

\subsection{Diagrammatic expansion of the spectral action}
\label{sect:sa}
%

Viewing the spectral action as a classical action, and following the background field method, the vertices of degree $n$ in the corresponding quantum theory should correspond to $n\th$-order functional derivatives of the spectral action. However, in the paradigm of noncommutative geometry, a base manifold is absent, and functional derivatives do not exist in the local sense. Therefore, a more abstract notion of a vertex is needed. The brackets $\br{\cdot}$ that power the expansion of the spectral action in Theorems \ref{thm:main thm} and \ref{thm:asymptotic expansion} are by construction cyclic and multilinear extensions of the derivatives of the spectral action, and as such provide an appropriate notion of \textit{noncommutative vertices}.
%
%
We define a noncommutative vertex with $V_1,\ldots,V_n\in\cup_K H_K$ on the external edges by
\begin{align}
\raisebox{-43pt}{\scalebox{0.55}{
\begin{tikzpicture}[thick]
	\draw[edge] (0,2) to (2,2);
	\draw[edge] (1,0) to (2,2);
	\draw[edge] (1,4) to (2,2);
	\draw[edge] (3,4) to (2,2);
	\draw[edge] (4,2) to (2,2);
	\ncvertex{2,2}
	\draw[line width=2pt, line cap=round, dash pattern=on 0pt off 3\pgflinewidth] (2,0) arc (-90:0:1.5cm);
	\node at (-0.5,2) {\huge $V_1$};
	\node at (0.8,4.4) {\huge $V_2$};
	\node at (3.2,4.4) {\huge $V_3$};
	\node at (4.5,2) {\huge $V_4$};
	\node at (0.7,-0.4) {\huge $V_n$};
\end{tikzpicture}}}
\quad
:=\quad\br{V_1,\ldots,V_n}.
  \label{eq:bracket}
\end{align}

In contrast to a normal vertex of a Feynman diagram, a noncommutative vertex is decorated with a  cyclic order on the edges incident to it. By convention, the edges are attached clockwise with respect to this cyclic order. 
%
As such, with perturbations $V_1,\ldots,V_n$ decorating the external edges, the diagram \eqref{eq:bracket} reflects the cyclicity of the bracket: $\br{V_1,\ldots,V_n}=\br{V_n,V_1,\ldots,V_{n-1}}$, the first property of Lemma \ref{cycl bracket}. 
In order to diagramatically represent the second property of Lemma \ref{cycl bracket} as well, we introduce the following notation. Wherever a gauge edge meets a noncommutative vertex we can insert a dashed line decorated with an element $a\in\A$ before or after the gauge edge, with the following meaning:
\begin{align*}
\raisebox{-10pt}{\scalebox{0.55}{
\begin{tikzpicture}[thick]
	\draw (0,0) arc (110:70:4cm);
	\draw (0,0.025) arc (110:70:4cm);
	\draw (0,0.05) arc (110:70:4cm);
	\draw (0,0.075) arc (110:70:4cm);
	\draw[edge] (2.65,0) to (2.65+0.5,2);
	\draw[streepjes] (2.65,0) to (2.3,2);
	\node at (2.2,2.5) {\huge $a$};
	\node at (3.3,2.5) {\huge $V$};
\end{tikzpicture}}}
\quad
\raisebox{10pt}{
:=
}
\raisebox{-10pt}{\scalebox{0.55}{
\begin{tikzpicture}[thick]
	\draw (0,0) arc (110:70:4cm);
	\draw (0,0.025) arc (110:70:4cm);
	\draw (0,0.05) arc (110:70:4cm);
	\draw (0,0.075) arc (110:70:4cm);
	\draw[edge] (2.65,0) to (2.65+0.5,2);
	\node at (3.3,2.5) {\huge $aV$};
\end{tikzpicture}}}
\quad
\raisebox{10pt}{
,
}
\qquad\qquad
\raisebox{-10pt}{\scalebox{0.55}{
\begin{tikzpicture}[thick]
	\draw (0,0) arc (110:70:4cm);
	\draw (0,0.025) arc (110:70:4cm);
	\draw (0,0.05) arc (110:70:4cm);
	\draw (0,0.075) arc (110:70:4cm);
	\draw[edge] (0,0) to (-0.5,2);
	\draw[streepjes] (0,0) to (0.35,2);
	\node at (0.45,2.5) {\huge $a$};
	\node at (-0.65,2.5) {\huge $V$};
\end{tikzpicture}}}
\quad
\raisebox{10pt}{
:=
}
\quad
\raisebox{-10pt}{
\scalebox{0.55}{\begin{tikzpicture}[thick]
	\draw (0,0) arc (110:70:4cm);
	\draw (0,0.025) arc (110:70:4cm);
	\draw (0,0.05) arc (110:70:4cm);
	\draw (0,0.075) arc (110:70:4cm);
	\draw[edge] (0,0) to (-0.5,2);
	\node at (-0.65,2.5) {\huge $Va$};\end{tikzpicture}}}.
\end{align*}
With this notation, the equation
\begin{align}\label{eq:classical Ward}
\br{a V_1,\ldots,V_n}-\br{V_1,\ldots,V_n a}
        =\br{V_1,\ldots,V_n,[D,a]},
\end{align}
is represented as
\begin{align}
\label{eq:ward}
\raisebox{-10pt}{\scalebox{0.55}{
\begin{tikzpicture}[thick]
	\draw (0,0) arc (110:70:4cm);
	\draw (0,0.025) arc (110:70:4cm);
	\draw (0,0.05) arc (110:70:4cm);
	\draw (0,0.075) arc (110:70:4cm);
	\draw[streepjes] (2.65,0) to (2.3,2);
	\node at (2.2,2.5) {\huge $a$};
\end{tikzpicture}}}
\quad\,\,\,
-
\,\,\quad
\raisebox{-10pt}{\scalebox{0.55}{
\begin{tikzpicture}[thick]
	\draw (0,0) arc (110:70:4cm);
	\draw (0,0.025) arc (110:70:4cm);
	\draw (0,0.05) arc (110:70:4cm);
	\draw (0,0.075) arc (110:70:4cm);
	\draw[streepjes] (0,0) to (0.35,2);
	\node at (0.45,2.5) {\huge $a$};
\end{tikzpicture}}}
\quad
=
\quad
\raisebox{-10pt}{\scalebox{0.55}{
\begin{tikzpicture}[thick]
	\draw (0,0) arc (110:70:4cm);
	\draw (0,0.025) arc (110:70:4cm);
	\draw (0,0.05) arc (110:70:4cm);
	\draw (0,0.075) arc (110:70:4cm);
	\draw[edge] (1.325,0.3) to (1.325,2);
	\node at (1.325,2.5) {\huge $[D,a]$};
\end{tikzpicture}}}
\quad
,
\end{align}
and is as such referred to as the \textit{Ward identity}.
The examples in the next subsection will serve to explain the power of this diagrammatic notation.
\subsubsection{Expansions for arbitrary brackets}
We recall that, whenever $V\in \mB(\H)$ and $f\in\E_s^{1/2}$, we have
	\begin{align}\label{eq:spectraction brackets}
		\Tr(f(D+V)-f(D))=\sum_{n=1}^\infty\frac{1}{n}\br{V,\ldots,V}.
	\end{align}
We can generalize the spectral action by taking the right-hand side of \eqref{eq:spectraction brackets} as our starting point, and replacing the bracket $\br{\cdot}$ by a more abstract version, denoted $\brr{\cdot}$. The following proposition shows that the algebraic results of Sections \ref{sct:CC and UF underlying the SA} and \ref{sct:main thm} continue to hold in this general situation.
\begin{prop}\label{prop:expansions for arbitrary brackets}
	Let $\brr{\cdot}$ denote a collection of functions $\mB(\H)^{\times n}\to\R$, $n\in\N$, satisfing
	\begin{enumerate}[label=\textnormal{(\Roman*)}]
	\item $\brr{V_1,\ldots,V_n}~=~\brr{V_n,V_1,\ldots,V_{n-1}},$\label{cyclicity general}
	\item $\brr{aV_1,\ldots,V_j,\ldots,V_n}-\brr{V_1,\ldots,V_na}~=~\brr{V_1,\ldots,V_n,[D,a]}$\label{commutation general}
\end{enumerate}		
	 (cf. Lemma \ref{cycl bracket}). Define 
	\begin{align*} 
	 \boldsymbol\phi_n(a_0,\ldots,a_n):=~\brr{a_0[D,a_1],[D,a_2],\ldots,[D,a_n]},\qquad
	 \boldsymbol\psi_n:=\boldsymbol\phi_n-\tfrac{1}{2}B_0\boldsymbol\phi_{n+1},
	\end{align*}	 
	and $\tilde{\boldsymbol{\psi}}_{2k-1}:=(-1)^{k-1}\tfrac{(k-1)!}{(2k-1)!}\boldsymbol\psi_{2k-1}$. Then $(\tilde{\boldsymbol{\psi}}_1,\tilde{\boldsymbol{\psi}}_3,\ldots)$ and $(\boldsymbol\phi_2,\boldsymbol\phi_4,\ldots)$ are $(b,B)$-cocycles. Moreover, for $A\in\Omega^1(\A)$, $V=\pi_D(A)$, we asymptotically have
	\begin{align*}
		\sum_{n=1}^\infty\frac{1}{n}\!\brr{V,\ldots,V}~\sim\,\sum_{k=1}^\infty\left(\int_{\boldsymbol\psi_{2k-1}}\cs_{2k-1}(A)+\frac{1}{2k}\int_{\boldsymbol\phi_{2k}}F^k\right),
	\end{align*}
	in the sense that, for every $K\in\N$, there exist forms $\omega_l\in\Omega^l(\A)$ for $l=K+1,\ldots,2K+1$ such that
	\begin{align*}
		\sum_{n=1}^K\frac{1}{n}\!\brr{V,\ldots,V}&-\sum_{k=1}^K\left(\int_{\boldsymbol\psi_{2k-1}}\cs_{2k-1}(A)+\frac{1}{2k}\int_{\boldsymbol\phi_{2k}}F^k\right)
		=\sum_{l=K+1}^{2K+1}\int_{\boldsymbol\phi_l}\omega_l.
	\end{align*}
	
\end{prop}
\begin{proof}
	The first statement follows by following step by step the arguments in \textsection\ref{sct:Cyclic cocycles associated to multiple operator integrals}, namely the proofs of Lemmas \ref{lem:b} and \ref{lem:c} and Proposition \ref{prop:bB}. The second statement follows by carefully walking through the proofs of Proposition \ref{prop:k} and Theorem \ref{thm:asymptotic expansion}.
\end{proof}

The above proposition realizes explicitly that any bracket satisfying \ref{cyclicity} and \ref{commutation} defines cyclic cocycles that, when evaluated at the respective Chern--Simons and Yang--Mills forms, give an asymptotic expansion of the spectral action as in Theorem \ref{thm:asymptotic expansion}. Because the properties \ref{cyclicity} and \ref{commutation} can be expressed diagrammatically, we conclude that there exists diagrammatic proofs of the algebraic results of Sections \ref{sct:CC and UF underlying the SA} and \ref{sct:main thm}.

To illustrate, let us give the relevant lower order computations.
The cyclic cocycles are expressed in terms of diagrams as
\begin{align}
  \label{eq:bracket-cochain}
  \int_{\phi_n} a^0 da^1 \cdots da^n 
  &=
  \raisebox{-41pt}{\scalebox{0.55}{
\begin{tikzpicture}[thick]
	\draw[edge] (0,2) to (2,2);
	\draw[edge] (1,0) to (2,2);
	\draw[edge] (1,4) to (2,2);
	\draw[edge] (3,4) to (2,2);
	\draw[edge] (4,2) to (2,2);
	\ncvertex{2,2}
	\draw[line width=2pt, line cap=round, dash pattern=on 0pt off 3\pgflinewidth] (2.5,-0.5) arc (-85:-15:2cm);
	\node at (-1.4,2) {\huge $a^0[D,a^1]$};
	\node at (0.5,4.4) {\huge $[D,a^2]$};
	\node at (3.5,4.4) {\huge $[D,a^3]$};
	\node at (5.1,2) {\huge $[D,a^4]$};
	\node at (0.7,-0.4) {\huge $[D,a^n]$};
\end{tikzpicture}}}.
\end{align}

\noindent For one external edge we find, writing $A=\sum_j a_jdb_j$ and suppressing summation over $j$,
\begin{align}
\br{V} =   \br{a_j [D,b_j]} &= 
\,\,
\raisebox{-4pt}{\scalebox{.55}{
\begin{tikzpicture}[thick]
	\draw[edge] (1,0) to (3,0);
	\ncvertex{3,0}
	\node at (0.5,0.6) {\huge $a_j[D,b_j]$};
\end{tikzpicture}}}
\quad
= 
\int_{\phi_1} A.
\end{align}
For two external edges, we apply the Ward identity \eqref{eq:ward} and derive
\begin{align*}
  \br{V,V}&=
\quad
\raisebox{-5pt}{\scalebox{0.55}{
\begin{tikzpicture}[thick]
	\draw[edge] (1,0) to (3,0);
	\draw[streepjes] (3.45,0) to (5,2);
	\draw[edge] (3,0) to (5,0);
	\ncvertex{3,0}
	\node at (-0.5,0) {\huge $a_j[D,b_j]$};
	\node at (6,0) {\huge $[D,b_{j'}]$};
	\node at (5.6,2.1) {\huge $a_{j'}$};
\end{tikzpicture}}}  
  \\[1mm]
  &=
  \quad
\raisebox{-5pt}{\scalebox{0.55}{
\begin{tikzpicture}[thick]
	\draw[edge] (1,0) to (3,0);
	\draw[streepjes] (2.55,0) to (1,2);
	\draw[edge] (3,0) to (5,0);
	\ncvertex{3,0}
	\node at (-0.5,0) {\huge $a_j[D,b_j]$};
	\node at (6,0) {\huge $[D,b_{j'}]$};
	\node at (0.4,2.1) {\huge $a_{j'}$};
\end{tikzpicture}}}
\quad
+
\quad
\raisebox{-5pt}{\scalebox{0.55}{
\begin{tikzpicture}[thick]
	\draw[edge] (1,0) to (3,0);
	\draw[edge] (3,0) to (3,1.6);
	\draw[edge] (3,0) to (5,0);
	\ncvertex{3,0}
	\node at (-0.5,0) {\huge $a_j[D,b_j]$};
	\node at (6,0) {\huge $[D,b_{j'}]$};
	\node at (3,2.1) {\huge $[D,a_{j'}]$};
\end{tikzpicture}}}
\\[4mm]
  &= \int_{\phi_2} A^2 
  + \int_{\phi_3} A d A.
\end{align*}

\subsubsection{The propagator}           
An important part of the quantization process introduced here is to find a mathematical formulation for the propagator. In other words, we need to introduce more general diagrams than the one-vertex diagram in \eqref{eq:bracket}, and assign each an amplitude. As usual in quantum field theory, the amplitudes depend on a cutoff $N$ and are possibly divergent as $N\to\infty$. 

What we will call a \textit{noncommutative Feynman diagram} (or, for brevity, a diagram) is a Feynman diagram in which every internal vertex $v$ is decorated with a cyclic order on the edges incident to $v$. These decorated vertices are what we call the noncommutative vertices, and are denoted as in \eqref{eq:bracket}. 
The edges of a diagram are always drawn as wavy lines. They are sometimes called gauge edges to distinguish them from any dashed lines in the diagram, which do not represent physical particles, but are simply notation. The \textit{loop order} is defined to be $L:=1-V+E$, where $V$ is the amount of (noncommutative) vertices and $E$ is the amount of (gauge) edges. We also say the noncommutative Feynman diagram is $L$-loop, e.g., the noncommutative Feynman diagram in \eqref{eq:bracket} is zero-loop. When the respective multigraph is planar, $L$ corresponds to the number of internal faces. Following physics terminology, these faces are referred to as \textit{loops}.
As usual for Feynman diagrams, the external edges are marked, say by the numbers $1,\ldots,n$.


Note that, by our definition, a noncommutative Feynman diagram is almost the same as a ribbon graph, the sole difference being that ribbons are sensitive to twisting, whereas our edges are not.

Each nontrivial noncommutative Feynman diagram will be assigned an \textit{amplitude}, as follows. Here \textit{nontrivial} means that every connected component contains at least one vertex with nonzero degree.

\begin{defi}\label{def:Propagator}
Let $N\in\N$ and let $f\in C^\infty$ satisfy $(f')^{[1]}(\lambda_i,\lambda_j)>0$ for $i,j\leq N$. Given a nontrivial $n$-point noncommutative Feynman diagram $G$ with external vertices marked by $1,\ldots,n$, its \textbf{amplitude} at level $N\in\N$ on the gauge fields $V_1,\ldots,V_n\in\cup_K H_K$ is denoted $\Gamma_N^G(V_1,\ldots,V_n)$, and is defined recursively as follows. When $G$ has precisely one vertex and the markings $1,\ldots,n$ respect its cyclic order, we set $\Gamma_N^G(V_1,\ldots,V_n):=\br{V_1,\ldots,V_n}$. Suppose the amplitudes of diagrams $G_1$ and $G_2$ with external edges $1,\ldots,n$ and $n+1,\ldots,m$ are defined. Then to the disjoint union $G$ of the diagrams we assign the amplitude
\begin{align*}
	\Gamma_N^{G}(V_1,\ldots,V_m):=\Gamma_N^{G_1}(V_1,\ldots,V_n)\Gamma_N^{G_2}(V_{n+1},\ldots,V_m).
\end{align*}
 Suppose the amplitude of a diagram $G$ is defined. Then, for any two distinct numbers $i,j\in\{1,\ldots,n\}$, let $G'$ be the diagram obtained from $G$ by connecting the two external edges $i$ and $j$ by a gauge edge (a propagator). We then define the amplitude of $G'$ as
\begin{align*}
	\Gamma_N^{G'}(V_1,\ldots,\widehat{V_i},\ldots,\widehat{V_j},\ldots,V_{n}):=-\frac{\int_{H_N}\Gamma_N^G(V_1,\ldots,\overset{i}{Q},\ldots,\overset{j}{Q},\ldots,V_{n})e^{-\tfrac12\br{Q,Q}}dQ}{\int_{H_N}e^{-\tfrac12\br{Q,Q}}dQ}.
\end{align*}
\end{defi}

Well-definedness is a straightforward consequence of Fubini's theorem. Note that, in general, $\Gamma^G_N$ is not cyclic in its arguments, as was the case in \eqref{eq:bracket}.

\begin{figure}[h!]
\hspace{30pt}
\scalebox{0.55}{
\hspace{24pt}
\begin{subfigure}[t]{0.4\textwidth}
	\image{0.75}{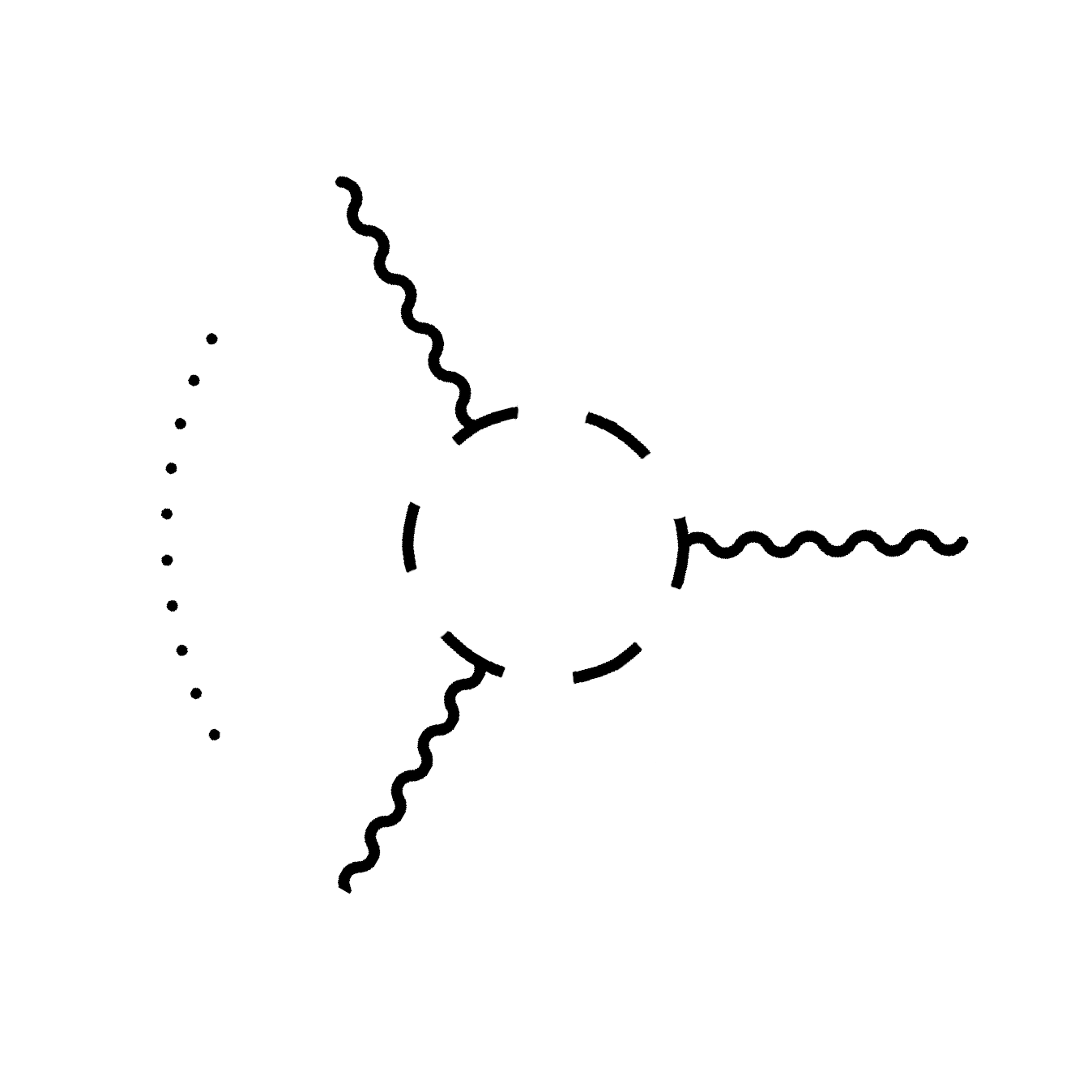}
	\put(-104,116){\Large $V_n$}
	\put(-104,11){\Large $V_1$}
	\put(-10,63){\Large $Q$}
	\put(-74,62){\Large $G_1$}
\end{subfigure}
\begin{subfigure}[t]{0.4\textwidth}
	\image{0.75}{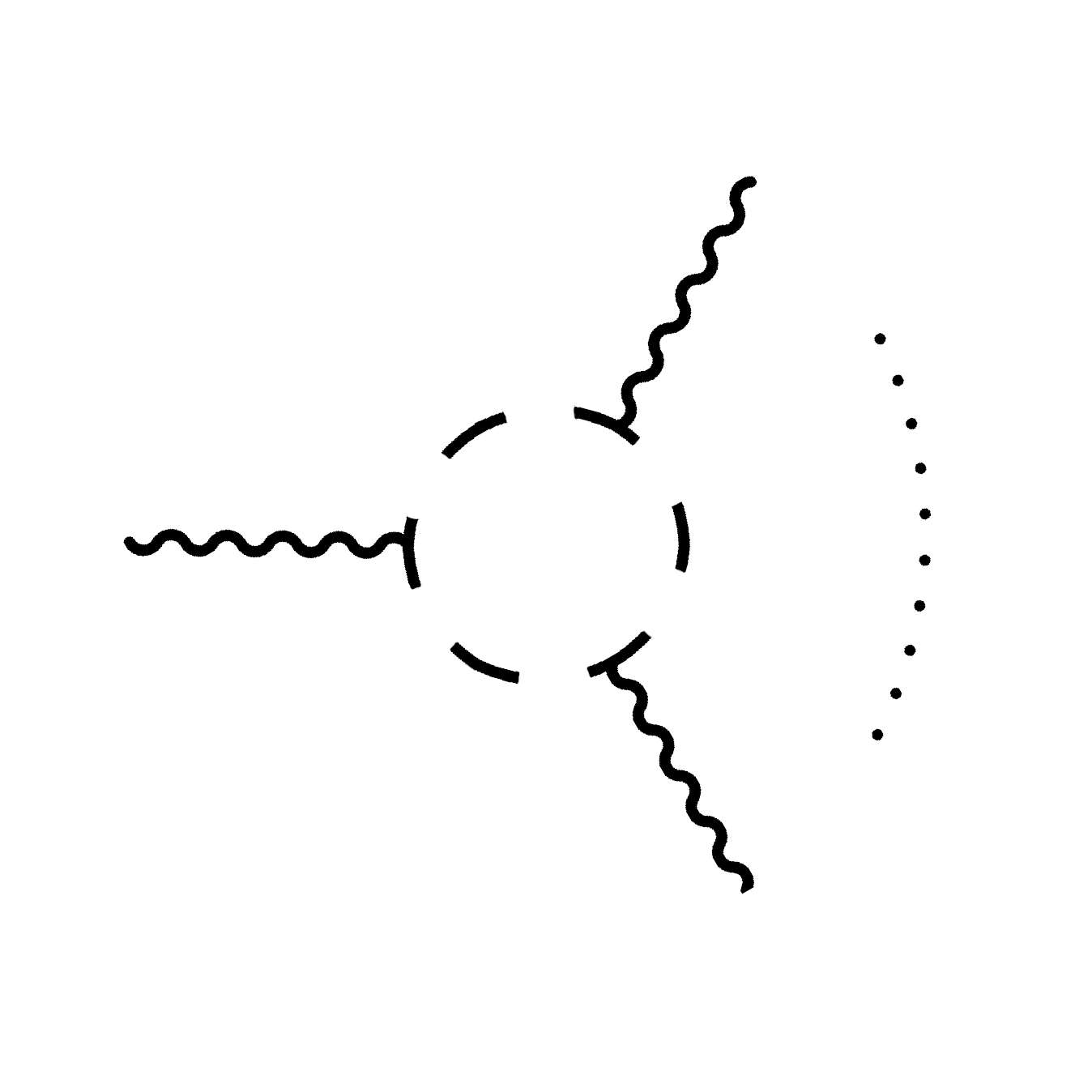}
	\put(-37,114){\Large $V_{n+1}$}
	\put(-37,12){\Large $V_m$}
	\put(-132,63){\Large $Q$}
	\put(-74,62){\Large $G_2$}
\end{subfigure}
\begin{subfigure}[t]{\textwidth}
	\image{0.55}{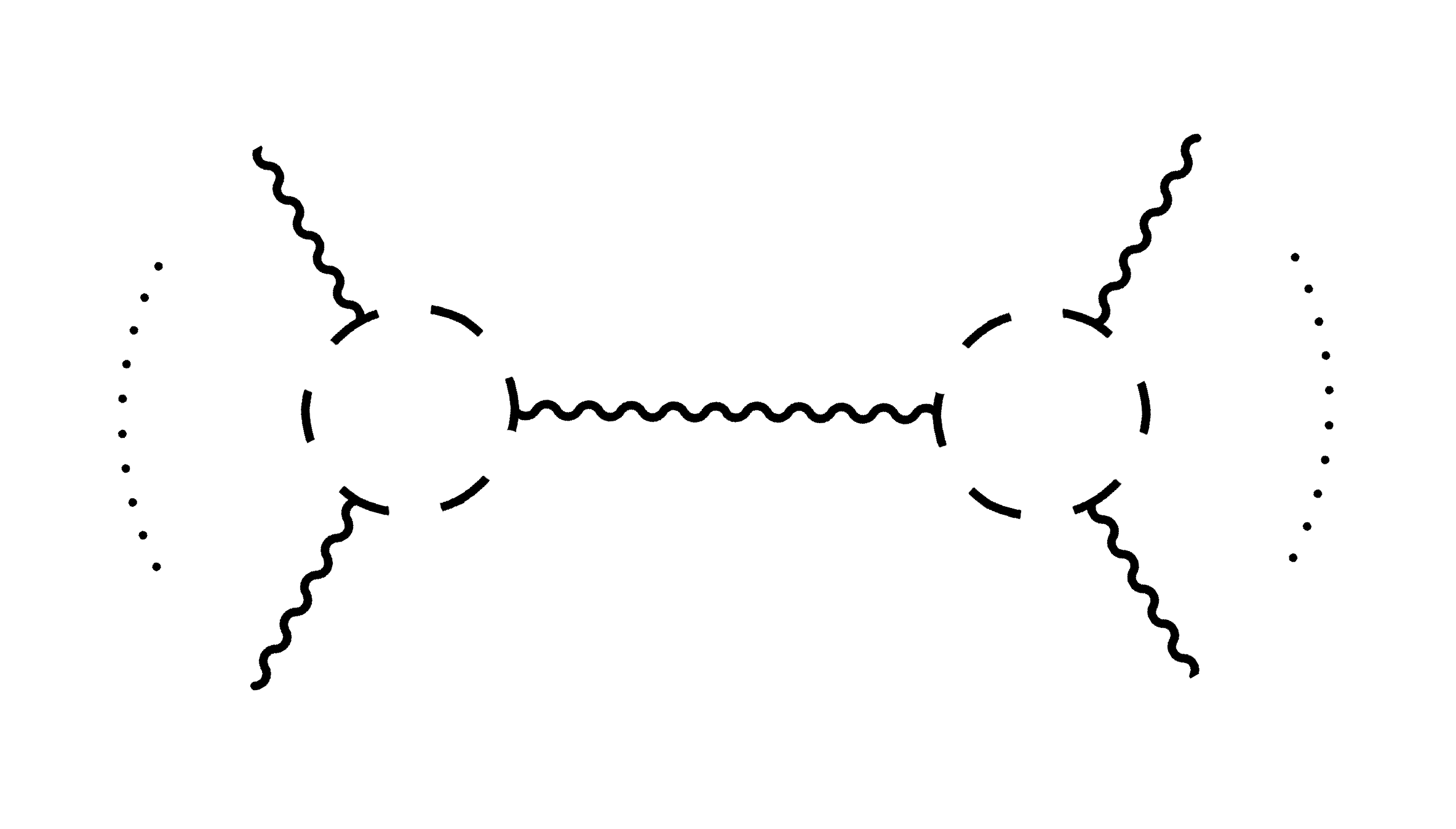}
	\put(-214,118){\Large $V_n$}
	\put(-215,11){\Large $V_1$}
	\put(-180,65){\Large $G_1$}
	\put(-37,117){\Large $V_{n+1}$}
	\put(-37,14){\Large $V_m$}
	\put(-76,65){\Large $G_2$}
\end{subfigure}
}
\caption{Constructing the propagator.}
\end{figure}
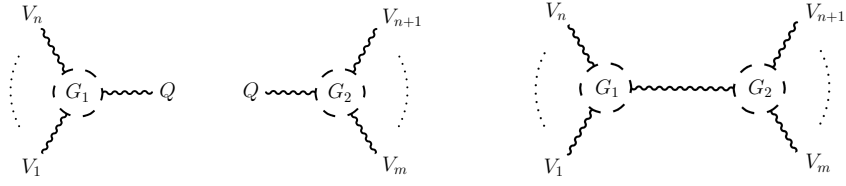

The assumption that $(f')^{[1]}(\lambda_i,\lambda_j)>0$ for $i,j\leq N$ can be accomplished by allowing $f$ to be unbounded, and replacing the spectral action
$$\Tr(f(D))$$
with the regularized version
$$\Tr(f_N(D))$$
where $f_N:=f\Phi_N$ for a sequence of bump functions $\Phi_N$ ($N\in\N$) that are 1 on $\{\lambda_k:~k\leq N\}$. As quantization takes place on the finite level (for a finite $N$), it is natural to also regularize the classical action before we quantize.
 Because we can now easily require
 	$$(f_N')^{[1]}(\lambda_k,\lambda_l)=(f')^{[1]}(\lambda_k,\lambda_l)>0,$$ for all $k,l\leq N$, Definition \ref{def:Propagator} makes sense and can be studied by Gaussian integration as in  \cite[Section 2]{BIZ80}. 

\subsection{Loop corrections to the spectral action}
 	
	To obtain the propagator, we have chosen the approach of random noncommutative geometries (as done in \cite{AK19,KP20}, see \cite{BG16,GS20} for computer simulations) in the sense that the integrated space in Definition \ref{def:Propagator} is the whole of $H_N$. Other approaches are conceivable by replacing $H_N$ by a subspace of gauge fields particular to the gauge theory under consideration, but this should also take into account gauge fixing, and will quickly become very involved.

 	In our case, the propagator becomes quite simple, and can be explicitly expressed by the following result.

\begin{lem}\label{lem:propagator}
Let $f\in C^\infty$ satisfy $(f')^{[1]}(\lambda_k,\lambda_l)>0$ for $k,l\leq N$. For $k,l,m,n\in\{1,\ldots,N\}$, we have
$$
\frac{ \int_{H_N} Q_{kl}  Q_{mn} e^{-\tfrac12\br{Q,Q}} dQ} {\int_{H_N} e^{-\tfrac12\br{Q,Q} } dQ} = 
\delta_{kn}\delta_{lm} G_{kl},
$$
in terms of $G_{kl} := \frac{1}{(f')^{[1]}(\lambda_k, \lambda_l)}$.
\end{lem}
\begin{proof}
By \eqref{eq:SA divdiff} we have the finite sum
\begin{align*}
\br{Q,Q}=\sum_{k,l} (f')^{[1]}(\lambda_k,\lambda_l)\left((\Re (Q_{kl}))^2+(\Im(Q_{kl}))^2\right),
\end{align*}
for all $Q\in H_N$. Moreover, we have
\begin{align*}
\int_{H_N} Q_{kl}  Q_{mn} e^{-\frac 12 \br{Q,Q}} dQ =&  \int_{H_N} (\Re(Q_{kl})\Re(Q_{mn})-\Im(Q_{kl})\Im(Q_{mn}))e^{-\frac 12 \br{Q,Q}} dQ\\
&+i\int_{H_N}(\Re(Q_{kl})\Im(Q_{mn}) + \Im(Q_{kl})\Re(Q_{mn}))e^{-\frac 12 \br{Q,Q}} dQ.
\end{align*}
The second integral on the right-hand side vanishes because its integrand is an odd function in at least one of the coordinates of $H_N$. The same holds for the first integral whenever $\{k,l\}\neq\{m,n\}$. Otherwise, we use that $\Re(Q_{lk})=\Re(Q_{kl})$ and $\Im(Q_{lk})=-\Im(Q_{kl})$ and see that the two terms of the first integral cancel when $k=m$ and $l=n$. When $k=n\neq l=m$, we instead find that these terms give the same result when integrated. By using symmetry of $(f')^{[1]}$ and integrating out all trivial coordinates, we obtain
\begin{align*}
\frac{\int_{H_N} Q_{kl}  Q_{mn} e^{-\frac 12 \br{Q,Q}} dQ}{\int_{H_N} e^{-\frac 12 \br{Q,Q}} dQ} =&  \delta_{kn}\delta_{lm}\frac{2\int_\R (\Re(Q_{kl}))^2 e^{-(f')^{[1]}(\lambda_k,\lambda_l)(\Re(Q_{kl}))^2}d\Re(Q_{kl})}{\int_\R e^{-(f')^{[1]}(\lambda_k,\lambda_l)(\Re(Q_{kl}))^2}d\Re(Q_{kl})},
\end{align*}
a Gaussian integral that gives the $G_{kl}$ required by the lemma. When $k=l=n=m$, the result follows similarly.
\end{proof}

The above lemma allows us to leave out all integrals from the subsequent computations. In place of those integrals, we use the following notation.
\begin{defi}\label{def:propagator}
We  define, with slight abuse of notation,
\begin{align*}
	\wick{\c Q_{kl} \c Q_{mn} }:=\delta_{kn}\delta_{lm} G_{kl},
\end{align*}
and refer to $G_{kl}$ as the \textit{propagator}.
\end{defi}

As an example and to fix terminology, we will now compute the amplitudes of the three most basic one-loop diagrams with two external edges.
These are given in Figure \ref{fig:1loop-2pt}.
Using Lemma \ref{lem:propagator} and Definition \ref{def:propagator}, we find the amplitude for the first diagram to be
\begin{align}
  \label{eq:2ptA}
 \hspace{-3pt}\raisebox{-10pt}{\scalebox{0.55}{
\begin{tikzpicture}[thick]
	\draw[edge] (0.7,0) to (1.8,0);
	\draw[edge] (1.8,0) to (2.9,0);
	\draw[edge] (3.1,0) to (4,0);
	\draw[edge] (1.8,0) arc (0:180:-0.6cm);
	\draw[edge] (3,0) arc (180:0:0.6cm);
	\draw[edge] (4,0) to (5.3,0);
	\ncvertex{2,0}
	\ncvertex{4,0}
	\node at (0.3,0) {\huge $V_1$};
	\node at (5.7,0) {\huge $V_2$};
\end{tikzpicture}}}
 &= \sum_{\begin{smallmatrix} i,j,k,l, \\ m,n\leq N\end{smallmatrix}} (f')^{[2]}(\lambda_i,\lambda_j,\lambda_k)(V_1)_{ij} \wick{ \c1 Q_{jk} \c2 Q_{ki} (f')^{[2]}(\lambda_l,\lambda_m,\lambda_n)(V_2)_{lm} \c1 Q_{mn} \c2 Q_{nl} }  \nonumber      \nonumber \\
 &  = \sum_{i,k\leq N} (f')^{[2]}(\lambda_i, \lambda_i,\lambda_k)(f')^{[2]}(\lambda_i, \lambda_k,\lambda_k)(V_1)_{ii} (V_2)_{kk} (G_{ik})^2  .
  \end{align}
As $V_1$ and $V_2$ are assumed of finite rank, the above expression converges as $N\to\infty$. To see this explicitly, let $K$ be such that $V_1,V_2\in H_K$, and let $G$ be the diagram on the left-hand side of \eqref{eq:2ptA}. We then obtain
\begin{align}\label{eq:irrelevant diagram}
	\lim_{N\to\infty}\Gamma^G_N(V_1,V_2)=\sum_{i,k\leq K}(f')^{[2]}(\lambda_i,\lambda_i,\lambda_k)(f')^{[2]}(\lambda_i,\lambda_k,\lambda_k)(V_1)_{ii}(V_2)_{kk}(G_{ik})^2,
\end{align}
a finite number. In general we can say that if all summed indices of an amplitude occur in a matrix element of any of the perturbations (e.g., $(V_1)_{ii}$ and $(V_2)_{kk}$) then the amplitude remains finite even when the size $N$ of the random matrices $Q$ is sent to $\infty$. In physics terminology, the first diagram in Figure \ref{fig:1loop-2pt} is \textit{irrelevant}, and can be disregarded for renormalization purposes.

We then turn to the second diagram in Figure \ref{fig:1loop-2pt}, and compute
\begin{align}
\raisebox{-8pt}{\scalebox{0.55}{
\begin{tikzpicture}[thick]
	\draw[edge] (1,0) to (2.1,0);
	\draw[edge] (2.1,0) to[out=90,in=90] (3.9,0);
	\draw[edge] (2.1,0) to[out=-90,in=-90] (3.9,0);
	\draw[edge] (3.9,0) to (5,0);
	\ncvertex{2,0}
	\ncvertex{4,0}
	\node at (0.5,0) {\huge $V_1$};
	\node at (5.5,0) {\huge $V_2$};
\end{tikzpicture}}}
&= \sum_{\begin{smallmatrix} i,j,k,l, \\ m,n\leq N\end{smallmatrix}} (f')^{[2]}(\lambda_i,\lambda_j,\lambda_k) (V_1)_{ij} \wick{ \c1 Q_{jk} \c2 Q_{ki} (f')^{[2]}(\lambda_l,\lambda_m,\lambda_n) (V_2)_{lm} \c2 Q_{mn} \c1 Q_{nl} } \nonumber \\
 &= \sum_{i,j,k\leq N} (f')^{[2]}(\lambda_i, \lambda_j,\lambda_k)^2 (V_1)_{ij} (V_2)_{ji} G_{ik}G_{kj} .
\label{eq:vertexcontr}
\end{align}
This diagram is planar, and the indices $i,j,k$ correspond to regions in the plane, assuming the external edges are regarded to stretch out to infinity. The index $k$ corresponds to the region within the loop, and is called a \textit{running loop index}. As the index $k$ is not restricted by $V_1$ and $V_2$ as in \eqref{eq:2ptA}, we find that in general the amplitude \eqref{eq:vertexcontr} diverges as $N \to \infty$.
In physical terms, this is a \textit{relevant} diagram.

The amplitude of the final diagram becomes
\begin{align}
 \raisebox{-15pt}{\scalebox{0.55}{ 
\begin{tikzpicture}[thick]
	\draw[edge] (0,-1) to (2,0);
	\draw[edge] (4,-1) to (2,0);
	\draw[edge] (1.9,0) arc (-90:270:0.8cm);
	\ncvertex{2,0}
	\node at (-.4,-1) {\huge $V_1$};
	\node at (4.4,-1) {\huge $V_2$};
\end{tikzpicture}}}
\quad&= -\sum_{ i,j,k,l\leq N} (f')^{[3]}(\lambda_i,\lambda_j,\lambda_k,\lambda_l)(V_1)_{ij} \wick{ \c Q_{jk} \c Q_{kl} } (V_2)_{li}   \nonumber \\
 &= -\sum_{i,j,k\leq N} (f')^{[3]}(\lambda_i, \lambda_j,\lambda_j,\lambda_k)(V_1)_{ij} (V_2)_{ji} G_{jk} .
 \label{eq:vertexcontr2}
\end{align}
Again, this amplitude contains a running loop index and is therefore potentially divergent in the limit $N \to \infty$. 

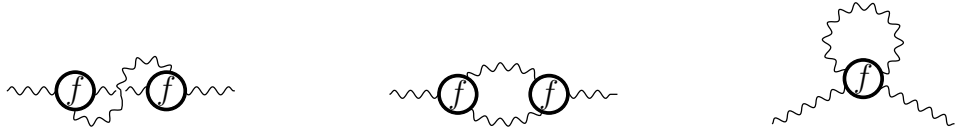
\begin{figure}
\hspace{24pt}
    \begin{tabular}{p{.3\linewidth}p{.3\linewidth}p{.3\linewidth}}
 \scalebox{.6}{
    \begin{tikzpicture}[thick]
	\draw[edge] (0.5,0) to (1.8,0);
	\draw[edge] (1.8,0) to (2.9,0);
	\draw[edge] (3.1,0) to (4.2,0);
	\draw[edge] (1.8,0) arc (0:180:-0.6cm);
	\draw[edge] (3,0) arc (180:0:0.6cm);
	\draw[edge] (4.2,0) to (5.5,0);
	\ncvertex{2,0}
	\ncvertex{4,0}
\end{tikzpicture}}
 &\scalebox{.6}{
    \begin{tikzpicture}[thick]
	\draw[edge] (0.5,0) to (2.1,0);
	\draw[edge] (2.1,0) to[out=90,in=90] (3.9,0);
	\draw[edge] (2.1,0) to[out=-90,in=-90] (3.9,0);
	\draw[edge] (3.9,0) to (5.5,0);
	\ncvertex{2,0}
	\ncvertex{4,0}
\end{tikzpicture}}
 &
 \scalebox{.6}{
 \begin{tikzpicture}[thick]
	\draw[edge] (0,-1) to (2,0);
	\draw[edge] (4,-1) to (2,0);
	\draw[edge] (1.9,0) arc (-90:270:0.8cm);
	\ncvertex{2,0}
\end{tikzpicture}}
 \end{tabular}
 
  \caption{Two-point diagrams with one loop. The first one is irrelevant, the second and third are relevant.}
  \label{fig:1loop-2pt}
\end{figure}

\subsubsection{One-loop counterterms to the spectral action}


Because we are interested in the behavior of the one-loop quantum effective spectral action as $N\to\infty$, we wish to consider only one-loop noncommutative Feynman diagrams whose amplitudes involve a running loop index. For example, the final two diagrams in Figure \ref{fig:1loop-2pt}, but not the first. 

As dictated by the background field method, in order to obtain a quantum effective action we should further restrict to one-particle-irreducible diagrams whose vertices have degree $\geq 3$.


Let us fix a one-loop one-particle-irreducible diagram $G$ in which all vertices have degree $\geq3$, and investigate whether the amplitude of $G$ contains a running loop index. Fix a noncommutative vertex $v$ in $G$. The vertex $v$ will have precisely two incident edges that belong to the loop of the diagram, and at least one external edge. Each index associated with $v$ is associated specifically with two incident edges of $v$. If one of these edges is external, the index will not run, because it will be fixed by the gauge field attached. A running index can only occur if the two incident loop edges of $v$ succeed one another, and the index is placed in between them. The latter of these two loop edges will attach to another noncommutative vertex, $w$, and the possibly running index will also be associated with the succeeding edge in $w$, which also has to be a loop edge if the index is to run. This process may continue throughout the loop until we end up at the original vertex $v$. By this argument, the amplitude of $G$ will contain a running loop index if and only if $G$ can be drawn in a plane with all noncommutative vertices oriented clockwise and all external edges extending outside the loop.

The wonderful conclusion is that the external edges of the relevant diagrams obtain a natural cyclic order. This presents us with a natural one-loop quantization of the bracket $\br{\cdot}$, and thus with a natural proposal for the one-loop quantization of the spectral action.

\begin{defi}\label{def:quantum effective SA}
	Let $N\in\N$ and let $f\in C^\infty$ satisfy $(f')^{[1]}(\lambda_i,\lambda_j)>0$ for $i,j\leq N$. We define 
	\begin{align*}
		\bbr{V_1,\ldots,V_n}:=\sum_G \Gamma_N^G(V_1,\ldots,V_n),
	\end{align*}	
	where the sum is over all planar one-loop one-particle-irreducible $n$-point noncommutative Feynman diagrams $G$ with clockwise vertices of degree $\geq3$ and external edges outside the loop and marked cyclically. The \textbf{one-loop quantum effective spectral action} is defined to be the formal series
	$$
\sum_{n=1}^\infty \frac{1}{n} \bbr{V,\ldots,V}.
     $$
\end{defi}
Directly from the definition of $\bbr{\cdot}$, we see that
\begin{align*}
	\bbr{V_2,\ldots,V_n,V_1}=\bbr{V_1,\ldots,V_n}.
\end{align*}
The next subsection serves to illustrate why an analogue of \eqref{eq:classical Ward} holds for $\bbr{\cdot}$ as well.

\subsubsection{Ward identity for the gauge propagator}
In addition to the Ward identity \eqref{eq:ward} for the noncommutative vertex, we claim that we also have the following Ward identity for the gauge edge:
\begin{align}
  \label{eq:ward-gauge}
\raisebox{-20pt}{\scalebox{0.55}{
\begin{tikzpicture}[thick]
	\draw (0,0) arc (-20:20:4cm);
	\draw (-0.025,0) arc (-20:20:4cm);
	\draw (-0.05,0) arc (-20:20:4cm);
	\draw (-0.075,0) arc (-20:20:4cm);
	\draw[edge] (0.15,1.4) to (2.95,1.4);
	\draw (3.1,0) arc (20:-20:-4cm);
	\draw (3.125,0) arc (20:-20:-4cm);
	\draw (3.15,0) arc (20:-20:-4cm);
	\draw (3.175,0) arc (20:-20:-4cm);
	\draw[streepjes] (2.9,1.4) to (2.3,2.9);
	\node at (2.1,3.2) {\huge $a$};
\end{tikzpicture}}}
\quad-\quad
\raisebox{-20pt}{\scalebox{0.55}{
\begin{tikzpicture}[thick]
	\draw (0,0) arc (-20:20:4cm);
	\draw (-0.025,0) arc (-20:20:4cm);
	\draw (-0.05,0) arc (-20:20:4cm);
	\draw (-0.075,0) arc (-20:20:4cm);
	\draw[edge] (0.15,1.4) to (2.95,1.4);
	\draw (3.1,0) arc (20:-20:-4cm);
	\draw (3.125,0) arc (20:-20:-4cm);
	\draw (3.15,0) arc (20:-20:-4cm);
	\draw (3.175,0) arc (20:-20:-4cm);
	\draw[streepjes] (0.2,1.4) to (0.75,2.9);
	\node at (0.9,3.2) {\huge $a$};
\end{tikzpicture}}}
\quad=\quad
\raisebox{-20pt}{\scalebox{0.55}{
\begin{tikzpicture}[thick]
	\draw (0,0) arc (-20:20:4cm);
	\draw (-0.025,0) arc (-20:20:4cm);
	\draw (-0.05,0) arc (-20:20:4cm);
	\draw (-0.075,0) arc (-20:20:4cm);
	\draw[edge] (0.15,1.4) to (1.5,1.4);
	\draw[edge] (1.5,1.4) to (2.95,1.4);
	\draw (3.1,0) arc (20:-20:-4cm);
	\draw (3.125,0) arc (20:-20:-4cm);
	\draw (3.15,0) arc (20:-20:-4cm);
	\draw (3.175,0) arc (20:-20:-4cm);
	\draw[edge] (1.5,1.4) to (1.5,2.8);
	\ncvertex{1.5,1.4}
	\node at (1.5,3.2) {\huge $[D,a]$};
\end{tikzpicture}}}
\end{align}
Indeed, the left-hand side yields terms
  \begin{align*}
    \sum_{m\leq N}\big(\wick{ \c Q_{ik} \c Q_{lm} a_{mn} }-   \wick{ a_{im} \c  Q_{mk}  \c Q_{ln}}\big)
    &=  \sum_{m\leq N}\big(G_{ik} \delta_{im} \delta_{kl} a_{mn} - G_{ln} \delta_{mn} \delta_{kl} a_{im}\big) \\
    &= ( G_{ik}- G_{nk} )\delta_{kl} a_{in},
   \end{align*}
   for arbitrary values of $i$, $k$, $l$, and $n$ determined by the rest of the diagram.
  The right-hand side, by the defining property of the divided difference, and because every internal edge adds a minus sign, yields the terms
   \begin{align*}
  &-\sum_{p,q,r\leq N}\wick{ \c Q_{ik} (f')^{[2]}(\lambda_p, \lambda_q, \lambda_r) \c Q_{pq}} [D,a]_{qr}\wick{ \c Q_{rp} \c  Q_{ln}} \\
    &\qquad= - \sum_{p,q,r\leq N}(f')^{[2]}(\lambda_p, \lambda_q, \lambda_r)(\lambda_q -\lambda_r)a_{qr}G_{ik} \delta_{iq} \delta_{kp} G_{rp} \delta_{rn} \delta_{pl}  \\
    &\qquad=   \left( (f')^{[1]}(\lambda_k, \lambda_n) -  (f')^{[1]}(\lambda_i, \lambda_k) \right)G_{ik} G_{nk}  \delta_{kl} a_{in}.
    \end{align*}
Because $G_{kl}=1/(f')^{[1]}(\lambda_k,\lambda_l)$ (see Lemma \ref{lem:propagator}) the two expressions coincide for every value of $i$, $k$, $l$, and $n$, thereby allowing us to apply the rule \eqref{eq:ward-gauge} whenever it comes up as part of a diagram. An example is below.

%

The Ward identity for the gauge propagator allows us to compute $\bbr{aV_1,\ldots,V_n}-\bbr{V_1,\ldots,V_na}$ diagrammatically.
For example, the contribution of the second diagram in Figure \ref{fig:1loop-2pt} to $\llangle a V_1, V_2\rrangle_N^{1L} - \llangle V_1, V_2 a \rrangle_N^{1L} $ is
\begin{align}\label{eq:example quantum Ward}
&\raisebox{-20pt}{\scalebox{.55}{
    \begin{tikzpicture}[thick]
	\draw[edge] (0.5,0) to (2.1,0);
	\draw[edge] (2.1,0) to[out=90,in=90] (3.9,0);
	\draw[edge] (2.1,0) to[out=-90,in=-90] (3.9,0);
	\draw[edge] (3.9,0) to (5.5,0);
	\draw[streepjes] (1.6,0) to (1,-1.5);
	\ncvertex{2,0}
	\ncvertex{4,0}
	\node at (0,0) {\huge $V_1$};
	\node at (6,0) {\huge $V_2$};
	\node at (0.9,-1.7) {\huge $a$};
\end{tikzpicture}}}
\quad
-
\raisebox{-20pt}{\scalebox{.55}{
    \begin{tikzpicture}[thick]
	\draw[edge] (0.5,0) to (2.1,0);
	\draw[edge] (2.1,0) to[out=90,in=90] (3.9,0);
	\draw[edge] (2.1,0) to[out=-90,in=-90] (3.9,0);
	\draw[edge] (3.9,0) to (5.5,0);
	\draw[streepjes] (4.4,0) to (5,-1.5);
	\ncvertex{2,0}
	\ncvertex{4,0}
	\node at (0,0) {\huge $V_1$};
	\node at (6,0) {\huge $V_2$};
	\node at (5.1,-1.7) {\huge $a$};
\end{tikzpicture}}}\\
&\raisebox{-20pt}{
\quad\vspace{20pt}\raisebox{-15pt}{
=
}
\raisebox{-35pt}{\scalebox{.55}{
    \begin{tikzpicture}[thick]
    {\color{red}
	\draw[edge] (0.5,0) to (2.1,0);
	\draw[edge] (2.1,0) to[out=90,in=90] (3.9,0);
	\draw[edge] (2.1,0) to[out=-90,in=-90] (3.9,0);
	\draw[edge] (3.9,0) to (5.5,0);
	\draw[edge] (2,0) to (1.2,-1.5);
	\ncvertex{2,0}
	\ncvertex{4,0}}
	\node at (0,0) {\huge $V_1$};
	\node at (6,0) {\huge $V_2$};
	\node at (1,-2) {\huge $[D,a]$};
\end{tikzpicture}}}
\,\,\raisebox{-15pt}{
+
}\,\,
\raisebox{-45pt}{\scalebox{.55}{
 \begin{tikzpicture}[thick]
 	{\color{blue}
	\draw[edge] (-.2,2.7) to (1,2.7);
	\draw[edge] (2,-0.1) to (2,1);
	\draw[edge] (2,1) to (1,2.7);
	\draw[edge] (1,2.7) to (3,2.7);
	\draw[edge] (3,2.7) to (2,1);
	\draw[edge] (3,2.7) to (4.2,2.7);
	\ncvertex{1,2.7}
	\ncvertex{2,1}
	\ncvertex{3,2.7}}
	\node at (-0.7,2.7) {\huge $V_1$};
	\node at (4.7,2.7) {\huge $V_2$};
	\node at (2,-0.5) {\huge $[D,a]$};
\end{tikzpicture}}}
\,\,\raisebox{-15pt}{
+
}\,\,
\raisebox{-35pt}{\scalebox{.55}{
    \begin{tikzpicture}[thick]
    {\color{red}
	\draw[edge] (0.5,0) to (2.1,0);
	\draw[edge] (2.1,0) to[out=90,in=90] (3.9,0);
	\draw[edge] (2.1,0) to[out=-90,in=-90] (3.9,0);
	\draw[edge] (3.9,0) to (5.5,0);
	\draw[edge] (4,0) to (4.8,-1.5);
	\ncvertex{2,0}
	\ncvertex{4,0}}
	\node at (0,0) {\huge $V_1$};
	\node at (6,0) {\huge $V_2$};
	\node at (5,-2) {\huge $[D,a]$};
\end{tikzpicture}}}
\!\!.
}\nonumber
\end{align}

The third two-point diagram in Figure \ref{fig:1loop-2pt} has two possible markings of its external lines. The respective contributions to $\llangle a V_1, V_2\rrangle_N^{1L} - \llangle V_1, V_2 a \rrangle_N^{1L} $ are
\begin{align*}
\raisebox{-20pt}{\scalebox{0.55}{
\begin{tikzpicture}[thick]
	\draw[edge] (0,-1) to (2,0);
	\draw[edge] (4,-1) to (2,0);
	\draw[edge] (1.9,0) arc (-90:270:0.8cm);
	\draw[streepjes] (1.6,-0.2) to (1.8,-1.8);
	\ncvertex{2,0}
	\node at (-0.4,-1.4) {\huge $V_1$};
	\node at (4.4,-1.4) {\huge $V_2$};
	\node at (2,-2) {\huge $a$};
\end{tikzpicture}}}
\quad 
- \quad
\raisebox{-20pt}{\scalebox{0.55}{
\begin{tikzpicture}[thick]
	\draw[edge] (0,-1) to (2,0);
	\draw[edge] (4,-1) to (2,0);
	\draw[edge] (1.9,0) arc (-90:270:0.8cm);
	\draw[streepjes] (2.4,-0.2) to (2.2,-1.8);
	\ncvertex{2,0}
	\node at (-0.4,-1.4) {\huge $V_1$};
	\node at (4.4,-1.4) {\huge $V_2$};
	\node at (2,-2) {\huge $a$};
\end{tikzpicture}}}
\quad 
= 
\quad
\raisebox{-20pt}{\scalebox{.55}{
\begin{tikzpicture}[thick]
	{\color{green}
	\draw[edge] (0,-1) to (2,0);
	\draw[edge] (4,-1) to (2,0);
	\draw[edge] (2,-1.6) to (2,0);
	\draw[edge] (1.9,0) arc (-90:270:0.8cm);
	\ncvertex{2,0}}
	\node at (-0.4,-1.4) {\huge $V_1$};
	\node at (4.4,-1.4) {\huge $V_2$};
	\node at (2,-2) {\huge $[D,a]$};
\end{tikzpicture}}}
\end{align*}
and
\begin{align*}
&\raisebox{-20pt}{\scalebox{0.55}{
\begin{tikzpicture}[thick]
	\draw[edge] (0,-1) to (2,0);
	\draw[edge] (4,-1) to (2,0);
	\draw[edge] (1.9,0) arc (-90:270:0.8cm);
	\draw[streepjes] (2.4,-0.2) to (4,0.5);
	\ncvertex{2,0}
	\node at (-0.4,-1.4) {\huge $V_2$};
	\node at (4.4,-1.4) {\huge $V_1$};
	\node at (4.4,0.6) {\huge $a$};
\end{tikzpicture}}}
\quad 
-
\quad
\raisebox{-20pt}{\scalebox{0.55}{
\begin{tikzpicture}[thick]
	\draw[edge] (0,-1) to (2,0);
	\draw[edge] (4,-1) to (2,0);
	\draw[edge] (1.9,0) arc (-90:270:0.8cm);
	\draw[streepjes] (1.6,-0.2) to (0,0.5);
	\ncvertex{2,0}
	\node at (-0.4,-1.4) {\huge $V_2$};
	\node at (4.4,-1.4) {\huge $V_1$};
	\node at (-0.4,0.6) {\huge $a$};
\end{tikzpicture}}}
\\[5pt]
&\quad
= 
\quad
\raisebox{-20pt}{\scalebox{.55}{
\begin{tikzpicture}[thick]
	{\color{green}
	\draw[edge] (0,-1) to (2,0);
	\draw[edge] (4,-1) to (2,0);
	\draw[edge] (2,0) to (4,0);
	\draw[edge] (1.9,0) arc (-90:270:0.8cm);
	\ncvertex{2,0}}
	\node at (-0.4,-1.4) {\huge $V_1$};
	\node at (5,0) {\huge $[D,a]$};
	\node at (4.4,-1.4) {\huge $V_2$};
\end{tikzpicture}}}
\quad\,
+
\,\,
\raisebox{-25pt}{\scalebox{.55}{
\begin{tikzpicture}[thick]
	{\color{red}
	\draw[edge] (0,-1) to (2,0);
	\draw[edge] (4,-1) to (2,0);
	\draw[edge] (2,1.5) to (2,2.8);
	\draw[edge] (1.9,0) arc (-90:270:0.8cm);
	\ncvertex{2,0}
	\ncvertex{2,1.5}}
	\node at (-0.4,-1.4) {\huge $V_1$};
	\node at (2,3.3) {\huge $[D,a]$};
	\node at (4.4,-1.4) {\huge $V_2$};
\end{tikzpicture}}}
\,\,
+
\quad\,
\raisebox{-20pt}{\scalebox{.55}{
\begin{tikzpicture}[thick]
	{\color{green}
	\draw[edge] (0,-1) to (2,0);
	\draw[edge] (4,-1) to (2,0);
	\draw[edge] (1.9,0) arc (-90:270:0.8cm);
	\draw[edge] (2,0) to (0,0);
	\ncvertex{2,0}}
	\node at (-0.4,-1.4) {\huge $V_1$};
	\node at (-1,0) {\huge $[D,a]$};
	\node at (4.4,-1.4) {\huge $V_2$};
\end{tikzpicture}}}
.
\end{align*}
We have colored the noncommutative Feynman diagrams on the right-hand sides according to their topology, i.e., without markings on the external edges and as they appear in Figure \ref{fig:1loop-3pt}. 
\begin{figure}
\hspace{24pt}
  \begin{tabular}{p{.3\linewidth}p{.3\linewidth}p{.3\linewidth}}
  \scalebox{.6}{ 
\begin{tikzpicture}[thick]
	\color{red}
	\draw[edge] (0.3,0.3) to (2,2);
	\draw[edge] (0.3,3.7) to (2,2);
	\draw[edge] (2,2) to[out=90,in=90] (4,2);
	\draw[edge] (2,2) to[out=-90,in=-90] (4,2);
	\draw[edge] (4,2) to (6,2);
	\ncvertex{4,2}
	\ncvertex{2,2}
\end{tikzpicture}}
 &\scalebox{.6}{
\begin{tikzpicture}[thick]
	\color{green}
	\draw[edge] (0,-1) to (2,0);
	\draw[edge] (4,-1) to (2,0);
	\draw[edge] (2,-2) to (2,0);
	\draw[edge] (1.9,0) arc (-90:270:0.8cm);
	\ncvertex{2,0}
\end{tikzpicture}}
 &
 \scalebox{.6}{
 \begin{tikzpicture}[thick]
 	\color{blue}
	\draw[edge] (0,3.2) to (1,2.7);
	\draw[edge] (2,-0.1) to (2,1);
	\draw[edge] (2,1) to (1,2.7);
	\draw[edge] (1,2.7) to (3,2.7);
	\draw[edge] (3,2.7) to (2,1);
	\draw[edge] (3,2.7) to (4,3.2);
	\ncvertex{1,2.7}
	\ncvertex{2,1}
	\ncvertex{3,2.7}
\end{tikzpicture}}
    \end{tabular}
  \caption{Relevant three-point diagrams at one-loop.}
  \label{fig:1loop-3pt}
  \end{figure}
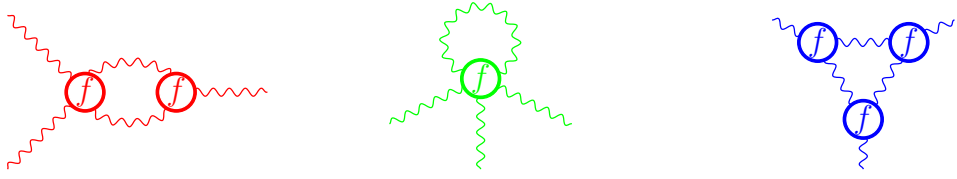
One then readily sees that the diagrams conspire to yield all cyclic permutations of $V_1,V_2,[D,a]$ as external fields on all relevant one-loop diagrams with three external edges. We obtain a two-point \textit{quantum Ward identity}, namely
\begin{align*}
	\bbr{aV_1,V_2}-\bbr{V_1,V_2a}=\bbr{V_1,V_2,[D,a]}.
\end{align*}

Proving the quantum Ward identity in general is the key to obtain the main theorem of Section \ref{sct:One-Loop}, which is formulated as follows.

\medskip

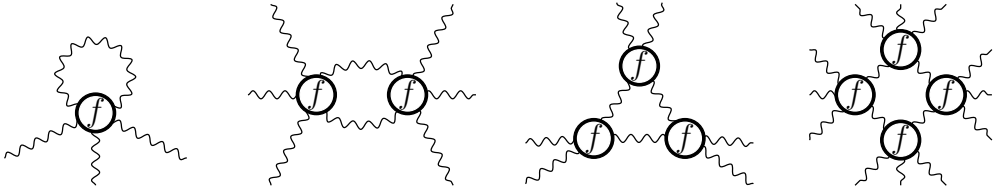
\begin{figure}
\hspace{.05\linewidth}
 \begin{tabular}{p{.181\linewidth}p{.21\linewidth}p{.19\linewidth}p{.20\linewidth}}
\scalebox{.6}{
\begin{tikzpicture}[thick]
	\draw[edge] (0,-1) to (2,0);
	\draw[edge] (4,-1) to (2,0);
	\draw[edge] (2,-1.6) to (2,0);
	\draw[edge] (1.9,0) arc (-90:270:0.8cm);
	\ncvertex{2,0}
\end{tikzpicture}}
&\scalebox{.6}{
\begin{tikzpicture}[thick]
	\draw[edge] (1,0) to (2,2);
	\draw[edge] (0.5,2) to (2,2);
	\draw[edge] (1,4) to (2,2);
	\draw[edge] (2,2) to[out=90,in=90] (4,2);
	\draw[edge] (2,2) to[out=-90,in=-90] (4,2);
	\draw[edge] (4,2) to (5,4);
	\draw[edge] (4,2) to (5.5,2);
	\draw[edge] (4,2) to (5,0);
	\ncvertex{2,2}
	\ncvertex{4,2}
\end{tikzpicture}}
   	&
 \scalebox{.6}{
\begin{tikzpicture}[thick]
	\draw[edge] (0.5,0.9) to (2,1);
	\draw[edge] (0.5,0) to (2,1);
	\draw[edge] (2,1) to (4,1);
	\draw[edge] (2,1) to (3,2.6);
	\draw[edge] (3,2.6) to (4,1);
	\draw[edge] (3,2.6) to (2.5,4);
	\draw[edge] (3,2.6) to (3.5,4);
	\draw[edge] (4,1) to (5.5,0.9);
	\draw[edge] (4,1) to (5.5,0);
	\ncvertex{2,1}
	\ncvertex{4,1}
	\ncvertex{3,2.6}
\end{tikzpicture}}
  &\quad
 \scalebox{.6}{
\begin{tikzpicture}[thick]
	\draw[edge] (0,1) to (1,2);
	\draw[edge] (0,2) to (1,2);
	\draw[edge] (0,3) to (1,2);
	\draw[edge] (1,2) to (2,1);
	\draw[edge] (2,1) to (3,2);
	\draw[edge] (1,2) to (2,3);
	\draw[edge] (2,3) to (3,2);
	\draw[edge] (2,3) to (1,4);
	\draw[edge] (2,3) to (2,4);
	\draw[edge] (2,3) to (3,4);
	\draw[edge] (3,2) to (4,3);
	\draw[edge] (3,2) to (4,2);
	\draw[edge] (3,2) to (4,1);
	\draw[edge] (2,1) to (1,0);
	\draw[edge] (2,1) to (2,0);
	\draw[edge] (2,1) to (3,0);
	\ncvertex{1,2}
	\ncvertex{2,1}
	\ncvertex{3,2}
	\ncvertex{2,3}
\end{tikzpicture}}
 \end{tabular}
 \caption{Relevant one-loop $n$-point functions with increasing number of vertices.}
 \label{table:skel-1l}
\end{figure}

\begin{thm}
There exist $(b,B)$-cocycles $\phi^N$ and $\tilde\psi^N$ (namely, those defined by taking $\brr{\cdot}=\bbr{\cdot}$ in Proposition \ref{prop:expansions for arbitrary brackets}) 
   for which the one-loop quantum effective spectral action can be expanded as
   $$
\sum_{n=1}^\infty \frac{1}{n} \bbr{V,\ldots,V} 
   \sim \sum_{k=1}^\infty \left( \int_{\psi_{2k-1}^N}  \!\!\!\!\! \cs_{2k-1} (A) +\frac 1 {2k} \int_{\phi_{2k}^N} \!\!\!\!\! F^{k} \right),
     $$
    in the sense of Proposition \ref{prop:expansions for arbitrary brackets}. As before, $\tilde\psi_{2k-1}^N=(-1)^{k-1}\tfrac{(k-1)!}{(2k-1)!}\psi_{2k-1}^N$. 
\end{thm}
\begin{proof}

Applying Definition \ref{def:quantum effective SA}, and combining two sums, we obtain
\begin{align*}
	\bbr{aV_1,\ldots,V_n}-\bbr{V_1,\ldots,V_na}=\sum_G \left(\Gamma_N^G(aV_1,\ldots,V_n)-\Gamma_N^G(V_1,\ldots,V_na)\right),
\end{align*}
where the sum is over all \textit{relevant} diagrams $G$, by which we mean the planar one-loop one-particle-irreducible $n$-point noncommutative Feynman diagrams $G$ with clockwise vertices of degree $\geq3$ and external edges outside the loop and marked cyclically.
Let $G$ be a relevant diagram marked $1,\ldots,n$. We let $I(G)$ denote the set of diagrams one can obtain from $G$ by inserting a single gauge edge at any of the places one visits when walking along the outside of the diagram from the external edge $n$ to the external edge $1$. To be precise, if the edges $n$ and $1$ attach to the same noncommutative vertex $v$, we set 
	$$I(G):=\{G'\},$$
where $G'$ is the diagram obtained from $G$ by inserting an external edge marked $n+1$ at $v$ between the edges marked $n$ and $1$. If the edges $n$ and $1$ attach to different vertices $v$ and $w$, respectively, then the edge $e$ succeeding the edge marked $n$ on $v$ necessarily attaches to $w$, preceding the edge marked $1$. In this case, we set
$$I(G):=\{G_n,G_e,G_1\},$$
where $G_n$ is obtained from $G$ by inserting an external edge marked $n+1$ at $v$ between $n$ and $e$, $G_e$ is obtained from $G$ by inserting a noncommutative vertex $v_0$ along $e$ and inserting an external edge marked $n+1$ along the outside of $v_0$, and $G_1$ is obtained from $G$ by inserting an external edge marked $n+1$ at $w$ between $e$ and $1$. For example, on the right-hand side of \eqref{eq:example quantum Ward} we see the elements of $E(G)$ for a $G$ which is shown on the left-hand side (disregarding the decorations and dashed lines).
More generally, by construction of $I(G)$, we find
\begin{align*}
&\llangle a  V_1, \ldots, V_n \rrangle_N^{1L} - 
\llangle V_1, \ldots, V_n  a \rrangle_N^{1L} = \sum_{G}\sum_{G'\in I(G)}
\Gamma_N^{G'}(V_1, \ldots ,V_n,[D,a]).
\end{align*}
The sum over $G$ and $G'$ yields all relevant $n+1$-point diagrams, and, moreover, any relevant $n+1$-point diagram with labels $V_1, \ldots ,V_n,[D,a]$ is obtained in a unique manner from an insertion of an external edge in an $n$-point diagram, as described above. We are therefore left precisely with 
\begin{align*}
\llangle a  V_1, \ldots, V_n \rrangle_N^{1L} - 
\llangle V_1, \ldots, V_n  a \rrangle_N^{1L} = \bbr{V_1,\ldots,V_n,[D,a]}.
\end{align*}
In combination with cyclicity, $\llangle V_1, \ldots, V_n \rrangle_N^{1L} = \llangle V_n , V_1, \ldots, V_{n-1} \rrangle_N^{1L}$, this identity allows us to apply Proposition \ref{prop:expansions for arbitrary brackets}. We thus arrive at the conclusion of the theorem.
\end{proof}

We conclude that the passage to the one-loop renormalized spectral action can be realized by a transformation in the space of cyclic cocycles, sending $\phi \mapsto \phi+ \phi^N$ and $\psi \mapsto \psi+ \psi^N$. One could say the theory is therefore one-loop renormalizable in a generalized sense, allowing for infinitely many counterterms, as in \mbox{\cite{GW96}}. Most notably, 
we have stayed within the spectral paradigm of noncommutative geometry.

It would be interesting to connect our approach to that of the Grosse--Wulkenhaar model \cite{GW05}, which is somewhat more specific, but naturally also allows for stronger results.
One of the main differences is that one there considers so-called non-local matrix models \cite{GW05b} with a quartic vertex, while instead we have introduced a model with vertices of arbitrary degree, by taking the spectral action as our starting point.

\end{fmffile}














\addtocontents{toc}{\protect\newpage}
\part{C*-algebraic Quantization for Lattice Gauge Theory}
\label{part:II}

\chapter{Classical and Quantized Resolvent Algebras on the Torus}
\chaptermark{Classical \& Quantized Resolvent Algebras on $\T^n$}
\label{ch:cylinder}

In this chapter, adapted from \cite{vNS20}, we define an analogue of the resolvent algebra \cite{BG} on the cotangent bundle $T^*\T^n$ of the $n$-torus by first generalizing the commutative resolvent algebra from \cite{vN19}, and subsequently applying Weyl quantization.
We prove that this quantization is almost strict (in the sense of Rieffel and Landsman) and show that our resolvent algebra shares many features with the original resolvent algebra.
We demonstrate that both our classical and quantized algebras are closed under the time evolutions corresponding to large classes of potentials.
The algebras are exceptionally convenient for lattice gauge theory.

Results in this chapter were obtained in collaboration with Ruben Stienstra.
The discussion here is smoothened by referring to \cite{vNS20} and \cite{Stienstra} for technical proofs of results that already appeared in Stienstra's dissertation \cite{Stienstra}.

\section{Introduction}
\label{subsec:classical_resolvent_algebra_introduction}

	Quantum theories are often obtained or studied via their classical limits.	
	This holds true not only for gauge theory, but for statistical mechanics, quantum gravity, and other parts of physics as well.
Showing that a classical theory is indeed the limit of the quantum theory at hand can be done at various levels of rigor.
The most precise way to establish this limit is by strict deformation quantization, where one `quantizes' a classical (commutative) Poisson algebra into a quantum (noncommutative) C*-algebra \cite{landsman98,Rieffel89} (cf. \cite[p. 5]{Hawkins} for an overview of the various definitions in the literature).

	A pair of a classical and a quantum C*-algebra connecting in this rigorous fashion is not easy to construct, but efforts are made to give more and more examples (\cite{ bieliavsky15, landsman98, rieffel90, Rieffel}, to name a few) in order to deal with the various configuration spaces that appear in applications. In abelian lattice gauge theory, the $n$-dimensional torus arises as configuration space, and one may look for strict quantizations of subspaces of $C_\textnormal{b}(T^*\T^n)$. As we will discuss, the known examples were too limited in certain specific respects. In this chapter, we will define a quantum observable algebra on the torus, i.e., a C*-algebra $A_\hbar\subseteq \B(L^2(\T^n))$ which satisfies the following properties:
\begin{enumerate}[label=P\arabic{*}:, ref=P\arabic{*}]
	\item The algebra $A_\hbar$ has a classical counterpart $A_0$ and can be obtained from this commutative algebra through (strict) quantization.\label{property:quantization}
	\item The algebra $A_\hbar$ is closed under the time evolution associated to the potential $V$ for each $V \in C(\T^n)_\sa$.
The classical analogue $A_0$ satisfies a similar condition.\label{property:time}
	\item The classical and quantum algebras associated to a given system are both sufficiently large to accommodate natural embeddings of the respective algebras corresponding to their subsystems.\label{property:direct}
	\item The algebras $A_0$ and $A_\hbar$ contain the algebra $C_0(T^\ast \mathbb{T}^n)$ and its quantization $\mathcal K(L^2(\T^n))$, respectively, without being larger than necessary.
	\label{property:small}
\end{enumerate}

\noindent
An observable algebra satisfying only \ref{property:quantization}, \ref{property:time} and \ref{property:small} has long been known, namely the compact operators $\mathcal K(L^2(\T^n))$, with $C_0(T^*\T^n)$ as its classical limit (cf. \cite{landsman98}, in particular Sections II.3.4, III.3.6 and III.3.11).
We now sketch how the need for \ref{property:direct} arises in quantum lattice gauge theory.
More details can be found in \cite[Section 5.1]{Stienstra}.

\paragraph*{Lattice gauge theory.} In the Hamiltonian lattice gauge theory by Kogut and Susskind \cite{KS}, one approximates a time-slice of spacetime by a finite `lattice', or more accurately, an oriented graph $\Lambda$.
The vertices, contained in the set $\Lambda^0$, are points in the time slice, while the oriented edges, contained in $\Lambda^1$, are paths between these points.
A gauge field corresponding to some connection on a principal fiber bundle over spacetime with gauge group $G$ (a compact Lie group) is approximated by the parallel transport maps along the edges of $\Lambda$.
After choosing a trivialization of the restriction of the principal fiber bundle to $\Lambda^0$, the set of all possible parallel transporters can be identified with $G^{\Lambda^1}$; this is the configuration space of the Hamiltonian lattice gauge theory, and it carries a natural action of $G^{\Lambda^0}$ (endowed with the obvious group structure).
The latter group represents the approximate gauge transformations.


The Hilbert space of the corresponding quantum lattice gauge theory is $\H = L^2(G^{\Lambda^1})$, where $G^{\Lambda^1}$ is endowed with the normalized Haar measure.
The field algebra of the system is some C*-algebra $A_\Lambda$ that is represented on $\H$, from which the observable algebra can be obtained by applying a reduction procedure with respect to the gauge group (cf. \cite{KR05,SvS}).
The observable algebra is accordingly represented on the set of elements of $\H$ that are invariant under gauge transformations.
Since the distinction between field and observable algebras is irrelevant with regard to the issue that motivates the present investigation -- the embedding maps take the same form in both cases -- we will continue to refer to $A_\Lambda$ as the observable algebra in what follows.

In the context of lattice gauge theory, one is interested in constructing an algebra of the continuum system from the above algebras $A_\Lambda$.
This is done by considering direct systems of lattices, and we are naturally led to consider the following situation.
Suppose that $\Lambda_1$ and $\Lambda_2$ are both lattices approximating a time slice, and that $\Lambda_2$ is a better approximation than $\Lambda_1$, i.e., $\Lambda_1^0 \subseteq \Lambda_2^0$, the graph $\Lambda_2$ contains more edges than $\Lambda_1$, and each edge in $\Lambda_1$ can be written as a concatenation of edges in $\Lambda_2$. 
We should be able to find a corresponding embedding map $A_{\Lambda_1} \hookrightarrow A_{\Lambda_2}$.
The embedding map takes a simple form if $\Lambda_2$ is obtained from $\Lambda_1$ by only adding edges: in that case, we have $\H_2 = \H_1 \hatotimes \H_1^c$, where $\H_1^c = L^2(G^{\Lambda_2^1 \backslash \Lambda_1^1})$, and the embedding is given by the restriction of the map
\begin{equation}\label{eq:tensoring with 1}
\B(\H_1)\rightarrow \B(\H_2) \cong \B(\H_1) \hatotimes \B(\H_1^c),\qquad a \mapsto a \otimes \1,
\end{equation}
to $A_{\Lambda_1}$, where $\1$ denotes the identity on $\H_1^c$, and $\hatotimes$ denotes the von Neumann algebraic tensor product.

A first guess for the observable algebras of the two quantum systems could be $\mathcal K(\H_1)$ and $\mathcal K(\H_2)$, the algebras of compacts.
However, except in trivial cases, the Hilbert space $\H_1^c$ will be infinite-dimensional, which means that $a\otimes \1$ will not be a compact operator.
Thus the algebra $\mathcal K(\H_2)$ is too small to accommodate these embeddings.
This problem was already noticed by Stottmeister and Thiemann in  \cite{ST}.
Simply adding $\1$ to the compacts causes its own problems, as for instance noted in \cite{BG}.

In \cite{ASS}, the problem concerning \eqref{eq:tensoring with 1} was not encountered since different embedding maps were used.
However, as \cite[Chapter 8]{Stienstra} points out, these embedding maps have problems of their own. 
The argument presented there is not specific to lattice gauge theory, but can be made for any physical system that is comprised of smaller subsystems.

Another guess for the observable algebra of the composite system could be the one generated by the embedded algebras of all subgraphs, as is done in \cite{GR}. However, this raises questions about regulator independence of this procedure in situations where one takes limits corresponding to an infinite volume or continuum limit of a collection of systems parametrized by a cutoff.
As this problem is beyond our scope, we will refer the reader to the discussion in \cite[Section 5.1]{Stienstra}.
The main point is that there is ample reason to try to solve the problem through an appropriate choice of algebras, i.e., algebras that satisfy \ref{property:direct}.

\paragraph*{The resolvent algebra on $\R^n$.} In the case where the configuration space is $\R^n$, there already exists an algebra satisfying \ref{property:quantization}, \ref{property:time}, \ref{property:direct} and \ref{property:small}: the resolvent algebra $\mR(\R^{2n},\sigma_n)$.
The resolvent algebra $\mathcal{R}(X,\sigma)$ on a symplectic vector space $(X,\sigma)$ is a C*-algebra that was introduced by Buchholz and Grundling in \cite{BG07}, and subsequently studied in greater detail in \cite{BG} and \cite{buchholz14} by the same authors.
Before we adapt this algebra to the case where $\T^{n}$ instead of $\R^{n}$ is the underlying configuration space, let us recall the main idea behind the construction of the resolvent algebra.

The resolvent algebra is constructed as the completion of a *-algebra with respect to a certain C*-seminorm \cite[Definition 3.4]{BG}; the *-algebra is defined in terms of generators and relations.
To each pair $(\lambda, f) \in (\R \backslash \{0\}) \times X$, a generator $\mathcal{R}(\lambda, f)$ is associated.
Such a generator is thought of as the resolvent (depending on $\lambda$) corresponding to some unbounded operator $\phi(f)$ associated to the vector $f$, where $\phi$ denotes a linear map from $X$ to a space of operators on a dense subspace of a Hilbert space on which $\mathcal{R}(X, \sigma)$ can be represented faithfully.

For example, suppose that $(X, \sigma)$ is $\R^2$ endowed with the standard symplectic form.
Then $\mathcal{R}(X, \sigma)$ admits a faithful representation on $L^2(\R)$ such that the unbounded operators corresponding to the vectors $(1,0)$ and $(0,1)$ are the standard position and momentum operators respectively (up to a factor of $\hbar$ in the latter case), see \cite[Corollary 4.4 and Theorem 4.10]{BG}.
Both of these unbounded operators can be defined on the (invariant) dense subspace $C^\infty_\textnormal{c}(\R)$, on which they are essentially self-adjoint.

For each $f \in X$, the generator $\mathcal{R}(\lambda, f)$ is mapped to the bounded operator $(i\lambda \unit - \phi(f))^{-1}$; in particular, taking $f = 0$, we see that $\mathcal{R}(X, \sigma)$ is unital.
The relations defining the *-algebra from which the resolvent algebra is constructed serve to encode the fact that $\mathcal{R}(\lambda, f)$ behaves like the resolvent of the unbounded operator $\phi(f)$, as well as the linearity of $\phi$.
Last but not least, the canonical commutation relations (CCR) are introduced by the defining relations of $\mathcal{R}(X, \sigma)$ in which the symplectic form appears, thereby justifying the term ``canonical quantum systems'' in the title of \cite{BG}.

The resolvent algebra is not the only approach to the reformulation of the CCR in a framework based on bounded operators; another is obtained through exponentiation of the unbounded operators of interest, leading to the Weyl form of the CCR and the Weyl algebra.
There is a bijection between certain classes of representations of these two algebras \cite[Corollary 4.4]{BG}.
In particular, generators of the resolvent algebras can be expressed in terms of generators of the Weyl algebra by means of the Laplace transform, as is done in \cite{BG07}.
By changing the representation in that definition to the usual representation on $L^2(\R)$ of the Weyl algebra on $\R^2$, one obtains the representation mentioned earlier.

Buchholz and Grundling note that their resolvent algebra has some desirable qualities not shared by the Weyl algebra, such as the presence of observables corresponding to bounded functions in regular representations.
Also with respect to time evolution, i.e., \ref{property:time}, the resolvent algebra is a superior alternative to the Weyl algebra.
For example, \cite[Proposition 6.1]{BG} shows that the resolvent algebra associated to $\R^2$ endowed with the standard symplectic form is closed under (quantum) time evolution for a large class of Hamiltonians, while the Weyl algebra only admits free time evolution.
The resolvent algebra is also stable under dynamics in the context of oscillating lattice systems \cite{buchholz17} and nonrelativistic Bose fields \cite{buchholz18}. 

As regards \ref{property:direct}, by \cite[Theorem 5.1]{BG}, for any symplectic vector space $(X,\sigma)$ and any decomposition $X=S \oplus S^\perp$ into nondegenerate subspaces (where $S^\perp$ denotes the complement of $S$ with respect to $\sigma$) the resolvent algebra $\mathcal{R}(X, \sigma)$ naturally contains a copy of $\mathcal{R}(S, \sigma_{|S}) \hatotimes \mathcal{R}(S^\perp, \sigma_{|S^\perp})$; with respect to corresponding faithful representations of these three resolvent algebras, the embeddings of $\mathcal{R}(S, \sigma_{|S})$ and $\mathcal{R}(S^\perp, \sigma_{|S^\perp})$ are given by the analogues of the aforementioned embedding map for lattice gauge theory.
Here, $\hatotimes$ denotes any C*-algebraic tensor product (nuclearity of the resolvent algebra is shown in \cite{buchholz14}).\\

\noindent
We have seen how properties \ref{property:time} and \ref{property:direct} hold for the resolvent algebra. A proof of \ref{property:quantization} also exists, and forms the basis of our construction in the case of the torus. Indeed, it is shown in \cite{vN19} that the resolvent algebra arises as the strict deformation quantization of an algebra that can be considered the observable algebra of a classical system in the sense of Rieffel and Landsman, i.e., the C*-algebra generated by the image of a dense Poisson subalgebra of the classical algebra under a quantization map \cite{landsman98}.
In particular, when $(X, \sigma)$ is $\R^{2n}$ endowed with the standard symplectic form, there is a corresponding commutative C*-algebra $C_{\mathcal{R}}(\R^{2n})$, which is the C*-subalgebra of $C_\textnormal{b}(\R^{2n})$ generated by functions of the form
\begin{equation*}
x \mapsto (i\lambda - x \cdot v)^{-1}, \quad
\lambda \in \R \backslash \{0\}, \:
v \in \R^{2n},
\end{equation*}
where $\cdot$ denotes the standard inner product.
Similar to the way in which the algebra $C_0(\R^{2n})$ may be quantized into the compact operators on $L^2(\R^n)$ by considering the dense Poisson subalgebra $\mathcal{S}(\R^{2n})$ of Schwartz functions and defining Weyl or Berezin quantization on them, we consider a dense Poisson subalgebra of $C_{\mathcal{R}}(\R^{2n})$ defined by
\begin{equation*}
\mathcal{S}_{\mathcal{R}}(\R^{2n})
:= \text{\normalfont span}_{\mathbb{\C}} \set{g \circ P_V}{V \subseteq \R^{2n} \text{ is linear, } g \in \mathcal{S}(V)} \, ,
\end{equation*}
where $P_V$ denotes the orthogonal projection onto $V$.
The Weyl quantization of $g \circ P_V$ is defined using the Fourier transform of $g$ as a function on $V$ \cite[Section 3.2]{vN19}, but is otherwise equal to the definition of the Weyl quantization of ordinary Schwartz functions on $\R^{2n}$.
It is then argued that the Weyl quantization map admits a (unique) linear extension to $\mathcal{S}_{\mathcal{R}}(\R^{2n})$.
Furthermore, it is shown that the images of $\mathcal{S}_{\mathcal{R}}(\R^{2n})$ under Weyl and Berezin quantization are both dense subspaces of $\mathcal{R}(\R^{2n}, \sigma)$.
The resulting algebra $C_{\mathcal{R}}(\R^{2n})$ is accordingly referred to as the \textit{commutative resolvent algebra on} (the cotangent bundle of) $\R^{n}$.
As is shown in \cite{vN19}, these definitions are easily extended to infinite dimensions.

In addition to being the classical counterpart of the resolvent algebra as defined by Buchholz and Grundling, the commutative resolvent algebra offers an interesting perspective on our earlier discussion on embeddings of observable algebras.
In some sense, $C_{\mathcal{R}}(\R^n)$ is the smallest C*-subalgebra of $C_\textnormal{b}(\R^n)$ that contains $C_0(\R^n)$, whilst also containing its analogues associated to linear subspaces of $\R^n$.
This may be formalized as follows.
Consider the category whose objects are finite-dimensional real vector spaces, and whose morphisms consist of projections of a vector space onto one of its subspaces.
Then there is a contravariant functor $C_\textnormal{b}$ from this category to the category of C*-algebras that maps an object $V$ to the space $C_\textnormal{b}(V)$, and that maps morphisms to their pullbacks between these spaces.
It is now consistent with the definition of the commutative resolvent algebra to define $C_{\mathcal{R}}$ as the smallest subfunctor of $C_\textnormal{b}$ with the property that the image of every object $V$ contains $C_0(V)$.
Note that this implies that $C_\mathcal{R}(\R^n)$ is unital, as it contains the embedding of $C_0(\{0\})$.
This makes precise in which sense \ref{property:small} holds for the resolvent algebras on $\R^{n}$.

\noindent
\paragraph*{Resolvent algebras on the torus.}
In this chapter we introduce an analogue of the resolvent algebra where the configuration space $\R^{n}$ is replaced by $\T^n$. Our main motivation is abelian lattice gauge theory, where the gauge group, and therefore the configuration space, is a compact abelian Lie group, and therefore isomorphic to $\T^n$.
The resolvent algebra of the torus, in contrast with the one of Buchholz and Grundling, is not introduced by means of generators and relations.
Rather, we first identify a commutative resolvent algebra $C_{\mathcal{R}}(T^\ast \T^n)$ by generalizing the definition of \cite{vN19}.
We then give a concrete characterization of $\Cr(T^*\T^n)$. Namely, identifying $T^*\T^n$ with $\T^n\times\R^n$, we prove that $\Cr(T^*\T^n)$ equals $C(\T^n)\!\hatotimes\! \Wr(\R^n)$, where $\Wr(\R^n)$ is the C*-algebra generated by the functions 
	\begin{align}\label{generators}x\mapsto 1/(i+x\cdot v)\quad\text{ and }\quad x\mapsto e^{ix\cdot v},\quad\text{ for all }\quad v\in\R^n.
	\end{align}
In addition, we identify a dense *-subalgebra $\mathcal{S}_{\mathcal{R}}(T^\ast \T^n)\subseteq\Cr(T^*\T^n)$ carrying a natural Poisson structure.
The algebra is spanned by functions of the form $e_b\otimes h$, where $e_b[x]:=e^{2\pi ib\cdot x}$, and $h$ is a smooth function that is a product of an element of $\mathcal{S}_{\mathcal{R}}(\R^n)$ and a function of the form $x \mapsto e^{i \xi \cdot x}$ for some $\xi \in \R^n$.

To define a quantum counterpart, we apply Weyl quantization, making \ref{property:quantization} integral to the definition of the (quantum) resolvent algebra on $\T^n$.
Our Weyl quantization map $\QW:\Sr(T^*\T^n)\rightarrow \B(L^2(\T^n))$ is an extension of the usual (\cite[Section II.3.4]{landsman98}) Weyl quantization on (a subalgebra of) $C_0(T^*\T^n)$, when we see $\T^n$ as a Riemannian manifold with its corresponding Levi-Civita connection. The same quantization map $\QW$ equivalently arises by viewing $\T^n$ as a quotient of Euclidean space, and adapting the Weyl quantization of $\R^{2n}$ accordingly. The most explicit characterization of $\QW$ is obtained by writing $\Cr(T^*\T^n)$ as the tensor product $C(\T^n)\!\hatotimes\! \Wr(\R^n)$. We then have
\begin{align}\label{eq:formula QW intro}
	\QW(e_b\otimes h)\psi_a=h(2\pi\hbar(a+\tfrac12 b))\psi_{a+b}\,,
\end{align}
where $e_b\otimes h\in\Sr(T^*\T^n)$, and $\psi_b$ is the equivalence class of $e_b\in C(\R^n)$ in $L^2(\mathbb{R}^n)$ for each $b \in \mathbb{Z}^n$.
Using this generalized Weyl quantization map $\QW$, we define the (quantum) resolvent algebra on the torus as
	$$A_\hbar:=C^*(\QW(\Sr(T^*\T^n)))\subseteq \B(L^2(\T^n)),$$
before remarking that $A_\hbar\cong A_{\hbar'}$ for all $\hbar,\hbar'\in(0,\infty)$. The property \ref{property:direct} turns out to follow from this explicit description of $\QW$ and the fact that \ref{property:direct} holds for $\Cr(T^*\T^n)$, which is readily seen. \ref{property:small} is satisfied by definition of $\Cr(T^*\T^n)$.

The main contribution of this chapter is that \ref{property:time} also holds for our algebras, both the classical and the quantum one, in the following very strong sense. Our commutative resolvent algebra $C_{\mathcal{R}}(T^\ast \T^n)$ is closed under the classical time evolution associated to the potential $V$ for each $V \in C^1(\T^n)_\sa$ with Lipschitz continuous derivative.
Our quantum resolvent algebra $A_\hbar$ is closed under the quantum time evolution associated to the potential $V$ for each $V \in C(\T^n)_\sa$.
In both cases, the free part of the Hamiltonian is the usual one.
Unlike the analogous result in \cite{BG} in which a similar result is established only for $\R^{2n}$ with $n = 1$, our results hold for arbitrary $n \in \N$.\\

\noindent 
This chapter is structured as follows.
In Section \ref{sct:Basic Results}, we first define the commutative resolvent algebra $C_{\mathcal{R}}(T^\ast \T^n)$ by extending the definition of \cite{vN19}.
We proceed by analyzing its structure, culminating in the more practical characterization $C_{\mathcal{R}}(T^\ast \T^n)=C(\T^n) \hatotimes \Wr(\R^n)$.
Furthermore, we identify a dense *-subalgebra that carries a Poisson structure.
In Section \ref{sec:quantisation}, we give a well-motivated definition of Weyl quantization on this dense *-subalgebra, and define the quantum resolvent algebra on $\T^n$ by use of this quantization map.
%
Section \ref{sec:classical time evolution} proves the fact that $C_{\mathcal{R}}(T^\ast \T^n)$ is invariant under classical time evolutions in the general setting mentioned mentioned above.
In Section \ref{sec:quantum time evolution}, we show that the quantum resolvent algebra on $\T^n$ is invariant under quantum time evolutions.
Finally, Section \ref{sct:almost strict} shows that our quantization map fulfills almost all conditions of a strict deformation quantization, which will lead us naturally to Chapter \ref{ch:SDQ}.


\section{Definition and basic results}
\label{sct:Basic Results}

\noindent
On the phase space $\R^{2n}$, we already have a commutative C*-algebra that satisfies \ref{property:time}, \ref{property:direct} and \ref{property:small} mentioned in the introduction and forms the classical part of a strict deformation quantization, namely the commutative resolvent algebra $\Cr(\R^{2n})$ defined in \cite{vN19}.
We begin this section by adapting its definition to $T^\ast \T^n$.
As mentioned in the introduction, we identify $T^\ast \T^n$ with $\T^n \times \R^n$, and we note that the latter space carries a natural left action of $\R^{2n} = \R^n \times \R^n$ by translation.

\begin{defi}\label{def:resolvent algebra on the torus}
For each $(v,w) \in \R^n \times \R^n = \R^{2n}$, let $(\T^n \times \R^n)/\{(v,w)\}^\perp$ be the space of orbits of the restriction of the action of $\R^{2n}$ to $\{(v,w)\}^\perp \subseteq \R^{2n}$, and let $$\pi_{(v,w)} \colon \T^n \times \R^n \rightarrow (\T^n \times \R^n)/\{(v,w)\}^\perp$$ be the corresponding canonical projection.
The \textbf{commutative} (or, classical) \textbf{resolvent algebra on} $\T^n$, denoted $\Cr(T^\ast \T^n)$, is the smallest C*-subalgebra of $C_\textnormal{b}(\T^n\times\R^n)$ generated by the set of functions
\begin{equation*}
\left\{ f \circ \pi_{(v,w)} :~ (v,w) \in \R^{2n}, \: f \in C_0((\T^n \times \R^n)/\{(v,w)\}^\perp) \right\},
\end{equation*}
that is, the set of continuous functions invariant under the action of $\{(v,w)\}^\perp \subseteq \R^{2n}$ for which the induced map on $(\T^n \times \R^n)/\{(v,w)\}^\perp$ vanishes at infinity.
\end{defi}

\noindent
To establish the link with the definition of $\Cr(\R^n)$ given in \cite{vN19}, note that there is an immediate generalization of the above definition to arbitrary topological spaces $M$ carrying a left action of $\R^{m}$ for some $m \in \N$.
Taking $M = \R^n$ and $m = n$ then yields the definition of $\Cr(\R^n)$.
Unfortunately, $T^*G$ does not have an appropriate action of $\R^{2n}$ for a nonabelian Lie group $G$ that would enable us to unambiguously generalize this construction.

The definition of the commutative resolvent algebra $C_{\mathcal{R}}(T^\ast \T^n)$ is clearly motivated, but very unwieldy in practice.
Our first task is therefore to find an alternative, more elementary characterization of $C_{\mathcal{R}}(T^\ast \T^n)$.

Recall that the algebra $\Weylalg{\R^n}$ of \textit{almost periodic functions on $\R^n$} is the C*-subalgebra of $C_\textnormal{b}(\R^n)$ generated by the functions $x\mapsto e^{i \xi \cdot x}$ for $\xi \in \R^n$.


\begin{defi}
Let $n \in \N$.
We define the algebra $\Wr(\R^n)$ as the C*-subalgebra of $C_\textnormal{b}(\R^n)$ generated by the commutative resolvent algebra $C_{\mathcal{R}}(\R^n)$ and the algebra of almost periodic functions $\Weylalg{\R^n}$ on $\R^n$.
\end{defi}

\noindent 
The next theorem will unveil $\Cr(T^*\T^n)$ as a tensor product of two algebras. We regard the algebraic tensor product of two C*-algebras $A\subseteq C_\textnormal{b}(X)$ and $B\subseteq C_\textnormal{b}(Y)$ as a subset of $C_\textnormal{b}(X\times Y)$ via $(f\otimes g)(x,y)=f(x)g(y)$, and denote its corresponding completion by $A\hatotimes B$.
Since commutative C*-algebras are nuclear,
 this is equivalent to any other C*-algebraic tensor product.

The following theorem is proven in \cite[Theorem 5]{vNS20} and \cite[Theorem 5.7]{Stienstra}.

\begin{thm}\label{thrm:tractable_resolvent_algebra}
For each $n \in \N$, we have
\begin{equation*}
C_{\mathcal{R}}(T^\ast \T^n) = C(\T^n) \hatotimes \Teunalg{\R^n}.
\end{equation*}
\end{thm}

\noindent 
We finish this section by defining a smooth subspace of $C_{\mathcal{R}}(T^\ast \T^n)$.

\begin{defi}\label{def:resolvent_algebra_Schwartz_functions}
Recall that $e_b \colon \T^n \rightarrow \C,$ $[x] \mapsto e^{2\pi i b \cdot x}$ for all $b\in\Z^n$.
For each subspace $U \subseteq \R^n$, for each $\xi \in U^\perp$, and for each Schwartz function $g \in \mathcal{S}(U)$, let
\begin{equation*}
h_{U,\xi,g} \colon \R^n \rightarrow \C, \quad 
p \mapsto e^{i\xi \cdot p} g ( P_U(p)),
\end{equation*}
where $P_U \colon \R^n \rightarrow U$ denotes the orthogonal projection onto $U$.
We define 
	$$\mathcal{S}_{\mathcal{R}}(T^\ast \T^n):=\spn\{e_b\otimes h_{U,\xi,g}:~b\in\Z^n,~U\subseteq\R^n\text{ linear, }\xi\in U^\perp,~g\in\S(U)\}.$$
\end{defi}
The following result is proven in \cite[Proposition 7]{vNS20} and \cite[Proposition 5.9]{Stienstra}.
\begin{prop}
%
The vector space $\mathcal{S}_{\mathcal{R}}(T^\ast \T^n)$ is a subspace of $C_{\mathcal{R}}(T^\ast \T^n)$ that is closed under multiplication and partial differentiation, and is consequently a Poisson subalgebra of $C^\infty(T^\ast \T^n)$.
Moreover, $\mathcal{S}_{\mathcal{R}}(T^\ast \T^n)$ is a norm-dense *-subalgebra of $C_{\mathcal{R}}(T^\ast \T^n)$.
\label{prop:Poisson_subalg_for_arbitrary_n}
\end{prop}

\section{Quantization of the resolvent algebra}
\label{sec:quantisation}

Having discussed the nice properties of $\Cr(T^*\T^n)$, we now ask whether there exists a quantum version of this algebra.
Contrary to the resolvent algebra $\mathcal{R}(\R^{2n}, \sigma)$ of Buchholz and Grundling, on $T^*\T^n$ it is hard -- if not impossible -- to define an algebra in terms of generators and relations implementing canonical commutation relations.
We therefore take a different approach.

We will define our quantization of the algebra $C_{\mathcal{R}}(T^\ast \T^n)$ as an algebra represented on $L^2(\T^n)$, using a version of Weyl quantization directly related to the definition of Landsman \cite[Section II.3.4]{landsman98} for general Riemannian manifolds.
By contrast, Rieffel's algebras on $T^*\T^n$ in \cite{Rieffel}, apart from being quantizations of (subalgebras of) $C_u(T^\ast \T^n)$, are defined as universal objects from which a physical quantum system is obtained as the image of one of its irreducible representations, and it is not always clear which representation corresponds to the physical system that one wishes to model.
These algebras have many inequivalent irreducible representations due to the fact that $\T$ is not simply connected, see e.g. \cite[Example 10.6]{Rieffel} and the discussion in \cite[Section 7.7]{landsman17}.
In \cite{Stienstra}, it is argued that such universal objects might still be suited as quantum observable algebras, but we will not pursue that path here.

An indisputable advantage of quantizing $C_{\mathcal{R}}(T^\ast \T^n)$ as an algebra of operators on $L^2(\T^n)$ lies in the explicit formula for the quantizations of the generators of $C_{\mathcal{R}}(T^\ast \T^n)$, which simplifies calculations.

Before we get to that formula, we motivate our extended definition of Weyl quantization by deriving it from its analogue on Euclidean phase space.


\subsection{Definition of the quantization map}
\label{subsec:definition_of_quantisation_map}

\noindent 
Let us first recall the basics of Weyl quantization in $\R^{2n}$, the quantization procedure in \cite{weyl27} conceived by Weyl.
Given say, a Schwartz function $f \in \mathcal{S}(\R^{2n})$, one associates an operator $\mathcal{Q}^\textnormal{W}_\hbar(f) \in \B(L^2(\R^n))$ to it as follows.
First, one expresses $f$ in terms of functions of the form
\begin{equation*}
\R^{2n} = \R^n \times \R^n \rightarrow \C, \quad 
(q,p) \mapsto e^{i(x \cdot q + y \cdot p)},
\end{equation*}
where $x,y \in \R^n$, by considering the Fourier transform of $f$.
One subsequently substitutes these exponential functions with the operators
\begin{equation*}
e^{i(x \cdot Q + y \cdot P)},
\end{equation*}
where $Q,P$ are vectors whose components are the essentially self-adjoint operators on $\S(\R^n)$ $\subseteq L^2(\R^n)$ defined by $Q_j\psi(x):=x_j\psi(x)$ and $P_j\psi(x):=-i\hbar\frac{d\psi}{dx_j}(x)$.
Thus, the Weyl quantization of a function $f$ is informally given by the expression
\begin{align*}
&\int_{\R^n} \int_{\R^n}\hat{f}(x,y)e^{ix\cdot Q+iy\cdot P} \: dx \: dy \\
&\quad = (2\pi)^{-2n} \int_{\R^n} \int_{\R^n} \int_{\R^n} \int_{\R^n} f(q,p) e^{i\hbar\frac{x \cdot y}{2}} e^{ix \cdot (Q - q)} e^{iy \cdot (P - p)} \: dq \: dp \: dx \: dy,
\end{align*}
where we take $\hbar > 0$.
To define the above integrals rigorously, we can insert a function $\psi \in \mathcal{S}(\R^n)$ on the right-hand side of the integrand, and check that the resulting expression is well defined and that it defines a bounded operator on $\mathcal{S}(\R^n)$ viewed as a subspace of $L^2(\R^n)$.
Since $\mathcal{S}(\R^n)$ is dense in $L^2(\R^n)$, the operator has a unique bounded extension to $L^2(\R^n)$, which we define to be $\mathcal{Q}^\textnormal{W}_\hbar(f)$.
Using standard identities for Fourier transforms of functions, and performing a number of substitutions, it can be shown that
\begin{equation*}
(\mathcal{Q}^\textnormal{W}_\hbar(f)\psi)(x)
= (2 \pi \hbar)^{-n} \int_{\R^n} \int_{\R^n} f\left(x + \frac{y}{2}, p \right) e^{-i\frac{ y \cdot p}{\hbar}} \psi(x + y) \: dp \: dy,
\end{equation*}
for each $\psi \in \mathcal{S}(\R^n)$ and each $x \in \R^n$.

We now adapt the Weyl quantization formula to $T^\ast \T^n$ in such a way that we can quantize elements of $C_{\mathcal{R}}(T^\ast \T^n)$.
We already identified a dense Poisson algebra of $\Cr(T^*\T^n)$ in Section \ref{sct:Basic Results}, namely the space $\mathcal{S}_{\mathcal{R}}(T^\ast \T^n)$ of finite linear combinations of functions of the form $e_b \otimes h_{U,\xi,g}$; see Proposition \ref{prop:Poisson_subalg_for_arbitrary_n}.
These are the functions that we will quantize.
Because these functions do not have to vanish at infinity, we need to do some extra work. We take inspiration from \cite{Rieffel}, regarding the integrals in the above formula as oscillatory integrals, and regularizing the expression by inserting a factor in the integrand in the form of a member of a net of functions that converges pointwise to the constant function $1_{\R^n}$, as in part (1) of the next proposition.
The proof of the proposition below is found in \cite[Proposition 16]{vNS20} and \cite[Proposition 7.1]{Stienstra} and is inspired by \cite[Proposition 1.11]{Rieffel}.

\begin{prop}

\noindent 
\begin{enumerate}[label=\textnormal{(\arabic*)}]
\item Let $f \in \mathcal{S}_{\mathcal{R}}(T^\ast \T^n)$, let $\hbar > 0$, and let $\psi \in C(\T^n)$.
Then for each $[x] \in \T^n$, the limit
\begin{equation}
\lim_{\delta \to 0} (2 \pi \hbar)^{-n} \int_{\R^n} \int_{\R^n} f\left(\left[x + \tfrac{1}{2}y\right], p \right) e^{-\frac{\delta}{2}p^2} e^{-i\frac{ y \cdot p}{\hbar}} \psi[x + y] \: dp \: dy,
\label{eq:Weyl_quantisation_formula}
\end{equation}
exists.
\item Suppose that $f = e_b \otimes h_{U,\xi,g}$ is a function as described in Definition \ref{def:resolvent_algebra_Schwartz_functions}, and consider $\psi_a[x]:=e^{2\pi i a\cdot x}$ for some $a \in \Z^n$.
Then the expression in equation \eqref{eq:Weyl_quantisation_formula} is equal to
\begin{equation*}
h_{U,\xi,g}(2\pi \hbar(a + \tfrac12 b)) \psi_{a+b}[x],
\end{equation*}
and the map defined on $\spn_{a\in\Z^n}\{\psi_a\}$ sending $\psi$ to the function on $\T^n$ that assigns to a point $[x] \in \T^n$ the limit in \eqref{eq:Weyl_quantisation_formula} extends in a unique way to a bounded linear operator on $L^2(\T^n)$ with norm $\leq \|g\|_\infty$.
\end{enumerate}
\label{prop:Weyl_quantisation_is_well-defined}
\end{prop}

\noindent 
The above proposition justifies the following definitions.
\begin{defi}
For each $\hbar > 0$, we define the \textbf{Weyl quantization} $\QW(f)$ of $f\in\Sr(T^*\T^n)$ to be the unique bounded linear extension of the operator on $\spn_{a\in\Z^n}\{\psi_a\}$ defined by the formula
\begin{align}\label{eq:formula}
	\QW(e_b\otimes h)\psi_a:=h(2\pi\hbar(a+\tfrac12 b))\psi_{a+b}\,.
\end{align}
We thus obtain a map $\mathcal{Q}^\textnormal{W}_\hbar \colon \mathcal{S}_{\mathcal{R}}(T^\ast \T^n) \rightarrow \B(L^2(\T^n))$, for each $\hbar > 0$.
We define the (quantum) \textbf{resolvent algebra on} $\T^n$ to be the C*-subalgebra $A_\hbar$ of $\B(L^2(\T^n))$ generated by the image of $\Sr(\T^*\T^n)$ under $\mathcal{Q}^\textnormal{W}_\hbar$.
\end{defi}

The Weyl quantization can easily be seen to restrict to the Weyl quantization defined in \cite[Definition II.3.4.4]{landsman98}. In other words, the two approaches given either by seeing $\T^n$ as a Riemannian manifold, or by seeing it as a quotient of $\R^n$, are equivalent. We thus uncover \eqref{eq:formula} as an effective way to quantize more functions than just the ones vanishing at infinity, which, as argued in the introduction, is crucial for obtaining an infinite dimensional limit. The following proposition shows further properties of our quantization map. For the proof we refer to \cite[Proposition 18]{vNS20} and \cite[Proposition 7.4]{Stienstra}.

\begin{prop}\label{prop:Weyl quantization properties}
Let $\hbar > 0$.
\begin{enumerate}[label=\textnormal{(\arabic*)}]
\item The Weyl quantization map is linear and *-preserving;\label{it:linear and star preserving}

\item For each $\hbar^\prime > 0$, we have $A_\hbar = A_{\hbar^\prime}$;

\item The image of
\begin{equation*}
\spn \set{e_b \otimes g}{b \in \Z^n, \: g \in \mathcal{S}(\R^n)}
\subseteq \mathcal{S}_{\mathcal{R}}(T^\ast \T^n) \cap C_0(T^\ast \T^n) ,
\end{equation*}
under $\mathcal{Q}^\textnormal{W}_\hbar$ lies dense in the space $\mK(L^2(\T^n))$ of compact operators;

\item Under the canonical embedding
\begin{equation*}
\B(L^2(\T^n))\hookrightarrow \B(L^2(\T^{n+m})) \cong \B(L^2(\T^n)) \hatotimes \B(L^2(\T^m)) , \quad a \mapsto a \otimes \1 ,
\end{equation*}
induced by the projection at the level of configuration spaces $\T^{n+m}\rightarrow\T^{n}$ onto the first $n$ coordinates, the image of the resolvent algebra on $\T^{n}$ is a subalgebra of the resolvent algebra on $\T^{n+m}$.
(Here, $\hatotimes$ denotes the von Neumann algebraic tensor product.)

\item Let $\rho_0$ be the group representation of $\T^n$ on $C_\textnormal{b}(T^\ast \T^n)$ given by
\begin{align*}
\rho_0[x]f:=\left(\,(q,p) \mapsto f(-x+q,p)\,\right),
\end{align*}
and let $\rho_\hbar$ be the group representation of $\T^n$ on $\B(L^2(\T^n))$ given by
\begin{equation*}
\rho_\hbar[x]a
:=  L^*_{[x]} a L^*_{[-x]}\,,
\end{equation*}
where $L^* \colon \T^n \rightarrow U(L^2(\T^n))$ denotes the left regular representation of $\T^n$.
Then both $C_{\mathcal{R}}(T^\ast \T^n)$ and $\mathcal{S}_{\mathcal{R}}(T^\ast \T^n)$ are invariant under $\rho_0$.
Furthermore, the Weyl quantization map is equivariant with respect to these representations.
\end{enumerate}
\label{prop:quantisation_notable_properties}
\end{prop}

\begin{rema}
Because of part (2) of this proposition, we will write $A_\hbar$ for the C*-algebra generated by $\mathcal{Q}^\textnormal{W}_{\hbar^\prime}(\mathcal{S}_{\mathcal{R}}(T^\ast \T^n))$ for any value of $\hbar^\prime > 0$ without specifying $\hbar$.
Part (3) is the analogue of the first part of \cite[Corollary II.2.5.4]{landsman98} in the present setting. Part (4) indicates that our quantum algebras respect the embedding maps corresponding to the addition of edges in lattice gauge theory. Part (5) is the analogue of \cite[Theorem II.2.5.1]{landsman98}, and indicates in particular that gauge transformations can be easily incorporated into our framework.
\end{rema}

Having displayed the basic useful properties of $\Cr(T^*\T^n)$, of $A_\hbar$, and of the quantization map $\QW$ between them, we are ready to prove our two most surprising results. They state that both $\Cr(T^*\T^n)$ and $A_\hbar$ are preserved by all respective time evolutions, including an arbitrary interaction.

\section{Classical time evolution}\label{sec:classical time evolution}
%

\noindent In this section, we prove that $\Cr(T^*\T^n)$ is preserved under the (time) flow induced by the Hamiltonian
	$$H(q,p)=\tfrac{1}{2}p^2 + V(q),$$
for each potential $V\in C^1(\T^n)_\sa$ such that $\nabla V$ is Lipschitz continuous.
This is arguably the most natural assumption on $V$; the Picard--Lindel\"of theorem then ensures that the Hamilton equations have unique solutions. 

Precisely stated, for every $(q_0,p_0)\in \T^n\times\R^n$, there exist unique functions $q \colon \R \rightarrow \T^n$ and $p \colon \R \rightarrow \R^n$ that satisfy
\begin{equation}
\label{eq:q and p}
\left\{
\begin{alignedat}{2}
(\dot{q}(t),\dot{p}(t)) &= (p(t),-\nabla V(q(t))) \qquad && t \in \R, \\
(q(0),p(0)) &= (q_0, p_0). \qquad &&
\end{alignedat}
\right. 
\end{equation}
Note that the expression on the right-hand side of the first line of equation \eqref{eq:q and p} is the Hamiltonian vector field $X_H$ corresponding to $H$ evaluated at $(q(t),p(t))$.
For each $t \in \R$, the time evolution of the system after time $t$ is the map
\begin{equation*}
\Phi^t_V \colon \T^n\times\R^n \rightarrow \T^n\times\R^n, \quad 
(q_0,p_0)\mapsto(q(t),p(t)),
\end{equation*}
which is the flow corresponding to $X_H$ evaluated at time $t$; it is well-known to be a homeomorphism.

Note that we have already made the notation of the flow less cumbersome by writing $\Phi^t_V$ instead of $\Phi^t_{X_H}$.
In what follows, we restrict our attention to the case $t = 1$, further simplifying the notation by defining $\Phi_V:=\Phi^1_V$.
The following lemma shows that we may do so without loss of generality:

\begin{lem}\label{lem:t=1 or arbitrary t}
	The algebra $\Cr(T^*\T^n)$ is preserved under the pullback of $\Phi_V$ for each $V$ if and only if it is preserved under the pullback of $\Phi_V^t$ for each $V$, for each $t\in\R$.
\end{lem}
\begin{proof}
	For any $t\neq0$ (as $t=0$ is trivial), we make the following transformation on phase space
	$$\phi(q,p):=(q,tp).$$
	Because the momentum part of $\phi$ is linear, its pullback preserves the commutative resolvent algebra. Given an integral curve $(q(t),p(t))$ of the vector field $X_H$ corresponding to the potential $V$, i.e., a solution of equation \eqref{eq:q and p}, one can easily check that $s\mapsto\phi(q(ts),p(ts))$ is an integral curve corresponding to the potential $t^2V$. We therefore conclude that
	$$\Phi_V^t(q_0,p_0)=\phi^{-1} \circ \Phi^1_{t^2V} \circ \phi(q_0,p_0),$$
	which implies the claim.
\end{proof}
\noindent We prove our main theorem in three steps: taking $V=0$; taking $V$ trigonometric; and finally taking general $V$.
In the second and third step we will need the following consequence of Gronwall's inequality.
Let $d$ denote the canonical distance function on $\T^n$ as well as on $\T^n\times\R^n$. (Note that these distance functions are the ones induced by the canonical Riemannian metrics on $\T^n$ and $T^*\T^n\cong\T^n\times\R^n$, respectively.)
\begin{lem}\label{lem:Gronwall}
	Let $f,g\colon\T^n\times\R^n\rightarrow\R^{2n}$ be Lipschitz continuous functions, let $c$ be the Lipschitz constant of $f$, and let $y,z\colon[0,1]\rightarrow \T^n\times\R^n$ be curves that satisfy $\dot{y}(t)=f(y(t))$ and $\dot{z}(t)=g(z(t))$ for each $t\in[0,1]$.
Finally, suppose that $\varepsilon > 0$ is a number such that $\supnorm{f-g}\leq\varepsilon$.
Then we have
	$$d(y(t),z(t))\leq(d(y(0),z(0))+t\varepsilon)e^{tc}.$$
\end{lem}
\begin{proof}
	By translation invariance of the metric on $\T^n\times\R^n$, we have
	\begin{align*}
		d(y(t),z(t))&\leq d((y(t)-y(0))-(z(t)-z(0)),0)+d(y(0),z(0))\\
		&\leq\int_0^t\norm{f(y(s))-g(z(s))}\:ds+d(y(0),z(0))\\
		&\leq c\int_0^t d(y(s),z(s))\:ds+t\varepsilon+d(y(0),z(0)).
	\end{align*}
	With the integral version of Gronwall's inequality, this implies the lemma.
\end{proof}

\subsection{Free time evolution}\label{sct:free time evolution}

\noindent
For each pair $(q_0,p_0)\in \T^n\times\R^n$, we have $q(t)=q_0+tp_0$ and $p(t)=p_0$, denoting the usual action of $\R^n$ on $\T^n$ by $+$.
The latter notation, explicitly written as $[x]+p=[x+p]$ for $x,p\in\R^n$, will be used in the remainder of this dissertation.
We find that $\Phi_0(q_0,p_0)=(q_0+p_0,p_0)$, and obtain the following preliminary result.
Let ${}^*$ denote the pullback.
\begin{lem}\label{lem:free time evolution}
	Free time evolution preserves the commutative resolvent algebra, i.e.,
	$$\Phi_0^*(\Cr(T^*\T^n))\subseteq\Cr(T^*\T^n).$$
\end{lem}
\begin{proof}
We have
	$$\Phi_0^*(e_b\otimes h_{U,\xi,g})(q_0,p_0)=e_b(q_0)e^{ 2\pi i b\cdot p_0}e^{i\xi\cdot p_0}g(P_U (p_0)).$$
Defining $\tilde{g}\in C_0(U)$ by $\tilde{g}(p):=e^{2\pi iP_U(b)\cdot p}g(p)$, and $\tilde{\xi}:=\xi+2\pi P_{U^\perp}(b)$, we obtain
	$$\Phi_0^*(e_b\otimes h_{U,\xi,g})=e_b\otimes h_{U,\tilde{\xi},\tilde{g}}.$$
Thus the generators of $\Cr(T^*\T^n)$ are mapped into $\Cr(T^*\T^n)$ by $\Phi_0^*$, and since this map is a *-homomorphism, the lemma follows.
\end{proof}

\subsection{Trigonometric potentials}\label{sct:trigonometric potentials}

We say that $V$ is a \textit{trigonometric potential} if it is real-valued and of the form $V=\sum_{b\in\mN}a_be_b$, for some coefficients $a_b\in\C$ and a finite subset $\mN\subseteq\Z^n$. The main trick used to establish time invariance of the commutative resolvent algebra is to use induction on the size of $\mN$. The induction basis, $\mN=\emptyset$, corresponds to free time evolution. In order to carry out the induction step we fix a vector $b\in\mN$, and compare the dynamics corresponding to $V$ with the dynamics corresponding to $V-V_b$, where $$V_b:=a_be_b+a_{-b}e_{-b}.$$
Similar to the already defined curves $q \colon [0,1]\rightarrow\T^n$ and $p \colon [0,1]\rightarrow\R^n$, the dynamics corresponding to $V-V_b$ of the point $(q_0,p_0)$ is encapsulated by the curves $\tilde{q} \colon[0,1]\rightarrow\T^n$ and $\tilde{p} \colon [0,1]\rightarrow\R^n$ satisfying
\begin{equation}
\label{eq:tilde_q and tilde_p}
\left\{
\begin{alignedat}{2}
(\dot{\tilde{q}}(t),\dot{\tilde{p}}(t)) &= (\tilde{p}(t),-\nabla (V - V_b)(\tilde{q}(t))) \qquad && t \in \R, \\
(\tilde{q}(0),\tilde{p}(0)) &= (q_0, p_0). \qquad &&
\end{alignedat}
\right. 
\end{equation}
We compare the two dynamics in the following proposition.

\begin{prop}\label{lem:long range convergence of dynamics}
	Let $b\in\Z^n$ and $\delta>0$. There exists a $D_b>0$ such that for each $(q_0,p_0)\in \T^n\times\R^n$ satisfying $|b\cdot p_0|>D_b$, we have
		$$ d\left(\Phi_V(q_0,p_0),\Phi_{V-V_b}(q_0,p_0)\right)<\delta.$$
\end{prop}
\begin{proof}
	Note that the statement is vacuously true for any $D_b > 0$ if $b = 0$.
We therefore fix a nonzero $b\in\Z^n$. Throughout the proof, we use a variation of big O notation, expanding in the variable $\dt:=|b\cdot p_0|^{-1}$, uniformly in $q_0$. That is, we write $f(q_0,p_0)=\O(\dt^d)$ if there exist $N,C>0$ such that for all $q_0,p_0$ with $|b\cdot p_0|>N$ we have $|f(q_0,p_0)|\leq C|b\cdot p_0|^{-d}$.
	Therefore, to prove the proposition, it suffices to show that
	\begin{align}\label{suffice to show big O}
		d\left(\vect{q(1)}{p(1)},\vect{\tilde{q}(1)}{\tilde{p}(1)}\right) = \O(\dt).
	\end{align}
	Assume that $\dt\in(0,1)$. We divide the time interval $[0,1]$ into $m$ intervals of length $\dt$, where $m:=\lfloor \frac{1}{\dt}\rfloor$, and a final interval of length $1-m\dt$.
For each $t\in[0,\dt]$ and each $j\in\{0,\ldots, m\}$ (these will be the assumptions on $t$ and $j$ throughout the rest of the proof) let
\begin{equation*}
q^j(t) := q(j\dt + t), \quad p^j(t) := p(j\dt + t),
\end{equation*}
and define the curves $\tilde{q}^j$ and $\tilde{p}^j$ analogously.
Note that $(q^j,p^j)$ and $(\tilde{q}^j,\tilde{p}^j)$ satisfy the differential equations \eqref{eq:q and p} and \eqref{eq:tilde_q and tilde_p} respectively, but with different initial conditions.
Furthermore, for every $j$, we define the curve $\gamma^j \colon [0,\dt] \rightarrow \T^n$ as the unique solution to the initial value problem
\begin{equation}
\label{eq:the_middle_man_gamma}
\left\{
\begin{alignedat}{2}
(\dot{\gamma}^j(t),\ddot{\gamma}^j(t)) &= (\dot{\gamma}^j(t),-\nabla (V - V_b)(\gamma^j(t))) \qquad && t \in \R, \\
(\gamma^j(0),\dot{\gamma}^j(0)) &= (q^j(0), p^j(0)), \qquad &&
\end{alignedat}
\right. 
\end{equation}
where on the first line, we have emphasized the similarity of this equation with the equations \eqref{eq:q and p} and \eqref{eq:tilde_q and tilde_p} by including $\dot{\gamma}^j(t)$.
We do not introduce any special notation for $\dot{\gamma}^j$, however.

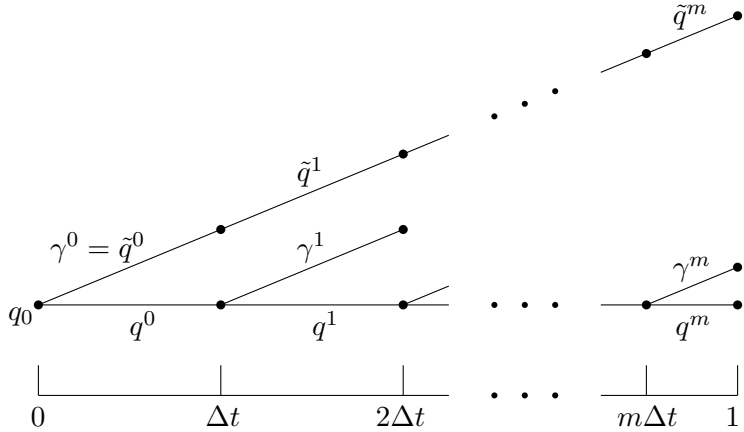
\begin{figure}[!htb]
	\centering
	\setlength{\unitlength}{0.2cm}
	\def \ds {0.6}
	\newcounter{counter}
	\setcounter{counter}{40}
	\def \x {16.66667} 
	\def \y {19.16667} 
	\begin{picture}(46,28)(0,-8) 

		\put(0,0){\circle*{\ds}}
		\put(0,0){\line(1,0){27}}
		\put(0,0){\line(12,5){27}}
			\put(12,5){\circle*{\ds}}
			\put(24,10){\circle*{\ds}}
		\put(12,0){\circle*{\ds}}
			\put(12,0){\line(12,5){12}}
			\put(24,5){\circle*{\ds}}
		\put(24,0){\circle*{\ds}}
			\put(24,0){\line(12,5){3}}

		\put(-2,-1){$q_0$}
		\put(6,-2){$q^0$}
		\put(18,-2){$q^1$}
		\put(5,3.2){\hspace{-2.2em}$\gamma^0=\tilde{q}^0$} 
		\put(17,8.2){$\tilde{q}^1$}
		\put(17,3.2){$\gamma^1$}
		\put(-0.5,-8){$0$}
		\put(11,-8){$\dt$}
		\put(22,-8){\hspace{0.1em}$2\dt$}
		\put(\thecounter,-8){\hspace{-1em}$m\dt$}
		\put(\thecounter,-2){\hspace{1em}$q^m$}
		\put(\thecounter,1.9){\hspace{0.9em}$\gamma^m$}
		\put(\thecounter,18.6){\hspace{0.9em}$\tilde{q}^m$}		
		\put(\thecounter,-8){\hspace{2.7em}$1$}
		
		\put(30,-6){\circle*{0.4}}
		\put(32,-6){\circle*{0.4}}
		\put(34,-6){\circle*{0.4}}
		\put(30,0){\circle*{0.4}}
		\put(32,0){\circle*{0.4}}
		\put(34,0){\circle*{0.4}}
		\put(30,12.5){\circle*{0.4}}
		\put(32,13.333){\circle*{0.4}}
		\put(34,14.167){\circle*{0.4}}
		
		\put(\thecounter,0){\line(-1,0){3}}
		\put(\thecounter,0){\circle*{\ds}}
		\put(\thecounter,0){\line(1,0){6}}
		\put(\thecounter,0){\line(12,5){6}}
		\put(\thecounter,\x){\line(-12,-5){3}}
		\put(\thecounter,\x){\circle*{\ds}}
		\put(\thecounter,\x){\line(12,5){6}}

\put(\thecounter,-6){\line(0,1){2}}
\addtocounter{counter}{6}
		\put(\thecounter,0){\circle*{\ds}}
		\put(\thecounter,2.5){\circle*{\ds}}
		\put(\thecounter,\y){\circle*{\ds}}
		
		\put(0,-6){\line(1,0){27}}
		\put(0,-6){\line(0,1){2}}
		\put(12,-6){\line(0,1){2}}
		\put(24,-6){\line(0,1){2}}
		\put(\thecounter,-6){\line(0,1){2}}
		\put(\thecounter,-6){\line(-1,0){9}}
	\end{picture}
	\caption{The position functions $q^j$,$\gamma^j$ and $\tilde{q}^j$. Sloping lines correspond to $V-V_b$, whereas the horizontal line that depicts $q$ corresponds to $V$.\label{fig: dynamica}}
\end{figure}	
	\noindent As depicted in Figure \ref{fig: dynamica}, the curve $\gamma^j\colon[0,\dt]\rightarrow\T^n$ plays a key role in 
comparing $q^j$ with $\tilde{q}^j$; the curve $(\gamma^j, \dot{\gamma}^j)$ is an integral curve along the same Hamiltonian vector field as $(\tilde{q}^j, \tilde{p}^j)$, but with the same initial conditions as $(q^j,p^j)$.
	
	We now expand our expressions in orders of $\dt$. Using equation \eqref{eq:q and p} and the fundamental theorem of calculus, we obtain
	\begin{align}\label{p_j}
		\norm{p^j(t)-p^j(0)}&\leq \int_0^{\dt}\norm{\nabla V (q^j(s))}\:ds\leq \supnorm{\nabla V}\dt
		=\O(\dt).
	\end{align}
	In particular, taking $t=\dt$, we get $\norm{p^{j+1}(0)-p^j(0)}=\O(\dt)$, and therefore by induction
	\begin{align}\label{pj(0)-p0}
		\norm{p^j(0)-p_0}=\O(1),
	\end{align}
	for every $0\leq j\leq m$. Equations \eqref{p_j} and \eqref{pj(0)-p0} give us
	\begin{align}\label{q^j}
		d(q^j(t),q^j(0)+tp_0) &\leq\norm{\int_0^{t}(p^j(s)-p_0)\:ds}\nonumber\\
		&\leq\int_0^{\dt}\norm{p^j(s)-p^j(0)}+\norm{p^j(0)-p_0}\:ds\nonumber\\
		&=\O(\dt).
	\end{align}
	A result similar to \eqref{p_j} exists for $\dot{\gamma}^j$ instead of $p^j$, and hence
	\begin{align}\label{p and Gamma dot}
		\norm{p^j(t)-\dot{\gamma}^j(t)}=\O(\dt),
	\end{align}
	which implies
	\begin{align}
		d(q^j(t),\gamma^j(t))&=\O(\dt^2).\label{q^j and Gamma^j}
	\end{align}
	Using the definitions of $V_b$ and $\dt$, we show that the distance between $p^j(\dt)$ and $\dot{\gamma}^j(\dt)$ is in fact of order $\dt^2$. We first note that
	\begin{align*}
		\norm{p^j(\dt)-\dot{\gamma}^j(\dt)}&=\norm{\int_0^{\dt}(\nabla V(q^j(s))-\nabla(V-V_b)(\gamma^j(s)))\:ds}\\
		&\leq\int_0^{\dt}\norm{\nabla(V-V_b)(q^j(s))-\nabla(V-V_b)(\gamma^j(s))}\:ds\\
		&\quad+\norm{\int_0^{\dt}\nabla V_b(q^j(s))\:ds}.
	\end{align*}
	By \eqref{q^j and Gamma^j}, the first term is $\O(\dt^3)$. For the second term we can use \eqref{q^j} and the observation that
		$$\int_0^{\dt}\nabla V_b(q^j(0)+sp_0)\:ds=0.$$
	Hence the second term is $\O(\dt^2)$. All in all, we obtain the estimate
	\begin{align*}
		\norm{p^j(\dt)-\dot{\gamma}^j(\dt)}=\O(\dt^2).
	\end{align*}
	This estimate, together with \eqref{q^j and Gamma^j}, implies
	\begin{align}\label{vect Gamma}
		d\left(\vect{\gamma^{j+1}(0)}{\dot{\gamma}^{j+1}(0)},\vect{\gamma^j(\dt)}{\dot{\gamma}^j(\dt)}\right)=d\left(\vect{q^{j}(\dt)}{p^{j}(\dt)},\vect{\gamma^j(\dt)}{\dot{\gamma}^j(\dt)}\right)=\O(\dt^2).
	\end{align}
	Since $\gamma^j$ and $\tilde{q}^j$ satisfy the same differential equation, say with associated Lipschitz constant $c$, Lemma \ref{lem:Gronwall} (with $f=g:(q,p)\mapsto(p,-\nabla(V-V_b)(q))$) implies that
	\begin{align}\label{Gamma and gamma time evolved}
		d\left(\vect{\gamma^j(t)}{\dot{\gamma}^j(t)},\vect{\tilde{q}^{j}(t)}{\tilde{p}^{j}(t)}\right)\leq e^{ct}d\left(\vect{\gamma^j(0)}{\dot{\gamma}^j(0)},\vect{\tilde{q}^j(0)}{\tilde{p}^j(0)}\right).
	\end{align}
	Taking $t=\dt$, we by definition have
	\begin{align}\label{vect Gamma and gamma}
		d\left(\vect{\gamma^j(\dt)}{\dot{\gamma}^j(\dt)},\vect{\tilde{q}^{j+1}(0)}{\tilde{p}^{j+1}(0)}\right)\leq e^{c\dt}d\left(\vect{\gamma^j(0)}{\dot{\gamma}^j(0)},\vect{\tilde{q}^j(0)}{\tilde{p}^j(0)}\right).
	\end{align}
	Combining \eqref{vect Gamma} and \eqref{vect Gamma and gamma}, we find that
	\begin{align*}
		d\left(\vect{\gamma^{j+1}(0)}{\dot{\gamma}^{j+1}(0)},\vect{\tilde{q}^{j+1}(0)}{\tilde{p}^{j+1}(0)}\right)\leq e^{c\dt}d\left(\vect{\gamma^j(0)}{\dot{\gamma^j}(0)},\vect{\tilde{q}^j(0)}{\tilde{p}^j(0)}\right)+\O(\dt^2).
	\end{align*}
	Because $e^{jc\dt}=\O(1)$, repeated use of the above equation gives
	\begin{align}\label{Gamma^m and gamma^m}
		d\left(\vect{\gamma^m(0)}{\dot{\gamma}^m(0)},\vect{\tilde{q}^{m}(0)}{\tilde{p}^{m}(0)}\right)=\O(\dt).
	\end{align}
	Let $t:=1-m\dt$. Using \eqref{Gamma and gamma time evolved}, we find
	\begin{align*}
		d\left(\vect{q(1)}{p(1)}, \vect{\tilde{q}(1)}{\tilde{p}(1)}\right)
		&\leq d\left(\vect{q(1)}{p(1)}, \vect{\gamma^m(t)}{\dot{\gamma}^m(t)}\right) + d\left(\vect{\gamma^m(t)}{
		\dot{\gamma}^m(t)}, \vect{\tilde{q}(1)}{\tilde{p}(1)}\right) \\
		&\leq d\left(q(1), \gamma^m(t)\right) + \|p(1) - \dot{\gamma}^m(t)\| \\
		&\quad+ e^{ct}
		d\left(\vect{\gamma^m(0)}{\dot{\gamma}^m(0)}, \vect{\tilde{q}^m(0)}{ \tilde{p}^m(0)}\right).
	\end{align*}
	The first term is $\O(\dt^2)$ by \eqref{q^j and Gamma^j}, the second is $\O(\dt)$ by \eqref{p and Gamma dot}, and the last term is $\O(\dt)$ by \eqref{Gamma^m and gamma^m}. This implies \eqref{suffice to show big O}, and thereby the proposition.
\end{proof}
\noindent Proposition \ref{lem:long range convergence of dynamics} expresses a property of the classical time evolution associated to a trigonometric potential in terms of points in phase space. To translate this result to the world of observables, we fix $\epsilon>0$ and notice that any $g\in\Cr(T^*\T^n)$ is uniformly continuous. Hence for every $b\in\mN$ we may fix a $D_b$ such that 
	\begin{align} \label{tau_V*g-tau_(V-Vb)*g}
		\sup_{x\in U_b}|\Phi_V^*g(x)-\Phi_{V-V_b}^*g(x)|\leq \varepsilon,
	\end{align}
	where
		\[U_b:=\T^n\times\set{x\in\R^n}{ |b\cdot x|>D_b}.\]
	We also define the open sets
	\begin{align*}
		W_b&:=\T^n\times\set{x\in\R^n}{ |b\cdot x|>2D_b}\,;\\
		U_\infty&:=\T^n\times\set{x\in\R^n}{ |b\cdot x|<4D_b\text{ for all }b\in\mN }\,;\\
		W_\infty&:=\T^n\times\set{x\in\R^n}{ |b\cdot x|<3D_b\text{ for all }b\in\mN },
	\end{align*}
	and remark that $\{U_i\}_{i\in I}$ and $\{W_i\}_{i\in I}$ are open covers satisfying $\overline{W_i}\subseteq U_i$ for all $i\in I:=\mN\cup\{\infty\}$. 
	Since we already know how $\Phi_V^*g$ approximately behaves on $\bigcup_{b\in\mN}U_b$, let us see how it behaves on $U_\infty$.
	
	\begin{lem}\label{lem:f_infty}
		There exists an $f_\infty\in \Cr(T^*\T^n)$ that equals $\Phi_V^*g$ on $U_\infty$.
	\end{lem}	
	\begin{proof}
		Let $S:=\spn_\R~\mN$. We write our phase space as a product of topological spaces
			$$\T^n\times\R^n=(\T^n\times S)\times S^\perp,$$
		and note that
			$$C_0(\T^n\times S) \hatotimes \Wr(S^\perp)$$
		is an ideal in $\Cr(T^*\T^n)$.
On the other hand, regarding our phase space as a coproduct of abelian Lie groups
			$$\T^n\times\R^n=(\T^n\times S) \oplus S^\perp,$$
		we define $\phi^t$ as the restriction of $\Phi^t_V$ to $\T^n\times S$ for each $t\in\R$. Because $\nabla V \perp S^\perp$, we have $\dot{p}(t)\perp S^\perp$, and hence
			$$\phi^t \colon\T^n\times S\rightarrow\T^n\times S\,$$
		is a well-defined homeomorphism. Moreover, we find the equation
		\begin{align*}
			\Phi^t_V(q,p_\|+p_\perp)&=\phi^t(q,p_\|)+(tp_\perp,p_\perp),\quad\text{for all}\quad p_\|\in S,~p_\perp\in S^\perp,
		\end{align*}
		because its two sides solve the same differential equation.
		Using the above relation in a straightforward calculation on generators, one can show that
			$$\Phi_V^*(C_0(\T^n\times S)\otimes\Wr(S^\perp))\subseteq C_0(\T^n\times S)\otimes\Wr(S^\perp).$$
		Actually, the same holds for $\Phi_V^{-1}$, which implies that $\Phi_V^*$ is a *-automorphism of the ideal $C_0(\T^n\times S)\hatotimes\Wr(S^\perp)$. Now note that $U_\infty$ is of the form $K \times S^\perp$ for some compact subset $K \subseteq \T^n\times S$. By Urysohn's lemma, we may choose a function $\tilde{g}\in C_0(\T^n\times S)\otimes\Wr(S^\perp)$ that is $1$ on $U_\infty$, and define $f_\infty:=\tilde{g}\cdot\Phi_V^*g$. We then find that
			$$f_\infty=((\tilde{g}\circ\Phi_V^{-1})\cdot g)\circ\Phi_V\in C_0(\T^n\times S)\hatotimes\Wr(S^\perp),$$
		and therefore $f_\infty\in\Cr(T^*\T^n)$.
	\end{proof}

\noindent We can finally prove that our commutative resolvent algebra is invariant under any time evolution corresponding to a trigonometric potential.

\begin{prop}\label{prop: time evolution}
	For every trigonometric potential $V\colon\T^n\rightarrow\R$ and $g\in \Cr(T^*\T^n)$ we have $\Phi_V^*g\in\Cr(T^*\T^n)$.
\end{prop}
\begin{proof}
We use induction on the size of $\mN$ in $V=\sum_{b\in\mN} a_be_b$ (while assuming that $\mN$ is chosen minimally). The induction base is precisely Lemma \ref{lem:free time evolution}.

We now carry out the induction step.
The induction hypothesis says that time evolution with respect to $V-V_b$ preserves $\Cr(T^*\T^n)$, for each $b\in\mN$. Therefore, writing $f_b:=\Phi_{V-V_b}^*g$, we have $f_b\in\Cr(T^*\T^n)$.
	Fixing $f_\infty$ as in Lemma \ref{lem:f_infty}, we have $f_i\in\Cr(T^*\T^n)$, and equation \eqref{tau_V*g-tau_(V-Vb)*g} implies that 
	\begin{align}\label{sup on U_i}
		\supnorm{f_i|_{U_i}-\Phi_V^*g|_{U_i}}<\epsilon,
	\end{align}
	for each $i\in I=\mN\cup\{\infty\}$.
	We now construct a partition of unity $\{\eta_i\}$ subordinate to the open cover $\{U_i\}$ of $\T^n \times \R^n$, to patch together the functions $\{f_i\}$ and obtain a single function in $\Cr(T^*\T^n)$.
	We start by defining nonnegative functions $\zeta_i\in \Cr(T^*\T^n)$ that are $1$ on $W_i$ and $0$ outside of $U_i$.
Explicitly, for each $b \in \mathcal{N}$, we take $\zeta_b := \indicator_{\T^n} \otimes (g_b\circ P_{\spn(b)})$ for some bump function $g_b$ on $\spn(b)$, and we take $\zeta_\infty := \indicator_{\T^n} \otimes (g_\infty \circ P_S)$ for some bump function $g_\infty$ on $S$. Because $\{W_i\}$ is a cover of $\T^n \times \R^n$, the sum $\sum_i\zeta_i\in\Cr(T^*\T^n)$ is bounded from below by 1, hence it is invertible in $\Cr(T^*\T^n)$, and therefore every function
		$$\eta_i:=\frac{\zeta_i}{\sum_j\zeta_j},$$
	also lies in $\Cr(T^*\T^n)$. Now \eqref{sup on U_i} gives us $$\supnorm{\Phi_V^*g-\sum_i f_i\eta_i}<\varepsilon.$$
Since $\varepsilon > 0$ was arbitrary and $\Cr(T^*\T^n)$ is norm-closed, the assertion follows.
\end{proof}

\subsection{Arbitrary potentials}\label{sct:arbitrary potentials}

\noindent 
Having covered the trigonometric case, we now wish to tackle the general case.
The following lemma provides the required approximation of a generic potential by trigonometric ones.

\begin{lem}\label{lem:approximating nabla V}
	Let $V \in C^1(\T^n)$.
Then there exists a sequence $(V_m)_{m=1}^\infty$ of trigonometric polynomials such that $(\nabla V_m)_{m=1}^\infty$ converges uniformly to $\nabla V$.
Furthermore, if $V$ is real-valued, then every $V_m$ can be chosen to be real-valued as well.
\end{lem}

\begin{proof}
We construct the sequence $(V_m)$ by convolving $V$ with the $n$-dimensional analogues of the family of \textit{Fej\'er kernels}.
We first recall that for each $m \geq 1$, the $m$-th Fej\'er kernel is given by
\begin{equation*}
F_{1,m} \colon \T \rightarrow \R, \quad 
q=[x] \mapsto \frac{1}{m} \sum_{k = 0}^{m - 1} \sum_{j = -k}^k e^{2\pi i j x}
= \frac{1}{m} \frac{\sin^2(\pi m x)}{\sin^2(\pi x)} \, ,
\end{equation*}
where the most right expression in this definition is understood to be equal to $m$ for $x=0$.
The sequence $(F_{1,m})_{m \geq 1}$ is an approximation to the identity, i.e., for every continuous function $f$ on $\T$, the sequence $(F_{1,m} \ast f)_{m \geq 1}$ converges uniformly to $f$, where $\ast$ denotes the operation of convolution of functions \cite[Sections 2.4 and 2.5.2]{stein03}.

Next, we define the $n$-dimensional analogues of these functions:
\begin{equation*}
F_{n,m} \colon \T^n \rightarrow \R, \quad 
q = (q_1,\ldots,q_n) \mapsto \prod_{l = 1}^n F_{1,m}(q_l) \, .
\end{equation*}
Using the corresponding fact for one-dimensional kernels, it is elementary to show that the sequence $(F_{n,m})_{m \geq 1}$ is an approximation to the identity.

We now define
\begin{equation*}
V_m := F_{n,m} \ast V \, ,
\end{equation*}
for each $m \geq 1$.
Because every $F_{n,m}$ is trigonometric, and $e_b*f=\hat{f}(b)e_b$ for every $f\in C(\T^n)$ and $b\in\Z^n$, the sequence $(V_m)_{m \geq 1}$ consists of trigonometric polynomials.
Moreover, by a general property of convolutions, we have
\begin{equation*}
\frac{\partial V_m}{\partial q_l}
= \frac{\partial}{\partial q_l} (F_{n,m} \ast V)
= F_{n,m} \ast \frac{\partial V}{\partial q_l},
\end{equation*}
and since $(F_{n,m})_{m \geq 1}$ is an approximation to the identity, the right-hand side converges uniformly to $\frac{\partial V}{\partial q_l}$ as $m \to \infty$, for $l = 1,\ldots,n$.
It follows that $(\nabla V_m)_{m \geq 1}$ converges uniformly to $\nabla V$.
The final assertion is a consequence of the fact that the family of Fej\'er kernels (as well as its higher-dimensional analogues) consists of real-valued functions.
\end{proof}

\noindent We now extend Proposition \ref{prop: time evolution} to general $V$, thereby arriving at our final result.

\begin{thm}\label{thm:classical time evolution}
	Let $V\in C^1(\T^n)_\sa$, and suppose that $\nabla V$ is Lipschitz continuous. Then we have 
		$$(\Phi^t_V)^*(\Cr(T^*\T^n))=\Cr(T^*\T^n),$$
	for every $t\in\R$.
\end{thm}

\begin{proof}
	It suffices to show that $(\Phi^t_V)^*(\Cr(T^*\T^n)) \subseteq \Cr(T^*\T^n)$; we can replace $t$ by $-t$ and note that $(\Phi_V^{-t})^\ast$ is the inverse of $(\Phi_V^t)^\ast$ to obtain the reverse inclusion.
By Lemma \ref{lem:t=1 or arbitrary t}, we may assume without loss of generality that $t = 1$.

Let $g \in \Cr(T^*\T^n)$.
By Lemma \ref{lem:approximating nabla V}, there exists a sequence of trigonometric potentials $(V_m)$ on $\T^n$ such that $(\nabla V_m)$ converges uniformly to $\nabla V$.
We show that this implies that $(\Phi_{V_m}^\ast (g))$ converges uniformly to $\Phi_V^\ast (g)$; since $\Phi_{V_m}^\ast (g) \in \Cr(T^*\T^n)$ by Proposition \ref{prop: time evolution} and since $\Cr(T^*\T^n)$ is norm-closed, the theorem will follow from this.

Let $\varepsilon > 0$, and let $c$ be the Lipschitz constant of $(q,p)\mapsto(p,-\nabla V(q))$.
Since $g$ is uniformly continuous, there exists $\delta > 0$ such that $|g(x) - g(y)| < \varepsilon$ for each $x,y \in \T^n \times \R^n$ with $d(x,y) < \delta$.
By assumption, there exists an $N \in \N$ such that for each $m \geq N$, we have $\|\nabla V - \nabla V_m\|_\infty < \delta e^{-c}$.
It follows from Lemma \ref{lem:Gronwall} that $d(\Phi_V(x),\Phi_{V_m}(x)) < \delta$ for each $x \in \T^n \times \R^n$ and each $m \geq N$, hence $\|\Phi_V^\ast(g) - \Phi_{V_m}^\ast(g)\|_\infty \leq \varepsilon$.
Thus $(\Phi_{V_m}^\ast (g))$ converges uniformly to $\Phi_V^\ast (g)$, as desired.
\end{proof}

\section{Quantum time evolution}
\label{sec:quantum time evolution}
\noindent 
Our next task is to show that $A_\hbar=C^*\!\left(\QW(\Sr(T^*\T^n))\right)$ is invariant under time evolution for each Hamiltonian with potential $V \in C(\T^n)$.
The general proof strategy resembles that of Buchholz and Grundling in \cite[Proposition 6.1]{BG}.
However, the present setting differs from theirs in two important ways, each of which introduces its own technical problems.
First of all, our configuration space is $\T^n$ rather than $\R^n$.
Secondly, we consider the problem of invariance under time evolution for arbitrary $n \in \N$, whereas Buchholz and Grundling only discuss the case $n = 1$.
We start with the simplest type of time evolution:

\begin{lem}\label{lem:free quantum time evolution}
Let $\hbar > 0$.
The algebra $A_\hbar$ is closed under the quantum time evolution corresponding to the free Hamiltonian $H_0$ that is the unique self-adjoint extension of the essentially self-adjoint operator $-\frac{\hbar^2}{2} \sum\frac{d^2}{dx_j^2}$ with domain $C^\infty(\T^n)$.
\end{lem}


\begin{proof}
We show that the quantum time evolution corresponding to $H_0$ maps the set of quantizations of the generators $e_b \otimes h_{U,\xi,g}$ of $C_{\mathcal{R}}(T^\ast \T^n)$ into itself; since the time evolution consists of a family of automorphisms of C*-algebras, the lemma will follow from this.

Let $e_b \otimes h_{U,\xi,g}$ be such a generator.
Note that for each $a \in \Z^n$, we have
\begin{equation}\label{eq:exp of H0 on eigenvector}
e^{-\frac{itH_0}{\hbar}} \psi_a
= e^{-2\pi^2 i t \hbar \|a\|^2}\psi_a.
\end{equation}
Using \eqref{eq:formula}, we obtain
\begin{align*}
&e^{\frac{itH_0}{\hbar}}\mathcal{Q}^\textnormal{W}_\hbar\!\left(e_b\otimes h_{U,\xi,g}\right) e^{-\frac{itH_0}{\hbar}} \psi_a \\
&\quad = e^{2\pi^2 i t \hbar (\|a + b\|^2 - \|a\|^2)} e^{2\pi \hbar i(a+b/2) \cdot \xi} g \circ P_U(2\pi \hbar(a+\tfrac12 b)) \psi_{a+b} \\
&\quad = e^{2\pi i\hbar (a+b/2) \cdot (\xi + 2 \pi t b)} g \circ P_U(2\pi \hbar(a+\tfrac12 b))\psi_{a+b} \\
&\quad = \mathcal{Q}^\textnormal{W}_\hbar \!\left(e_b \otimes h_{U,\tilde{\xi},\tilde{g}}\right) \psi_a,
\end{align*}
for each $a \in \Z^n$, where
\begin{equation*}
\tilde{\xi} := \xi + 2 \pi t P_{U^\perp}(b) \in U^\perp,
\end{equation*}
and
\begin{equation*}
\tilde{g} \colon U \rightarrow \C, \quad 
p \mapsto e^{2 \pi i t P_U(b) \cdot p} g(p),
\end{equation*}
is again a Schwartz function on $U$, so $e_b \otimes h_{U,\tilde{\xi},\tilde{g}}$ is a generator of $C_{\mathcal{R}}(T^\ast \T^n)$.
It follows that the set of generators of $A_\hbar$ is indeed invariant under the free quantum time evolution.
\end{proof}

\begin{rema}
Comparing the proof of Lemma \ref{lem:free quantum time evolution} with the proof of the analogous Lemma \ref{lem:free time evolution}, we see that (for $t=1$) $\tilde{\xi}$ and $\tilde{g}$ are both the same. Indeed, one can easily obtain $$\QW\circ (\Phi^t_0)^*=\tau_t^0\circ\QW,$$ which is analogous to a known result for Weyl quantization on $\R^{2n}$ (proved in higher generality in \cite[Theorem II.2.5.1]{landsman98}). There is generally no such result for non-free time evolution.
\end{rema}
\noindent
In order to deal with the general quantum time evolution, we recall some basic theory about lattices that we need due to the appearance of the lattice $\Z^n$ in $\T^n=\R^n/\Z^n$.
A set of linearly independent vectors $v_1,\ldots,v_k$ in a lattice $\Lambda$ is called \textit{primitive in $\Lambda$} if $\spn_\Z(v_1,\ldots,v_k)=\spn_\R(v_1,\ldots,v_k)\cap\Lambda$.
For instance, every $\Z$-basis of a lattice $\Lambda$ is primitive in $\Lambda$. Furthermore, we have the following result:

\begin{lem}\label{lem: extend primitive set}
	Let $\Lambda \subset \R^m$ be a lattice.
Every primitive set $v_1,\ldots,v_k$ in $\Lambda$ can be extended to a $\Z$-basis $v_1,\ldots,v_k,v_{k + 1},\ldots,v_m$ of $\Lambda$.
\end{lem}

\begin{proof}
	This is exactly \cite[\textsection 1.3, Theorem 5]{Lekkerkerker69}.
\end{proof}

\noindent 
This will help us prove the main theorem of this section:

\begin{thm}\label{thm:quantum time evolution}
Let $V \in C(\T^n)_\sa$, and define the self-adjoint operator $H := H_0 + M_V$ with domain $\dom H_0$ (see Lemma \ref{lem:free quantum time evolution}). I.e., $$H\psi=-\frac{\hbar^2}{2} \sum\frac{d^2\psi}{dx_j^2}+V\psi$$ for $\psi\in C^\infty(\T^n)$.
Let $\big(e^{\frac{-itH}{\hbar}}\big)_{t \in \R}$ be the corresponding one-parameter group implementing the quantum mechanical time evolution on $L^2(\T^n)$, and let $(\tau_t)_{t \in \R}$ be the associated one-parameter group of automorphisms on $\B(L^2(\T^n))$.
Then $$\tau_t(A_\hbar)=A_\hbar$$ for all $t\in\R$.
\end{thm}

\begin{proof}
We claim that for each $t \in \R$, we have
\begin{equation*}
e^{\frac{itH_0}{\hbar}} e^{\frac{-itH}{\hbar}} \in A_\hbar .
\end{equation*}
Suppose for the moment that this claim holds true.
Then for each $O \in A_\hbar$ and each $t \in \R$, we have
\begin{equation*}
\tau_t(O)
= e^{\frac{itH}{\hbar}} O e^{\frac{-itH}{\hbar}}
= \left( e^{\frac{itH_0}{\hbar}} e^{\frac{-itH}{\hbar}} \right)^\ast \tau^0_t(O) \left( e^{\frac{itH_0}{\hbar}} e^{\frac{-itH}{\hbar}} \right).
\end{equation*}
By assumption, the first and the third factors of the right-hand side (those within parentheses) are elements of $A_\hbar$, and the second factor is an element of $A_\hbar$ by Lemma \ref{lem:free quantum time evolution}.
It then follows that $\tau_t(O) \in A_\hbar$.

Thus it remains to prove the claim.
As in the proof of \cite[Proposition 6.1]{BG}, we use the fact that the product of two of the elements of the different one parameter groups can be written as a norm-convergent Dyson series, i.e.,
\begin{equation}
e^{\frac{itH_0}{\hbar}} e^{\frac{-itH}{\hbar}}
= \sum_{m = 0}^\infty (i\hbar)^{-m} \int_0^t \int_0^{t_1} \dots \int_0^{t_{m - 1}} \tau^0_{t_1} (M_V) \cdots \tau^0_{t_m}( M_V ) \: dt_m \cdots dt_2 \: dt_1.
\label{eq:Dyson_series}
\end{equation}
The integrals in the above expression can be defined in the following way.
First, observe that the function
\begin{equation*}
\R \rightarrow \B(L^2(\T^n)) , \quad 
t \mapsto \tau^0_t (M_V) ,
\end{equation*}
is bounded and strongly continuous.
It follows that the function
\begin{equation*}
\R^m \rightarrow \B(L^2(\T^n)) , \quad 
(t_1, \ldots, t_m) \mapsto \tau^0_{t_1} (M_V) \cdots \tau^0_{t_m} (M_V) ,
\end{equation*}
is bounded and strongly continuous.
For each $\psi \in L^2(\T^n)$, one can therefore define the integral
\begin{equation}\label{Dyson term}
\int_0^t \int_0^{t_1} \dots \int_0^{t_{m - 1}} \tau^0_{t_1}( M_V) \cdots \tau^0_{t_m}(M_V)\psi \: dt_m \cdots dt_2 \: dt_1 ,
\end{equation}
using Bochner integration, and it is easy to check that the norm of the corresponding operator is less than or equal to $(m!)^{-1} |t|^m \|V\|_\infty^m$, so that the Dyson series is indeed norm-convergent.
As in \cite{BG}, because \eqref{Dyson term} is continuous in $V$ it suffices to prove the claim for potentials $V$ that lie in a dense subset of $C(\T^n)$. If we assume that $V$ is in the span of $\set{e_b}{ b\in\Z^n}$, we can write \eqref{Dyson term} as a sum of relatively explicit expressions.
Thus, we are left to show that for each $t \in \R$ and each $b_1,\ldots,b_m \in \Z^n$, the operator
\begin{equation*}
O:=\int_0^t \int_0^{t_1} \dots \int_0^{t_{m - 1}} \tau^0_{t_1} (M_{e_{b_1}}) \cdots \tau^0_{t_m}( M_{e_{b_m}} ) \: dt_m \cdots dt_1 ,
\end{equation*}
lies in $A_\hbar$. 
A quick computation using \eqref{eq:exp of H0 on eigenvector} gives us
\begin{align*}
	\tau^0_t(M_{e_b})\psi_a
&= M_{e_b}e^{2\pi^2 i\hbar (\norm{a+b}^2-\norm{a}^2)}\psi_a\\
	&= M_{e_b} e^{2\pi^2it\hbar \norm{b}^2}e^{4\pi^2it\hbar b\cdot a}\psi_a,
\end{align*}
which shows that, for any $\psi\in L^2(\T^n)$ and $[x]\in\T^n$, we have
\begin{align*}
	(\tau^0_t(M_{e_b})\psi)[x] = e^{2\pi ix\cdot b} e^{2\pi^2i\hbar t\norm{b}^2}\psi \left[ x + 2\pi\hbar tb \right] .
\end{align*}
Applying this formula many times, we find a function $f_0\in C_\textnormal{b}(\R^m)$ that takes values on the unit circle such that
\begin{align*}
	\tau^0_{t_1} (M_{e_{b_1}}) \cdots \tau^0_{t_m}( M_{e_{b_m}} )\psi[x] = e^{2\pi i x\cdot\sum b_i} f_0(t_1,\ldots,t_m)\psi\left[x+2\pi\hbar\sum t_ib_i\right].
\end{align*}
The operator $O$ looks like an integral operator, in the sense that we perform an integral over the variables $t_i$ that appear as $\sum t_i b_i$ in the argument of $\psi$. However, the $b_i$'s may both fail to constitute a linearly independent and a complete set of vectors in $\R^n$. Still, we can relate $O$ to an integral operator, which will be the subject of the rest of the proof. 

We use a special case of Lemma \ref{lem: extend primitive set} (extending an empty primitive set) to find a $\Z$-basis $v_1,\ldots,v_k$ of $\spn_\R(b_1,\ldots,b_m)\cap\Z^n$. Because the $b_i$'s are integral, this is also an $\R$-basis of $\spn_\R(b_1,\ldots,b_m)$. Expressing the $b_i$'s in terms of $v_j$'s as
	$$ b_i = \sum_{j=1}^k c_{ij}v_j,$$
we obtain
\begin{align*}
	\psi\bigg[x+2\pi\hbar\sum_{i=1}^m t_i b_i\bigg]&=\psi\bigg[x+2\pi\hbar\sum_{j=1}^k\sum_{i=1}^m t_i c_{ij}v_j\bigg]\\
	&= \psi\bigg[x+2\pi\hbar\sum_{j=1}^k T_0(t_1,\ldots, t_m)_jv_j\bigg],
\end{align*}
for a unique surjective linear map $T_0 \colon \mathbb{R}^m \rightarrow \mathbb{R}^k$. By surjectivity, the map $T_0$ admits a lift to an invertible linear map $T \colon \mathbb{R}^m \rightarrow \mathbb{R}^m$ with respect to the projection $\mathbb{R}^m \rightarrow \mathbb{R}^k$ onto the first $k$ coordinates. Fix such a $T$, and perform a change of variables, replacing $(t_1,\ldots, t_m)$ with $T^{-1}(s)$. We get
\begin{align*}
	O\psi[x] = e^{2\pi i x\cdot\sum b_i}\left|\det T\right|^{-1}\int_K f_0\big(T^{-1}s\big)\psi\bigg[x+2\pi\hbar\sum_{j=1}^k s_j v_j\bigg] \:ds,
\end{align*}
for some compact subset $K\subseteq \R^m$. Let $K'$ be the image of $K$ under the projection $\R^m\rightarrow\R^k$ onto the first $k$ coordinates, and define the function $f_1\colon\R^k\rightarrow \C$ by
\begin{align*}
	f_1\colon s_{(1)}\mapsto \left|\det T\right|^{-1}\int_{\R^{m-k}}\indicator_{K}(s_{(1)}\oplus s_{(2)})f_0\big(T^{-1}(s_{(1)}\oplus s_{(2)})\big)\:ds_{(2)}.
\end{align*}
One easily finds that $f_1\in L^\infty(\R^k)$. We are now left with the integral
\begin{align*}
	O\psi[x]=e^{2\pi ix\cdot\sum b_i}\int_{K'} f_1(s)\psi\bigg[x+2\pi\hbar\sum_{j=1}^k s_j v_j\bigg]\:ds.
\end{align*}
We want to relate the above integral to an integral over the first $k$ components in $\T^n$.
For this purpose, we apply Lemma \ref{lem: extend primitive set} once more to extend $v_1,\ldots,v_k$ to a $\Z$-basis $v_1,\ldots,v_n$ of $\Z^n$, and let $S$ be the matrix whose columns are the vectors $v_1,\ldots,v_n$.
Since $S$ and its inverse are matrices in $GL_n(\Z)$, we find that $\det S=\pm1$.
Moreover, $S$ induces the group automorphism $[x] \mapsto [Sx]$ of $\T^n$, which we can pull back to the unitary map
	\[U\colon L^2(\T^n)\rightarrow L^2(\T^n),\quad U\psi[x]:=\psi[Sx],\]
for which it is straightforward to check (on generators of $A_\hbar$) that $U^{-1} A_\hbar 
U\subseteq A_\hbar$. For $\varphi=\varphi_1\otimes\varphi_2\in L^2(\T^k)\otimes L^2(\T^{n-k})$ we have, denoting $b:=\sum_i b_i$,
\begin{align*}
	U M_{e_{-b}}OU^{-1}\varphi[x]
	&= \int_{K'} f_1(s)U^{-1}\varphi\bigg[S(x)+2\pi\hbar\sum_{j=1}^k s_j S(e_j)\bigg]\:ds\\
	&=\int_{K'} f_1(s)\varphi\left[x+2\pi\hbar(s\oplus0)\right]\:ds\\
	&= \int_{K'} f_1(s)\varphi_1\left(x_{(1)}+2\pi\hbar s+\Z^k\right)\varphi_2\left(x_{(2)}+\Z^{n-k}\right)\:ds\\
	&=\int_{\T^k} f_2\left(x_{(1)}+\Z^k,s\right)\varphi_1(s)\:ds\:\varphi_2\left(x_{(2)}+\Z^{n-k}\right),
\end{align*}
where $x=x_{(1)}\oplus x_{(2)}$ and $f_2\in L^\infty(\T^k\times\T^k) \subseteq L^2(\T^k\times\T^k)$ denotes the function
\begin{equation*}
f_2(r,s)
:= \sum_{M\in\Z^k} f_1\left(\frac{\iota(s-r)+M}{2\pi\hbar}\right),
\end{equation*}
where $\iota$ denotes the canonical map $\T^k\rightarrow[0,1)^k$.
Note that the above sum has only finitely many nonzero terms since $f_1$ is compactly supported.

In conclusion, we have proved that
\begin{align*}
	O = M_{e_{b}}U^{-1}(F\otimes\1)U,
\end{align*}
for a Hilbert--Schmidt integral operator $F\in \mK(L^2(\T^k))$ with kernel $f_2\in L^2(\T^k\times\T^k)$. By part (3) of Proposition \ref{prop:Weyl quantization properties}, any compact operator, like $F$, is inside the resolvent algebra on $\T^k$. By part (4) of Proposition \ref{prop:Weyl quantization properties}, this implies that $F\otimes\1\in A_\hbar$, and hence $U^{-1}(F\otimes\1)U\in A_\hbar$. As $M_{e_{b}}$ is the quantization of $e_{b}\otimes \indicator_{\R^n}$, we find $O\in A_\hbar$. As we have seen, linearity and continuity of the Dyson series imply that $e^{\frac{itH_0}{\hbar}}e^{\frac{-itH}{\hbar}}\in A_\hbar$, and this implies the theorem itself.
\end{proof}


\section{Weyl quantization on $T^*\T^n$ is almost strict}
\label{sct:almost strict}
The following definition is equivalent to \cite[Definition II.1.1.2]{landsman98}, other definitions of strict deformation quantization are reviewed in \cite[Section 2]{Hawkins}.
\begin{defi}
Let $\A_0$ be a complex Poisson algebra densely contained in a C*-algebra $A_0$, satisfying $\{f,g\}^*=\{f^*,g^*\}$. A \textbf{strict deformation quantization} of $\A_0$ consists of a subset $I\subseteq\R$ with $0\in I\cap\overline{I\setminus\{0\}}$, a collection of C*-algebras $\{A_\hbar\}_{\hbar\in I}$ (with norms $\|\cdot\|_\hbar$) and a collection of injective linear *-preserving maps $Q_\hbar:\A_0\to A_\hbar$ ($\hbar\in I$) such that $Q_0$ is the identity map, $Q_\hbar(\A_0)$ is a dense *-subalgebra of $A_\hbar$ ($\hbar\in I$) and for all $f,g\in\A_0$:
\begin{align*}
		\lim_{\hbar\to0}\norm{Q_\hbar(f)Q_\hbar(g)-Q_\hbar(fg)}=0&\qquad\textbf{(von Neumann's condition);}\\				\lim_{\hbar\to0}\norm{(-i\hbar)^{-1}[Q_\hbar(f),Q_\hbar(g)]-Q_\hbar(\{f,g\})}=0&\qquad\textbf{(Dirac's condition)};\\
		\text{the map}\quad I\to\R,\quad\hbar\mapsto\norm{Q_\hbar(f)}\quad\text{is continuous}&\qquad\textbf{(Rieffel's condition).}
	\end{align*}
\end{defi}

Looking at this definition, we immediately see that $Q_\hbar=\QW$ is not a strict deformation quantization, because $\QW$ is not injective. One does however have many of the above properties. The following theorem is proven in \cite[Theorem 22]{vNS20} and \cite[Theorem 7.8]{Stienstra}.
\begin{thm}\label{thm:almost strict}
	Let $I:=[0,\infty)$. Then, except for continuity at $\hbar>0$, the triple
		$$\left(I,\{A_\hbar\}_{\hbar\in I},\{\QW:\Sr(T^*\T^n)\to A_\hbar\}_{\hbar\in I}\right),$$
	is a strict quantization of the Poisson algebra $\A_0=\Sr(T^*\T^n)$, i.e., $\QW$ is a linear *-preserving map such that $\mathcal{Q}^\textnormal{W}_0$ is the identity map, $\QW(\A_0)$ is a dense *-subalgebra of $A_\hbar$, both von Neumann's condition and Dirac's condition hold, and the map
	$$ I\to\R,\quad\hbar\mapsto\norm{\mathcal{Q}_\hbar(f)}$$
	is continuous at $0$.
\end{thm}
The above theorem justifies us in calling the map $\QW$ a strict quantization map, because, when defined on the index set $I:=\{0\}\cup 1/\N$, it is an actual strict quantization as defined in \cite[Definition II.1.1.1]{landsman98}. This palliative discretizing of $I$ is also needed in geometric quantization, as discussed in \cite{Hawkins}, and is needed here because of the following lemma.

\begin{lem}\label{lem:failure of Rieffel's condition}
Rieffel's condition away from zero does not hold for $\QW$ defined on an index set $I$ that contains an open subset. 
\end{lem}
\begin{proof}
Let $\hbar_0 > 0$ be arbitrary, and consider the function $f = e_0 \otimes h$, where the function $h$ is defined as follows:
\begin{equation*}
h:\R^n\to\R, \quad 
p = (p_1,p_2,\ldots, p_n) \mapsto \sin \left( \frac{p_1}{\hbar_0} \right).
\end{equation*}
Note that $h$ can be written as the sum of two generators of $\Weylalg{\R^n} \subseteq \Teunalg{\R^n}$, so $f \in \mathcal{S}_{\mathcal{R}}(T^*\T^n)$.
Furthermore, $h$ vanishes at each point in $2 \pi \hbar_0 \cdot \Z^n$, hence $\mathcal{Q}^\textnormal{W}_{\hbar_0}(f) = 0$ by the definition \eqref{eq:formula}, or equivalently, $\|\mathcal{Q}^\textnormal{W}_{\hbar_0}(f)\| = 0$.
On the other hand, for each $N \in \N \backslash \{0\}$, let
\begin{equation*}
\hbar_N := \hbar_0 \left( 1 + \frac{1}{4N} \right).
\end{equation*}
Then $\|\mathcal{Q}^\textnormal{W}_{\hbar_N}(f)\| = 1$; indeed, we have $\|\mathcal{Q}^\textnormal{W}_{\hbar_N}(f)\| \leq \|h\|_\infty = 1$, and equality holds since
\begin{equation*}
\mathcal{Q}^\textnormal{W}_{\hbar_N}(f) \psi_{(N,0,0,\ldots,0)}
= \psi_{(N,0,0,\ldots,0)}\,.
\end{equation*}
Thus, while $\lim_{N \to \infty} \hbar_N = \hbar_0$, we also have
\begin{equation*}
\lim_{N \to \infty} \|\mathcal{Q}^\textnormal{W}_{\hbar_N}(f)\|
= 1 \neq 0
= \|\mathcal{Q}^\textnormal{W}_{\hbar_0}(f)\|,
\end{equation*}
so the function $\hbar \to \|\mathcal{Q}^\textnormal{W}_\hbar(f)\|$ fails to be continuous at $\hbar_0$.
\end{proof}

In the next chapter, we will reveal the failure of both injectivity and Rieffel's condition as being artifacts of the regularization procedure that tries to describe a gauge theory on a finite lattice. The above remarks about strict (deformation) quantization play a key role in obtaining the correct continuum limit, and, vice versa, the continuum limit will solve the problems addressed above in a remarkable way.

\chapter{Strict Deformation Quantization of Abelian Lattice Gauge Fields}
\chaptermark{S.D.Q. of Abelian Lattice Gauge Fields}
\label{ch:SDQ}

	This chapter, adapted from \cite{vN21}, shows how to construct classical and quantum field C*-algebras modeling a $U(1)^n$-gauge theory in any dimension using a novel approach to lattice gauge theory, while simultaneously constructing a strict deformation quantization between the respective field algebras.
	The construction starts with quantization maps defined on operator systems (instead of C*-algebras) associated to the lattices, in a way that quantization commutes with all lattice refinements, therefore giving rise to a quantization map on the continuum (meaning ultraviolet and infrared) limit. Although working with operator systems at the finite level, in the continuum limit we obtain genuine C*-algebras. We also prove that the C*-algebras (classical and quantum) are invariant under time evolutions related to the electric part of abelian Yang--Mills. Our classical and quantum systems at the finite level are essentially the ones of Chapter \ref{ch:cylinder}, which admit completely general dynamics, and we briefly discuss ways to extend this powerful result to the continuum limit. 

\section{Introduction}

As mentioned in the introduction, constructing a field algebra for the continuum limit of a quantum abelian lattice gauge theory roughly comes down to two problems, firstly to define the system at the finite level, and secondly to define the limit.

The resolvent algebra on the torus introduced in Chapter \ref{ch:cylinder} seems to be the ideal field algebra on the finite level. If the gauge group is $\T^n=U(1)^n$ and the lattice has $k$ edges, the resolvent algebra on $\T^{nk}$ is a C*-subalgebra of $\mB(L^2(\T^{nk}))$ containing $\mK(L^2(\T^{nk})$ and a copy of the crossed product algebra $C(\T^{nk})\rtimes \T^{nk}$, while being preserved under time evolutions by Theorem \ref{thm:quantum time evolution}. Moreover, like the crossed product algebra it contains, the resolvent algebra on $\T^{nk}$ admits a gauge group action, admits embedding maps related to addition of edges, and has a reasonably good notion of a classical limit. However, Section \ref{sct:almost strict} also uncovered a few problems with that same classical limit, namely that the respective quantization map lacks injectivity as well as Rieffel's condition, and therefore does not define a strict deformation quantization. 

	With regards to constructing a continuum limit, another challenging problem has been identified in \cite{Stienstra}. Whereas the embedding map for adding an edge to the lattice is easily defined by construction of the resolvent algebra, there are severe obstructions against a *-homomorphic embedding map for the subdivision of an edge, as explained in \cite[pages 247--249]{Stienstra}. These problems are more or less independent from the field algebra one chooses at the finite level; one can easily see that the same problems arise for the crossed product algebra $C(\T^{nk})\rtimes \T^{nk}$, and that the situation for $\mK(L^2(\T^{nk}))$ is even worse.

 The current chapter solves all of the above problems simultaneously, by letting go of the need for multiplicativity of the embedding maps.
  On each lattice, we restrict ourselves to a subspace of the classical C*-algebra on which the quantization map \eqref{eq:formula} is injective. This subspace and its image under quantization turn out to be only operator systems, and not algebras. At first, this appears to distance us from the powerful C*-algebraic approach. However, on these operator systems, both the classical and quantum embedding maps are now naturally defined and commute with quantization. Moreover, the ensuing limit of operator systems turns out be a *-algebra lying dense in a C*-algebra, thus recovering the C*-algebraic approach.   
  
  This `operator systemic' method has many advantages. The obtained quantum embedding maps respect the gauge action, which becomes very important when one wishes to make the step from field algebras to observable algebras.
  Moreover, in the continuum limit, the quantum and classical theory behave even better than in the case on the lattice, in the sense that they form a strict deformation quantization, satisfying all conditions of \cite[Definition II.1.1.1 and II.1.1.2]{landsman98}.

For these reasons, the operator systemic method seems to improve upon the existing literature. In most operator algebraic approaches to lattice gauge theory (e.g., \cite{ASS,BS19a,BS19b,GR,Stienstra,ST}) one uses inductive limits of C*-algebras instead. We validate our deviation in \textsection \ref{sct:quantum systems and embedding maps}.

The emergence of a strict deformation quantization counts as another validation of our method, but is also remarkable in itself. Most notably, it involves two limits; besides the usual limit $\hbar\to0$ also the limit of lattice spacing tending to zero becomes important. The interaction between these two limits complicates the proof at most places, but in other places is the very reason the result holds.

Section \ref{sct:2} constructs the classical C*-algebra on the continuum, the quantization map on the continuum, and the quantum C*-algebra on the continuum. The classical and quantum C*-algebras are shown to be invariant under time evolution related to the electric part of abelian Yang--Mills \cite{KS} in \textsection\ref{sct:time evolution}. Section \ref{sct:SDQ} gives the proof of strict deformation quantization, and forms by far the most technical part of this chapter. 



\vspace{-10pt}
\paragraph*{Notation.} We denote $G:=\T^n$, $\g:=\R^n$ and $\g^*:=\R^n$. Elements of $G^l$ for a finite set $l$ 
are usually denoted by $q$ or $[x]$ where $[x]:=x+(\Z^n)^l$ for $x\in(\R^n)^l$. We denote by $L_q$ the left-translation on $G^l$, i.e., $L_{[x]}[y]=[x+y]$. We denote by $e^{i\xi\cdot}$ the function $x\mapsto e^{i\xi\cdot x}$ and by $e_a$ the function $[x]\mapsto e^{2\pi i a\cdot x}$ for $a\in(\Z^n)^l$. We denote by $\psi_a$ the equivalence class of $e_a$ in $L^2(G^l)$. In any metric space, $B_d(x)$ is the open ball around $x$ with radius $d$. We let $B:=B_{1/2\pi}(0_\g)\subseteq\g$, remarking that $x\mapsto [x]$ is a diffeomorphism on $B$. By an operator system we mean a linear subspace of a unital C*-algebra that is preserved under $*$ and contains the unit. We do not require operator systems to be closed.

\section{Operator systems and limit C*-algebras}
\label{sct:2}

\paragraph*{Lattices.}
For simplicity, we take our time-slice to be $\R^D$, although any metric space would work.
Throughout this chapter, a \textit{lattice} is a finite subset $l\subseteq\R^D\times\R^D$ such that, using the lexicographical ordering of $\R^D$, we have $x<y$ for all $(x,y)\in l$, and, we have $tx+(1-t)y\neq sz+(1-s)w$ for all $(x,y),(z,w)\in l$ and all $0<t,s<1$. The elements $e=(x,y)$ of a lattice $l$ are interpreted as directed straight edges from $x$ to $y$. Thus, all we ask of a lattice is that its edges do not intersect, except possibly at their boundaries. The set of all lattices becomes a directed set, denoted $(\L,\leq)$, when we agree that $l\leq m$ if and only if the lattice $m$ can be obtained from $l$ by adding and subdividing edges in the sense of \cite{ASS}. Put precisely, $l\leq m$ if and only if for all $(x_1,x_2)\in l$ there exists $N\in\N_0$ and $0<t_1<\cdots<t_N<1$ such that for $y_s:=(1-t_s)x_1+t_sx_2$ we have $(x_1,y_1),(y_1,y_2),\ldots,(y_{N-1},y_N),(y_N,x_2)\in m$. We say these edges $(x_1,y_1),\ldots, (y_N,x_2)$ in $m$ are obtained \textit{by subdivision} from the edge $(x_1,x_2)$ in $l$. Of course, $m$ can contain elements not obtained in this way, which we say are \textit{added} in passing from $l$ to $m$.
We endow every edge $e=(x,y)$ with a length $d_e:=\norm{x-y}$.

%

 Let us compare our notation with the one in \cite{ASS,Stienstra}, in which an index set $I$ is used, and $\{\Lambda_i\}_{i\in I}$ is the net of finite lattices, including a set of vertices $\Lambda_i^0$, a set of edges $\Lambda_i^1$, and a set of plaquettes $\Lambda_i^2$. In our situation, the elements $l\in\L$ can be identified with the sets of edges $\Lambda^1_i$. Because we will not reduce to the gauge group and only discuss the electric part of Yang--Mills dynamics, the vertices and plaquettes will play no role. By our definition of $l\in\L$ and simply following set notation, $G^l$ denotes the set of functions from the edges in $l$ to elements in $G$, or equivalently ordered tuples of length $|l|$ with elements in $G$.

%
%
%

\subsection{The finite and continuum classical systems}

\paragraph*{The continuum phase space.}
	 Throughout this chapter, the configuration space associated to each edge of a lattice is the compact abelian Lie group $G=\T^n$. Its Lie algebra is $\g=\R^n$, and the associated exponential map $\g\to G$ is given by $x\mapsto [x]$. The phase space $X^l$ associated to a lattice $l\in\L$ is given by the cotangent bundle of the Lie group $G^l$, i.e., $X^l:=T^*G^l\cong G^l\times(\g^*)^l$. In order to define connecting maps between $X^l$ and $X^m$, for lattices $l\leq m\in\L$, we use the fact that $m$ can be obtained from $l$ by recursively applying two operations: adding an edge to the lattice and subdividing an edge of length $d$ into two edges of lengths $d_1$ and $d_2$ with $d_1+d_2=d$. In that manner, we define connecting maps
\begin{align*}
	\gamma_{lm}=(\gamma_{lm}^\conf,\gamma_{lm}^\mom):G^m\times(\g^*)^m\to G^l\times(\g^*)^l
\end{align*}
by recursively composing embedded versions of the maps $\gamma_{\text{add}}=(\gamma^{\conf}_{\text{add}},\gamma^{\mom}_{\text{add}}):G^2\times(\g^*)^2\to G\times\g^*$ and $\gamma_{\text{sub}}=(\gamma^{\conf}_{\text{sub}},\gamma^{\mom}_{\text{sub}}):G^2\times(\g^*)^2\to G\times\g^*$ defined by
%
%
%
\begin{align*}
	\gamma^\conf_\add([x_1],[x_2])&:=[x_1];&\gamma^\mom_\add(v_1,v_2)&:=v_1;\\
	\gamma^\conf_\sub([x_1],[x_2])&:=[x_1+x_2];&\gamma^\mom_\sub(v_1,v_2)&:=\frac{d_1v_1+d_2v_2}{d}.
\end{align*}
These embedding maps arise naturally by interpreting $x_e\in G$ as the parallel transport along the edge $e$ and $v_e\in\g$ as the average rate of change along $e$. One could replace `average' by `total', at the cost of a slightly different quantization map. By construction, the surjective maps $\gamma_{lm}:X^m\to X^l$ for $l\leq m\in\L$ define an inverse system of topological spaces.
The ensuing inverse limit is denoted as
\begin{align*}
	X^\infty:=\lim_{\leftarrow} X^l=\lim_{\leftarrow}G^l\times\lim_{\leftarrow}(\g^*)^l,\quad
	\gamma_l=(\gamma_l^{conf},\gamma_l^{mom})
	:X^\infty\to X^l.
\end{align*}
One can naturally embed any phase space of classical gauge fields (the continuous ones, the smooth and compactly supported ones, etcetera) in $X^\infty$. Because any element of $X^l$ can be extended to such a field, we find that $\gamma_l:X^\infty\to X^l$ is surjective.

\paragraph*{Operator systems.}
The classical system on the lattice $l\in\L$ can be described by the commutative C*-algebra introduced in Chapter \ref{ch:cylinder}, namely
\begin{align*}
	A^l_0:=\Cr(T^*G^l)=C(G^l)\hatotimes \Wr((\g^*)^l),
\end{align*}
where $\Wr((\g^*)^l)$ is the C*-subalgebra of $C_\textnormal{b}((\g^*)^l)$ generated by the commutative Weyl C*-algebra $\W^0((\g^*)^l)$ from \cite{BHR} and the commutative resolvent algebra $\Cr((\g^*)^l)$ from \cite{vN19}. The reason to work with the unital C*-algebra $A^l_0$ is that $A_0^l$ and its Weyl quantization are conserved under fully general dynamics in the sense of \cite{vNS20}. In contrast, the C*-subalgebra $C(G^l)\hatotimes\W^0((\g^*)^l)\subseteq A^l_0$, where $\W^0((\g^*)^l):=\overline{\spn}\{e^{i\xi\cdot}:~\xi\in\g^l\}$, is only conserved under `free' time evolution \cite{vNS20}. As explained in Chapter \ref{ch:cylinder}, $A_0^l$ is the closure of the *-algebra $\A_0^l$ defined by
\begin{align*}
	\A_0^l:=\spn\left\{e_b\otimes e^{i\xi\cdot}(g\circ P_U)\mid
	\begin{aligned}
	&b\in(\Z^n)^l,~ U\subseteq \g^l\text{ linear},\\
	&g\in \S(U),~\hat{g}\in C_\textnormal{c}^\infty(U^*),~\xi\in \g^l
	\end{aligned}\right\}.
\end{align*}
For this chapter, we will only need that any element of $\A_0^l$ can be written as $\sum_{k=1}^K g_k\otimes h_k$ with $h_k\in C_\textnormal{b}((\g^*)^l)$ an (inverse) Fourier transform $h_k=\check{\mu}_k:=\int d\mu_k(\xi)e^{i\xi\cdot}$ of a compactly supported finite complex Borel measure $\mu_k$ on $\g^l$. We can thus define the operator system
\begin{align*}
	\M_0^l:=\spn\{g\otimes \check\mu\in\A_0^l:~\supp(\mu)\subseteq B^l\}\subseteq\A_0^l,
\end{align*}
where $B=B_{1/2\pi}(0_\g)$.
The *-algebras $\A^l_0$ are endowed with the connecting maps $\gamma_{lm}^*:\A_0^l\to\A_0^m$, whose restrictions to the operator systems $\M_0^l$ we denote as 
	$$F^{ml}_C:=\gamma_{lm}^*|_{\M^l_0}:\M_0^l\to\M_0^m,$$
and refer to as the classical embedding maps.
We define the *-algebraic direct limit 
	$$\A_0^\infty:=\lim_{\rightarrow}\A_0^l,$$
and identify $\A_0^\infty\subseteq C_\textnormal{b}(X^\infty)$ by identifying the universal map $F_C^l:\A_0^l\to\A_0^\infty$ with the restriction of the isometry $\gamma_l^*:C_\textnormal{b}(X^l)\to C_\textnormal{b}(X^\infty)$. To describe $\A_0^\infty$, it turns out we only need to regard the operator systems $\M_0^l$. To prove this, we first introduce the following useful notation.

\begin{defi}\label{def:l_R}
	For a lattice $l$ and a positive integer $R$, we let $l^R\geq l$ be the lattice obtained by subdividing every edge of $l$ into $R$ edges of equal length. 
\end{defi}

\begin{lem}
The direct limit of *-algebras $\A_0^l$ is also the direct limit of the operator systems $\M_0^l$, in the sense that we have
\begin{align}\label{eq:classical algebra vs operator system}
	\A_0^\infty=\{f\circ\gamma_l:~l\in \L,~f\in\M_0^l\}.
\end{align}
\end{lem}
\begin{proof}
By recursively composing the maps
\begin{align}\label{S^ml}
	S^\sub(\xi):=\left(\frac{d_1}{d}\xi,\frac{d_2}{d}\xi\right),\qquad S^\add(\xi):=(\xi,0),
\end{align}
we obtain a direct system of linear maps $S^{ml}:\g^l\to\g^m$ ($l\leq m\in\L$) allowing us to write the classical embedding maps as
\begin{align}\label{eq:F_C^ml in termen van S^ml}
	F_C^{ml}(g\otimes \check\mu)=(g\circ\gamma^\conf_{lm})\otimes (S^{ml}_*\mu)\check{~}.
\end{align}
For every $F_C^l(f)=f\circ\gamma_l\in\A_0^\infty$ we can write $f=\sum_k g_k\otimes \check{\mu}_k$ for compactly supported measures $\mu_k$. Choose $R$ such that $\supp(\mu_k)\subseteq B_{R/2\pi}(0_\g)^l$ for all $k$, and consider the lattice $l^R\geq l$. Every edge $e\in l^R$ satisfies $d_e=\tfrac{1}{R} d_{e'}$ for the edge $e'\in l$ it lies in. Hence 
	$$(S^{l^Rl}\xi)_e=\tfrac{1}{R}\xi_{e'}\,,$$
and so $S^{l^Rl}(\supp(\mu_k))\subseteq S^{l^Rl}(B_{R/2\pi}(0_\g)^l)\subseteq B_{1/2\pi}(0_\g)^{l^R}=B^{l^R}$. As $S^{l^Rl}$ is a closed map, we therefore obtain $\supp(S^{l^Rl}_*\mu_k)\subseteq B^{l^R}$ for all $k$. Then \eqref{eq:F_C^ml in termen van S^ml} gives $f\circ\gamma_{ll^R}\in\M^{l^R}_0$, so $F_C^l(f)=(f\circ\gamma_{ll^R})\circ\gamma_{l^R}$ is in the set on the right-hand side of \eqref{eq:classical algebra vs operator system}.
\end{proof}

\begin{rema}\label{rema:supremum trick}
	Two arbitrary functions in $\A_0^\infty$ can be written as $f_1\circ\gamma_l,f_2\circ\gamma_l$ for a certain $l\in \L$. Indeed, given $f'_1\circ\gamma_{l_1},f'_2\circ\gamma_{l_2}\in\A_0^\infty$, one takes the supremum $l$ of $l_1$ and $l_2$ (this corresponds to the coarsest lattice that is finer than both $l_1$ and $l_2$), and writes $f'_j\circ\gamma_{l_j}=(f'_j\circ\gamma_{l_jl})\circ\gamma_l\equiv f_j\circ\gamma_l$. The same goes for $k$ functions $f_1\circ\gamma_l,\ldots,f_k\circ\gamma_l$.
\end{rema} 
 
The first use of this remark is in defining a Poisson structure on $\A_0^\infty$. The Poisson bracket of $f_1\circ \gamma_l$ and $f_2\circ\gamma_l$ is defined as
\begin{align*}
	\{f_1\circ \gamma_l,f_2\circ\gamma_l\}:=\{f_1,f_2\}\circ\gamma_l,
\end{align*}
in terms of the Poisson bracket on $\A_0^l$, which is a Poisson subalgebra of $C^\infty(X^l)$. To show that the above bracket on $\A_0^\infty$ is well defined, it suffices to show that $\{f_1\circ\gamma_{lm},f_2\circ\gamma_{lm}\}=\{f_1,f_2\}\circ\gamma_{lm}$ for all $l\leq m$. This follows from the analogous statement for $\gamma_{\text{add}}$ and $\gamma_{\text{sub}}$, which can be straightforwardly checked.

\subsection{The quantum systems and quantum embedding maps}
\label{sct:quantum systems and embedding maps}
To each lattice $l\in\L$ we will associate an operator system modeling the quantum system. This operator system is defined as a quantization of $\M_0^l$ under a quantization map $\Ql$ that defines an extension of Weyl quantization. Recall that every $f\in\M_0^l$ can be written as $f=\sum_k g_k \otimes \check\mu_k$ for $g_k\in C^\infty(G^l)$ and $\supp(\mu_k)\subseteq B^l\subseteq\g^l$, where $B=B_{1/2\pi}(0_\g)$. Notice that $\hbar\xi\in B^l$ for every $\hbar\in[-1,1]$ and $\xi\in\supp(\mu_k)$. 
We define the quantization map on the lattice $l$ to be
\begin{align}\label{eq:Ql}
	\Ql:\M_0^l&\to\mB(L^2(G^l)),\nonumber\\
	\Ql\bigg(\sum_{k=1}^Kg_k\otimes\check\mu_k\bigg)\psi[y]&:=\sum_{k=1}^K \int_{\g^l} d\mu_k(\xi)g_k[y+\tfrac12\hbar\xi]\psi[y+\hbar\xi].
\end{align}

A simple calculation shows that, acting on the wave functions $\psi_a[x]:=e^{2\pi i a\cdot x}$ ($a\in(\Z^n)^l$), this quantization map has the simple form
\begin{align}\label{eq:asymptotically Ruben}
	\Ql(e_b\otimes h)\psi_a=h(2\pi \hbar(a+\tfrac12 b))\psi_{a+b},
\end{align}
and therefore coincides with the one in Chapter \ref{ch:cylinder}. Moreover, when $b$ is small enough, it coincides with Weyl quantization on the Riemannian manifold $\T^{n|l|}$ as given in \cite[Definition II.3.4.4]{landsman98}, as the cutoff function $\kappa$ used there becomes 1 when we restrict to $\M_0^l$. The insight used by this chapter is that, restricted to the operator system $\M_0^l$, the quantization map is injective. To see this, let $f=\sum_j e_{b_j}\otimes h_j\in\M_0^l$ be such that $\Ql(f)=0$. By \eqref{eq:asymptotically Ruben}, $h_j(2\pi\hbar(a+\tfrac12 b_j))=0$ for all $j$ and all $a\in(\Z^n)^l$. For a fixed $j$, since $\supp(\mu_j)\subseteq B_{1/(2\pi\hbar)}(0_\g)^l$, the Whittaker--Nyquist--Shannon theorem implies that $h_j=\check{\mu} _j$ is determined by its values on the points $\pi\hbar b_j+2\pi\hbar a$, $a\in(\Z^n)^l$. Therefore $h_j=0$ for all $j$, and therefore $f=0$. Hence $\Ql$ is injective on $\M_0^l$.

The quantum system associated to $l$ is defined by
	$$\M_\hbar^l:=\Ql(\M_0^l).$$
As $\Ql$ is linear, unital, and *-preserving, $\M_\hbar^l$ is an operator system.

\begin{exam}
A notable subset of $\M^l_0$ is $\mathfrak{W}^l:=\spn\{g\otimes e^{i\xi\cdot}:~g\in C^\infty(G^l),\xi\in  B^l\}$. This subset generates the C*-algebra $C(G^l)\hatotimes \W^0((\g^*)^l)$, which can be seen as a classical Weyl C*-algebra on the torus \cite{BHR,vN19,vNS20} that lies inside $A^l_0=\overline{\A_0^l}$. The image of $\mathfrak{W}^l$ under the above quantization map generates the crossed product C*-algebra $C(G^l)\rtimes G^l$. Indeed, we have $\Ql(g\otimes e^{i\xi\cdot})=M_{g\circ L_{[\hbar\xi/2]}}L^*_{[\hbar\xi]}$.
\end{exam}

\begin{rema}\label{rem:non-abelian Ql}
	The above example in particular suggests that our approach is generalizable to nonabelian groups, although this might require considerable analytical effort. It is in any case reasonable to require the quantization map $\Ql$ on a nonabelian group to satisfy $\Ql(g \otimes e^{i\xi\cdot})=M_{g\circ L_{\exp(\hbar\xi/2)}}L^*_{\exp(\hbar\xi)}$, where $\exp$ is the Lie-theoretic exponential map.
\end{rema}

%

\paragraph*{Direct limit of Hilbert spaces.}
To model the quantum system in infinite degrees of freedom, we will eventually construct a noncommutative C*-algebra that is canonically represented on a Hilbert space. This Hilbert space is the limit of the following direct system of Hilbert spaces:
\begin{align*}
	\quad\H^l:=L^2(G^l),\quad u^{ml}:=(\gamma^{conf}_{lm})^*:\H^l\to\H^m.
\end{align*}
Passing to the direct limit, we denote
\begin{align*}
	\H^{\infty}:=&\lim_{\rightarrow}\H^l,\quad
	u^l=(\gamma^{conf}_l)^*:\H^l\to\H^\infty.
\end{align*}
To define a direct limit of the operator systems $\M^l_\hbar$, we need to define the embedding maps and show that they satisfy the needed properties.

\paragraph*{Quantum embedding maps.}
The quantum embedding maps are defined by quantizing the classical embedding maps, i.e., for all $l\leq m\in\L$ and all $f\in\M_0^l$ we define
\begin{align*}
	F^{ml}_Q:\M^l_\hbar&\to\M^m_\hbar,\\
	F^{ml}_Q(\Ql(f))&:=\Qm(F^{ml}_C(f)),
\end{align*}
which is unambiguous by injectivity of $\Ql$.
\begin{exam}
	The embedding map $F_Q^\add$ is given by tensoring with $1$, which exemplifies why our quantum systems should be unital. The embedding map $F_Q^\sub$ is best understood on elements of 
	$C(G^l)\rtimes G^l.$ As depicted in Figure \ref{fig:QE1}, we have
	$$F_Q^\sub(M_gL^*_{[\xi]})=M_{g\circ \mu}L^*_{\left[\frac{d_1}{d}\xi,\frac{d_2}{d}\xi\right]},$$
where $g\in C(G)$, $\xi\in B^l$ and $\mu:\T^n\times\T^n\to\T^n$ is given by $\mu([x_1],[x_2]):=[x_1+x_2]$. One sees the metric structure (encoded in $d_1,~d_2,$ and $d$) at work, 
and notices that the well-definedness of the UV-limit hinges on the use of operator systems.
\end{exam}

\begin{rema}
	As in Remark \ref{rem:non-abelian Ql}, we briefly touch upon the nonabelian case here. The embedding maps for adding an edge do not have to be altered when $G=\T^n$ is replaced by a nonabelian group. The embedding maps for subdivision generalize as well, giving in particular $F_Q^\sub(M_gL^*_{\exp(\xi)})=M_{g\circ \mu}L^*_{\exp\left(\frac{d_1}{d}\xi,\frac{d_2}{d}\xi\right)},$
where $\mu$ is the group multiplication. We see this as an encouraging sign, but will again take $G$ abelian in the rest of this chapter to keep the discussion simple and the results as strong as possible.
\end{rema}


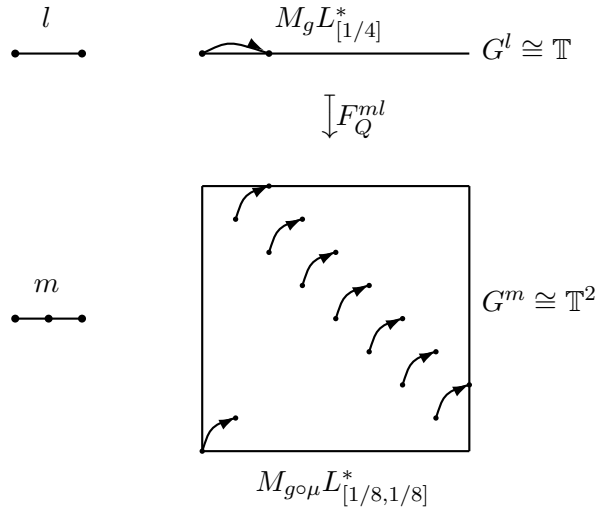
\begin{figure}[!htb]
	\centering
	\thicklines
	\begin{picture}(150,180)(-50,-20)
		\put(101,148){ $G^l\cong\T$}	
		\put(101,53){ $G^m\cong\T^2$}
		
		\put(-70,150){\line(1,0){25}}
		\put(-70,150){\circle*{3}}
		\put(-45,150){\circle*{3}}
		\put(-60,160){$l$}
		\put(-70,50){\line(1,0){25}}
		\put(-70,50){\circle*{3}}
		\put(-57.5,50){\circle*{3}}
		\put(-45,50){\circle*{3}}
		\put(-63,60){$m$}

		\put(0,150){\line(1,0){100}}

		\put(0,0){\line(1,0){100}}
		\put(0,0){\line(0,1){100}}
		\put(100,0){\line(0,1){100}}	
		\put(0,100){\line(1,0){100}}
			
		\put(45,135){\rotatebox{270}{$\longmapsto$}}
		\put(50,123){$F^{ml}_Q$}		
		
		\put(0,150){\circle*{2.5}}
		\put(25,150){\circle*{2.5}}
		\cbezier(0,150)(10,155)(15,155)(25,150)
		\put(20.5,152.5){\vector(16,-8){2.5}}
		
		\put(0,0){\circle*{2}}
		\put(12.5,12.5){\circle*{2}}
		\cbezier(0,0)(2.5,7.5)(5,10)(12.5,12.5)	
		\put(10,11.5){\vector(2,1){2.5}}
		
		\put(87.5,12.5){\circle*{2}}
		\put(100,25){\circle*{2}}
		\cbezier(87.5,12.5)(90,20)(92.5,22.5)(100,25)	
		\put(97.5,24){\vector(2,1){2.5}}
		
		\put(75,25){\circle*{2}}
		\put(87.5,37.5){\circle*{2}}
		\cbezier(75,25)(77.5,32.5)(80,35)(87.5,37.5)	
		\put(85,36.5){\vector(2,1){2.5}}
		
		\put(62.5,37.5){\circle*{2}}
		\put(75,50){\circle*{2}}
		\cbezier(62.5,37.5)(65,45)(67.5,47.5)(75,50)	
		\put(72.5,49){\vector(2,1){2.5}}
		
		\put(50,50){\circle*{2}}
		\put(62.5,62.5){\circle*{2}}
		\cbezier(50,50)(52.5,57.5)(55,60)(62.5,62.5)	
		\put(60,61.5){\vector(2,1){2.5}}
		
		\put(37.5,62.5){\circle*{2}}
		\put(50,75){\circle*{2}}
		\cbezier(37.5,62.5)(40,70)(42.5,72.5)(50,75)	
		\put(47.5,74){\vector(2,1){2.5}}
		
		\put(25,75){\circle*{2}}
		\put(37.5,87.5){\circle*{2}}
		\cbezier(25,75)(27.5,82.5)(30,85)(37.5,87.5)	
		\put(35,86.5){\vector(2,1){2.5}}
		
		\put(12.5,87.5){\circle*{2}}
		\put(25,100){\circle*{2}}
		\cbezier(12.5,87.5)(15,95)(17.5,97.5)(25,100)	
		\put(22.5,99){\vector(2,1){2.5}}
		
		\put(27,160){$M_gL^*_{[1/4]}$}
		\put(20,-15){$M_{g\circ\mu} L^*_{[1/8,1/8]}$}		
	\end{picture}	
	\caption{A depiction of an operator $M_gL^*_{[1/4]}\in\mathcal{M}^l_\hbar$ and its image under the quantum embedding map, where $G=\T$, $l$ has one edge, and $m=l^2$. For the picture, $g$ is supported closely around $[1/4]\in G^l$. The embedding map clearly respects the gauge action coming from the central vertex of $m$.}
	\label{fig:QE1}
\end{figure}		

Our quantum embedding maps contrast with those used in the existing literature on C*-algebraic lattice gauge theory \cite{ASS,BS19a,BS19b,GR,Stienstra,ST} because ours do not define a direct system (inductive system) of *-algebras. They therefore warrant some motivation.

We assume the situation of Figure \ref{fig:QE1} and Figure \ref{fig:QE2}, where $G=\T=U(1)$ and a lattice $l$ consisting of a single edge is compared to a lattice $m\geq l$ with two edges. 
There exist multiple observables on the lattice $m$ that have the same behavior when restricted to $l$. This can be seen in Figure \ref{fig:QE2}, in the formulas, or by interpreting the gauge field as rigid rotors associated to every edge, as in \cite{KS}. 
Indeed, the two rotors associated to the two edges of $m$ can either both be turned clockwise by a quarter circle, or both anti-clockwise by a quarter circle. When describing the gauge field by a single rotor, the two operations appear as the same observable (see Figure \ref{fig:QE2}).

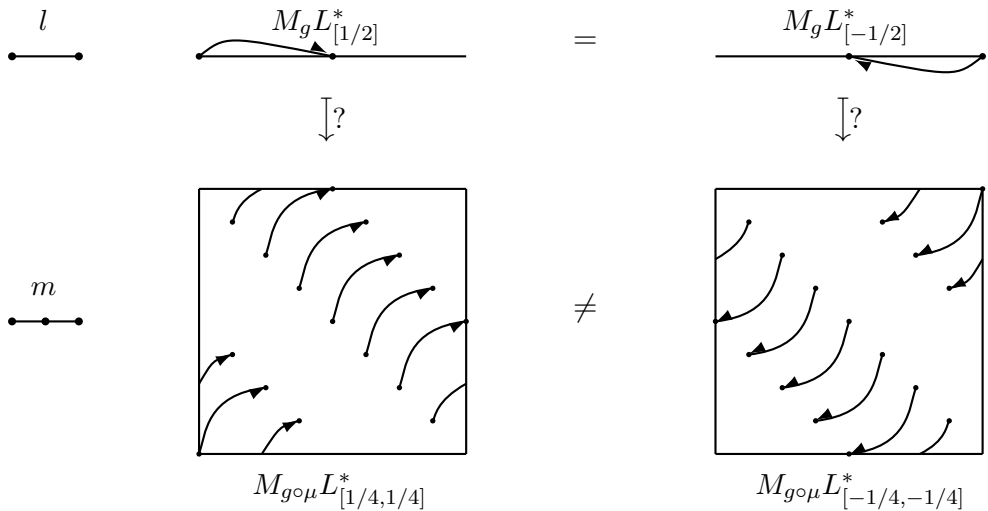
\begin{figure}[!htb]
	\centering
	\thicklines
	\begin{subfigure}{.49\textwidth}
	\centering
	\begin{picture}(135,200)(-50,-20)
		\put(-70,150){\line(1,0){25}}
		\put(-70,150){\circle*{3}}
		\put(-45,150){\circle*{3}}
		\put(-60,160){$l$}
		\put(-70,50){\line(1,0){25}}
		\put(-70,50){\circle*{3}}
		\put(-57.5,50){\circle*{3}}
		\put(-45,50){\circle*{3}}
		\put(-63,60){$m$}

		\put(0,150){\line(1,0){100}}
		
		\put(0,150){\circle*{2.5}}
		\cbezier(0,150)(10,158)(40,158)(50,150)
		\put(45.5,152.5){\vector(16,-8){2.5}}
		\put(50,150){\circle*{2.5}}

		\put(0,0){\line(1,0){100}}
		\put(0,0){\line(0,1){100}}
		\put(100,0){\line(0,1){100}}	
		\put(0,100){\line(1,0){100}}
			
		\put(45,135){\rotatebox{270}{$\longmapsto$}}
		\put(50,123){$?$}	

		\put(0,0){\circle*{2}}
		\put(25,25){\circle*{2}}
		\cbezier(0,0)(2.5,7.5)(17.5,22.5)(25,25)	
		\put(22.5,24){\vector(2,1){2.5}}
		
		\put(75,25){\circle*{2}}
		\put(100,50){\circle*{2}}
		\cbezier(75,25)(77.5,32.5)(92.5,47.5)(100,50)	
		\put(97.5,49){\vector(2,1){2.5}}
		
		\put(62.5,37.5){\circle*{2}}
		\put(87.5,62.5){\circle*{2}}
		\cbezier(62.5,37.5)(65,45)(80,60)(87.5,62.5)	
		\put(85,61.5){\vector(2,1){2.5}}

		\put(50,50){\circle*{2}}
		\put(75,75){\circle*{2}}
		\cbezier(50,50)(52.5,57.5)(67.5,72.5)(75,75)	
		\put(72.5,74){\vector(2,1){2.5}}
		
		\put(37.5,62.5){\circle*{2}}
		\put(62.5,87.5){\circle*{2}}
		\cbezier(37.5,62.5)(40,70)(55,85)(62.5,87.5)	
		\put(60,86.5){\vector(2,1){2.5}}
		
		\put(25,75){\circle*{2}}
		\put(50,100){\circle*{2}}
		\cbezier(25,75)(27.5,82.5)(42.5,97.5)(50,100)	
		\put(47.5,99){\vector(2,1){2.5}}
		
		\put(12.5,87.5){\circle*{2}}
		\cbezier(12.5,87.5)(13.5,89)(15,95)(23.5,100)

		\put(37.5,12.5){\circle*{2}}
		\cbezier(23.5,0)(30,10)(35,10.5)(37.5,12.5)
		\put(35,11.5){\vector(2,1){2.5}}

		\put(12.5,37.5){\circle*{2}}
		\cbezier(0,26.5)(5,35)(10,35.5)(12.5,37.5)
		\put(10,36.5){\vector(2,1){2.5}}
		
		\put(87.5,12.5){\circle*{2}}
		\cbezier(87.5,12.5)(88.5,14)(90,20)(100,26.5)
		
		\put(140,153){\large $=$}
		\put(140,52){\large $\neq$}
		
		\put(27,160){$M_gL^*_{[1/2]}$}
		\put(20,-15){$M_{g\circ\mu} L^*_{[1/4,1/4]}$}		
	\end{picture}
	\end{subfigure}
	\begin{subfigure}{.49\textwidth}
	\centering
	\begin{picture}(225,200)(-65,-20)

		\put(0,150){\line(1,0){100}}
		
		\put(100,150){\circle*{2.5}}
		\cbezier(100,150)(90,142)(60,142)(50,150)
		\put(54.5,147.5){\vector(-16,8){2.5}}
		\put(50,150){\circle*{2.5}}

		\put(0,0){\line(1,0){100}}
		\put(0,0){\line(0,1){100}}
		\put(100,0){\line(0,1){100}}	
		\put(0,100){\line(1,0){100}}			
		
		\put(45,135){\rotatebox{270}{$\longmapsto$}}
		\put(50,123){$?$}		
					
		\put(100,100){\circle*{2}}
		\put(75,75){\circle*{2}}
		\cbezier(100,100)(97.5,92.5)(82.5,77.5)(75,75)	
		\put(77.5,76){\vector(-2,-1){2.5}}
		
		\put(25,75){\circle*{2}}
		\put(0,50){\circle*{2}}
		\cbezier(25,75)(22.5,67.5)(7.5,52.5)(0,50)	
		\put(2.5,51){\vector(-2,-1){2.5}}

		\put(37.5,62.5){\circle*{2}}
		\put(12.5,37.5){\circle*{2}}
		\cbezier(37.5,62.5)(35,55)(20,40)(12.5,37.5)	
		\put(15,39){\vector(-2,-1){2.5}}
		
		\put(50,50){\circle*{2}}
		\put(25,25){\circle*{2}}
		\cbezier(50,50)(47.5,42.5)(32.5,27.5)(25,25)	
		\put(27.5,26){\vector(-2,-1){2.5}}

		\put(62.5,37.5){\circle*{2}}
		\put(37.5,12.5){\circle*{2}}
		\cbezier(62.5,37.5)(60,30)(45,15)(37.5,12.5)	
		\put(40,14){\vector(-2,-1){2.5}}
		
		\put(75,25){\circle*{2}}
		\put(50,0){\circle*{2}}
		\cbezier(75,25)(72.5,17.5)(57.5,2.5)(50,0)	
		\put(52.5,1){\vector(-2,-1){2.5}}
	
		\put(87.5,12.5){\circle*{2}}
		\cbezier(87.5,12.5)(86.5,11)(85,5)(76.5,0)

		\put(62.5,87.5){\circle*{2}}
		\cbezier(76.5,100)(70,90)(65,89.5)(62.5,87.5)
		\put(65,88.5){\vector(-2,-1){2.5}}

		\put(87.5,62.5){\circle*{2}}
		\cbezier(100,73.5)(95,65)(90,64.5)(87.5,62.5)
		\put(90,63.5){\vector(-2,-1){2.5}}
		
		\put(12.5,87.5){\circle*{2}}
		\cbezier(12.5,87.5)(11.5,86)(10,80)(0,73.5)

		\put(25,160){$M_gL^*_{[-1/2]}$}
		\put(15,-15){$M_{g\circ\mu} L^*_{[-1/4,-1/4]}$}		
	\end{picture}
	\end{subfigure}
	\caption{The quantum embedding map does not extend in a multiplicative way from $\mathcal{M}^l_\hbar$ to the algebra $\A^l_\hbar$ generated by $\mathcal{M}^l_\hbar$. If we would try, we would end up with two representations of the same observable in $\mathcal A^l_\hbar$ being mapped to two different observables in $\mathcal A^m_\hbar$.}
	\label{fig:QE2}
\end{figure}

 Therefore, if one wants to interpret an observable on a lattice $l$ as an observable on the continuum, a choice has to be made. We make this choice by restricting at any finite level to observables that rotate any rotor less than a certain amount, so that an embedding of such an observable can be made by fairly distributing that rotation over the smaller rotors that make up the original one. Clearly, this means giving up on multiplicative structure. This is not against the C*-algebraic philosophy, however,
which states that one can describe any physical system once we have a sufficiently rich C*-algebra of observables. The set of observables at a finite level makes up but a subset of the full algebra, and is therefore not required to completely describe a physical system. Only the full set of observables, with arbitrary lattice size, can discern between any two gauge fields, and can therefore be expected to form a *-algebra (lying densely in a C*-algebra). That is indeed what we will prove in Proposition \ref{prop:algebra}.


As further motivation of our quantum embedding maps, and to be used later, we show that they intertwine the direct system of Hilbert spaces given by $u^{ml}:\H^l\to\H^m$.

\begin{lem}\label{lem:F_Q and u}
	For $l\leq m\in\L$ and $O\in\M_\hbar^l$ we have
\begin{align*}
	F_Q^{ml}(O)u^{ml}=u^{ml}O.
\end{align*}
\end{lem}
\begin{proof}
	Similar to \eqref{S^ml}, we define
	\begin{align*}
		T^\add(\xi):=(\xi,0)\qquad T^\sub(\xi):=(\xi,\xi),
	\end{align*}	
	to obtain a direct system of linear maps $T^{ml}:\g^l\to\g^m$. We account here that
	\begin{align}\label{eq:S and T}
		\gamma^\mom_{lm}(X)\xi=X S^{ml}(\xi);\qquad \gamma^\conf_{lm}(q) \xi=q T^{ml}(\xi);\qquad \gamma_{lm}^\mom\circ T^{ml}=\id_{\g^l},
	\end{align}
	such that, in particular, $u^{ml}\psi_a=\psi_{T^{ml}(a)}$ for all $a\in(\Z^n)^l$. For $f=e_b\otimes h$, we get
	\begin{align*}
		\Qm(f\circ\gamma_{lm})u^{ml}\psi_a&=\Qm(f\circ\gamma_{lm})\psi_{T^{ml}(a)}\\
		&=h(\gamma^\mom_{lm}(2\pi\hbar(T^{ml}(a)+\tfrac12 T^{ml}(b))))\psi_{T^{ml}(a)+T^{ml}(b)}\\
		&=h(2\pi\hbar(a+\tfrac12 b))u^{ml}\psi_{a+b},
	\end{align*}
	so $\Qm(f\circ\gamma_{lm})u^{ml}\psi_a=u^{ml}\Ql(f)\psi_a$, which implies the lemma.
\end{proof}

\subsection{The continuum quantization map and quantum system}

To define $\Q$, we define $\Q(f\circ\gamma_l)\in\mB(\H^\infty)$ by its action on $u^m\psi\in\H^\infty$, where $m\geq l$, namely
\begin{align*}
	\Q(f\circ\gamma_l)u^m\psi:=u^mQ^m_\hbar(f\circ\gamma_{lm})\psi\qquad (\psi\in\H^m).
\end{align*}
To show that this is well defined, we use Lemma \ref{lem:F_Q and u} and find, for all $n\geq m\geq l$ and $\psi\in\H^m$,
\begin{align*}
	u^n\Qn(f\circ\gamma_{ln})u^{nm}\psi &= u^n F_Q^{nm}(\Qm(f\circ\gamma_{lm}))u^{nm}\psi\\
	&= u^m\Qm(f\circ\gamma_{lm})\psi,
\end{align*}
and conclude that $\Q(f\circ\gamma_l)$ is well defined on the dense subset $\cup_m u^m\H^m\subseteq\H^\infty$. If we write $f=\sum_k g_k\otimes \check\mu_k$ we obtain, 
\begin{align*}
	\quadnorm{\Q(f\circ\gamma_l)u^m\psi} &= \quadnorm{\Qm(f\circ\gamma_{lm})\psi}\leq \sum_k\supnorm{g_k}\absnorm{\mu_k}\quadnorm{u^m\psi}. 
\end{align*}
Therefore $\Q:\A_0^\infty\to\mB(\H^\infty)$ is well defined, and $\norm{\Q(f\circ\gamma_l)}\leq\sum\supnorm{g_k}\absnorm{\mu_k}$, independently from $\hbar$. The above also shows that 
\begin{align}
	\norm{\Q(f\circ\gamma_l)}=\sup_{m\geq l}\norm{\Qm(f\circ\gamma_{lm})}=\lim_{m}\norm{\Qm(f\circ\gamma_{lm})}.\label{eq:Q^infty sup over Q^j}
\end{align}
We define $$\A^\infty_\hbar:=\Q(\A^\infty_0)\equiv \{F^l_Q(O):~l\in\L,~ O\in\M_\hbar^l\}.$$ We write $\A_\hbar^\infty$ instead of $\M_\hbar^\infty$ to suggest it is in fact an algebra.

\begin{prop}\label{prop:algebra}
	The operator system $\A_\hbar^\infty=\Q(\A_0^\infty)$ is a *-algebra.
\end{prop}
\begin{proof}
	By Remark \ref{rema:supremum trick}, we only have to show that $\Q(f_1\circ\gamma_l)\Q(f_2\circ\gamma_l)$ is in $\A_\hbar^\infty$. Write $O_1:=\Ql(f_1)$ and $O_2:=\Ql(f_2)$. Because we cannot take their product in the operator system $\M^l_\hbar$, we first subdivide the edges of $l$ to obtain the lattice $l^2$ defined by Definition \ref{def:l_R}. A straightforward computation shows firstly that
		\begin{align*}
			F_Q^{l^2l}(O_1)F_Q^{l^2l}(O_2)\in\M_\hbar^{l^2},
		\end{align*}
		and secondly that
		\begin{align*}
			F_Q^{ml}(O_1)F_Q^{ml}(O_2)=F_Q^{ml^2}(F_Q^{l^2l}(O_1)F_Q^{l^2l}(O_2)),
		\end{align*}
		for all $m\geq l^2$.
		Using this formula and Lemma \ref{lem:F_Q and u}, we obtain
		\begin{align*}
			\Q(f_1\circ\gamma_l)\Q(f_2\circ\gamma_l)u^m\psi&=u^mF_Q^{ml}(O_1)F_Q^{ml}(O_2)\psi\\
			&=F_Q^{l^2}(F_Q^{l^2l}(O_1)F_Q^{l^2l}(O_2))u^m\psi,
		\end{align*}
		for all $u^m\psi\in\H^\infty$. Hence, $\Q(f_1\circ\gamma_l)\Q(f_2\circ\gamma_l)=F_Q^{l^2}(F_Q^{l^2l}(O_1)F_Q^{l^2l}(O_2))\in\A_\hbar^\infty$.
\end{proof}

Taking the closures of $\A_0^\infty\subseteq C_\textnormal{b}(X^\infty)$ and $\A_\hbar^\infty\subseteq\mB(\H^\infty)$, we therefore obtain C*-algebras $A_0^\infty$ and $A_\hbar^\infty$. By Theorem \ref{thm:sdq infinity}, we are justified in saying that the noncommutative C*-algebra $A_\hbar^\infty$ is obtained by strict deformation quantization of $A_0^\infty$.

\subsection{Time evolution}\label{sct:time evolution}
Before moving on to strict deformation quantization, we state two promising results with respect to time evolution. They show that our C*-algebras are invariant under the natural extension of free time evolution to the continuum limit.
On the finite level, these results are strengthened to invariance under \textit{all} time evolutions in Theorems \ref{thm:classical time evolution} and \ref{thm:quantum time evolution}.
The combination of the results here and in Chapter \ref{ch:cylinder} indicates that we are on the right track to obtaining classical and quantum C*-algebras that are invariant under respectively classical and quantum Yang--Mills time evolution.
\begin{thm}\label{thm:classical time evolution infty}
	The C*-algebra $A^\infty_0\subseteq C_\textnormal{b}(X^\infty)$ is conserved by the time evolution given on a lattice $l\in\L$ by the Hamiltonian $H_l:T^*G^l\to\R$, $H_l(q,v):=\sum_{e\in l}d_ev_e^2$, where $d_{(x,y)}=\norm{x-y}$.
\end{thm}
\begin{proof}
	Every Hamiltonian $H_l$ induces a time-evolution $\tau^0_{l}: \R\times A_0^l\to A_0^l$ by \cite[Lemma 10]{vNS20}. It can be checked that $H_l\circ\gamma_{lm}=H_m$, and therefore $\tau_m^0(t,f\circ\gamma_{lm})=\tau_l^0(t,f)\circ\gamma_{lm}$. We conclude that the time-evolution
	\begin{align*}
		\tau_\infty^0:A_0^\infty\to A_0^\infty,\qquad\tau_\infty^0(t,f\circ\gamma_l):=\tau_l^0(t,f)\circ\gamma_l
	\end{align*}
	is well defined.
\end{proof}

\begin{thm}\label{thm:quantum time evolution infty}
	The C*-algebra $A^\infty_\hbar\subseteq\mB(\H^\infty)$ is conserved by the time evolution given on a lattice $l\in\L$ by the Hamiltonian $\hat{H}_l:=\sum_{e\in l}d_e\Delta_e$, where $\Delta_e$ is the Laplace operator on the $e^{\text{th}}$ copy of $G$ in $G^l$.
\end{thm}
\begin{proof}
	These Hamiltonians define a continuum Hamiltonian $\hat H_\infty$ in $\H^\infty$ with domain 
	\begin{align*}
		\dom \hat H_\infty:=\bigcup_{l\in\L} u^l(\dom \hat H_l)=\bigcup_{l\in\L} u^l(C^\infty(G^l)),
	\end{align*}
	namely $\hat{H}_\infty u^l\psi:=u^l\hat H_l\psi$. Straightforwardly, one checks well-definedness and essential self-adjointness. By \cite[Remark 27]{vNS20}, we have
	\begin{align*}
		e^{it\hat H_\infty}\Q(f\circ\gamma_l)e^{-it\hat H_\infty}u^m\psi&=u^me^{it\hat H_m}\Qm(f\circ\gamma_{lm})e^{-it\hat H_m}\psi\\
		&=u^m\Qm(\tau^0_m(t,f\circ\gamma_{lm}))\psi\\
		&=u^m\Qm(\tau^0_l(t,f)\circ\gamma_{lm})\psi\\
		&=\Q(\tau^0_l(t,f)\circ\gamma_l)u^m\psi.
	\end{align*}
	Therefore, $e^{it\hat H_\infty}\Q(f\circ\gamma_l)e^{-it\hat H_\infty}=\Q(\tau^0_l(t,f)\circ\gamma_l)\in A_\hbar^\infty$ for every $t$.
\end{proof}

\section{Strict deformation quantization}
\label{sct:SDQ}

In this section we prove our main theorem, which is formulated as follows.

\begin{thm}\label{thm:sdq infinity}
	Let $Q_0^\infty:=\id_{\A^\infty_0}$. Together with the subset $I=[-1,1]$ and the C*-algebras $\{A^\infty_\hbar\}_{\hbar\in I}$, the maps $\{\Q:\A_0^\infty\to A^\infty_\hbar\}_{\hbar\in I}$ form a strict deformation quantization of $\A_0^\infty$.
\end{thm}

For readability, the proof of Theorem \ref{thm:sdq infinity} is split up into Propositions \ref{prop:star-preserving}, \ref{prop:injective}, \ref{prop:von Neumann}, \ref{prop:Dirac}, \ref{prop:Rieffel0}, and \ref{prop:Rieffel1}.

\begin{prop}\label{prop:star-preserving}
	The map $\Q:\A_0^\infty\to\A_\hbar^\infty$ is linear and *-preserving for all $\hbar\in I$.
\end{prop}
\begin{proof}
	Linearity is obvious, so we are left to prove that $\Q(f)^*=\Q(\overline{f})$ for $f\in\A_0^\infty$. Given $f\circ\gamma_l\in\A_0^\infty$ and $u^m\psi^m,u^n\psi^n\in\H^\infty$, choose $p\geq l,m,n$. By using that $\Ql:\A_0^l\to \A_\hbar^l$ is star-preserving by Proposition \ref{prop:Weyl quantization properties} (which can also be derived directly from \eqref{eq:Ql}) we get
		\begin{align*}
			\p{\Q(f\circ\gamma_l)u^m\psi^m}{u^n\psi^n}
			&=\p{u^pQ_\hbar^p(f\circ \gamma_{lp})u^{pm}\psi^m}{u^pu^{pn}\psi^n}\\
			&=\p{\Qp(f\circ\gamma_{lp})u^{pm}\psi^m}{u^{pn}\psi^n}\\
			&=\p{u^{pm}\psi^m}{\Qp(\overline{f}\circ\gamma_{lp})u^{pn}\psi^n}\\
			&= \p{u^m\psi^m}{\Q(\overline{f}\circ\gamma_l)u^n\psi^n}.
		\end{align*}
		Therefore $\Q(f\circ\gamma_l)^*$ equals $\Q(\overline{f}\circ\gamma_l)$ on a dense subset of $\H^\infty$, hence on the whole of $\H^\infty$ by boundedness.
\end{proof}

\begin{prop}\label{prop:injective}
	The map $\Q:\A_0^\infty\to A_\hbar^\infty$ is injective for all $\hbar\in I$.
\end{prop}
\begin{proof}
	Suppose $\Q(f\circ\gamma_l)=0$ for some $f\in\M_0^l$. Then
		\begin{align*}
			0&=\Q(f\circ\gamma_l)u^l\psi=u^l\Ql(f)\psi
		\end{align*}
		for all $\psi\in\H^l$. So $0=\Ql(f)$, and, by injectivity of $\Ql$, $f=0$.
\end{proof}

\begin{prop}\label{prop:von Neumann}
	\textbf{(von Neumann's condition)} For all $f,g\in\A_0^\infty$, we have
			\begin{align*}
				\lim_{\hbar\to0}\norm{\Q(f)\Q(g)-\Q(fg)}=0.
			\end{align*}
\end{prop}
\begin{proof}
	The proof is based on that of \cite[Theorem 22(2)]{vNS20}, but more complicated because $\Q$ is defined on $\H^\infty$, which includes all $\H^m$. Therefore, estimating an operator norm in $\mB(\H^\infty)$ amounts to taking a supremum over $m$. For two lattices $l\leq m\in\L$ and a function $e_b\otimes h\in\M_0^l$ we have,
		\begin{align*}
			(e_b\otimes h)\circ\gamma_{lm}=e^{2\pi ib\cdot \gamma^{\text{conf}}_{lm}(\cdot)}\otimes (h\circ\gamma_{lm}^{\text{mom}})
			=e_{T^{ml}(b)}\otimes (h\circ\gamma_{lm}^{\text{mom}})
		\end{align*}
		where we used \eqref{eq:S and T}. Combining this with \eqref{eq:asymptotically Ruben} and \eqref{eq:S and T}, we find
		\begin{align*}
			\Qm((e_b\otimes h)\circ\gamma_{lm})\psi_a&=h(\gamma^\mom_{lm}(2\pi\hbar(a+\tfrac{1}{2}T^{ml}(b))))\psi_{a+T^{ml}(b)}\\
			&=h(2\pi\hbar(\gamma^\mom_{lm}(a)+\tfrac{1}{2}b))\psi_{a+T^{ml}(b)}\,.
		\end{align*}
		Fix $f_1=e_{b_1}\otimes h_1,f_2=e_{b_2}\otimes h_2\in\M_0^l$ for an $l\in\L$. By bilinearity and Remark \ref{rema:supremum trick} it suffices to prove the proposition for $f=f_1\circ\gamma_l$ and $g=f_2\circ\gamma_l$. We note that if $O\psi_a=F(a)\psi_{a+b}$ for some $F\in C_\textnormal{b}((\Z^n)^m)$ and $b\in(\Z^n)^m$, then clearly $\norm{O}=\sup_{a\in(\Z^n)^m}\quadnorm{O\psi_a}$. We find, 
		\begin{align*}
			&\sup_{m\in\L_{\geq l}}\sup_{a\in(\Z^n)^{m}}\quadnorm{\left(\Qm(f_1f_2\circ\gamma_{lm})-\Qm(f_1\circ\gamma_{lm})\Qm(f_2\circ\gamma_{lm})\right)\psi_a}\\
			&\quad\leq \sup_{m\in\L_{\geq l}}\sup_{a\in(\Z^n)^m}\Big|h_1(2\pi\hbar(\gamma_{lm}^\mom(a)+\tfrac{1}{2}b_1+\tfrac{1}{2}b_2))h_2(2\pi\hbar(\gamma_{lm}(a)+\tfrac{1}{2}b_1+\tfrac{1}{2}b_2))\\
			&\qquad-h_1(2\pi\hbar(\gamma_{lm}^\mom(a)+\tfrac{1}{2}b_1+b_2))h_2(2\pi\hbar(\gamma^\mom_{lm}(a)+\tfrac{1}{2}b_2))\Big|\\
			&\quad\leq \supnorm{h_1}\pi|\hbar|\supnorm{\partial_{b_1} h_2} + \pi|\hbar|\supnorm{\partial_{b_2} h_1}\supnorm{h_2}\to 0\quad(\hbar\to0),
		\end{align*}
		where $\partial_b$ denotes the directional derivative. By \eqref{eq:Q^infty sup over Q^j}, this completes the proof.
\end{proof}

\begin{prop}\label{prop:Dirac}
	\textbf{(Dirac's condition)} For all $f,g\in\A_0^\infty$, we have
			\begin{align*}
				\lim_{\hbar\to0}\norm{(-i\hbar)^{-1}[\Q(f),\Q(g)]-\Q(\{f,g\})}=0.
			\end{align*}
\end{prop}
\begin{proof}
	Similar to the proof of Proposition \ref{prop:von Neumann}, 
	we obtain
		\begin{align*}
			&\sup_{m\in\L_{\geq l}}\sup_{a\in(\Z^n)^{m}}\norm{\left(\frac{i}{\hbar}[\Qm(f_1\circ\gamma_{lm}),\Qm(f_2\circ\gamma_{lm})]-\Qm(\{f_1\circ\gamma_{lm},f_2\circ\gamma_{lm}\})\right)\psi_a}\\
			&\quad\leq \sup_{m\in\L_{\geq l}}\sup_{a\in(\Z^n)^{m}}\Big|\frac{i}{\hbar}\Big(h_1\big(2\pi\hbar(\gamma_{lm}^\mom(a)+b_2+\tfrac12 b_1)\big)h_2\big(2\pi\hbar(\gamma_{lm}^\mom(a)+\tfrac12 b_2)\big)\\
			&\qquad\quad-h_1\big(2\pi\hbar(\gamma_{lm}^\mom(a)+\tfrac12 b_1)\big)h_2\big(2\pi\hbar(\gamma_{lm}^\mom(a)+b_1+\tfrac12 b_2)\big)\Big)\\
			&\qquad -2\pi i\Big(\partial_{b_2}h_1\cdot h_2-h_1\cdot\partial_{b_1}h_2\Big)\Big(2\pi\hbar(\gamma_{lm}^\mom(a)+\tfrac12(b_1+b_2))\Big)\Big|\\
			&\quad\to0\quad(\hbar\to0),
		\end{align*}
		which by \eqref{eq:Q^infty sup over Q^j} completes the proof.
\end{proof}

\subsection{Rieffel's condition at zero}

Rieffel's condition is all that remains to prove in order to establish our main theorem. Its proof is by far the most difficult component of this chapter, and is split into two parts, the first part giving continuity around $\hbar=0$ and the second part giving continuity elsewhere. 

For the first part we will use the following lemma.

\begin{lem}\label{lem:supnorm from operators}
	Let $f=\sum_{k=1}^Kg_k\otimes \check\mu_k\in\M_0^l$ for $l\in\L$. 
	For every $m\geq l$, we have
	\begin{align*}
		\supnorm{f}=\supnorm{F_C^{ml}(f)}=\sup_{q\in G^m}\norm{\sum_{k=1}^K\int_{\g^l}d\mu_k(\xi)g_k(\gamma_{lm}^\conf(q))L^*_{S^{ml}(\hbar\xi)}},
	\end{align*}
	where, on the right-hand side, the norm is the operator norm on $\mB(L^2(\g^m))$ and the integral is interpreted strongly.
\end{lem}
\begin{proof}
	The first equality is immediate, as $F^{ml}_C=(\gamma_{lm})^*$ and $\gamma_{lm}$ is surjective. By \eqref{eq:F_C^ml in termen van S^ml}, it now suffices to prove the lemma in the case where $l=m$, so $\gamma_{lm}^\conf=\id$ and $S^{ml}=\id$.	We obtain
	\begin{align*}
		\supnorm{f}&=\sup_{q\in G^l}\supnorm{\sum g_k(q)\int d\mu_k(\xi)e^{i\hbar\xi\cdot}}\\
		&=\sup_{q\in G^l}\sup_{\substack{\psi\in L^2((\g^*)^l)\\\quadnorm{\psi}=1}}\quadnorm{\sum \int d\mu_k(\xi)g_k(q) e^{i\hbar\xi\cdot}\psi(\cdot)}\\
		&=\sup_{q\in G^l}\sup_{\substack{\check\psi\in L^2(\g^l)\\\quadnorm{\check\psi}=1}}\quadnorm{\sum \int d\mu_k(\xi)g_k(q)\check\psi(\cdot+\hbar\xi)},
	\end{align*}
	by using Parseval's identity twice in the last step. The lemma follows.
\end{proof}

\begin{prop}\label{prop:Rieffel0}
	\textbf{(Rieffel's condition at 0)} For each $f\in\A_0^\infty$, we have
			\begin{align*}
				\lim_{\hbar\to0}\norm{\Q(f)}=\supnorm{f}.
			\end{align*}
\end{prop}
\begin{proof}
	Let $f\circ\gamma_l\in\A^\infty_0$ be arbitrary, for arbitrary $l\in\L$ and $f\in\M^l_0$. Write $f=\sum_{k=1}^K g_k\otimes \check\mu_k$. 
		We need to prove that $\norm{\Q(f\circ\gamma_l)}$ converges to $\supnorm{f\circ\gamma_l}=\supnorm{f}$, which by \eqref{eq:Q^infty sup over Q^j} comes down to showing that $\norm{\Qm(f\circ\gamma_{lm})}$ converges to $\supnorm{f}$ \textit{uniformly} in $m$. 
		
		For proving $\lim_{\hbar\to0} \norm{\Q(f\circ\gamma_l)}\geq \supnorm{f}$, we can simply use the similar statement for $\Ql$. Indeed, \cite[Theorem 7.8(1)]{Stienstra} gives
		\begin{align*}
			\lim_{\hbar\to0}\norm{\Q(f\circ\gamma_l)}=\lim_{\hbar\to0}\sup_{m\in\L_{\geq l}}\norm{\Qm(f\circ\gamma_{lm})}\geq\lim_{\hbar\to0}\norm{\Ql(f)}\geq \supnorm{f}.
		\end{align*}
		
		The reverse inequality, however, is considerably more difficult. For any $\epsilon>0$, we will need to construct an $\hbar_0>0$ such that for all $|\hbar|\leq\hbar_0$ we have $\norm{\Qm(f\circ\gamma_{lm})}\leq \supnorm{f}+\epsilon$, independently of $m$.

		Let $\epsilon>0$ be arbitrary. We define
		\begin{align}\label{eq:Q}
			Q:=\sum_{k=1}^K\supnorm{g_k}\absnorm{\mu_k},	
		\end{align}
		and remark that $\norm{\Q(f\circ\gamma_l)}\leq Q$ for all $\hbar\in[-1,1]$. Pick $N\in\N$ and distinct points $x_1,\ldots,x_N\in G^l$ such that for $$r:=\sup_{y\in G^l}\inf_{j=1}^N d(y,x_j),$$
		we have $B_{r}[0_{\g^l}]\subseteq (B_{1/2\pi}[0_\g])^l$, as well as $r<1/4$ and
		\begin{align}\label{eq:uniform continuity}
			d(x,y)<2r \Rightarrow |g_k(x)-g_k(y)|<\frac{\epsilon}{12Q\sum_k\absnorm{\mu_k}}.
		\end{align}
		We define, for all $j\in\{1,\ldots,N\}$ and $\delta\geq0$, the sets
		\begin{align*}
			V_{\delta,j}:=\{y\in G^l:~d(y,x_j)+\delta\leq d(y,x_{j'})\text{ for all } j'\neq j\}.
		\end{align*}
		We have $V_{\delta,j}\subseteq V_{0,j}\subseteq L_{x_j}(B_{1/2\pi}[0]^l)$.
		Choose $\delta>0$ such that $\delta\leq r$ and 
		\begin{align*}
			\vol(G^l\setminus\cup_{j=1}^N V_{\delta,j})<\frac{\epsilon}{3Q^2}.
		\end{align*}
		Choose $\hbar_0>0$ such that 
		\begin{align}\label{eq:hbar0}
			\max_{\xi\in\cup_k\supp(\mu_k)}\norm{\hbar_0\xi}<\frac{\delta}{2}.
		\end{align}
		Let $\hbar\in[-1,1]$ with $|\hbar|<\hbar_0$ be arbitrary. 
		Let $n\geq l$ be arbitrary. Let $m\in\L$ be the (unique) lattice for which $l\leq m\leq n$, $m\subseteq n$ and $m\setminus\{e\}\ngeq l$ for all $e\in m$, i.e., $m$ is made from $l$ by subdivisions, and $n$ is made from $m$ by additions of edges. As $F^\add_Q$ is isometric, 
		\begin{align}\label{eq:j or k same same}
			\norm{\Qn(f\circ\gamma_{ln})}=\norm{F_Q^{nm}(\Qm(f\circ\gamma_{lm}))}=\norm{\Qm(f\circ\gamma_{lm})},
		\end{align}
		so it suffices to prove that $\norm{\Qm(f\circ\gamma_{lm})}\leq\supnorm{f}+\epsilon$.
		Define 
		\begin{align*}
			\tilde V_{\delta,j}:=(\gamma_{lm}^\conf)^{-1}(V_{\delta,j}),\qquad\tilde V:=\cup_{j=1}^N\tilde V_{\delta,j},
		\end{align*}
		as depicted in Figure \ref{fig:1a}. 
		It is easily checked that $U\mapsto  (\gamma_{lm}^\conf)^{-1}(U)$ preserves volume. Hence $\vol(G^m\setminus\tilde V)<\epsilon/(3Q^2)$.
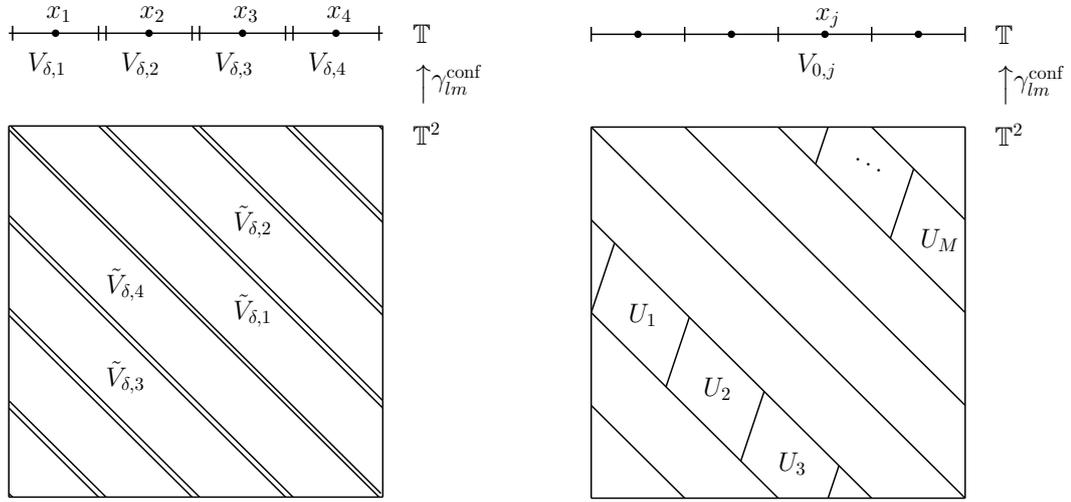
\begin{figure}[!htb]
	\centering
	\thicklines
	\begin{subfigure}{.49\textwidth}
		\centering
		\scalebox{0.7}{
		\begin{picture}(200,260)(0,0)
			\put(0,0){\line(1,0){200}}
			\put(0,0){\line(0,1){200}}
			\put(200,0){\line(0,1){200}}	
			\put(0,200){\line(1,0){200}}
	
			\put(0,2){\line(1,-1){2}}
			\put(0,48){\line(1,-1){48}}
			\put(0,52){\line(1,-1){52}}
			\put(0,98){\line(1,-1){98}}
			\put(0,102){\line(1,-1){102}}
			\put(0,148){\line(1,-1){148}}
			\put(0,152){\line(1,-1){152}}
			\put(0,198){\line(1,-1){198}}
			\put(1,200){\line(1,-1){198}}
			\put(48,200){\line(1,-1){152}}
			\put(52,200){\line(1,-1){148}}
			\put(98,200){\line(1,-1){102}}
			\put(102,200){\line(1,-1){98}}
			\put(148,200){\line(1,-1){52}}
			\put(152,200){\line(1,-1){48}}
			\put(198,200){\line(1,-1){2}}
			
			\put(52,60){\Large $\tilde V_{\delta,3}$}
			\put(52,110){\Large $\tilde V_{\delta,4}$}
			\put(120,96){\Large $\tilde V_{\delta,1}$}
			\put(120,144){\Large $\tilde V_{\delta,2}$}
			
			\put(0,250){\line(1,0){200}}
			\put(2,254){\line(0,-1){8}}
			\put(25,250){\circle*{4}}
			\put(20,258){\Large $x_1$}
			\put(10,230){\Large $V_{\delta,1}$}
			\put(48,254){\line(0,-1){8}}
			\put(52,254){\line(0,-1){8}}
			\put(75,250){\circle*{4}}
			\put(70,258){\Large $x_2$}
			\put(60,230){\Large $V_{\delta,2}$}
			\put(98,254){\line(0,-1){8}}
			\put(102,254){\line(0,-1){8}}
			\put(125,250){\circle*{4}}
			\put(120,258){\Large $x_3$}
			\put(110,230){\Large $V_{\delta,3}$}
			\put(148,254){\line(0,-1){8}}
			\put(152,254){\line(0,-1){8}}
			\put(175,250){\circle*{4}}
			\put(170,258){\Large $x_4$}
			\put(160,230){\Large $V_{\delta,4}$}
			\put(198,254){\line(0,-1){8}}	
			
			\put(216,245){\Large $\T$}
			\put(216,190){\Large $\T^2$}
			\put(220,211){\Large \rotatebox{90}{$\longrightarrow$}}
			\put(226,221){\Large $\gamma_{lm}^\conf$}
		\end{picture}
		}	
		\caption{The subspaces $\tilde V_{\delta,j}:=(\gamma_{lm}^\conf)^{-1}(V_{\delta,j})$.}
		\label{fig:1a}
	\end{subfigure}
	\begin{subfigure}{0.49\textwidth}
		\centering
		\scalebox{0.7}{
		\begin{picture}(200,260)(0,0)
			\put(0,0){\line(1,0){200}}
			\put(0,0){\line(0,1){200}}
			\put(200,0){\line(0,1){200}}	
			\put(0,200){\line(1,0){200}}
	
			\put(0,250){\line(1,0){200}}
			\put(0,254){\line(0,-1){8}}
			\put(25,250){\circle*{4}}
			\put(50,254){\line(0,-1){8}}
			\put(75,250){\circle*{4}}
			\put(100,254){\line(0,-1){8}}
			\put(125,250){\circle*{4}}
			\put(120,258){\Large $x_j$}
			\put(110,230){\Large $V_{0,j}$}
			\put(150,254){\line(0,-1){8}}
			\put(175,250){\circle*{4}}
			\put(200,254){\line(0,-1){8}}	
			
			\put(0,50){\line(1,-1){50}}
			\put(0,100){\line(1,-1){100}}
			\put(0,150){\line(1,-1){150}}
			\put(0,200){\line(1,-1){200}}
			\put(50,200){\line(1,-1){150}}
			\put(100,200){\line(1,-1){100}}
			\put(150,200){\line(1,-1){50}}
			
			\put(0,100){\line(1,3){12.5}}
			\put(40,60){\line(1,3){12.5}}
			\put(80,20){\line(1,3){12.5}}
			\put(126.666667,0){\line(1,3){5.833333}}
			\put(120,180){\line(1,3){6.666667}}
			\put(160,140){\line(1,3){12.5}}
			\put(20,95){\Large $U_1$}
			\put(60,55){\Large $U_2$}
			\put(100,15){\Large $U_3$}
			\put(140,175){\Large $\ddots$}
			\put(176,135){\Large $U_M$}
			\put(216,245){\Large $\T$}
			\put(216,190){\Large $\T^2$}
			\put(220,211){\Large \rotatebox{90}{$\longrightarrow$}}
			\put(226,221){\Large $\gamma_{lm}^\conf$}
		\end{picture}
		}
		\caption{The subspaces $U_1,\ldots,U_M$ for a fixed $j$.}
		\label{fig:1b}
	\end{subfigure}
	\caption{Dividing the configuration space $G^m\cong\T^2$ into small subspaces when $m$ has two edges (of different length) and $l$ has one (so that $G^l\cong\T$).}
\end{figure}
		Choose $\psi\in\H^m$ such that $\quadnorm{\psi}=1$ and
		\begin{align}\label{eq:eps over 3}
			\quadnorm{\Qm(f\circ\gamma_{lm})\psi}^2\geq \norm{\Qm(f\circ\gamma_{lm})}^2-\frac{\epsilon}{3}.
		\end{align}
		We now claim that there exists a point $q_0\in G^m$ such that 
		\begin{align}\label{eq:contr}
			\int_{G^m\setminus L_{q_0}(\tilde V)}\!dq\,|\Qm(f\circ\gamma_{lm})\psi(q)|^2\leq \vol(G^m\setminus\tilde V)Q^2<\frac{\epsilon}{3}.
		\end{align}
		Indeed, if there were no such $q_0\in G^m$, we would obtain
		\begin{align*}
			\vol(G^m\setminus\tilde V)Q^2&<\int_{G^m}dq_0\int_{G^m\setminus L_{q_0}(\tilde V)}\! dq\,|\Qm(f\circ\gamma_{lm})\psi(q)|^2\\
			&=\int_{G^m\setminus \tilde V}\!dq\int_{G^m}dq_0\,|\Qm(f\circ\gamma_{lm})\psi(q_0q)|^2\\
			&=\vol(G^m\setminus\tilde V)\quadnorm{\Qm(f\circ\gamma_{lm})\psi}^2\leq\vol(G^m\setminus\tilde V)Q^2,
		\end{align*}
		which is a contradiction. Therefore a $q_0\in G^m$ satisfying \eqref{eq:contr} does exist, and is fixed throughout the rest of the proof. Using \eqref{eq:eps over 3}, we conclude
		\begin{align}
			\norm{\Qm(f\circ\gamma_{lm})}^2-\frac{2\epsilon}{3}&\leq\quadnorm{\Qm(f\circ\gamma_{lm})\psi}^2-\frac{\epsilon}{3}\nonumber\\
			&<\sum_{j=1}^N\int_{L_{q_0}(\tilde V_{\delta,j})}dq\,|\Qm(f\circ\gamma_{lm})\psi(q)|^2.
		\label{eq:Rieffel0 phase 1}
		\end{align}
		For all $j\in\{1,\ldots,N\}$, we define 
		\begin{align*}
			\psi_j:=\psi \indicator_{L_{q_0}(\tilde V_{0,j})}.
		\end{align*}
		By $|\hbar|\leq\hbar_0$ and \eqref{eq:hbar0}, we have $\norm{\hbar\xi}<\delta/2$ for all $\xi\in\cup_k\supp(\mu_k)$. By using $\gamma_{lm}^\conf[S^{ml}(\hbar\xi)]=[\hbar\xi]$ we infer that $q\in L_{q_0}(\tilde V_{\delta,j})$ implies $q+S^{ml}(\hbar\xi)\in L_{q_0}(\tilde V_{0,j})$. Therefore, by using
		\begin{align}\label{eq:Qj}
			\Qm(f\circ\gamma_{lm})\psi(q)=\sum_{k=1}^K\int d\mu_k(\xi) g_k(\gamma_{lm}^\conf(q+\tfrac12 S^{ml}(\hbar\xi)))\psi(q+S^{ml}(\hbar\xi)),
		\end{align}
		we obtain that, for all $q\in L_{q_0}(\tilde{V}_{\delta,j})$,
		\begin{align*}
			\Qm(f\circ\gamma_{lm})\psi(q)=\Qm(f\circ\gamma_{lm})\psi_j(q).
		\end{align*}				
		Hence, \eqref{eq:Rieffel0 phase 1} becomes
		\begin{align*}
			\norm{\Qm(f\circ\gamma_{lm})}^2-\frac{2\epsilon}{3}\leq  \sum_{j=1}^N\int_{L_{q_0}(\tilde V_{\delta,j})}dq\,|\Qm(f\circ\gamma_{lm})\psi_j(q)|^2.
		\end{align*}
		By an argument similar to how we found $q_0$ (finding a contradiction if it would not exist) now using $\sum\|\psi_j\|_2^2=\|\psi\|_2^2$, we may fix a $j\in\{1,\ldots,N\}$ such that
		\begin{align}\label{eq:bound psi_m}
			\int_{L_{q_0}(\tilde V_{\delta,j})}dq\,|\Qm(f\circ\gamma_{lm})\psi_j(q)|^2\geq\quadnorm{\psi_j}^2\left(\norm{\Qm(f\circ\gamma_{lm})}^2-\frac{2\epsilon}{3}\right).
		\end{align}
		
		We fix subspaces $U_1,\ldots,U_M\subseteq\tilde V_{0,j}$ and points $y_1,\ldots,y_M\in (\gamma_{lm}^\conf)^{-1}(\{x_j\})\subseteq\tilde V_{0,j}$ such that $y_s\in U_s$ for all $s=1,\ldots,M$ and such that
		\begin{enumerate}[label=\textnormal{(\alph*)}]
			\item\label{item:a} $\bigcup_{s}U_s=\tilde V_{0,j}$ and the $U_s$ are disjoint;
			\item\label{item:b} $L_{[S^{ml}\xi]}(U_s\cap\tilde V_{\delta,j})\subseteq U_s$ for all $\xi\in B_{\delta/2}(0)\subseteq\g^l$;
			\item\label{item:c} $U_s\subseteq L_{y_s}(B^m)$ for all $s$. 
		\end{enumerate}
		An example of such sets is depicted in Figure \ref{fig:1b}.
		Define, for all $s$,
		\begin{align*}
			\psi_{j,s}:=\psi_j\indicator_{L_{q_0}(U_s)}=\psi\indicator_{L_{q_0}(U_s)}.
		\end{align*}
		By \ref{item:a}, we have
		\begin{align*}
			&\int_{L_{q_0}(\tilde V_{\delta,j})}dq\,|\Qm(f\circ\gamma_{lm})\psi_j(q)|^2
			=\sum_{s=1}^M\int_{L_{q_0}(U_s\cap \tilde V_{\delta,j})}dq\,|\Qm(f\circ\gamma_{lm})\psi_j(q)|^2.
		\end{align*}
		Notice that, for all $\xi\in\cup_k\supp(\mu_k)$, we have $\hbar\xi\in B_{\delta/2}(0)$. Therefore, by \ref{item:b}, we find that $q \in L_{q_0}(U_s\cap\tilde V_{\delta,j})$ implies that $q+ S^{ml}(\hbar\xi)\in L_{q_0}(U_s)$. Hence \eqref{eq:Qj} gives
		\begin{align*}
			&\int_{L_{q_0}(\tilde V_{\delta,j})}dq\,|\Qm(f\circ\gamma_{lm})\psi_j(q)|^2
			=\sum_{s=1}^M\int_{L_{q_0}(U_s\cap \tilde V_{\delta,j})}dq\,\Big|\Qm(f\circ\gamma_{lm})\psi_{j,s}(q)\Big|^2.
		\end{align*}
		Therefore, \eqref{eq:bound psi_m} gives
		\begin{align*}
			\sum_{s=1}^M\int_{L_{q_0}(U_s)}dq\,\Big|\Qm(f\circ\gamma_{lm})\psi_{j,s}(q)\Big|^2\geq\quadnorm{\psi_j}^2\left(\norm{\Qm(f\circ\gamma_{lm})}^2-\frac{2\epsilon}{3}\right).
		\end{align*}
		Again arguing by contradiction, and using that $\sum_{s=1}^M\quadnorm{\psi_{j,s}}^2=\quadnorm{\psi_j}^2$, we may fix an $s$ such that
		\begin{align}\label{eq:psi_ms ineq}
			\int_{L_{q_0}(U_s)}dq\,\Big|\Qm(f\circ\gamma_{lm})\psi_{j,s}(q)\Big|^2\geq\quadnorm{\psi_{j,s}}^2\left(\norm{\Qm(f\circ\gamma_{lm})}^2-\frac{2\epsilon}{3}\right).
		\end{align}
		Using the function $\psi_{j,s}\in L^2(G^m)$ we constructed, which is supported in $L_{q_0}(U_s)$, we can subsequently construct a function $\tilde\psi\in L^2(\g^m)$, as follows. First define $\breve{U}:=L_{q_0}(U_s)$ and $\breve{y}:=q_0+y_s\in \breve U$, so that the support of $q\mapsto \psi_{j,s}(\breve y+q)$ lies in $L_{\breve y}^{-1}(L_{q_0}(U_s))=L_{y_s}^{-1}(U_s)\subseteq B^m=[B_{1/2\pi}(0_\g)^m]$ by \ref{item:c} above. Define
		\begin{align*}
			\tilde{\psi}(X):=\begin{cases}
			\psi_{j,s}(\breve y+X)\quad&\text{if $X\in B^m$}\\
			0&\text{if $X\notin B^m$},
			\end{cases}
		\end{align*}
		which implies $\|\tilde\psi\|_2^2=\quadnorm{\psi_{j,s}}^2$. Using \eqref{eq:psi_ms ineq} we get
		\begin{align*}
		\left(\norm{\Qm(f\circ\gamma_{lm})}^2-\frac{2\epsilon}{3}\right)\big\|\tilde{\psi}\big\|_2^2&\leq \int_{\breve U}dq\,|\Qm(f\circ\gamma_{lm})\psi_{j,s}(q)|^2,
		\end{align*}
		in which we can use \eqref{eq:Qj} and expand the square of the absolute value of the sum over $k$. For brevity, we write $\dot g_k:=g_k(\gamma_{lm}^\conf(\breve y))$ and $g_{k,\xi}^q:=g_k(\gamma_{lm}^\conf(q+\tfrac12 S^{ml}(\hbar\xi)))-g_k(\gamma_{lm}^\conf(\breve y))$. We obtain
		\begin{align*}
		&\left(\norm{\Qm(f\circ\gamma_{lm})}^2-\frac{2\epsilon}{3}\right)\big\|\tilde{\psi}\big\|_2^2\\
		&\quad\leq \bigg|\sum_{k,k'=1}^K\int_{\breve U}dq\int d\overline{\mu_k}(\xi)\,\overline{(\dot g_k+g_{k,\xi}^q)\psi_{j,s}(q+S^{ml}(\hbar\xi))}\\
		&\qquad\quad\int d\mu_{k'}(\xi')(\dot g_{k'}+g_{k',\xi'}^q)\psi_{j,s}(q+S^{ml}(\hbar\xi'))\bigg|\\
		&\quad\leq \int_{\breve U}dq\,\bigg|\sum_{k=1}^K \int d\mu_k(\xi)\dot g_k\psi_{j,s}(q+S^{ml}(\hbar\xi))\bigg|^2
		\\&\qquad
		+\sum_{k,k'=1}^K\int d|\mu_k|(\xi)\int d|\mu_{k'}|(\xi')\Big(2|\dot g_k|+\sup_{q\in\breve U}|g_{k,\xi}^q|\Big)\sup_{q\in\breve U}|g_{k',\xi'}^q|\norm{\psi_{j,s}}^2.
		\end{align*}
		Because $\norm{\tfrac12\hbar\xi}<r$, because $U_s\subseteq\tilde V_{0,j}$ and because $d(x,x_j)\leq r$ for all $x\in V_{0,j}$ we can apply \eqref{eq:uniform continuity} to find, for all $k$ and $\xi\in\supp(\mu_k)$,
		\begin{align*}
			\sup_{q\in\breve U}|g_{k,\xi}^q|<\frac{\epsilon}{12Q\sum_k\absnorm{\mu_k}}.
		\end{align*}
		Therefore, and by Lemma \ref{lem:supnorm from operators},
		\begin{align*}
			&\left(\norm{\Qm(f\circ\gamma_{lm})}^2-\frac{2\epsilon}{3}\right)\big\|\tilde{\psi}\big\|_2^2\\
			&\quad\leq \int_{\g^m} dX\bigg|\sum_{k=1}^K \int d\mu_k(\xi)g_k(\gamma^\conf_{lm}(\breve y)) \tilde\psi(X+S^{ml}(\hbar \xi))\bigg|^2\\
			&\qquad+\sum_{k,k'=1}^K\absnorm{\mu_k}4\supnorm{g_k}\absnorm{\mu_{k'}}\sup_{\xi'\in\supp(\mu_{k'})}\sup_{q\in\breve U}|g_{k',\xi'}^q|\quadnorm{\psi_{j,s}}^2\\
			&\quad\leq \sup_{q\in G^m}\norm{\sum_{k=1}^K\int d\mu_k(\xi)g_k(\gamma^\conf_{lm}(q))L^*_{S^{ml}(\hbar\xi)}}^2\big\|\tilde\psi\big\|_2^2
			+\frac{\epsilon}{3}\big\|\tilde\psi\big\|_2^2\\
			&\quad=\left(\supnorm{f}+\frac{\epsilon}{3}\right)\big\|\tilde\psi\big\|_2^2.
		\end{align*}
		By \eqref{eq:j or k same same} we conclude that $\norm{\Qn(f\circ\gamma_{ln})}^2\leq\supnorm{f}+\epsilon$. Since $n\geq l$ was arbitrary, we have $\norm{\Q(f\circ\gamma_{l})}^2\leq\supnorm{f}+\epsilon$, which concludes the proof.
\end{proof}

\subsection{Rieffel's condition away from zero}
Now that we have continuity of $\hbar\mapsto\norm{\Q(f\circ\gamma_{l})}$ at $\hbar=0$, we are left to prove continuity at an arbitrary $\hbar_1\in[-1,1]\setminus\{0\}$. In the rest of the chapter, we fix such an $\hbar_1$, as well as a function $f\in\M_0^{l}$, expanded as $f=\sum_{k=1}^{K}g_k\otimes \check\mu_k$. 


The reason that Rieffel's condition away from zero holds in the infinite dimensional case, as opposed to the case on a finite lattice (see Lemma \ref{lem:failure of Rieffel's condition} for a counterexample) is that $\norm{Q^\infty_{\hbar_1}(f\circ\gamma_l)}$ is given by a supremum over lattices $m\geq l$ as shown in \eqref{eq:Q^infty sup over Q^j}. Better yet: it is also given by a supremum over lattices $m\geq l^R$, with $l^R$ from Definition \ref{def:l_R}. If we choose $R$ large enough, the components of the $S^{l^Rl}(\xi)$'s, for $\xi\in\cup_k\supp(\mu_k)$, become arbitrarily small. We take advantage of this fact by the following construction.
	For every edge $e\in l$, we choose a single edge $e'\in l^R$ that lies inside $e$. The edge $e'$ has a length $1/R$ times the length of $e$, 
	so we have $S^{l^Rl}(\xi)_{e'}=\frac{1}{R}\xi_e$.
	We then define the projection
	\begin{align*}
		\chi_{ll^R}:G^{l^R}\to G^{l},\quad \chi_{ll^R}(q)_e:=q_{e'},
	\end{align*}		
	and note that it satisfies $\chi_{ll^R}[S^{l^Rl}(\xi)]=[\tfrac1R  \xi]$.
	
	Similarly to the proof of Proposition \ref{prop:Rieffel0}, we will define subsets of $G^l$, which in volume approximate the whole of $G^l$ but are topologically better behaved than $G^l$.
\begin{defi}\label{def:U_delta en V_delta}
For all $\delta\geq0$, define
	\begin{align*}
		U_\delta:=\{\xi \in\g^{l}:~\xi_e\in(-\tfrac12+\tfrac12\delta,\tfrac12-\tfrac12\delta)^n \text{ for all $e\in l$}\}.
	\end{align*}
	In particular, $U_0$ is the open unit cube around 0. Using $U_\delta$, we define a subset $V_\delta\subseteq G^{l}$ with volume $\vol(V_\delta)=(1-\delta)^{n|l|}$, $n=\dim G$, by setting
	\begin{align*}
		V_\delta:=\{[\xi]\in G^{l}:~\xi\in U_\delta\}.
	\end{align*}
\end{defi}
Using these we will define subsets of $G^m$, for a particular class of lattices $m\geq l^R$. 

\begin{figure}[!htb]
	\centering
	\thicklines
	\begin{subfigure}{.49\textwidth}
		\centering
		\scalebox{0.7}{
		\begin{picture}(200,320)
			\put(50,250){\line(0,1){50}}
			\put(50,250){\circle*{6}}
			\put(50,300){\circle*{6}}
			\put(100,250){\line(0,1){50}}
			\put(100,250){\circle*{6}}
			\put(100,275){\circle*{6}}
			\put(100,300){\circle*{6}}
			\put(150,250){\line(0,1){50}}
			\put(150,250){\circle*{6}}
			\put(150,275){\circle*{6}}
			\put(150,300){\circle*{6}}
			\put(47,220){\Large $l$}
			\put(97,220){\Large $l^R$}
			\put(145,220){\Large $m$}
		
			\put(0,0){\line(1,0){200}}
			\put(0,0){\line(0,1){200}}
			\put(200,0){\line(0,1){200}}	
			\put(0,200){\line(1,0){200}}

			\put(0,85){\line(1,0){200}}
			\put(0,115){\line(1,0){200}}
			
			\multiput(0,93)(21,0){10}{\line(1,0){11}}
			\multiput(0,107)(21,0){10}{\line(1,0){11}}
	
		\put(-35,155){\Large $\tilde V_{\delta_1}$}
		\put(219,146){\Large $\tilde V_{\delta_2}$}
\put(-10,85){\line(1,0){6}}
\put(-10,85){\line(0,-1){85}}
\put(-10,115){\line(1,0){6}}
\put(-10,115){\line(0,1){85}}
\put(-10,160){\line(-1,0){5}}

\put(210,93){\line(-1,0){6}}
\put(210,93){\line(0,-1){93}}
\put(210,107){\line(-1,0){6}}
\put(210,107){\line(0,1){93}}
\put(210,150){\line(1,0){5}}	
			
			\multiput(10,67)(20,0){10}{\vector(0,1){12}}
			\multiput(10,133)(20,0){10}{\vector(0,-1){12}}
			\multiput(10,33)(20,0){10}{\vector(0,1){7}}
			\multiput(10,167)(20,0){10}{\vector(0,-1){7}}
	
			\multiput(10,0)(20,0){10}{\circle*{3}}
			\multiput(10,200)(20,0){10}{\circle*{3}}

		\end{picture}
		}	
		\caption{Subdivision where $R=2$ and $m=l^R$.}
		\label{fig:2a}
	\end{subfigure}
	\begin{subfigure}{0.49\textwidth}
		\centering
		\scalebox{0.7}{
		\begin{picture}(200,320)
			\put(50,250){\line(0,1){50}}
			\put(50,250){\circle*{6}}
			\put(50,300){\circle*{6}}
			\put(100,250){\line(0,1){50}}
			\put(100,250){\circle*{6}}
			\put(100,300){\circle*{6}}
			\put(150,250){\line(0,1){50}}
			\put(150,250){\circle*{6}}
			\put(150,283){\circle*{6}}
			\put(150,300){\circle*{6}}
			\put(47,220){\Large $l$}
			\put(97,220){\Large $l^R$}
			\put(145,220){\Large $m$}

			\put(0,0){\line(1,0){200}}
			\put(0,0){\line(0,1){200}}
			\put(200,0){\line(0,1){200}}	
			\put(0,200){\line(1,0){200}}
			
			\put(0,85){\line(1,-1){85}}
			\put(0,115){\line(1,-1){115}}
			
			\put(85,200){\line(1,-1){115}}
			\put(115,200){\line(1,-1){85}}

	
			\multiput(80,190)(20,-20){6}{\vector(1,2){5}}
			\multiput(50,190)(20,-20){8}{\vector(1,2){3}}
			\multiput(16.5,183.5)(20,-20){10}{\circle*{3}}
			\multiput(23,137)(20,-20){7}{\vector(-1,-2){3}}
			\multiput(13,117)(20,-20){6}{\vector(-1,-2){5}}
			
			\multiput(0,70)(20,-20){4}{\vector(1,2){5}}
			\multiput(10,30)(20,-20){2}{\vector(1,2){3}}	
			\multiput(163,197)(20,-20){2}{\vector(-1,-2){3}}
			\multiput(133,197)(20,-20){4}{\vector(-1,-2){5}}
			
			\multiput(96,197)(15,-15){7}{\line(1,-1){10}}
			\multiput(112,195)(15,-15){6}{\line(1,-1){10}}
			\multiput(5,102)(15,-15){7}{\line(1,-1){10}}
			\multiput(3,90)(15,-15){6}{\line(1,-1){10}}
		\end{picture}
		}
		\caption{Subdivision where $R=1$ and $m>l^R$.}
		\label{fig:2b}
	\end{subfigure}
	\caption{Choosing subsets $\tilde V_\delta$ ($\delta>0$) within the configuration space $G^m\cong \T^2$ such that, when $\delta_1>\delta_2$, $\tilde V_{\delta_1}\subseteq\tilde V_{\delta_2}$. The bijection $F:\tilde V_{\delta_1}\to\tilde V_{\delta_2}$ expands the subset $\tilde V_{\delta_1}$ onto $\tilde V_{\delta_2}$ along the direction of $S^{ml}\circ\uvp$, as indicated by the arrows. Here $m$ has two edges (of possibly different length) and $l$ has one.}
	\label{fig:2}
\end{figure}
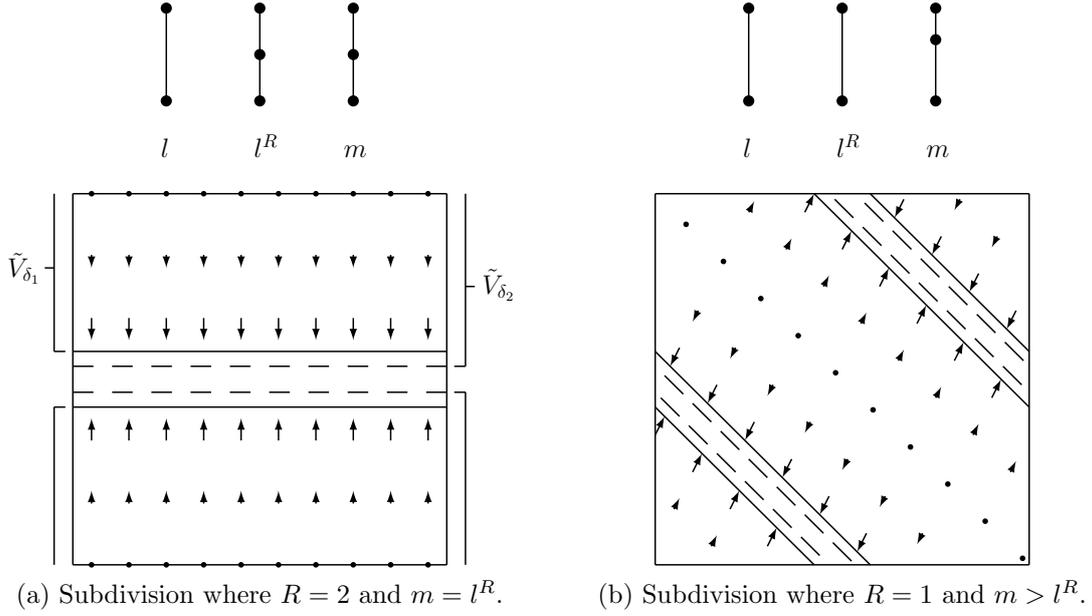

\begin{lem}\label{lem:phi smooth}
	Given a  lattice $m$ obtained from $l^R$ by subdivisions (hence in particular $l\leq l^R\leq m$) the map $\varphi_{ll^Rm}:G^m\to G^{l}$ defined by
		$$\varphi_{ll^Rm}:=\chi_{ll^R}\circ\gamma^\conf_{l^Rm}:G^m\to G^l$$
	is smooth, and $U\mapsto\varphi_{ll^Rm}^{-1}(U)$ preserves volume.
\end{lem}
\begin{proof}
	By first considering the elementary steps of adding and subdividing an edge, one finds that both $\chi_{ll^R}$ and $\gamma_{l^Rm}^\conf$ are smooth and preserve volume by inverse image.
\end{proof}

	 For any $\delta\geq0$, we set
	\begin{align*}
		\tilde{V}_{\delta}:=\varphi_{ll^Rm}^{-1}(V_{\delta})\subseteq G^m.
	\end{align*}
	For $\hbar\in[-1,1]$ of the same sign as $\hbar_1$, we define a map $F:\tilde V_0\to G^m$ by
	\begin{align}\label{eq:F}
		F(q):=q+R\left(\frac{\hbar}{\hbar_1}-1\right)S^{ml}(\uvp(q)),
	\end{align}
	where $\uvp:\tilde V_0\to\g^{l}$ is defined by
	\begin{align*}
		\uvp(q):=\xi\in\g^{l}\quad\text{if}\quad \varphi_{ll^Rm}(q)=[\xi]\in G^{l}\quad\text{for}\quad \xi\in U_0.
	\end{align*}
	
\begin{lem}\label{lem:F}
	Let $\hbar\in[-1,1]$ be of the same sign as $\hbar_1$ and let $\delta_1,\delta_2\in(0,1)$ satisfy
	\begin{align}\label{eq:hbars deltas}
		\hbar(1-\delta_1)=\hbar_1(1-\delta_2).
	\end{align}
	Then $F$ restricts to a diffeomorphism $F:\tilde{V}_{\delta_1}\to\tilde V_{\delta_2}$ satisfying for all $q\in\tilde V_{\delta_1}$:
	\begin{align*}
		|\det d_qF|=(\hbar/\hbar_1)^{n|l|}.
	\end{align*}
	Moreover, when $q+tS^{ml}(\hbar_1\xi)\in\tilde V_{\delta_1}$ for all $t \in[0,1]$, we have	
	\begin{align}\label{eq:the point of F}
		F(q+S^{ml}(\hbar_1\xi))=F(q)+S^{ml}(\hbar\xi).
	\end{align}
\end{lem}
\begin{proof}
	By Lemma \ref{lem:phi smooth}, $\uvp$ is smooth, which implies that $F$ is smooth. It follows from \eqref{eq:F} that the map $\tilde V_{\delta_1}\to \text{End}(\g^m)$,  $q\mapsto d_qF$ is constant, and that
	\begin{align*}
		\varphi_{ll^Rm}(F[x])=\left[(\hbar/\hbar_1)x\right],
	\end{align*}		
	which by \eqref{eq:hbars deltas} implies that $F:\tilde V_{\delta_1}\to\tilde V_{\delta_2}$ is bijective. Therefore $F$ is a diffeomorphism and 
$|\det d_q F|$ is given by
		$\vol(\tilde V_{\delta_1})/\vol(\tilde V_{\delta_2})=\vol(V_{\delta_1})/\vol(V_{\delta_2})=(\hbar/\hbar_1)^{n|l|},$
	by use of Definition \ref{def:U_delta en V_delta} and Lemma \ref{lem:phi smooth}. The last statement of the lemma is a simple check.
\end{proof}
In Figure \ref{fig:2}, two key examples show how $F$ maps the points of $\tilde V_{\delta_1}$ to $\tilde V_{\delta_2}$. We now have all the tools we need to establish the last part of our main result.

\begin{prop}\label{prop:Rieffel1}
	\textbf{(Rieffel's condition away from 0)} For each $f\in\A_0^\infty$, and each $\hbar_1\in[-1,1]\setminus\{0\}$, we have
			\begin{align*}
				\lim_{\hbar\to\hbar_1}\norm{\Q(f)}=\norm{Q^\infty_{\hbar_1}(f)}.
			\end{align*}
\end{prop}
\begin{proof}
	Let $f\in\M_0^{l}$ for some $l\in\L$, write $f=\sum_{k=1}^K g_k\otimes \check\mu_k$ for $g_k\in C^\infty(G^{l})$ and $\mu_k$ a finite complex measure supported in $B^{l}$, and let $\hbar_1\in[-1,1]\setminus\{0\}$. 
	By \cite[Proposition 24]{vNS20} we already have
	\begin{align*}
		\lim_{\hbar\to\hbar_1}\norm{\Q(f)}\geq\norm{Q^\infty_{\hbar_1}(f)}.
	\end{align*}
	In order to also prove 
	\begin{align*}
		\lim_{\hbar\to\hbar_1}\norm{\Q(f)}\leq\norm{Q^\infty_{\hbar_1}(f)},
	\end{align*}
	we let $\epsilon>0$ be arbitrary. By Definition \ref{def:U_delta en V_delta} and Lemma \ref{lem:phi smooth}
	we can choose $\delta\in(0,1)$ small enough such that, with $Q$ from \eqref{eq:Q},
	\begin{align}\label{eq:vol assump hbar1}
		\vol(G^m\setminus\tilde V_\delta)=\vol(G^{l}\setminus V_\delta)<\frac{\epsilon}{3Q^2}.
	\end{align}
	Choose a natural number $R\in\N$, big enough such that
	\begin{align}\label{eq:R}
		\frac{1}{R}\sqrt{|l|}<\delta,
	\end{align}
	where $|l|$ denotes the number of edges in $l$.
	For all $\xi\in\cup_k\supp(\mu_k)$, we have $\hbar_1\xi\in B^{l}$, which is an open set. We choose a number $c>0$ such that for all $\hbar\in[-1,1]$ with $|\hbar-\hbar_1|<c$ it holds that
	\begin{align}
		 1-\frac{\hbar}{\hbar_1}(1-\delta)\in(0,1);\qquad&\hbar\xi\in B^{l}\text{ for all $\xi\in\cup_k\supp(\mu_k)$;}\nonumber\\
		  1-\frac{\hbar}{\hbar_1}\left(1-\frac{\delta}{2}\right)\in(0,\delta);\qquad&\sgn(\hbar)=\sgn(\hbar_1);\nonumber\\
		  1-\frac{\hbar}{\hbar_1}\left(1-\frac{\delta}{4}\right)\in(0,1); \qquad&\supnorm{\nabla g_k}R\left|\frac{\hbar}{\hbar_1}-1\right|\sqrt{|l|}\leq \frac{\epsilon}{6Q\sum_k\absnorm{\mu_k}}.\label{eq:assumptions c}
	\end{align}
	Let $\hbar\in [-1,1]$ be arbitrary such that $|\hbar-\hbar_1|<c$. By 
	\eqref{eq:Q^infty sup over Q^j} and \eqref{eq:j or k same same}, it suffices to prove
	\begin{align*}
		\norm{Q^m_{\hbar}(f\circ\gamma_{lm})}\leq\norm{Q^m_{\hbar_1}(f\circ\gamma_{lm})}+\epsilon,
	\end{align*}
	for all lattices $m\geq l^R$ obtained from $l^R$ purely by subdivision of edges. We let $m$ be such a lattice in the following.
	We choose a $\psi\in\H^m$ such that $\quadnorm{\psi}=1$ and
	\begin{align*}
		\norm{\Qm(f\circ\gamma_{lm})}^2-\frac{\epsilon}{3}\leq \quadnorm{\Qm(f\circ\gamma_{lm})\psi}^2.
	\end{align*}
	By a proof by contradiction (as we gave several times in the proof of Proposition \ref{prop:Rieffel0}) using \eqref{eq:vol assump hbar1} we obtain a point $q_0\in G^m$ such that
	\begin{align*}
		\int_{L_{q_0}(G^m\setminus\tilde{V}_\delta)}dq\,|\Qm(f\circ\gamma_{lm})\psi(q)|^2\leq \frac{\epsilon}{3}.
	\end{align*}
	Therefore
	\begin{align}\label{eq:2epsilon/3}
		\norm{\Qm(f\circ\gamma_{lm})}^2-\frac{2\epsilon}{3}\leq \int_{L_{q_0}(\tilde V_\delta)}dq\left|\Qm(f\circ\gamma_{lm})\psi(q)\right|^2.
	\end{align}		
	

	Using \eqref{eq:assumptions c}, define $F_{q_0}:L_{q_0}(\tilde V_{\delta/4})\to G^m$ by $F_{q_0}(q):=F(q-q_0)+q_0$, so that when $q+tS^{ml}(\hbar_1\xi)\in L_{q_0}(\tilde V_{\delta/4})$ for all $t\in[0,1]$, \eqref{eq:the point of F} gives
	\begin{align}\label{eq:Fqo}
		F_{q_0}(q+S^{ml}(\hbar_1\xi))
		&=F_{q_0}(q)+S^{ml}(\hbar\xi).
	\end{align}
	As $\tilde V_\delta\subseteq \tilde V_{\delta/4}$, we may in particular define $\tilde\psi\in\H^m= L^2(G^m)$ by
	\begin{align*}
		\tilde\psi(q):=\begin{cases}
		\sqrt{(\hbar/\hbar_1)^{n|l|}}\psi(F_{q_0}(q))\quad&\text{if }q\in L_{q_0}(\tilde V_{\delta})\\
		0\quad&\text{if }q\in G^m\setminus L_{q_0}(\tilde V_{\delta}).
		\end{cases}
	\end{align*}
	From Lemma \ref{lem:F} and the first assumption of \eqref{eq:assumptions c} we obtain that $\|\tilde\psi\|_2^2\leq\norm{\psi}_2^2=1$.

	We have $\hbar_1\xi\in B^{l}$, so $\norm{\hbar_1\xi}\leq \sqrt{|l|}/2$ for all $\xi\in\cup_k\supp(\mu_k)$. By \eqref{eq:R}, and because $\varphi_{ll^Rm}[S^{ml}(X)]=[\tfrac1R X]$, we have
	\begin{align}\label{eq:what xi's do to V's}
		\norm{\uvp[S^{ml}(\hbar_1\xi)]}\leq \delta/2.
	\end{align}
	Therefore $q+S^{ml}(\hbar_1\xi)\in \tilde V_{\delta}$ implies $q\in\tilde V_{\delta/2}$. Translating this implication with $L_{q_0}$, we obtain,
	\begin{align}
		\norm{Q^m_{\hbar_1}(f\circ\gamma_{lm})\tilde\psi}^2=\int_{G^m}&dq\left|\sum \int d\mu_k(\xi)g_k(\gamma^\conf_{lm}(q+\tfrac12 S^{ml}(\hbar_1\xi)))\tilde\psi(q+S^{ml}(\hbar_1\xi))\right|^2\nonumber\\
		=\int_{L_{q_0}(\tilde V_{\delta/2})}&dq\left|\sum \int d\mu_k(\xi)\,\dot{g}_k^q\,\tilde\psi(q+S^{ml}(\hbar_1\xi))\right|^2\label{eq:bigint hbar1}
	\end{align}
	when we define, for all $q\in L_{q_0}(\tilde V_{\delta/2})$,
	\begin{align*}
		\dot g_k^q&:=g_k(\gamma_{lm}^\conf(q+\tfrac12 S^{ml}(\hbar_1\xi)));\\ 
		\overline g_{k,\xi}^q&:=g_k(\gamma^\conf_{lm}(F_{q_0}(q+\tfrac12 S^{ml}(\hbar_1\xi))))-g_k(\gamma_{lm}^\conf(q+\tfrac12 S^{ml}(\hbar_1\xi))).
	\end{align*}
	We choose $\delta_4$ such that $F:\tilde{V}_{\delta/2}\to\tilde V_{\delta_4}$ is a bijection by Lemma \ref{lem:F}, i.e., we define $\delta_4:=1-\hbar/\hbar_1(1-\delta/2)$. By \eqref{eq:assumptions c}, we have $\delta_4\in(0,\delta)$, and therefore $\tilde V_{\delta}\subseteq\tilde V_{\delta_4}$. When we apply a change of variables $q\mapsto F_{q_0}(q)$ to \eqref{eq:2epsilon/3} we obtain, by Lemma \ref{lem:F} and \eqref{eq:Fqo},
	\begin{align}
		&\norm{\Qm(f\circ\gamma_{lm})}^2-\frac{2\epsilon}{3}\nonumber\\
		&\quad\leq \int_{L_{q_0}(\tilde V_{\delta_4})} dq\left|\sum \int d\mu_k(\xi)g_k(\gamma_{lm}^\conf(q+\tfrac12 S^{ml}(\hbar\xi)))\psi\big(q+S^{ml}(\hbar\xi)\big)\right|^2\nonumber\\
		&\quad= \left(\frac{\hbar}{\hbar_1}\right)^{n|l|}\int_{L_{q_0}(\tilde V_{\delta/2})} dq\,\bigg|\sum\int d\mu_k(\xi)
		 (\dot{g}_k^q+\overline{g}_{k,\xi}^q)\psi\big(F_{q_0}(q)+S^{ml}(\hbar\xi)\big)\bigg|^2\nonumber\\
		&\quad=\int_{L_{q_0}(\tilde V_{\delta/2})}dq\left|\sum\int d\mu_k(\xi) (\dot{g}_k^q+\overline{g}_{k,\xi}^q)\tilde\psi(q+S^{ml}(\hbar_1\xi))\right|^2.\label{eq:bigint hbar}
	\end{align}
	The only difference between \eqref{eq:bigint hbar1} and \eqref{eq:bigint hbar} is now the appearance of $\overline g^q_{k,\xi}$ in the latter expression.
	For all $q\in L_{q_0}(\tilde V_{\delta/2})$, $\xi\in\cup_k\supp(\mu_k)$, and $t\in[0,1]$, we have $q+\tfrac t2 S^{ml}(\hbar_1\xi)\in L_{q_0}(\tilde V_{\delta/4})$ by \eqref{eq:what xi's do to V's}.
	By Lemma \ref{lem:F} and \eqref{eq:assumptions c}, we obtain
	\begin{align*}
		|\overline{g}^q_{k,\xi}|&\leq\supnorm{\nabla g_k} d\left(\gamma_{lm}^\conf(F_{q_0}(q+\tfrac12 S^{ml}(\hbar_1\xi))),\gamma_{lm}^\conf(q+\tfrac12 S^{ml}(\hbar_1\xi))\right)\\
		&\leq \supnorm{\nabla g_k}\norm{R\left(\frac{\hbar}{\hbar_1}-1\right)\uvp(q-q_0+\tfrac12 S^{ml}(\hbar_1\xi)}\\
		&\leq \supnorm{\nabla g_k}R\left|\frac{\hbar}{\hbar_1}-1\right|\frac{\sqrt{|l|}}{2}
		\leq\frac{\epsilon}{12Q\sum_k\absnorm{\mu_k}},
	\end{align*}
	for all $q,k,$ and $\xi$. Expanding the square of the absolute value in \eqref{eq:bigint hbar}, and using \eqref{eq:bigint hbar1},
		\begin{align*}
		&\norm{\Qm(f\circ\gamma_{lm})}^2-\frac{2\epsilon}{3}\\
		&\quad\leq \int_{L_{q_0}(\tilde V_{\delta/2})}dq\,\left|\sum_{k=1}^K \int d\mu_k(\xi)\dot\, g^q_k\,\tilde\psi(q+S^{ml}(\hbar_1\xi))\right|^2
		\\&\qquad\quad
		+\sum_{k,k'=1}^K\int d|\mu_k|(\xi)\Big(2\sup_q |\dot g^q_{k}|+\sup_{q}|\overline g_{k,\xi}^q|\Big)\int d|\mu_{k'}|(\xi')\sup_{q}|\overline g_{k',\xi'}^q|\big\|\tilde\psi\big\|_2^2\\
		&\quad\leq \quadnorm{Q^m_{\hbar_1}(f\circ\gamma_{lm})\tilde\psi}^2
		+4Q\sum_{k=1}^K\absnorm{\mu_k}\sup_{\xi}\sup_{q}|\overline g_{k,\xi}^q|\\
		&\quad\leq \norm{Q^m_{\hbar_1}(f\circ\gamma_{lm})}^2
		+\frac{\epsilon}{3}.
		\end{align*}
		Therefore $\norm{\Qm(f\circ\gamma_{lm})}^2\leq \norm{Q^m_{\hbar_1}(f\circ\gamma_{lm})}^2+\epsilon$ for all lattices $m\geq l^R\geq l$, which is what we needed to prove.
\end{proof}
We conclude that $\Q:\A^\infty_0\to\A^\infty_\hbar$ is a strict deformation quantization:

\begin{proof}[Proof of Theorem \ref{thm:sdq infinity}]
	Combine Propositions \ref{prop:algebra}, \ref{prop:star-preserving}, \ref{prop:injective}, \ref{prop:von Neumann}, \ref{prop:Dirac}, \ref{prop:Rieffel0}, and \ref{prop:Rieffel1}.
\end{proof}

\appendix
\chapter{Research Data Management}
\markboth{RESEARCH DATA MANAGEMENT}{}
	This thesis research has been carried out under the institute research data management policy of IMAPP, Radboud University.\\~\\
	\noindent As required by this policy, here follows the set of persistent identifiers for all datasets used in this thesis, alongside the corresponding chapters in which they appear:
	~\\
	~\\
	\begin{align*}
		\emptyset.
	\end{align*}
	~\\
	~\\
	This set is empty, for no data was used. Everything needed for peers to reproduce the results in this thesis are the accompanying proofs, and a bit of stamina. \\
	
	\noindent Although certainly grateful for the existence of $\emptyset$, the author wishes to apologize for writing down a triviality. 
	

\chapter{Samenvatting}
In dit proefschrift heb ik stellingen bewezen met behulp van wiskundige objecten genaamd C*-algebras. Zoals vaak in de wiskunde weten we na het bewijzen van deze stellingen nog niet direct waar ze allemaal toe zullen leiden. Wel heb ik dankzij mijn proefschriftonderzoek een heel concrete vraag beter leren begrijpen, namelijk de vraag: Wat is licht? 

Licht is elke dag om ons heen, en is voor heel veel mensen belangrijk, zonder dat ze precies kunnen vertellen waar het uit bestaat. Sommigen zullen je vertellen dat licht een deeltje is (klinkt heel logisch als je wel eens over de lichtdeeltjes genaamd fotonen hebt gehoord) terwijl sommigen zullen zeggen dat het een golf is (wat weer klopt met het feit dat licht een golflengte heeft, die bepaalt of het UV-licht, infrarood licht, of zichtbaar licht is, en van welke kleur dan wel). 

\subsection*{Deel I}
Licht is als een strandganger die na wat overgooien zijn strandbal de zee in ziet waaien. De strandganger wil de dobberende bal zo snel mogelijk pakken, maar twijfelt: zal ik recht op de bal af rennen en in dezelfde richting het laatste stuk zwemmen? Of ren ik eerst naar de plek in de branding waar de bal het dichtstbij is en zwem ik dan loodrecht op de waterlijn naar de bal toe? In het eerste geval is de totale afstand het kleinst, maar in het tweede geval is de te zwemmen afstand het kleinst, en voor de strandganger duurt zwemmen nou eenmaal langer dan rennen. 

Na een rekensom komt de strandganger met de oplossing: de snelste route ligt tussen de twee net genoemde routes in:  ren eerst naar een specifieke plek op de branding die afhangt van de rensnelheid en zwemsnelheid, en verander dan een beetje van richting: het pad `breekt'. Deze situatie is analoog aan de breking van licht wanneer het overgaat van lucht naar water, of van brillenglas naar lucht, naar je oog. 

De route die licht aflegt is altijd gebaseerd op het lokaal optimaliseren van een bepaalde optelsom van factoren, samengevat in een zogenaamde \textit{actiefunctionaal}. In veel gevallen staat het optimaliseren van deze actiefunctionaal gelijk aan het minimaliseren van de reistijd van punt A naar punt B, en dat is waarom het terughalen van de strandbal een goede analogie is voor breking van het licht: in water is licht bijvoorbeeld langzamer dan in lucht, net als de strandganger.


Actiefunctionalen zijn belangrijk in de natuurkunde omdat de optimalisatie hiervan niet alleen licht maar alle deeltjes en krachten kan beschrijven, als je maar de goede actiefunctionaal kiest. Afhankelijk van de ingewikkeldheid van de actiefunctionaal kan het wiskundig heel moeilijk zijn om te bepalen wat ermee gebeurt als je de invoer (zoals het pad naar de strandbal) varieert. Ik bewijs in mijn proefschrift stellingen over het vari\"eren van een uitdagende klasse van actiefunctionalen. Daarbij combineer ik verschillende wiskundige vakgebieden zoals \textit{Niet-commutatieve Meetkunde}, \textit{Multiple Operator Integration} en \textit{Cyclic Cohomology}. Uiteindelijk blijkt dat een actiefunctionaal genaamd de \textit{spectrale actie} (ook wel de spectrale actiefunctionaal genoemd) zich heel intrigerend gedraagt als je hem varieert. De spectrale actie is een actiefunctionaal die enkel en alleen van `spectrale' informatie afhangt, zoals de golflengte van licht. Het kan alle vier de fundamentele krachten perfect beschrijven, zolang we kwantummechanica even buiten  beschouwing laten. 

Ik bewijs in dit proefschrift dat de gevarieerde spectrale actie te schrijven is als de originele spectrale actie plus een oneindigheid van termen die alsmaar en alsmaar ingewikkelder worden. Maar gelukkig blijkt er structuur in te herkennen, en kun je een uitdrukking geven voor term nummer $k$, die op een simpele manier afhangt van $k$; in het bijzonder maakt het uit of $k$ even is of oneven. Sterker nog, die even en oneven termen hebben dezelfde vorm als beroemde actiefunctionalen die in de natuurkunde gebruikt worden, en waar (in het oneven geval iets rigoreuzer dan in het even geval) kwantumversies van bestaan.

Een van de theorie\"en die volledig beschreven kan worden door de spectrale actie is elektromagnetisme. Inderdaad, de theorie van elektriciteit en magnetisme. Maar deze theorie geeft ook een verklaring voor licht, namelijk als een golf in het elektromagnetische veld.

In deze zin hebben we in het eerste deel van dit proefschrift in het bijzonder het verschijnsel licht een beetje beter begrepen. Daarnaast hebben we een glimp opgevangen van kwantummechanica.

\subsection*{Deel II}
In Deel \ref{part:II} van dit proefschrift onderzoek ik een manier waarop elektromagnetisme ooit misschien verenigd zou kunnen worden met kwantummechanica, namelijk door eerst elektromagnetisme te benaderen met behulp van een rooster.


Eerst tekenen we een vierkant rooster over een deel van de ruimte waarop we een elektromagnetisch veld willen bekijken.
Om het elektromagnetische veld te begrijpen is het handig om te kijken wat voor effect het heeft op elektronen die zich in dat veld bewegen. Het blijkt heel lastig om een kwantumversie van het systeem bestaande uit elektronen en een elektromagnetisch veld wiskundig precies te maken.
 Daarom bekijken we eerst wat er gebeurt als we elektronen niet overal, maar alleen langs de lijntjes van het rooster laten bewegen.

Een elektron gedraagt zich dan als een mier die zich beweegt door een netwerk van parallele tunnels. Omdat een mier (bij benadering) geen zwaartekracht voelt weet hij op elk gegeven moment niet of hij zich boven, onder, of ergens aan de zijkanten van de tunnel bevindt. Oftewel, de mier heeft geen idee van zijn ori\"entatie. Om de anologie met het elektromagnetisch veld te kunnen maken hebben de tunnels `loopgroeven', parallele maar mogelijk gekromde lijnen langs de rand van de tunnel die bepalen hoe de ori\"entatie van een mier verandert als hij van de ene naar de andere kant van de tunnel loopt. Als de mier terugloopt, loopt hij langs dezelfde loopgroef weer terug. Omdat een mier even makkelijk langs muren en over het plafond loopt heeft de mier niet door dat hij van ori\"entatie verandert. 

Er is slechts \'e\'en manier waarop de mier, met behulp van een bevriende mier, erachter kan komen dat het tunnelstelsel \"uberhaupt bestaat, dat wil zeggen, dat er meer is dan alleen zijn eigen loopgroef. Als de twee mieren op dezelfde plek in het tunnelstelsel en in dezelfde loopgroef starten en vervolgens langs verschillende routes door het tunnelstelsel lopen kan het gebeuren dat de ene mier boven en de andere mier onder uitkomt. De mieren ruiken dat ze op dezelfde plek zijn, behalve dat ze onderling van ori\"entatie verschillen. Pas als ze langs de weg teruglopen waarlangs ze gekomen zijn komen ze elkaar weer tegen met dezelfde ori\"entatie.

Dit is ook wat er gebeurt met elektronen in een elektromagnetisch veld. De ori\"entatie in de buis wordt in dit geval de elektromagnetische fase van het elektron genoemd. Het elektromagnetische veld bepaald hoe de fase van het elektron verandert als het van punt naar punt beweegt. De absolute fase van een elektron kunnen we niet meten. De reden dat we toch het elektromagnetische veld kunnen meten is omdat we elektronen langs verschillende paden sturen, en we vervolgens kunnen opmerken dat ze relatief van fase verschillen.

Het blijkt dat de loopgroeven in de tunnels waardoor de mieren lopen niet stilstaan, maar kunnen draaien. Ze kunnen zich soms helemaal opdraaien, om vervolgens weer te ontdraaien, als een verticaal touw waaraan een object hangt dat heen en weer kan draaien. De tunnels draaien niet alleen, maar volgen hun buren. Als een tunnel ziet dat zijn buur-tunnel heel erg een kant op gedraait is zal hij zelf ook die kant op willen draaien. Hierdoor kan het gebeuren dat er een golf door het tunnelstelsel gaat van telkens opdraaiende en ontdraaiende tunnels. 
Zo'n golf is een klassieke beschrijving van licht.

In mijn proefschrift heb ik een nieuwe wiskundig nette manier gevonden om deze benadering van elektromagnetisme te beschrijven. Daarnaast heb ik gekeken naar de limiet waarin je het rooster (oftwel, het tunnelstelsel) steeds groter laat worden, maar ook steeds verfijnder. Hierdoor kom je dus steeds dichter bij een volledige beschrijving van elektromagnetisme.

Klassiek gezien kunnen we licht dus zien als een golf, maar het punt van deze beschrijving was om het makkelijker te maken om naar een kwantumbeschrijving toe te gaan. Dit blijkt ook heel goed te werken. Omdat de beweging van een stuk buis an sich een redelijk simpel kwantummechanisch analogon heeft, kunnen we elektromagnetisme op een rooster ook kwantummmechanisch beschrijven. In de limiet waarin we het rooster steeds groter en fijner maken moet nog wel wat gedaan worden.

 Op deze manier geeft dit proefschrift een nieuwe methode om een kwantumversie van elektromagnetisme te benaderen, die in bepaalde opzichten beter werkt dan voorgaande methoden. 
Of er een volledig waterdichte kwantumversie van elektromagnetisme bestaat blijft een open onderzoeksvraag, en ik vermoed dat deze vraag ons nog heel veel fascinerende wiskunde zal brengen.

\chapter*{Acknowledgements}

During the four years that led to this thesis I have been lucky to have had inmense support from the people mentioned below, both scientifically and personally. Most scientific support came from my promotor, Walter van Suijlekom, who believed enough in me to trust me with high-risk high-reward projects, which indeed led to high rewards. I am grateful for the ever positive support of Klaas Landsman, without whom my research would hardly be the same. Starting out as a PhD student, I was given a flying start by my collaboration with Ruben Stienstra, whom I thank for teaching me lattice gauge theory and having a keen sense for detail. I thank all my fellow PhD students for making the workspace into a welcoming and stimulating environment, and encouraging me to take frequent coffee breaks to talk about math, philosopy, language, and more. Thanks to Askar for giving me insight in the world of a physicist. In the early PhD days, it was a joy to reside at the same university I did my Bachelor's and Master's, mostly for hanging out with friends in the breaks. I apologize to all professors for celebrating too loudly after winning a round of klaverjas. I do not apologize to the losers.

Thanks goes to everyone that encouraged me to go bouldering, because it was the perfect sport to distract me from mathematics. I thank the group of bouldering mathematicians including Tommy, Leo, Mostafa, Ioan and Peter for making these sessions fun and for putting up with what I call flashing a pink route and Tommy refers to as toxic masculinity.

I am grateful to all the participants of the many conferences I attended, for everything I learned from them. I want to especially thank the organizers and participants of the workshop Noncommutative Calculus and the Spectral Action in Sydney, because it had a huge influence on Part \ref{part:I} of this thesis.

On the topic of hugely influencing Part \ref{part:I} of this thesis, I thank Anna Skripka for guiding me toward precisely the multiple operator integration techniques I needed to get an analytic handle on the spectral action. The mailing correspondence with her was the most productive one I ever had. It also helped to keep me indoors when the world needed me to stay indoors the most.

I thank all the members of the manuscript committee for reading this thesis and giving valuable feedback. I thank the members of the defense committee for playing a vital part in the defense.


I thank my friend Chris van de Ven for engaging conversation about strict deformation quantization, early research life, and everything else. 

I thank my high-school friends Koen, Simon, Marlyn, Rolf and Max for throwing the littest parties and reminding me not to click on dubious links.

Since starting my studies I have had a big and yet close group of friends in Nijmegen. To the outsiders: this friend group is most famous for the boy group and occasional boy band \textit{Gewoon} formed by five of its most handsome members. 
Big up to Jelte; he was the first one of this group that I met (and a fabulous meeting it was indeed) and we sure have a bromance. Another special thanks goes to Hannelore, my sister from another mister, who hosted so many dinner and/or klaverjas parties, and can dance like no other. I also want to thank Luuk, Bas, Nick, Dani\"elle, Reimer, Enya, Suzan, Paulien, Mirne, Tom, Britt, Vincent, Abel, Jeroen, and everyone else I am still forgetting, as the group is manifold and without boundary.

Over je familie praat je het best in je moedertaal, en dat doe ik ook hier.

Ik heb de beste ouders die je je maar kunt wensen. Nog belangrijker zelfs dan dat ze me geleerd hebben te klaverjassen hebben ze me geleerd om van mezelf te houden. Mam en Pap, jullie onuitputbare trots en liefde zullen me mijn leven blijven steunen. 

Zus! Dankjewel 
dat je altijd mijn goede voorbeeld bent. Het leven is zo veel leuker en makkelijker als je een persoon hebt om na te doen, vooral als die persoon zo fantastisch is als Annelies.

Marjolein, jouw bijdrages aan dit proefschrift zijn eigenlijk te veel om op te noemen. Niet in het minst omdat je elke regel van dit proefschrift hebt gelezen, behalve deze regels hier. 
Naast alle inhoudelijke bijdrages ben ik je dankbaar voor het feit dat je er altijd voor me bent geweest. Bedankt voor je geduld, je gevoel voor humor, je spitsvondigheid, je liefde, je tederheid, je eerlijkheid, en je betoverende lach. Ik kan niet zonder.


\end{document}